\documentclass[english,11pt]{article}

\usepackage[T1]{fontenc}
\usepackage{geometry}
\RequirePackage{amsthm,amsmath,amsfonts,amssymb}
\RequirePackage[authoryear]{natbib}
\usepackage{subcaption}
\usepackage{graphicx}
\usepackage{comment}

\usepackage[shortlabels]{enumitem}
\usepackage{array}
\usepackage{booktabs}
\usepackage[colorlinks,citecolor=blue,urlcolor=blue]{hyperref}
\usepackage{blkarray}
\usepackage{algorithm}
\usepackage{algorithmic}
\usepackage{float}
\usepackage{url}
\usepackage{longtable}
\usepackage{stmaryrd}
\usepackage{mwe}

\usepackage{multirow}
\DeclareMathAlphabet\mathbfcal{OMS}{cmsy}{b}{n}
\usepackage{comment}

\newcommand{\diag}{{\rm diag}}

\makeatletter
\newcommand*{\rom}[1]{\expandafter\@slowromancap\romannumeral #1@}
\makeatother

\usepackage{listings}
\usepackage[dvipsnames]{xcolor}

\setlist[itemize]{leftmargin=1em}
\setlist[enumerate]{leftmargin=1em}
\allowdisplaybreaks
\floatstyle{ruled}
\restylefloat{algorithm}

\newcommand{\argmin}{\mathop{\rm arg\min}}

\newcommand{\tr}{{\rm tr}}

\newcommand{\norm}[1]{\lVert#1\rVert}
\newcommand{\ran}{{\rm Ran}}

\newcommand{\h}{\mathsf{H}}

\newcommand{\splq}{\mathsf{L^2}(\mathbb{Q})}
\newcommand{\splp}{\mathsf{L^2}(\mathbb{P})}

\newcommand{\oplp}{\mathcal{L}_P}
\newcommand{\oplq}{\mathcal{L}_Q}
\newcommand{\oplpw}{\mathcal{L}_{P^w}}

\newcommand{\oplpnw}{\mathcal{L}_{P_n^w}}

\newcommand{\oplpiw}{\mathcal{L}_{P_i^w}}
\newcommand{\oplpjw}{\mathcal{L}_{P_j^w}}
\newcommand{\oplpwl}{\mathcal{L}_{P^w,\lambda}}
\newcommand{\oplpnwl}{\mathcal{L}_{P_n^w,\lambda}}

\newcommand{\oplpwo}{\mathcal{L}_{P^{w_0}}}

\newcommand{\opt}{\mathcal{T}}
\newcommand{\opid}{\mathcal{I}}
\newcommand{\opv}{\mathcal{V}}

\newcommand{\maid}{\mathbf{I}}
\newcommand{\mah}{\mathbf{h}}
\newcommand{\maw}{\mathbf{W}}
\newcommand{\mak}{\mathbf{K}}

\newcommand{\vkx}{\mathbf{k}_{\mathbf{x}}}
\newcommand{\vy}{\mathbf{y}}
\newcommand{\vvx}{\mathbf{v}_{\mathbf{x}}}
\newcommand{\va}{\mathbf{a}}
\newcommand{\vz}{\mathbf{z}}
\newcommand{\vone}{\mathbf{1}}

\newcommand{\pd}{\mathcal{N}_{\lambda}}

\newcommand{\pdn}{\mathcal{N}_{\lambda_n}}

\newcommand{\fpw}{f_{P^w}}

\newcommand{\hatf}{\hat{f}}

\def\beps{\boldsymbol{{\varepsilon}}}

\def\calA{{\cal  A}}

\def\calD{{\cal  D}} 
\def\calE{{\cal  E}}

\def\calL{{\cal  L}} 
\def\calM{{\cal  M}}

\def\calV{{\cal  V}}

\newcommand{\bfm}[1]{\ensuremath{\mathbf{#1}}}

\def\ba{\bfm a}

     \def\EE{\mathbb{E}}

\def\bk{\bfm k}   \def\bK{\bfm K}

     \def\PP{\mathbb{P}}
     \def\QQ{\mathbb{Q}}
     \def\RR{\mathbb{R}}

   \def\bW{\bfm W}  
\def\bx{\bfm x}     
\def\by{\bfm y}     
\def\bz{\bfm z}

\def\cs{\textsf{CS}}

\def\frakX{\mathfrak{X}}

\def\eps{\varepsilon}
\def\sfH{\mathsf{H}}
\def\sfL{\mathsf{L}}

\def\what{\widehat}
\def\wtilde{\widetilde}
\def\oracle{\textsf{ora}}

\theoremstyle{plain}
\newtheorem{lemma}{\textbf{Lemma}}
\newtheorem{proposition}{\textbf{Proposition}}
\newtheorem{theorem}{\textbf{Theorem}}
\newtheorem{corollary}{\textbf{Corollary}}
\newtheorem{assumption}{\textbf{Assumption}}

\newtheorem{definition}{\textbf{Definition}}
\newtheorem{remark}{\textbf{Remark}}

\numberwithin{equation}{section}
\numberwithin{theorem}{section}
\numberwithin{lemma}{section}
\numberwithin{proposition}{section}
\numberwithin{corollary}{section}
\numberwithin{definition}{section}
\numberwithin{example}{section}
\numberwithin{assumption}{section}

\begin{document}

\title{Nonparametric Goodness-of-fit Testing under Covariate Shift}
\author{Zhen Hou and  Dong Xia\footnote{The research was partially supported by  Hong Kong RGC Grant GRF 16303224.} \\
{\small Department of Mathematics}\\
{\small  Hong Kong University of Science and Technology}
}
\date{(\today)}

\maketitle

\begin{abstract}
This paper develops procedures for nonparametric goodness-of-fit testing under covariate shift, where labelled data are drawn from a \emph{source} population but goodness-of-fit is  evaluated for a \emph{target} population. The distribution mismatch is quantified by either a bounded moment condition or a sub-exponential tail condition on the target-to-source density ratio. Our method combines truncated importance-weighting kernel ridge regression with a multiplier bootstrap to construct confidence sets for the regression function. The truncation stabilizes the importance-weighting kernel ridge regression as well as the bootstrap calibration, making our approach applicable even when the density ratio has heavy tails.   We prove nonasymptotic validity and sharpness of the resulting confidence sets under suitable operator compatibility conditions, and establish explicit error rates for coverage probability under specific conditions on the target-to-source density ratio and on the spectral decay of the kernel integral operator. Numerical experiments corroborate our theoretical findings. 
\end{abstract}


\section{Introduction}
\label{sec:introduction}

\emph{Covariate shift} arises when labelled data are collected from a source population but scientific conclusions are sought for a target population with a different covariate distribution. This problem has become increasingly prominent as machine learning methods trained on labelled source data are deployed in specialized target domains \citep{Guan_2022}, where target-domain labels are often scarce or expensive \citep{wang2024comprehensive}. For instance, when general-purpose language models are deployed in specialized biomedical and legal domains, the prompts, documents, and questions can differ drastically from the broad corpora used for pretraining or instruction tuning \citep{gururangan2020dont}. These applications raise important statistical questions, for example, \emph{which statistical methods perform well under target distributions and are robust to target-to-source mismatch?} and \emph{is it possible to evaluate the performance of these methods under the target distribution even when they are trained under the source distribution?} While the first question has been extensively studied in the literature, the second is far less understood. This paper provides an affirmative answer through a general inference framework based on truncated importance-weighted kernel ridge regression (KRR) and the multiplier bootstrap \citep{bickel1981some,efron1994introduction,hall2013bootstrap}.

Let $\PP$ and $\QQ$ denote the \emph{source} and \emph{target} distributions, respectively, on the joint space $\frakX \times \RR$. For simplicity, we assume that both $\PP_X$ and $\QQ_X$ admit \emph{known} densities with respect to a common dominating measure, and that $\QQ_X$ is absolutely continuous with respect to $\PP_X$. Our methods and main theorems hold as long as the density ratio $\rho(x)=(d\QQ_X/d\PP_X)(x)$ is known. Moreover, $\rho(x)$ can be estimated from unlabelled source and target covariates by direct importance-estimation or distribution-matching methods, such as those in \cite{sugiyama2008direct,huang2007correcting}. See further discussion in Section~\ref{sec:discuss}. We assume that the marginal distributions of the covariates may differ, $\PP_X\neq \QQ_X$, whereas the conditional distributions of the response remain the same, $\PP_{Y|X}=\QQ_{Y|X}$.  With a slight abuse of notation, we shall write $\PP$ and $\QQ$ for the covariate marginal distributions $\PP_X$ and $\QQ_X$, respectively, whenever no ambiguity arises.
Suppose that an i.i.d. sample $\mathcal{D}=\{(X_i, Y_i)\}_{i=1}^n\subset \frakX\times \RR$ drawn from the source distribution $\PP$ is available, where the response and covariates satisfy the regression model $f_*(x) = \mathbb{E}[Y|X=x]$, where $f_{\ast}(\cdot): \frakX\mapsto \RR$ is an \emph{unknown} regression function.  Our goal is to evaluate $f_{\ast}$ under the target covariate distribution. More specifically, given a candidate function $f_0$, we aim to test
\begin{equation}
  \label{eq:testing-problem}
    H_0: \norm{f_0 - f_{\ast}}_{\splq} = 0 \quad { v.s.} \quad H_1: \norm{f_0 - f_{\ast}}_{\splq} > 0,
\end{equation}
where the $\splq$-norm is defined by $\|g\|_{\splq}^2:=\int g(x)^2 d\QQ(x)$. Equivalently, the null hypothesis states that $\QQ_X\big(\{ x: f_{\ast}(x)=f_0(x)\}\big)=1$.  For ease of exposition, we focus on the simple test, but our framework can be easily generalized to composite tests and two-sample tests. See  numerical experiments in Section~\ref{subsec:scf-monotonicity} and  the Appendix.

The seminal work \citet{shimodaira2000improving} proposed importance weighting (IW) for estimating the target function $f_{\ast}(x)$ under covariate shift in parametric likelihood models, showing that the asymptotically optimal weighting function is the target-to-source density ratio $\rho(x)$. 
Non-asymptotic analyses of importance weighting have also attracted considerable attention across diverse fields and applications, including the generalization bounds for statistical learning with function classes of finite pseudo-dimension \citep{cortes2010learning}, local convergence rates for learning Lipschitz functions under local doubling conditions \citep{schmidt2024local}, minimax-optimal rates for learning H\"older-continuous functions under covariate shift \citep{pathak2022new}, and kernel-based regression and classification \citep{ma2023optimally,wang2026pseudo}. 

Importance weighting is generally necessary for unbiased estimation of the target risk under covariate shift, but it may be unnecessary for identifying the Bayes-optimal predictor when the conditional distribution is invariant and the model is correctly specified. See, for example, \citet{gogolashvili2023importance} and the references therein. However, importance weighting can degrade finite-sample performance when the density ratio is unbounded or highly variable.  A common remedy is to truncate the density ratio, replacing $\rho(x)$ by $w_n(x):=\min\{\rho(x), \tau_n\}$ where \(\tau_n>0\) is a pre-specified truncation level.  \cite{ma2023optimally} showed that the importance weighted kernel ridge regression (IWKRR) estimator achieves the minimax optimal convergence rate under covariate shift using a truncated importance weight for variance reduction. 
Figure~\ref{fig:intro-iwkrr-truncation} illustrates the stabilizing effect of truncation when the density ratio is large in target-relevant regions. Indeed, a small number of source observations may receive extremely large importance weights, causing the IWKRR estimator to be highly \emph{variable}. Capping the weights can therefore substantially stabilize the $\sfL^2(\QQ)$ error.

\begin{figure}
\centering
\begin{subfigure}[t]{0.5\textwidth}
    \centering
    \includegraphics[width=\textwidth]{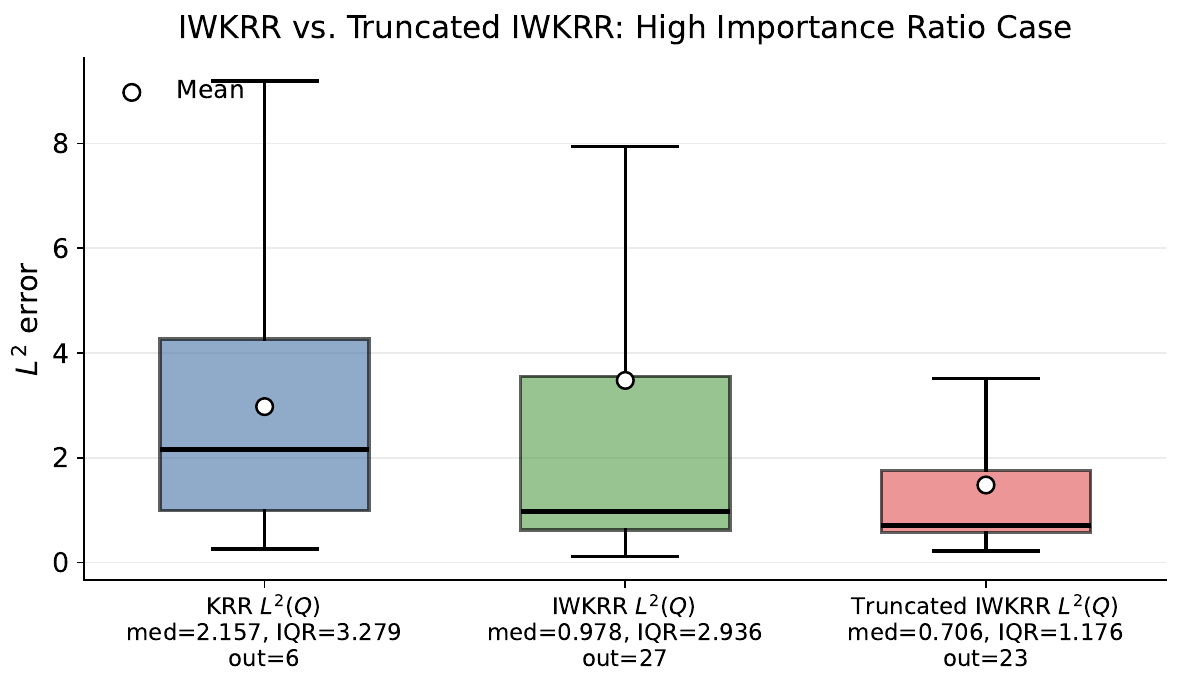}
    \caption{$\splq$-norm error rates: KRR, IWKRR, and IWKRR with truncation.}
\end{subfigure}
    \hfill
    \begin{subfigure}[t]{0.45\textwidth}
    	\centering
    	\includegraphics[width=\textwidth]{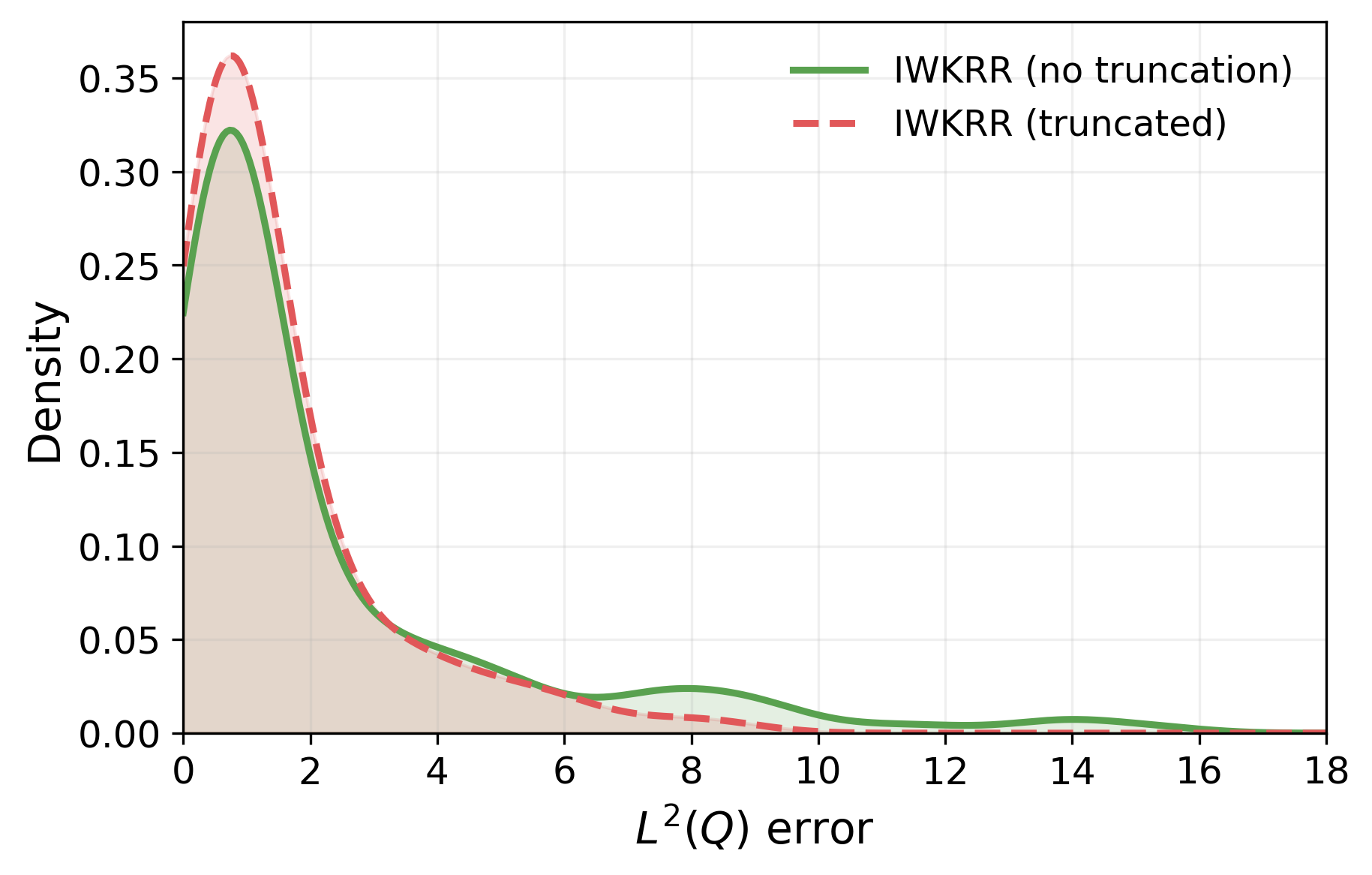}
    	\caption{IWKRR $\splq$-norm error distributions, without and with truncation.}
    \end{subfigure}
\caption{truncating large density-ratio weights $\rho(x):=(d\QQ/d\PP) (x)$ stabilizes the finite-sample performance of the importance-weighting kernel ridge regression (IWKRR). Source and target distributions are {\it truncated normal} ${\rm TN}(-1,0.25;[-2,2])$  and ${\rm TN}(1,0.25;[-2,2])$, respectively.  Panel (a) summarizes the validation-tuned errors of KRR, IWKRR, and truncated IWKRR. The true regression function is $f_{\ast}(x)=|x|$ . Panel (b) shows the Monte Carlo densities of the $\splq$ errors over 200 repetitions, using raw importance weights which differ only in truncation.}
\label{fig:intro-iwkrr-truncation}
\end{figure}

The aforementioned works primarily focus on efficient and robust estimation under covariate shift. While they provide practical guidance and theoretical justification for statistical methods that improve predictive performance under target distributions, it remains largely unclear how to evaluate the performance of these methods under the target distribution and how to quantify the uncertainty of the resulting estimators. 
For instance, in studies of conditional average treatment effects, one may wish not only to estimate the treatment-effect function but also to test whether treatment effects are constant over the target population \citep{crump2008nonparametric,ding2016randomization}. These inferential conclusions depend critically on the target and source distribution mismatch: heterogeneity that is negligible under the source distribution may be substantial under the target distribution if the regions in which the effect varies have substantially different probability densities under the two distributions.

In this paper, we develop statistical procedures to quantify the uncertainty of the truncated IWKRR estimator and to evaluate its performance under the target covariate distribution. More specifically, we construct confidence sets for $f_\ast$ with respect to the $\sfL^2(\QQ)$-norm distance based on the truncated IWKRR estimator, establish non-asymptotic bounds on the absolute coverage error of these confidence sets, and study the goodness-of-fit testing problem in \eqref{eq:testing-problem}. 
The main challenge is that the IWKRR estimator, the inferential loss, and the key technical tools naturally live in different geometries. The IWKRR estimator is governed by the truncated weighted-source risk, whereas coverage of the confidence sets is evaluated in the target $\sfL^2(\QQ)$ metric. On the technical side, the covariance operator of the stochastic fluctuation must be analyzed in the target $\sfL^2(\QQ)$ geometry, whereas its algebraic structure is most transparent through RKHS operators. 
We address this mismatch using unitary isomorphisms generated by the square roots of the target and weighted-source integral operators. These maps realize the relevant $\sfL^2$ geometries inside the RKHS while preserving Hilbert-space norms and inner products, thereby allowing the Gaussian and bootstrap approximations to be compared within a common RKHS operator calculus. This operator-theoretic treatment captures two effects of covariate shift: the bias of the truncated IWKRR estimator and its covariance structure relevant for coverage guarantees of the proposed confidence sets.

Our contributions are as follows. First, based on the truncated IWKRR estimator, trained on finite-sample data from source distribution, and multiplier bootstrap, we propose an inference procedure to construct confidence sets of the unknown regression function in the target distribution. The procedure is computationally efficient where computing each bootstrap statistics mainly involves matrix and vector multiplications. Second, we provide a general theoretical framework establishing the coverage guarantee of the bootstrap confidence sets,  where the error derives from the bias,  Gaussian approximation, and bootstrap approximation.     
The target-to-source distribution mismatch,  kernel integrals, and truncation level jointly determine the validity of multiplier bootstrap method.  We identify the {\it regularized transfer coverage} condition ensuring that the source risk is transferable to the target risk and the \emph{non-degenerate covariance} condition that prevents the Gaussian covariance from degenerating in the target geometry.  Third,  we establish explicit error rates for the coverage probability under distinct density ratio and kernel spectral conditions,  covering the cases of sub-exponential type density ratio and polynomially decaying kernel eigenvalues.  Finally,   our inference methods yield interesting findings in real data experiments.

\subsection{Related work}
\label{sec:related-work}
{\it Covariate shift.} The seminal work \cite{shimodaira2000improving} initiated this field by analyzing weighted likelihood methods for parametric predictive inference under covariate shift and identifying importance weighting as the asymptotically optimal weighting scheme. Finite-sample generalization bounds for importance weighting were developed by \citet{cortes2008sample}, covering kernel-based algorithms and function classes with finite pseudo-dimension. 
The unknown density ratio, as well as its variants, can be estimated using unlabelled samples from source and target populations such as the Kullback--Leibler importance estimation \citep{sugiyama2008direct} and kernel mean matching \citep{huang2007correcting}. 
The $\gamma$-transfer exponent was introduced by \cite{kpotufe2018marginal} as a measure of source and target similarity. 
A systematic review of machine learning methods under covariate shift is provided by \cite{book_covariate_shift}.
More recently, \cite{zamolodtchikov2026minimax} established the minimax rates for nonparametric regression with both source and target samples, characterized through the notion of a transfer function. 
Model selection and cross validation under covariate shift have also gained considerable attention. 
An information criterion was derived in \cite{shimodaira2000improving} for model selection under covariate shift.  \cite{sugiyama2007covariate} showed that the ordinary cross-validation can be biased under covariate shift and proposed importance-weighted cross-validation. 
The instability of importance-weighted cross-validation and its robust variants were studied by  \cite{kouw2019robust}.

{\it Nonparametric inference.} 
\cite{chernozhukov2014gaussian,chernozhukov2017central,Chernozhukov_2014} established nonasymptotic Gaussian and bootstrap approximations for suprema of empirical processes, with applications to confidence bands. 
\cite{gaussianapproximationbounded} and \cite{gaussianapproximationsubgaussian} proved Gaussian approximation results for sums of independent random vectors under bounded and sub-Gaussian conditions, respectively.
\cite{yang2017frequentistcoveragesupnormconvergence} studied posterior coverage and sup-norm credible sets in Gaussian process regression.
\cite{singh2025kernelridgeregressioninference} developed inference procedures for KRR without covariate shift.  Similarly,  \cite{shang2025bootstrap} studied data integration based on KRR and constructed adaptive confidence intervals for $f_{\ast}(x)$ via local and global Bootstrap procedures.  Constructing uniformly valid pointwise confidence intervals is generally more challenging than testing a global null such as $H_0:\|f_{\ast}-f_0\|_{\splp}=0$. The former requires uniform control of the estimation error, typically in sup norm, and is subject to well-known impossibility and lower-bound phenomena in nonparametric models. See, e.g., \cite{low1997nonparametric, genovese2008adaptive,gine2010confidence, cai2006adaptive,juditsky2003nonparametric,li1989honest,robins2006adaptive,spokoiny1996adaptive} on the methods and challenges of constructing honest and adaptive non-parametric confidence sets.

{\it Predictive intervals.}
A related but distinct line of work studies predictive intervals and sets under covariate shift.
\cite{tibshirani2019conformal} developed weighted conformal prediction when the likelihood ratio between the target and source covariate distributions is known or can be accurately estimated, and \cite{barber2023conformal} further studied conformal prediction beyond exchangeability.
These methods provide prediction sets for a future response and are largely distribution-free.
By contrast, our goal is function-level inference: we construct a confidence ball for the regression function \(f_*\), with a radius calibrated in \(\splq\)-norm, and therefore must control the bias, covariance structure, and bootstrap approximation of the  IWKRR estimator with truncated weights.

\subsection{Notations}
\label{sec:notation}
Throughout, \(\mathsf{H}\) denotes the RKHS associated with a kernel function \(k(\cdot, \cdot)\), \(\splq\) denotes the \(L^2\) space with respect to the target covariate distribution $\mathbb{Q}_X$, and we write \(k_x=k(x,\cdot)\). For a Hilbert space \(\mathsf{G}\), let  \(\langle\cdot,\cdot\rangle_{\mathsf{G}}\) and \(\|\cdot\|_{\mathsf{G}}\) denote its inner product and norm, respectively; for \(a,b\in\mathsf{G}\), \(a\otimes_{\mathsf{G}}b\) denotes the rank-one operator \((a\otimes_{\mathsf{G}}b)c=\langle b,c\rangle_{\mathsf{G}}\,a\) for \(c\in\mathsf{G}\). For a bounded linear operator \(\mathcal{A}:\mathsf{G}\to\mathsf{S}\), \(\|\mathcal{A}\|_{\mathsf{G}\to\mathsf{S}}\) denotes its operator norm and \(\|\mathcal{A}\|_{\mathrm{HS},\mathsf{G}\to\mathsf{S}}\) its Hilbert-Schmidt norm, when finite; if \(\mathcal{A}:\mathsf{G}\to\mathsf{G}\) is trace class, \(\tr_{\mathsf{G}}(\mathcal{A})\) denotes its trace; for self-adjoint operators \(\mathcal{A},\mathcal{B}\) on \(\mathsf{G}\), we write \(\mathcal{A}\preceq_{\mathsf{G}}\mathcal{B}\) if \(\mathcal{B}-\mathcal{A}\) is positive semidefinite. When the subscript is omitted, inner products, norms, tensor products, operator norms, Hilbert-Schmidt norms, traces, and Loewner order are understood in \(\mathsf{H}\). For a compact positive self-adjoint operator \(\mathcal{A}\) on \(\mathsf{G}\), \((\mu_k(\mathcal{A},\mathsf{G}),\phi_k(\mathcal{A},\mathsf{G}))\) denotes its \(k\)-th eigenpair, with eigenvalues arranged in nonincreasing order and counted with multiplicity; we use the shorthand \((\mu_k,\phi_k):=(\mu_k(\oplq,\splq),\phi_k(\oplq,\splq))\), define the spectral tail \(\sigma_{\mathsf{G}}^2(m,\mathcal{A}):=\sum_{j=m+1}^{\infty}\mu_j(\mathcal{A},\mathsf{G})\) for \(m\ge0\), and also write \(\sigma^2(m,\mathcal{A},\mathsf{G})\) to specify the ambient space. For a positive trace-class operator \(\mathcal{C}\) on \(\mathsf{G}\), \(\mathrm{N}_{\mathsf{G}}(0,\mathcal{C})\) denotes the centered Gaussian element in \(\mathsf{G}\) with covariance operator \(\mathcal{C}\). Finally, we write \(a\lesssim b\) if \(a\le Cb\) for a universal constant \(C>0\), and \(a\asymp b\) if both \(a\lesssim b\) and \(b\lesssim a\) hold; the empirical average is denoted by \(\mathbb{E}_n[g]=\EE_{\PP_n}[g]=n^{-1}\sum_{i=1}^n g_i\).

\section{Problem Formulation and Preliminaries}
\label{sec:problem-setup}

\subsection{Confidence balls in target distribution}
\label{sec:statistical-inference-under-covariate-shift}

A pair of random variables $(X, Y)\in \frakX\times \RR$ satisfies 
\begin{equation}\label{eq:fmodel}
    Y = f_{\ast}(X) + \eps,
\end{equation}
where $X$ and $Y$ denote the covariate and response, respectively, and the regression function $f_{\ast}$ is unknown.  The random error $\eps$ carries no systematic information in the sense that $\EE[\eps|X]=0$. 

Let $(X, Y), \{(X_i, Y_i)\}_{i=1}^n\in\frakX\times \RR$ be i.i.d. observations, where the covariate $X$ is drawn from a source distribution $\PP$ on $\frakX$. We aim to construct a confidence set for $f_{\ast}$ in the $\splq$ geometry. Given an $\splq$-norm consistent estimator $\what f$, a natural confidence set is the $\splq$-norm ball centered at $\what f$, with a radius $\hat c_{\alpha,n}$ calibrated  using the quantiles of $\|\what f-f_{\ast}\|_{\splq}$. More precisely, the inferential goal is to construct the $100(1-\alpha)\%$ confidence set
\begin{equation}
\cs_{\alpha}(\what f)
    =
    \Big\{f :
    \|\what f-f\|_{\splq}\le \hat c_{\alpha,n}
    \Big\},
\end{equation}
with asymptotically valid coverage probability; that is, $\PP\big(f_{\ast}\in\cs_{\alpha}(\what f)\big)\to 1-\alpha$ as $n\to\infty$. We suppress the dependence of $\what f$ and $\cs_{\alpha}(\what f)$ on the sample size $n$ for simplicity.  
Once a valid confidence set is available, testing the hypothesis in \eqref{eq:testing-problem} is straightforward. For example, one can reject the null hypothesis at significance level $\alpha$ if $f_0\notin \cs_{\alpha}(\what f)$.

We emphasize that the choice of metric is part of the inferential goal, rather than merely a technical convention. A confidence set contains precisely those candidate functions that are close to the point estimator in the chosen geometry. Consequently, different metrics may lead to different confidence sets. 
There are three natural metrics in KRR: the reproducing kernel Hilbert space (RKHS) norm, the \(\mathsf{L}^2(\PP)\)-norm induced by the source distribution, and the \(\splq\)-norm induced by the target distribution. The RKHS norm measures discrepancy in the smoothness geometry imposed by the kernel, whereas the \(\mathsf{L}^2(\PP)\)- and \(\splq\)-norms weight pointwise errors according to the source and target distributions, respectively. Under covariate shift, regions that are important under \(\QQ\) may have little mass under \(\PP\). Hence, a function can be close to \(f_\ast\) in \(\mathsf{L}^2(\PP)\) but not in \(\splq\).

\subsection{Reproducing kernel Hilbert space and integral operators}

Let $\mathsf{H}$ be a reproducing kernel Hilbert space (RKHS,  \cite{scholkopf2002learning}) over $\frakX$ with positive definite kernel $k(\cdot, \cdot): \frakX\times \frakX\mapsto \RR$.  Throughout,  we assume that $\mathfrak{X}$ is compact.  
 More precisely,  $\mathsf{H}$ is the Hilbert space of functions on $\mathfrak{X}$ obtained as the completion of the linear span of  $\big\{k_x:=k(x, \cdot)\}_{x \in \mathfrak{X}}$.   For any $\{x_i\}_{i=1}^p,  \{\tilde x_j\}_{j=1}^q\subset \frakX$ ,  consider functions of the form $f(\cdot)=\sum_{i=1}^p a_i k(x_i, \cdot)$ and $g(\cdot)=\sum_{j=1}^q b_j k(\tilde x_j, \cdot)$.  Their inner product in $\sfH$ is given by $\langle f,  g\rangle_{\sfH}=\sum_{i=1}^p\sum_{j=1}^q a_ib_jk(x_i, \tilde x_j)$.  
 
 For any $f \in \mathsf{H}$, the reproducing property holds:
$ f(x) = \langle f, k(x, \cdot) \rangle_{\mathsf{H}}$,  for $\forall x\in\frakX$.   Throughout the paper,  we assume that the kernel $k(\cdot, \cdot)$ is \emph{bounded},  namely,  $\sup_{x \in \mathfrak{X}} k(x, x) \leq \kappa^2$  for some constant  $\kappa > 1$.  This  immediately implies that point evaluation is bounded by the RKHS norm.   Indeed,  for any $f\in\sfH$,  we have $\sup_x|f(x)|=\sup_x \big|\langle f,  k(x,\cdot)\rangle_{\sfH} \big|\leq \|f\|_{\sfH} \cdot \sup_x\|k(x,\cdot)\|_{\sfH}\leq \kappa \|f\|_{\sfH}$.

The \emph{kernel integral operator} is central to the analysis of kernel methods, as its spectral structure encodes both the geometry of the data distribution and the effective complexity of the induced feature space.   
Given the probability measure $\QQ$,  define the integral operator $\calL_Q: \sfL^2(\QQ)\to \sfL^2(\QQ)$ by
$$
\oplq f = \int_{\mathfrak{X}} k(x, \cdot) f(x) d\QQ(x),\quad \forall f\in\sfL^2(\QQ). 
$$
Note that the integral operator $\calL_Q$ is closely related to the RKHS \emph{covariance operator} $\Sigma_{\QQ}=\EE_{\QQ}[k_X\otimes_{\sfH} k_X]: \sfH \to \sfH$,
where, for $g, h\in\sfH$,   the rank-one operator $g\otimes_{\sfH}h: \sfH\to\sfH$ is defined by $(g\otimes_{\sfH}h)f=\langle h, f\rangle_{\sfH}g$, $\forall f\in\sfH$. 

By the reproducing property, we have $\Sigma_{\QQ}f=\int_{\frakX} k(x,\cdot)f(x)d\QQ(x)$, $\forall f\in\sfH$. Under the bounded-kernel assumption, both $\calL_Q$ and $\Sigma_{\QQ}$ are bounded, positive, self-adjoint, and trace-class operators.  Since these two operators are given by the same integral expression but act on different spaces, we abuse notation and write $\oplq$ for both whenever the underlying domain is clear from context.

The spectral decomposition of $\oplq$ is often used to characterize the kernel $k(\cdot, \cdot)$, the distribution $\QQ$, and the associated RKHS $\sfH$.  

\begin{lemma}\label{lem:eigendecomposition_LQ}[\cite{mercer1909functions}]
Suppose that $\frakX$ is compact, $k(\cdot,\cdot)$ is continuous, symmetric, and positive definite, $\QQ$ is a Borel probability measure on $\frakX$, and $\operatorname{supp}(\QQ)=\frakX$.
Then there exist nonnegative eigenvalues $\{\mu_j\}_{j\ge 1}$ and corresponding eigenfunctions $\{\phi_j\}_{j\ge 1}\subset \sfL^2(\QQ)$, with
 $\langle \phi_i, \phi_j \rangle_{\splq} = \delta_{ij}$ and $k(x,y) = \sum_{j=1}^{\infty} \mu_j \phi_j(x) \phi_j(y)$, $x, y\in\frakX$.  The convergence is absolute and uniform on $\frakX\times\frakX$.
\end{lemma}
Mercer's theorem immediately implies several useful identities. For $\mu_i, \mu_j>0$, $\langle \phi_i, \phi_j \rangle_{\mathsf{H}} = \mu_i^{-1} \delta_{ij}$. Thus, $\big\{\mu_j^{1/2} \phi_j: \mu_j>0\big\}_{j=1}^{\infty}$ is an orthonormal basis of $\mathsf{H}$. In particular, if $f=\sum_{j=1}^{\infty}a_j\phi_j$, then $\|f\|_{\sfH}^2=\sum_{j:\mu_j>0}a_j^2\mu_j^{-1}$. 

The fractional powers of $\calL_Q$ are defined via the spectral theorem. For $r>0$, define the operator $\calL^r_Q: \sfL^2(\QQ)\to \sfL^2(\QQ)$ by 
$$
\calL^r_Q f= \sum_{j\geq 1}\mu_j^r \langle f, \phi_j\rangle_{\sfL^2(\QQ)}\phi_j,\quad f\in \sfL^2(\QQ). 
$$
Although $\calL_Q^{1/2}$ is defined as an operator on $\sfL^2(\QQ)$, its range admits a natural RKHS interpretation. Define $\ran(\oplq)=\big\{\sum_{j: \mu_j>0} a_j\phi_j: \sum_{j:\mu_j>0}a_j^2\mu_j^{-2}<\infty\big\}$.  
Indeed, if $f\in\overline{\ran(\oplq)}$,  then $\big\|\calL_Q^{1/2}f\big\|_{\sfH}^2=\sum_{j:\mu_j>0}\langle f, \phi_j\rangle_{\sfL^2(\QQ)}^2=\|f\|^2_{\sfL^2(\QQ)}$. Hence $\calL_Q^{1/2}$ can be viewed as an isometry from $\overline{\ran(\oplq)}$ to $\sfH$.

\begin{lemma}\label{lem:LQ_half_unitary_isomorphism}
 The operator $\oplq^{1/2}:\overline{\ran(\oplq)} \to \mathsf{H}$ is a
 unitary isomorphism onto $\mathsf H$. In particular, it preserves inner products:
$
    \langle f, g \rangle_{\splq}
    = \langle \oplq^{1/2} f, \oplq^{1/2} g \rangle_{\mathsf{H}},
$
for $f,  g\in \overline{\ran(\oplq)}$.  
\end{lemma}


We can similarly define the integral operator and covariance operator by the probability measure $\PP$ on $\frakX$,  denoted both by $\calL_P$ for notational simplicity.  These notations can be generalized to the sampling distribution $\PP_n$,  defined from the i.i.d.  sample $X_1,\cdots, X_n\sim \PP$.  For instance,  we write the integral operator 
$$
\calL_{P_n}f=\frac{1}{n}\sum_{i=1}^nf(X_i)k_{X_i},  \quad f\in \sfL^2(\PP),
$$
which may also represent the \emph{sample covariance operator} $f\mapsto n^{-1}\sum_{i=1}^n(k_{X_i}\otimes_{\sfH} k_{X_i})f,  \forall f\in\sfH$.  

Let $w:\frakX\to(0,\infty)$ be a weight function. Define the weighted measure $\PP^w$ on $\frakX$ by
$\PP^w(\calA)=\int_{\calA} w(x)\,d\PP(x), \calA\subseteq \frakX .$
For instance, if \(w\) is the target-to-source density ratio
$\rho(x)=(d\QQ/d\PP)(x)$, then $\PP^w=\QQ$, assuming that $\QQ\ll \PP$. The measure $\PP^w$ and its empirical counterpart induce the following integral operators:
$$
    \oplpw f
    :=
    \int_{\frakX} k_x f(x) w(x)\,d\PP(x),
    \quad {\rm and}\quad 
    \oplpnw f
    :=
    \frac{1}{n}\sum_{i=1}^n w(X_i) k_{X_i} f(X_i).
$$

\subsection{Importance weighted kernel ridge regression}
\label{sec:reproducing-kernel-hilbert-space}

Given the  i.i.d. sample $\mathcal{D}=\{(X_i, Y_i)\}_{i=1}^n\subset \frakX\times \RR$,  the kernel ridge regression (KRR) seeks a function $f\in\sfH$ that minimizes the empirical risk
$
n^{-1}\sum_{i=1}^n \big(f(X_i)-Y_i\big)^2
$
with an RKHS-norm penalization $\lambda_n \|f\|_{\sfH}^2$,  
where $\lambda_n>0$ is the regularization parameter controlling the trade-off between fitting the data and model complexity.   The regularization prevents KRR from overfitting,  especially when $\sfH$ is  dense in the space of continuous functions.  The error rates and optimality of KRR are well understood. See \cite{wahba1990spline, caponnetto2007optimal, vapnik2013nature} and references therein.

Importance weighting (IW,  \cite{sugiyama2007covariate}) is a technique for unbiased estimation of target risk under covariate shift.  In particular,  the importance-weighted kernel ridge regression (IWKRR) estimator is defined by 
\begin{equation}
  \label{eq:iwkrr-definition-setup}
  \what f
  := \argmin_{f\in\mathsf{H}}\ 
    \frac{1}{n}\sum_{i=1}^n w_n(X_i)\{Y_i-f(X_i)\}^2
    + \lambda_n \|f\|_{\mathsf{H}}^2,
\end{equation}
where  $w_n:\frakX\to(0,\infty)$ is a weight function.  If equipped with a proper weight function,  IW transports the source risk to the target risk.  For instance,  using the density ratio weight $w_n(x)=\rho(x)$,  we get $\EE_{\PP}[w_n(X)(Y-f(X))^2]=\EE_{\QQ}[(Y-f(X))^2]$.  This density-ratio weighting thus corrects the population mismatch at the level of the squared-error risk. For notational simplicity, we may write $\lambda$ and $w(\cdot)$ and suppress their dependence on $n$.  

However,  the IWKRR estimator can  have large finite-sample variances if the target-to-source density ratio is unbounded or has heavy tails,  in which case a few observations carry excessively large weights and destabilize the estimator.  
A common remedy is to use the truncated density ratio $w_n(x):=\min\{\rho(x), \tau_n\}$,  where $\tau_n>0$ is the truncation level.  The truncation introduces a controlled bias but reduces variance.  Under suitable density ratio and spectral conditions,  with appropriate choices of regularization parameter $\lambda_n$ and truncation parameter $\tau_n$,  the truncated IWKRR attains the minimax optimal rate for the target $\splq$-norm risk.   See \cite{ma2023optimally} and Section~\ref{sec:theory} for detailed discussions.

In this paper,  we aim to construct the confidence ball of $f_{\ast}$ with respect to the $\splq$-norm distance based on the truncated IWKRR estimator (\ref{eq:iwkrr-definition-setup}).  We begin with presenting the explicit matrix form of the IWKRR estimator.   Denote the $n\times n$ matrix $\bK$  with its $(i,j)$-th entry $k(X_i,  X_j)$,  the weight matrix $\bW={\rm diag}\big(w_n(X_1),\cdots, w_n(X_n)\big)$,  and $\by=(Y_1,\cdots,  Y_n)^{\top}$.  Then,  the IWKRR estimator can be compactly written as 
\begin{align}
  \label{eq:matrix-vector-form-of-iwkrr-estimator}
    \what f &= \vkx^{\top}(\maw\mak+n\lambda_n\maid)^{-1}\maw\vy=\vkx^{\top}(\mak+n\lambda_n\maw^{-1})^{-1}\vy,
\end{align}
where the vector function $\bk_{\bx}=\big(k(X_1,\cdot),\cdots,  k(X_n, \cdot)\big)^{\top}$ and the second equation holds provided that $\bW$ is nonsingular.

The IWKRR estimator is biased from two sources: the RKHS-norm regularization and the truncated density-ratio weight.  The population IWKRR estimator is 
\begin{align}\label{eq:psupw-def}
\fpw=\argmin_{f \in \mathsf{H}}\  \mathbb{E}_{\mathbb{P}} \left[ w_n(X) (f(X) - Y)^2 \right] + \lambda_n \|f\|_{\mathsf{H}}^2 =(\oplpw + \lambda_n \opid)^{-1} \oplpw f_*.
\end{align}
If the weight $w_n(x)=(d\QQ/d\PP)(x)$ without truncation,  the population IWKRR estimator reduces to 
\begin{align}\label{eq:psuq-def}
f_Q= \argmin_{f \in \mathsf{H}}   \mathbb{E}_{\mathbb{Q}} \left[ (f(X) - Y)^2 \right] + \lambda_n \|f\|_{\mathsf{H}}^2 =(\oplq + \lambda_n \opid)^{-1} \oplq f_*,
\end{align}
which minimizes the regularized risk in target distribution.  The explicit formulas in (\ref{eq:psupw-def}) and (\ref{eq:psuq-def}) follow from the first-order optimality conditions for the two population objectives.  
Finally,  we remark that the IWKRR estimator (\ref{eq:matrix-vector-form-of-iwkrr-estimator}) can be explicitly written as $\what f=(\oplpnw + \lambda_n \opid)^{-1} \left(n^{-1}\vkx^{\top} \maw\vy\right)$.

\section{Inference Procedure}
\label{sec:inference-procedure}

\subsection{Bias-variance decomposition and Bahadur representation}\label{sec:bahadue}
We begin with the bias-variance decomposition of $\what f$ and derive its Bahadur representation. 
Notice that $\sqrt{n}(\what f - f_{\ast})=\sqrt{n}(f_{P^w}-f_{\ast})+\sqrt{n}(\what f- f_{P^w})$, where the first and second terms capture the bias and variance, respectively.  
Moreover,  the bias term can be explicitly written as
$$
   \sqrt{n}(\fpw-f_*)= 
-\sqrt{n}\lambda_n(\oplpw+\lambda_n\opid)^{-1}f_* ,
$$
implying that the bias term can be estimated by $-\sqrt{n}\lambda_n(\oplpnw+\lambda_n\opid)^{-1}\what f$. However, we will show in Section~\ref{sec:theory} that the bias term is asymptotically negligible if $\lambda_n$ is small, i.e., under-smoothing IWKRR. Our bootstrap inference procedure therefore focuses on the variance term $\sqrt{n}(\what f-f_{P^w})$. Nevertheless, we also provide a bias-corrected inference procedure and numerical experiments in Appendix~\ref{app:ext-bCboot}.

The second term, $\sqrt{n}(\what f-\fpw)$, captures the stochastic fluctuation of the IWKRR estimator around its population counterpart.  We show in Section~\ref{sec:theory} that its leading term, in $\sfL^2(\QQ)$-norm, admits the Bahadur representation:
\begin{equation}\label{eq:bahadur}
  \sqrt{n}(\what f - \fpw)
  \approx_{\splq}
  \frac{1}{\sqrt{n}}\sum_{i=1}^n U_i,
\end{equation}
where
$
U_i
:=
(\oplpw+\lambda_n\opid)^{-1}
\big(
w_n(X_i)\eps_i k_{X_i}
+
\big(w_n(X_i)(k_{X_i}\otimes_{\h} k_{X_i})-\oplpw\big)(f_*-\fpw)
\big).
$
It suggests that $\sqrt{n}(\what f - \fpw)$ can be approximated by a centered Gaussian process in $\sfL^2(\QQ)$ with the covariance operator $\opv=\mathbb{E}[U_i\otimes_{\splq}U_i]$, and Gaussian multiplier bootstrap can be used to calibrate the radius of $\sfL^2(\QQ)$-norm confidence balls for $f_{\ast}$.

\subsection{Multiplier bootstrap}
\label{subsec:bootstrap-algorithm}
The random element $U_i$  is not directly applicable because it relies on the unknown $f_*$ and the population IWKRR estimator $f_{P^w}$.   
Nevertheless,  the Bahadur representation (\ref{eq:bahadur}) motivates us to consider the \emph{oracle} bootstrap function, defined by the anti-symmetric multiplier process: 
\begin{equation*}
    B_{\oracle}
    :=
    \frac{1}{n}\sum_{1\leq i\neq j\leq n} h_{ij}\cdot
    \frac{U_i-U_j}{\sqrt{2}}
    =
    \frac{1}{n}\sum_{1\leq i\neq j\leq n} h_{ij}\cdot 
    \frac{V_i-V_j}{\sqrt{2}},
\end{equation*}
where $h_{ij}\stackrel{{\rm i.i.d.}}{\sim} N(0, 1)$ are independent of the data $\calD$ and
\begin{equation*}
V_i:=U_i+(\oplpw+\lambda_n\opid)^{-1}\oplpw(f_*-\fpw)=
w_n(X_i)\big(Y_i-\fpw(X_i)\big)(\oplpw+\lambda_n\opid)^{-1}k_{X_i}.
\end{equation*}
Conditioned on $\calD$,  the oracle bootstrap function $B_{\oracle}$ is a centered Gaussian process in $\sfL^2(\QQ)$ with covariance operator $ \what{\opv}_{\oracle}=\EE_{\PP_n}[U_i\otimes_{\splq}U_i]-\EE_{\PP_n}[U_i]\otimes_{\splq}\EE_{\PP_n}[U_i] $,  which is the sample version of the covariance operator $\calV=\mathbb{E}[U_i\otimes_{\splq}U_i]$ defined after (\ref{eq:bahadur}).  The conditional covariance operator of $B_{\oracle}$ can be equivalently written as $ \what{\opv}_{\oracle}=\EE_{\PP_n}[V_i\otimes_{\splq}V_i]-\EE_{\PP_n}[V_i]\otimes_{\splq}\EE_{\PP_n}[V_i]$.

The oracle bootstrap function $B_{\oracle}$  remains unavailable in practice,  due to the unknown function $f_{P^w}$.  
By plugging in the empirical counterparts,  we define 
$$
    \what V_i:=
    w_n(X_i)\big(Y_i-\what f(X_i)\big)(\oplpnw+\lambda_n\opid)^{-1}k_{X_i}
    =
    w_n(X_i)\what{\eps}_i(\oplpnw+\lambda_n\opid)^{-1}k_{X_i},
$$
which leads to the computable multiplier bootstrap function: 
\begin{equation}\label{eq:B-def}
    B:=
    \frac{1}{n}\sum_{1\leq i\neq j\leq n} h_{ij}\cdot 
    \frac{\what V_i-\what V_j}{\sqrt{2}} =
    \sqrt n\,\vkx^{\top}(\mak+n\lambda_n\maw^{-1})^{-1}
    \diag(\what \beps)\vz,
\end{equation}
with $\mah:=(h_{ij})_{n\times n}$ and $\vz:=(2n)^{-1/2}(\mah-\mah^{\top})\vone_n\sim N(0,\maid_n-\vone\vone^{\top}/n)$.   Similarly,  conditioned on data,  $B$ is a centered Gaussian process in $\sfL^2(\QQ)$ with covariance operator 
$
    \what{\opv}:  =
    \EE_{\PP_n}[\what V_i\otimes_{\splq}\what V_i]
    -
    \EE_{\PP_n}[\what V_i]\otimes_{\splq}\EE_{\PP_n}[\what V_i].$

\emph{Fast computation of $\sfL^2(\QQ)$-norm.} We use bootstrap percentiles of $n^{-1/2}\|B\|_{\splq}$ to calibrate the radius of confidence balls.   It is computationally expensive to evaluate these norms by definition.  Fortunately,  $B$ is a linear combination of functions $k_{X_i}$'s,  which permits fast computation of $\|B\|_{\splq}^2$.  Indeed,  given the sampled vector $\bz$,  we have  $B=\sqrt n\,\vkx^{\top}\ba$ with $\va =(\mak+n\lambda_n\maw^{-1})^{-1}\diag(\what{\beps})\vz$. Then, $\norm{B}_{\splq}^2=n\norm{\vkx^{\top}\va}_{\splq}^2=n\va^{\top}\mak_{\QQ}\va$, where $\bK_{\QQ}$ is an $n\times n$ matrix with its $(i,j)$-th entry being $(\bK_{\QQ})_{ij}=\langle k_{X_i}, k_{X_j}\rangle_{\splq}$.  Since the matrix $\bK_{\QQ}$ is fixed throughout bootstrap repetitions, computing $\norm{B}_{\splq}^2$ amounts to evaluating a quadratic form.

\begin{algorithm}
\caption{Multiplier bootstrap for $\sfL^2(\QQ)$-norm confidence balls}
\label{alg:bootstrap-confidence-set}
\begin{algorithmic}[1]
\REQUIRE Data $\calD = \big\{(X_i, Y_i)\big\}_{i=1}^n$, kernel $k(\cdot, \cdot)$, density ratio $\rho(x)=q(x)/p(x)$, weight $w_n(x)=\min\{\rho(x), \tau_n\}$, confidence level $1-\alpha$,  bootstrap repetitions $M$, regularization parameter $\lambda_n$;
\STATE Form the matrix $\mak = \big(k(X_i, X_j)\big)_{n \times n}$, the weight matrix $\maw = \diag(w_n(X_1), \dots, w_n(X_n))$, the  vector $\vy = (Y_1, \dots, Y_n)^{\top}$, and the vector function $\bk_{\bx}=\big(k(X_1,\cdot),\cdots,  k(X_n, \cdot)\big)^{\top}$;
\STATE Compute the IWKRR estimator $\what f = \vkx^{\top} (\mak + n\lambda_n \maw^{-1})^{-1} \vy$;
\STATE Compute the residual vector $\what \beps = \big(Y_1 - \what f(X_1), \dots, Y_n - \what f(X_n)\big)^{\top}$;
\STATE Form the $n\times n$ matrix $\mak_{\QQ}$ with its $(i,j)$-th entry being $(\bK_{\QQ})_{ij}=\langle k_{X_i}, k_{X_j}\rangle_{\splq}$;
\FOR{$b = 1, \dots, M$}
  \STATE Sample $\vz^{(b)} \sim N(0, \maid_n - \vone\vone^{\top}/n)$.
  \STATE Set the coefficient vector
  $
  \va^{(b)} = (\mak + n\lambda_n \maw^{-1})^{-1} \diag(\what{\beps}) \vz^{(b)}.
  $
  \STATE Compute the bootstrap statistic
  $
  T^{(b)} = \norm{\vkx^{\top} \va^{(b)}}_{\splq} = \big((\va^{(b)})^{\top} \mak_{\QQ} \va^{(b)}\big)^{1/2}.
  $
\ENDFOR
\STATE Let $\hat{c}_{\alpha,n}$ be the empirical $(1-\alpha)\%$-percentile of $\{T^{(b)}\}_{b=1}^M$.
\STATE Define the confidence set
$
\cs_{\alpha}(\what f)
=
\big\{
f\in\mathsf{H}:
\norm{\what f-f}_{\splq}\le \hat c_{\alpha,n}
\big\}.
$
\ENSURE the IWKRR $\what  f$ and the radius $\hat c_{\alpha,n}$.
\end{algorithmic}
\end{algorithm}

The multiplier bootstrap procedure is detailed in Algorithm~\ref{alg:bootstrap-confidence-set}.  Let $\what f$ and $\hat c_{\alpha,n}$ be the IWKRR estimator and radius output by Algorithm~\ref{alg:bootstrap-confidence-set}. The resulting $100(1-\alpha)\%$ confidence ball is
\begin{equation}
  \label{eq:feasible-confidence-ball}
    \cs_{\alpha}(\what f)
    =
    \left\{
    f\in\mathsf{H}:
    \norm{\what f-f}_{\splq}\le \hat c_{\alpha, n}
    \right\}.
\end{equation}
The computational cost of (IW)KRR grows fast when the sample size $n$ increases. Fortunately, the multiplier bootstrap Algorithm~\ref{alg:bootstrap-confidence-set} requires solving IWKRR only once, and the subsequent bootstrap repetitions involve just matrix-vector multiplications.  We remark that the residual bootstrap \citep{hall2013bootstrap} is also applicable for constructing the $\splq$-norm confidence balls. See Appendix~\ref{app:ext-residual-bootstrap} for details.

\section{Theoretical Guarantees}
\label{sec:theory}

This section establishes the asymptotic validity of the bootstrap confidence balls (\ref{eq:feasible-confidence-ball}). The proof mainly involves four components. First, we introduce key covariate-shift \emph{compatibility conditions} connecting the target $\splq$ geometry with the regularized truncated weighted-source geometry $\sfL^2(\PP^w)$. Second, we provide a Bahadur
representation of the IWKRR estimator $\what f$ and use Gaussian coupling to approximate the stochastic fluctuation of
$\what f$ by a Gaussian element in \(\splq\). Third, a multiplier
bootstrap function is developed to couple with the same Gaussian element. Finally, we establish the coverage-error bounds for the bootstrap confidence balls using Gaussian  approximation and anti-concentration. 

\subsection{Assumptions and compatibility conditions}

\subsubsection{Statistical and spectral assumptions}

We first collect the statistical and spectral assumptions needed for Bahadur representation and Gaussian approximation under covariate shift and truncated importance weighting.

\begin{assumption}[uniformly bounded noise]\label{assump:error}
There exist constants \(0<\underline{\sigma}\leq\overline{\sigma}<\infty\)
such that
$$
  |\eps|\leq\overline{\sigma}\quad {\rm and}\quad 
  \mathbb E[\eps^2\mid X]\geq\underline{\sigma}^2
  \quad\text{almost surely}.
$$
\end{assumption}
For technical convenience, we assume the random error is uniformly bounded. It is possible to generalize the main theorems to sub-Gaussian noise, but we leave it for  future work. The major challenge is to control the bias and variance induced by truncated importance weighting.

\begin{assumption}[source condition]\label{assump:source-condition}
There exist \(s\in[0,1]\), \(\tau_0>0\), and \(g_{\ast}\in\mathsf H\) such that
$$
  f_*=\oplpwo^s (g_{\ast}),
  \  {\rm where}\ 
  \oplpwo(g)
  = \EE_{\PP}\big[w_0(X)g(X)k_X\big]\ {\rm with}\ w_0(X)=\min\big\{\rho(X), \tau_0\big\}. 
$$
\end{assumption}

The source condition is typical for KRR that characterizes the smoothness of $f_{\ast}$ relative to the kernel integral operator. Assumption~\ref{assump:source-condition} differs slightly from the standard ones in KRR literature  due to the truncated importance weighting. See, e.g., \cite{caponnetto2007optimal, ma2023optimally}. 
The larger smoothness parameter $s$ ensures sharper control on the $\splq$-norm bias and the
RKHS-norm quantities appearing in the bootstrap approximation. Note that we assume the truncation level used in Algorithm~\ref{alg:bootstrap-confidence-set} satisfies
$\tau_n\geq\tau_0$ so that $\oplpwo\preceq_{\mathsf H}\oplpw$.

\begin{assumption}[effective dimension]\label{assump:pseudo-dimn}
  There exists \(d \in [0,1]\) such that
\[
  E_d
  :=
  \max\Big\{
    1,\,
    \sup\nolimits_{\lambda\in(0,1]}
    \big(\pd\,\lambda^d\big)^{1/2}
  \Big\}
  <\infty,
  \ {\rm where}\ 
  \pd:
  =
  \tr_{\mathsf H}\!\left((\oplq+\lambda\opid)^{-1}\oplq\right).
\]
\end{assumption}

Assumption~\ref{assump:pseudo-dimn} always holds with $d=1$ under the bounded-kernel condition, since $\tr_{\sfH}(\calL_Q)=\sum_{i}\mu_i\leq \kappa^2$ by Lemma~\ref{lem:eigendecomposition_LQ}.  If $d\in(0,1)$, Assumption~\ref{assump:pseudo-dimn} holds if the eigenvalues of $\calL_{Q}$ decrease fast enough, for instance, $\mu_i=O(i^{-1/d})$; see \cite{caponnetto2007optimal} for more discussions.

\subsubsection{Compatibility conditions under covariate shift}

The covariate shift and truncated importance weighting require compatibility conditions on the integral operators defined on the measure $\QQ$ and $\PP^w$. 
The two compatibility conditions play distinct roles. The first compatibility condition ensures the \emph{regularized transfer coverage} so that the source-side risk and RKHS-norm regularization control the target-side risk.  The second compatibility condition guarantees a \emph{non-degenerate covariance} in the target distribution, which essentially prevents the Gaussian covariance from degenerating in target geometry and is generally necessary for anti-concentration properties. 

\begin{assumption}[regularized transfer coverage]
\label{assump:covariate-shift-estimation}
There exists a constant \(\mathfrak{C}_1\geq 1\), independent of \(n\), such that 
$
\big\|\oplq^{1/2}(\oplpw+\lambda_n \opid)^{-1/2}\big\|_{\mathsf{H}\to\mathsf{H}}
\leq \mathfrak{C}_1,
$ for the chosen \(\tau_n\) and \(\lambda_n\). 
\end{assumption}

Equivalently, Assumption~\ref{assump:covariate-shift-estimation} requires the \emph{regularized operator domination}: 
$
  \oplq
  \preceq_{\mathsf H}
  \mathfrak{C}_1^2
  (\oplpw+\lambda_n\opid),
$
or, in quadratic-form notation, 
$
  \norm{f}_{\splq}^2
  \leq
  \mathfrak{C}_1^2
  \big(
    \mathbb E_{\mathbb P}\ [w_n(X)f(X)^2]
    +\lambda_n\norm{f}_{\mathsf H}^2
  \big) 
$ for every \(f\in\mathsf H\). 
Moreover, by Lemma~\ref{lem:LQ_half_unitary_isomorphism}, the operator norm in 
Assumption~\ref{assump:covariate-shift-estimation} admits the equivalent
variational representation
\[
  \big\|\oplq^{1/2}(\oplpw+\lambda_n\opid)^{-1/2}\big\|_{\mathsf{H}\to\mathsf{H}}^2
  =
  \sup_{f\in\mathsf{H},\, f\ne0}
  \frac{\norm{f}_{\splq}^2}
  {\mathbb{E}_{\mathbb{P}}\!\left[w_n(X)f(X)^2\right]
  +\lambda_n\norm{f}_{\mathsf{H}}^{2}}.
\]
Therefore, Assumption~\ref{assump:covariate-shift-estimation} essentially states that the weighted source distribution $\PP^w$ covers the target distribution $\QQ$ in all kernel directions, up to the regularization level $\lambda_n$. If the constant $\mathfrak{C}_1$ is large, there exist directions that have small regularized size in the weighted source distribution but large size in the target distribution, which may lead to unstable transfer from source distribution to the target geometry.  The constant $\mathfrak{C}_1$ measures the residual mismatch
after importance weighting, truncation, and regularization; it is therefore
determined jointly by the distribution shift, the kernel geometry, and the
choices of \(\tau_n\) and \(\lambda_n\), rather than by the density ratio alone.

\begin{remark}[sufficient conditions of Assumption~\ref{assump:covariate-shift-estimation}]\label{rmk:density-ratio}
The following three typical density-ratio conditions are sufficient to ensure Assumption~\ref{assump:covariate-shift-estimation}. Recall the truncated weight function $w_n(x)=\min\{\rho(x), \tau_n\}$. 
\begin{enumerate}[(i)]
  \item \textit{Bounded density ratio.}
  Suppose that \(\rho(X)\leq \overline{\rho}\) \(\mathbb P\)-almost surely. If
  \(\tau_n\leq \overline{\rho}\), then
  $
    \rho(x)
    \leq
    (\bar\rho/\tau_n) w_n(x).
  $
  Therefore, 
  $
    \oplq
    \preceq_{\mathsf H}
    (\bar\rho/\tau_n)\oplpw
    \preceq_{\mathsf H}
    (\bar\rho/\tau_n)(\oplpw+\lambda\opid),
  $
  and Assumption~\ref{assump:covariate-shift-estimation} holds with
  $
    \mathfrak{C}_1
    \leq (\bar\rho/\tau_n)^{1/2}.
  $
  If \(\tau_n>\overline{\rho}\), then \(w_n(x)=\rho(x)\). Then, 
  \(\oplpw=\oplq\), and Assumption~\ref{assump:covariate-shift-estimation} holds with $\mathfrak{C}_1=1$. 
  Combining the two regimes gives
  $
    \mathfrak{C}_1
    \leq
    \max\big\{1, (\bar\rho/\tau_n)^{1/2} \big\}.
  $

  \item \textit{Sub-exponential density ratio.} The bounded density ratio condition is stringent. Assumption~\ref{assump:covariate-shift-estimation} still holds if the density ratio has light tails. 
  \begin{definition}[Bernstein-type moment condition]
    \label{def:bernstein-moment-condition-for-weight} There exist \(\gamma\geq 1, W>0\), and \(\sigma>0\) such that, for every
    integer \(m\geq2\),
    $
      \left(
        \mathbb E_{\mathbb Q}\!\left[
          \rho(X)^{\gamma(m-1)}
        \right]
      \right)^{1/\gamma}
      \leq
      (m!/2)W^{m-2}\sigma^2.
    $
  \end{definition}
  The density ratio  may be unbounded, but has light tails under Definition~\ref{def:bernstein-moment-condition-for-weight}.  Nevertheless, Lemma~\ref{lem:bound-a-under-infinite-moment} in the Appendix verifies  Assumption~\ref{assump:covariate-shift-estimation} with
  \(\mathfrak{C}_1\leq\sqrt2\), if the truncation level
  \[
    \tau_n
    \geq
    \left(
      2^{2\gamma-1}\kappa^{2\gamma}E_d^{2\gamma}
      m_0!W^{m_0-2}\sigma^2
    \right)^{\frac{1}{m_0-1}}
    \lambda^{-\frac{(1+d)\gamma}{m_0-1}},
  \]
  where $m_0\geq 2$ is some integer and $E_d$ is the effective dimension defined in Assumption~\ref{assump:pseudo-dimn}. 

  \item \textit{Heavy-tailed density ratio.} Even if the density ratio has heavy tails, Assumption~\ref{assump:covariate-shift-estimation} still holds if the truncation level is properly chosen. 
  \begin{definition}[Bounded $\theta$-th moment]
    \label{def:second-moment-condition-for-weight}
    For a $\theta\geq 2$, there exists some
    \(\Omega<\infty\) such that 
    $
      \mathbb E_{\mathbb P}\!\left[\rho(X)^{\theta}\right]
      \leq\Omega.
    $
    If $\theta=2$, then $
    \chi^2(\mathbb Q\,\|\,\mathbb P)\leq\Omega-1.
  $  
  \end{definition}
  Lemma~\ref{lem:bound-a-under-vartheta-moment} in the Appendix verifies
  Assumption~\ref{assump:covariate-shift-estimation} with
  \(\mathfrak{C}_1\leq\sqrt2\) for \(\theta>2\), if
  $
    \tau_n^{(\theta-1)}
    \geq
    4\kappa^2\Omega E_d^2\lambda^{-(1+d)},
  $
  where \(E_d\) is the effective dimension defined in
  Assumption~\ref{assump:pseudo-dimn}.  
\end{enumerate}
\end{remark}

Note that similar comparison inequality between integral operators already appeared in the literature to quantify the information
transferred from the source distribution to the target task. In particular,
\cite{wang2026pseudo} proposes an effective sample size 
defined by
\[
  n_{\mathrm{eff}}
  :=
  \sup\left\{
    t\leq n:
    \oplq\preceq_{\mathsf H}
    \frac{n}{t}\left(\oplp+\tilde{\lambda}\opid\right)
  \right\},
\]
with a prescribed regularization level \(\tilde{\lambda}\).
Equivalently, \(n/n_{\mathrm{eff}}\) is the smallest \(c\geq1\) such that
\(\oplq\preceq_{\mathsf H}c(\oplp+\tilde{\lambda}\opid)\). Suppose that
\(\tilde{\lambda}=\lambda_n\), the truncation \(w_n(X)\geq\underline{\rho}\) for some
\(\underline{\rho}\in(0,1]\)
\(\mathbb P\)-almost surely, and \(\tau_n\geq1\). Since
$
  \underline{\rho}(\oplp+\lambda_n\opid)
  \preceq_{\mathsf H}
  \oplpw+\lambda_n\opid
  \preceq_{\mathsf H}
  \tau_n(\oplp+\lambda_n\opid),
$
our Assumption~\ref{assump:covariate-shift-estimation} holds with 
\(\mathfrak{C}_1\) satisfying 
$
  n/(n_{\mathrm{eff}}\tau_n)
  \leq
  \mathfrak{C}_1^2
  \leq
  n/(n_{\mathrm{eff}}\underline{\rho}).
$
The effective sample size \(n_{\mathrm{eff}}\) measures source-to-target compatibility in the
regularized unweighted-source geometry, whereas
\(\mathfrak{C}_1\) measures the remaining mismatch
after incorporating truncated importance weighting.

\begin{remark}[other distribution shift measures]
There exist other measures of source--target mismatch serving different purposes, which do not directly supply the domination inequality in
Assumption~\ref{assump:covariate-shift-estimation}. Hypothesis-class
discrepancies compare distributions through the largest difference of expected
losses over a chosen class; for example, the discrepancy distance of
\cite{mansour2009domainadaptation} is used to establish generalization bounds for domain adaptation. While the loss-reliant discrepancies provide convenient risk transfer from source to target distributions, they cannot reflect the impact of the truncated weights or the regularization level.

Distribution-matching measures such as kernel mean matching \citep{huang2007correcting} choose weights to match the kernel mean embeddings in target and weighted source distributions,  i.e.,  controlling the first-order differences of the form \(\mathbb E_{\mathbb Q}f(X)-\mathbb E_{\mathbb P}[w_n(X)f(X)]\).  
Local nonparametric measures, including the transfer exponent of
\cite{kpotufe2018marginal} and the ball-probability
ratio measure of \cite{pathak2022new}, compare how
source and target mass scale on metric balls and are designed to characterize
local transfer rates for Lipschitz or H\"older function classes.  In contrast,  Assumption~\ref{assump:covariate-shift-estimation} goes beyond the source and target distributions: it is a kernel-, truncation-, and regularization-dependent domination condition for the specific RKHS geometry in
which the IWKRR estimator resides.
\end{remark}

In addition,  we require a non-degenerate covariance condition for valid inference,  expressed as compatibility between integral operators.   This lower bound condition on the variance is typical in high-dimensional statistical inference where estimators are usually biased.  See,  e.g.,  Assumption 4 in \cite{xia2021statistical}. 
This non-degenerate covariance condition is imposed to establish anti-concentration of the
norm of Gaussian element in $\splq$.

\begin{assumption}[non-degenerate covariance operator]
  \label{assump:lower-boundedness-of-covariance-operator}
There exists a constant \(\underline{\rho}>0\) such that \(\rho(X)\geq\underline{\rho}\) \(\mathbb P\)-almost surely and  \(\tau_n\geq\underline{\rho}\). Moreover, there exists a constant $\mathfrak{C}_2\geq 1$ such that $(\oplpw+\lambda_n \opid)^{-1} \oplpw (\oplpw+\lambda_n \opid)^{-1}
  \succeq_{\h} \mathfrak{C}_2^{-2}
  (\oplq + \lambda_n \opid)^{-1} \oplq (\oplq + \lambda_n \opid)^{-1}$ for the chosen $\tau_n, \lambda_n$. 
\end{assumption}

The uniform lower bound on the density ratio is operationally less restrictive than a uniform upper bound, especially when source
data are abundant. For instance, source observations in the low-ratio region \(A_\delta=\{x:\rho(x)<\delta\}\) may be screened out
entirely for a small $\delta>0$. This region has target probability
$
  \mathbb Q_X(A_\delta)
  =
  \mathbb E_{\mathbb P_X}
  \bigl[\rho(X)\mathbf 1\{X\in A_\delta\}\bigr]
  \leq \delta.
$
Alternatively, we can randomly retain source observations in this region with probability \(a_\delta(x):=\min\{1,\rho(x)/\delta\}\). Denote by \(\mathbb P_X^{(\delta)}\) the covariate distribution of the retained source sample.  The resulting density ratio
satisfies 
$
  \big(d\mathbb Q_X/d\mathbb P_X^{(\delta)}\big)(x)
  =
  Z_\delta\big(\rho(x)/a_\delta(x)\big)
  \geq Z_\delta\delta,
  \ \mathbb P_X^{(\delta)}\text{-almost surely},
$
where \(Z_\delta:=\mathbb E_{\mathbb P_X}[a_\delta(X)]\). 
Thus, partial thinning of source-rich yet target-irrelevant regions can enforce a uniform lower bound while retaining some observations from those
regions. By contrast, a uniform upper bound on \(\rho(x)\) rules out target-important regions
that are poorly represented by the source distribution, which is a more
consequential restriction for target-domain inference.

Assumption~\ref{assump:lower-boundedness-of-covariance-operator} lower bounds the regularized covariance operator $(\oplpw+\lambda_n\opid)^{-1}\oplpw(\oplpw+\lambda_n\opid)^{-1}$ by $(\oplq+\lambda_n\opid)^{-1}\oplq(\oplq+\lambda_n\opid)^{-1}$, which prevents the Gaussian approximation from becoming degenerate in the target geometry and thereby supports anti-concentration. In contrast, the regularized transfer coverage condition in Assumption~\ref{assump:covariate-shift-estimation} enables 
transfer of estimation and bootstrap errors to the target norm. 
Neither condition implies the other in general.

\subsection{General coverage guarantee of bootstrap confidence balls}
\label{subsec:approximation-framework}
Recall the bias-variance decomposition of IWKRR estimator $\sqrt n(\what f-f_*)=\sqrt n(\what f-\fpw)+\sqrt n(\fpw-f_*)$. 
The bootstrap confidence balls constructed by Algorithm~\ref{alg:bootstrap-confidence-set} are valid if the quantiles of $\|B\|_{\splq}^2$, conditioned on $\calD$, approximate the quantiles of $n\|\what f-f_{\ast}\|_{\splq}^2$. Broadly speaking,  we need to study three approximations to establish the connection between these two quantiles: 

\begin{enumerate}
\item[-] {\it Bias control of IWKRR}.  There exists $\Delta_{\rm bias}>0$ such that the deterministic bias satisfies
  \begin{equation}
    \label{eq:theory-bias-bound}
    \sqrt n\norm{\fpw-f_*}_{\splq}\leq \Delta_{\mathrm{bias}}.
  \end{equation}
  \item[-] {\it Gaussian approximation of IWKRR.} For any $\eta\in(0,1)$, there exist $\Delta_G>0$ and \(Z\sim\mathrm{N}_{\splq}(0,\opv)\) such that
  \begin{equation}
    \label{eq:theory-gaussian-approximation}
    P\!\left(
      \big\|\sqrt n(\what f-\fpw)-Z\big\|_{\splq}\leq \Delta_G
    \right)
    \geq1-\eta.
  \end{equation}
  \item[-] {\it Bootstrap approximation.} For any $\eta\in(0,1)$, there exist $\Delta_B>0$ and \(Z'\sim\mathrm{N}_{\splq}(0,\opv)\), conditionally coupled with \(B\), such
  that, with probability at least \(1-\eta\) over the data \(\mathcal D\),
  \begin{equation}
    \label{eq:theory-bootstrap-approximation}
    P\!\left(
      \norm{B-Z'}_{\splq}\leq \Delta_B
      \,\middle|\,\mathcal D
    \right)
    \geq1-\eta.
  \end{equation}
\end{enumerate}

The above approximation error terms \(\Delta_G,\Delta_B\) may depend on the probability tolerance \(\eta\). The total approximation error is $\Delta:=\Delta_G+\Delta_B+\Delta_{\mathrm{bias}}$, which plays the primary role in the coverage guarantee of bootstrap confidence balls. 
For ease of exposition, we first present the coverage guarantees assuming that the inequalities (\ref{eq:theory-bias-bound})-(\ref{eq:theory-bootstrap-approximation}) hold. Explicit bounds for \(\Delta_G\), \(\Delta_B\), and \(\Delta_{\mathrm{bias}}\) will be established in Section~\ref{sec:theory-explicit} for specific source-target and kernel spectral conditions.

Note that if the quantiles of $\|\what f-f_{\ast}\|_{\splq}^2$ is discontinuous, establishing the $\splq$-norm confidence balls may be infeasible. To rule out these ill-posed cases, we require that the distribution of $\|\what f-f_{\ast}\|_{\splq}^2$ does not concentrate too sharply near its quantiles. Due to the Gaussian approximation (\ref{eq:theory-gaussian-approximation}), it suffices to show that the density function of \(\norm{Z}_{\splq}\) is uniformly bounded. Recall the covariate operator $\calV$ defined in Section~\ref{sec:bahadue}.

\begin{lemma}
  \label{lem:density-upper-bound-for-gaussian-in-splq}
Suppose Assumptions~\ref{assump:error} - \ref{assump:lower-boundedness-of-covariance-operator} hold, \(\tau_n\geq\tau_0\), \(0<\lambda_n\leq1\), and $\oplq$ has either  polynomial-decay eigenvalues $\mu_i\asymp i^{-\beta}$ for some $\beta>1$ or exponential-decay eigenvalues $\mu_i\asymp \exp(-\beta i^{\gamma})$ for some $\beta, \gamma>0$. Let
$Z\sim\mathrm{N}_{\splq}(0,\opv)$ with $\opv=\mathbb E[U_i\otimes_{\splq}U_i]$. 
Then the density of \(\norm{Z}_{\splq}\) is uniformly upper bounded by $C_1\tau_n^{1/2}$, where the constant $C_1>0$ only depend on the noise and source-target distributions, kernel spectral condition, 
and compatibility constants.
\end{lemma}

The general coverage guarantee of bootstrap confidence balls is presented in the following theorem.

\begin{theorem}[general coverage guarantee]
\label{thm:main-coverage}
Let \(\alpha\in(0,1)\) be a miscoverage level. Suppose the conditions of Lemma~\ref{lem:density-upper-bound-for-gaussian-in-splq} hold and inequalities (\ref{eq:theory-bias-bound}) - (\ref{eq:theory-bootstrap-approximation}) hold with a probability tolerance satisfying \(0<\eta<\min\{\alpha,1-\alpha\}\).  Denote by \(\hat t_\alpha\) the bootstrap quantile such that $ P(\|B\|_{\splq} \ge \hat{t}_{\alpha} \mid\mathcal{D}) = \alpha$,
and set \(\hat c_{\alpha,n}=\hat t_{\alpha}/n^{1/2}\). The miscoverage of the confidence ball $\cs_{\alpha}(\what f)$ defined in (\ref{eq:feasible-confidence-ball}) satisfies
    \begin{equation*}
        \big|
        \PP\!\big(f_*\notin  \cs_{\alpha}(\what f)\big)-\alpha
        \big|
        \lesssim
        \eta + \sqrt{\tau_n}\Delta.
    \end{equation*}
\end{theorem}

Therefore, if $\eta_n\to0$ and $\sqrt{\tau_n}\,\{\Delta_{G,n}+\Delta_{B,n}+\Delta_{\mathrm{bias},n}\}\to 0$ as $n\to\infty$, then the coverage probability of the bootstrap confidence balls converges to the nominal level.

\subsection{Finite-sample bounds for coverage probability}\label{sec:theory-explicit}

We establish explicit bounds for the terms $\Delta_{G}, \Delta_{B}$, and $\Delta_{\rm bias}$, which characterize the coverage guarantee of confidence balls.  

\subsubsection{Bias control of IWKRR}

\begin{proposition}\label{prop:main-bias}
Suppose that Assumptions~\ref{assump:source-condition} and \ref{assump:covariate-shift-estimation} hold with \(\tau_n\geq\tau_0\). Then, 
	\begin{equation*}
		\sqrt{n}\norm{\fpw - f_*}_{\splq}
		\le
		\Delta_{\mathrm{bias}}
		:=
		\max\{\kappa^{2s-1},1\}\,
		\mathfrak{C}_1\sqrt n\,
		\lambda_n^{\min\{s+1/2,1\}}\norm{g_{\ast}}_{\sfH}.
	\end{equation*}
\end{proposition}

Thus the bias bound inequality (\ref{eq:theory-bias-bound}) holds with above choice of \(\Delta_{\mathrm{bias}}\),  which diminishes asymptotically if $n\lambda_n^{\min\{2s+1, 2\}}\to 0$ as $n\to\infty$.

\subsubsection{Gaussian approximation of IWKRR}

To establish the inequality (\ref{eq:theory-gaussian-approximation}), we must: (i) bound the approximation error of Bahadur representation (\ref{eq:bahadur}); (ii) bound the Gaussian coupling of the normalized sum of the influence elements. The Gaussian coupling proof uses techniques similar to those in \cite{singh2025kernelridgeregressioninference}, adapted
to the target \(\splq\) geometry and the covariance operator induced by
truncated importance weighting.

\begin{proposition}[Bahadur representation]
\label{prop:main-bahadur} 
Suppose Assumptions~\ref{assump:error},  \ref{assump:source-condition},  and \ref{assump:covariate-shift-estimation} hold with
\(\tau_n\geq\tau_0\),  and $\tau_n \to \infty,  \lambda_n \to 0$ as $n\to\infty$.  If $n> \max\big\{(\tau_n \kappa^2 + 1)\log(e+\pdn)/(C_2\lambda_n),  \tau_n/(\lambda_n \pdn)\big\}$,  where $C_2$ is a constant depending on $\log(1/\eta)$,  then with probability at least $1-\eta$,
\begin{equation*}
\bigg\|\sqrt{n}(\what f - \fpw)-\frac{1}{\sqrt{n}} \sum_{i=1}^n U_i\bigg\|_{\splq}
\lesssim
\Delta_{G,\mathrm{Bah}}
:=
\mathfrak{C}_1\tau_{n}\overline{\sigma}\,\kappa^{3}
\bigl(1+\norm{g_{*}}_{\sfH}\bigr) \sqrt{\frac{\pdn \log(e+\pdn)}{\lambda_n n}}\log^2\left(\frac{2}{\eta}\right)
\end{equation*}
where $U_i=
(\oplpw + \lambda_n \opid)^{-1}\big(w_n(X_i)\eps_i k_{X_i}
+ (w_n(X_i) k_{X_i} \otimes k_{X_i} - \oplpw)(f_* - \fpw)
\big)$.  Moreover,  $\|U_i\|_{\splq}\leq \bar U:= \tau_n\mathfrak{C}_1\lambda_n^{-1/2}(\kappa^2\|f_{\ast}\|_{\sfH}+\kappa\bar\sigma)$.  
\end{proposition}

We note that the upper bound $\bar U$ for $\|U_i\|_{\splq}$ will appear frequently in the subsequent theorems.  The Bahadur representation error $\Delta_{G,{\rm Bah}}$ diminishes asymptotically if $n\lambda_n \gg \tau_n^2\pdn\log(e+\pdn)$.

Recall the spectral tail $\sigma^2_{\splq}(m, \calV)=\sum_{j=m+1}^{\infty} \mu_j(\calV,  \splq)$,  where $\mu_j(\calV, \splq)$ denotes the $j$-th eigenvalue of $\calV$  in the $\splq$ space.  The spectral tail characterizes the smallest error of approximating $\calV$ by a finite-rank operator.  

\begin{proposition}[Gaussian coupling in target \(\mathsf{L^2}(\mathbb Q)\)]\label{prop:main-gaussian-coupling}
There exists \(Z\sim\mathrm{N}_{\splq}(0,\opv)\) such that, with probability at least \(1-\eta\),
\[
  \bigg\|
  \frac{1}{\sqrt{n}} \sum_{i=1}^n U_i - Z
  \bigg\|_{\splq}
  \lesssim
  \Delta_{G,\mathrm{coup}}:=
  \inf_{m\ge1}
  \bigg\{
    \log^{1/2}(6/\eta) \sigma_{\splq}(m,\opv)
    +
    \frac{m^2\bar U\log(m^2/\eta)}{\sqrt{n}}
  \bigg\}  
\]
\end{proposition}

The Gaussian coupling error $\Delta_{G,{\rm coup}}$ is determined by the decay rate of the eigenvalues of $\calV$ in $\splq$.  
By Propositions~\ref{prop:main-bahadur} and \ref{prop:main-gaussian-coupling},  the Gaussian approximation inequality (\ref{eq:theory-gaussian-approximation}) holds with 
$
  \Delta_G
  =
  \Delta_{G,\mathrm{Bah}}
  +
  \Delta_{G,\mathrm{coup}}.
$

\subsubsection{Bootstrap approximation}

To establish (\ref{eq:theory-bootstrap-approximation}),  we must: (i) bound the difference between the oracle bootstrap function $B_{\oracle}$ and the computable bootstrap function $B$ in $\splq$-norm; (ii) bound the Gaussian coupling of $B_{\oracle}$.  The regularized transfer coverage Assumption~\ref{assump:covariate-shift-estimation} is required in the first step  because both the residual and the empirical regularized operator must be transferred from the weighted-source geometry to target \(\splq\) geometry.

\begin{proposition}[difference between oracle and computable bootstrap functions]
\label{prop:main-feasible-bootstrap}
Suppose Assumptions~\ref{assump:error},  \ref{assump:source-condition},  and \ref{assump:covariate-shift-estimation} hold with
\(\tau_n\geq\tau_0\),  and $\tau_n \to \infty,  \lambda_n \to 0$ as $n\to\infty$.  If $n> \max\big\{(\tau_n \kappa^2 + 1)\log(e+\pdn)/(C_2\lambda_n),  \tau_n/(\lambda_n \pdn)\big\}$,  where $C_2$ is a constant depending on $\log(1/\eta)$, then there exists a data event $\calE$ with $\PP(\calE)\geq 1-\eta$ such that, conditioned on $\calE$ and data \(\calD\), with probability at least \(1-\eta\),
\begin{equation*}
\norm{B-B_{\oracle}}_{\splq}
\lesssim
\Delta_{B,\mathrm{feas}}:=\mathfrak{C}_1\,\overline{\sigma}\,\kappa^{2s+4}
\bigl(1+\norm{g_{\ast}}_{\sfH}\bigr)
\frac{\tau_n \pdn}{\sqrt{n\lambda_n}}\cdot \log^3(2/\eta),
\end{equation*}

where 
\begin{equation*}
B
= \frac{1}{n} \sum_{1\leq i\neq j\leq n}  h_{ij}\cdot \frac{\what{V}_i-\what{V}_j}{\sqrt{2}} \quad {\rm and} \quad 
B_{\oracle}= \frac{1}{n} \sum_{1\leq i\neq j\leq n} h_{ij}\cdot \frac{V_i-V_j}{\sqrt{2}}, 
\end{equation*}
with $V_i = w_n(X_i)\bigl(Y_i - \fpw(X_i)\bigr)\,(\oplpw+\lambda_n \opid)^{-1} k_{X_i}$ and  $\what{V}_i = w_n(X_i)\bigl(Y_i - \what f(X_i)\bigr)\,(\oplpnw+\lambda_n \opid)^{-1} k_{X_i}$. 
Moreover, $
B_{\oracle}\mid\mathcal{D}
\sim
\mathrm{N}_{\splq}(0,\what{\opv}_{\oracle})$
and 
$
B\mid\mathcal{D}
\sim
\mathrm{N}_{\splq}(0,\what{\opv}),
$
where $\what{\opv}_{\oracle}=\mathbb{E}_{\PP_n}[V_i\otimes_{\splq}V_i]-\mathbb{E}_{\PP_n}[V_i]\otimes_{\splq}\mathbb{E}_{\PP_n}[V_i]$ and $\what \calV=\mathbb{E}_{\PP_n}[\what V_i\otimes_{\splq}\what V_i]-\mathbb{E}_{\PP_n}[\what V_i]\otimes_{\splq}\mathbb{E}_{\PP_n}[\what V_i]$, respectively.
\end{proposition}

The error $\Delta_{B,{\rm feas}}$ diminishes asymptotically if $n\lambda_n\gg \tau_n^2\pdn^2$.  
Denote the trace $\tr_{\splq}(\calV)=\sigma^2_{\splq}(0,\calV)=\sum_{j=1}^{\infty}\mu_j(\calV, \splq)$. 

\begin{proposition}[bootstrap coupling in \(\mathsf{L^2}(\mathbb Q)\)]\label{prop:main-bootstrap-coupling}
Let  $g$ be the isonormal Gaussian process in $\splq$.  Then there exists a
random element \(G\stackrel{{\rm d.}}{=}\what{\opv}_{\oracle}^{1/2} g\big|\calD\) such that, conditioned on $\calD$, with total probability at
least \(1 - 6\eta\),
\begin{equation*}
  \left\| {\opv}^{1/2} g - G \right\|_{\splq}
  \lesssim
  \Delta_{B,\mathrm{coup}}:=\inf_{m \geq 1} \bigg[
  m^{1/4} \bigg(\frac{\bar U^2\tr_{\splq}(\calV)}{n} + \frac{\bar U^4}{n^2}\bigg)^{1/4}
  + \sigma_{\splq}(m, \opv)\bigg]\log^{3/2}(2/\eta) ,
\end{equation*}
where the covariance operator $\opv= \mathbb{E}\!\left[ U_i \otimes_{\splq} U_i \right]$. 
\end{proposition}

Putting together Propositions ~\ref{prop:main-feasible-bootstrap} and \ref{prop:main-bootstrap-coupling}, we can couple the Bootstrap function $B$ with a Gaussian element in $\splq$, conditioned on data $\calD$, which is formalized in the following corollary. 

\begin{corollary}[bootstrap approximation in \(\mathsf{L^2}(\mathbb Q)\)]
\label{cor:main-approximate-bootstrap}
Suppose the conditions of Propositions~\ref{prop:main-feasible-bootstrap} and~\ref{prop:main-bootstrap-coupling} hold. There exist a universal constant \(c>0\) and a random element $Z'\stackrel{{\rm d.}}{=} \mathrm{N}_{\splq}(0,\opv)\big|\calD$, conditionally coupled with $B$, such that with probability at least
\(1-\eta\) over the data,
\[
  P\!\left(
    \norm{B-Z'}_{\splq}
    \le
    \Delta_{B,\mathrm{feas}}
    +
    \Delta_{B,\mathrm{coup}}
    \,\middle|\,\mathcal D
  \right)
  \geq 1-\eta ,
\]
where \(\Delta_{B,\mathrm{feas}}\) and \(\Delta_{B,\mathrm{coup}}\) denote the
bounds in Propositions~\ref{prop:main-feasible-bootstrap}
and~\ref{prop:main-bootstrap-coupling}, respectively, evaluated with
tolerance \(\eta_{\mathrm{tol}}=c\eta^2\). 
\end{corollary}

By Corollary~\ref{cor:main-approximate-bootstrap}, the bootstrap approximation inequality (\ref{eq:theory-bootstrap-approximation}) holds with $ \Delta_B
  =
  \Delta_{B,\mathrm{feas}}
  +
  \Delta_{B,\mathrm{coup}}$, where \(\Delta_{B,\mathrm{feas}}\) and \(\Delta_{B,\mathrm{coup}}\) are evaluated at tolerance \(c\eta^2\).

\subsubsection{Examples}
\label{subsec:spectral-rate-consequences}

We now provide explicit error rates for the coverage probability of bootstrap confidence sets. By Propositions~\ref{prop:main-bias}-\ref{prop:main-bootstrap-coupling} and Corollary~\ref{cor:main-approximate-bootstrap}, the key roles in the error rates are played by: (i) target-source distribution mismatch; (ii) the spectral decay rate of the covariance operator $\calV$, or equivalently, of the integral operator $\calL_Q$.

For ease of exposition, we ignore the constants $\mathfrak{C}_1, \mathfrak{C}_2, \underline{\rho}$, and assume that Assumptions~\ref{assump:error} - \ref{assump:pseudo-dimn} hold with
\(s\geq 1/2\) and \(\tau_n\geq\tau_0\).  Similarly, the quantities $\gamma, W, \sigma,\bar{\rho}$, and $\Omega$, appearing in Remark~\ref{rmk:density-ratio}, will also be treated as constants. The probability parameter is fixed at $\eta\asymp n^{-1}$ so that $\log(1/\eta)\asymp \log n$. Recall that $\mu_k$ denotes the $k$-th eigenvalue of $\calL_Q$. 

\begin{corollary}[bounded density ratio]\label{cor:bd-density}
Suppose that \(\rho(X)\leq\overline\rho\) \(\mathbb P\)-almost surely, and let
\(\tau_n=C_0\) for a sufficiently large fixed constant
\(C_0\geq\max\{\overline\rho,\tau_0\}\). 
\begin{enumerate}
\item[(1)] if $\mu_k\asymp k^{-\beta}$ for some $\beta>5$ and all $k\geq 1$, by choosing the regularization parameter $\lambda_n\asymp n^{-\frac{2(\beta+1)}{3(\beta+3)}}$, we get 
$$
\big|\PP\big(f_{\ast}\notin \cs_{\alpha}(\what f)\big) -\alpha \big|\lesssim n^{-\frac{(\beta-5)(\beta+10)}{6(\beta+3)(2\beta+5)}}\log^3n;
$$
\item[(2)] if $\mu_k \asymp \exp(-c k^{\zeta})$ for some $c,\zeta>0$ and all $k\geq 1$, by choosing the regularization parameter $\lambda_n\asymp n^{-3/5}$, we get 
$$
\big|\PP\big(f_{\ast}\notin \cs_{\alpha}(\what f)\big) -\alpha \big|\lesssim n^{-1/10}(\log n)^{\frac32+\frac{1}{2\zeta}};
$$
\end{enumerate}
\end{corollary}

\begin{corollary}[sub-exponential tail]\label{cor:subexp-density}
Suppose that the density ratio satisfies the sub-exponential tail condition in Definition~\ref{def:bernstein-moment-condition-for-weight}.
\begin{enumerate}
\item[(1)] if $\mu_k\asymp k^{-\beta}$ for some $\beta>5$ and all $k\geq 1$, by choosing the truncation level
$\tau_n\asymp n^{\frac{\beta-5}{9(\beta+3)(2\beta+5)}}$ and the regularization parameter
$\lambda_n\asymp n^{-\frac{2(\beta+1)}{3(\beta+3)}}$, we get
$$
\big|\PP\big(f_{\ast}\notin \cs_{\alpha}(\what f)\big) -\alpha \big|\lesssim n^{-\frac{(\beta-5)(\beta+9)}{6(\beta+3)(2\beta+5)}}\log^3n;
$$
\item[(2)] if $\mu_k \asymp \exp(-c k^{\zeta})$ for some $c,\zeta>0$ and all $k\geq 1$, by choosing the truncation level
$\tau_n\asymp n^{1/150}$ and the regularization parameter
$\lambda_n\asymp n^{-3/5}$, we get
$$
\big|\PP\big(f_{\ast}\notin \cs_{\alpha}(\what f)\big) -\alpha \big|\lesssim n^{-\frac9{100}}\log^3n;
$$
\end{enumerate}
\end{corollary}

If the density ratio $\rho(X)$ has heavy tails in $\PP$, our methods can still provide valid inference as long as the eigenvalues of $\calL_{Q}$ decay sufficiently fast and the density ratio has a bounded higher moment. For ease of exposition, we only present the theoretical guarantees for the case when $\calL_Q$ has exponentially decaying eigenvalues. If $\calL_Q$ has polynomial-decay eigenvalues, the explicit error rate for coverage probability is complicated (see Appendix~\ref{app:moment-density-poly} for more details).

\begin{corollary}[bounded $\theta$-th moment]\label{cor:moment-density}
Suppose that
\(\EE_{\PP}[\rho(X)^\theta]\leq\Omega_\theta<\infty\) for some
\(\theta>7\) as in Definition~\ref{def:second-moment-condition-for-weight}. 
If \(\mu_k\asymp\exp(-c k^\zeta)\) for some \(c,\zeta>0\)
and all \(k\geq1\), then, by choosing
\[
 \lambda\asymp
 n^{-\frac{3(\theta-1)}{5\theta-2}}
 \quad\text{and}\quad
 \tau_n
 \asymp
 n^{\frac{3}{5\theta-2}}
 (\log n)^{\frac{1}{\zeta(\theta-1)}},
\]
we obtain
\[
\big|\PP\big(f_{\ast}\notin \cs_{\alpha}(\what f)\big)-\alpha\big|
\lesssim
n^{-\frac{\theta-7}{2(5\theta-2)}}
(\log n)^{\frac32+\frac{1}{2\zeta}
+\frac{5}{4\zeta(\theta-1)}}.
\]
\end{corollary}

\subsection{Hypothesis tests and power analysis}
\label{sec:power-analysis}

For a prescribed function
\(f_0\), we reject the null hypothesis in \eqref{eq:testing-problem} whenever
\(f_0\notin  \cs_{\alpha}(\what f)\), or equivalently whenever
\(\|\what f-f_0\|_{\splq}>\hat c_{\alpha, n}\).  The coverage guarantees in Sections~\ref{subsec:approximation-framework}-\ref{sec:theory-explicit} control the size of the inverted test under the null
hypothesis.

 We now provide the power guarantee under fixed
alternatives. The key observation is that, under \(H_1\), the target-domain
separation
$
    \delta:=\|f_*-f_0\|_{\splq}
$
is strictly positive. Therefore, the inverted test is powerful whenever the estimation
error and the scaled critical value are both negligible compared with this
fixed separation. 

\begin{theorem}[power guarantee]\label{thm:power}
Suppose \(H_1\) holds with a fixed separation \(\delta=\|f_*-f_0\|_{\splq}>0\) in the target domain. If
$
    \|\what f-f_*\|_{\splq}
    +
    \hat{c}_{\alpha, n}
    \xrightarrow{\PP}0$  as  $n\to\infty$,  
then 
\begin{equation*}
    \PP\!\big(\|\what f-f_0\|_{\splq}>\hat c_{\alpha,n}\big)
    \to 1, \quad \textrm{ as } n\to\infty.
\end{equation*}
\end{theorem}

\section{Numerical Experiments}
\label{sec:numerical-experiments}

\subsection{Simulation design}

We numerically examine the finite-sample properties of our
proposed method from three perspectives: (i) the calibration of bootstrap radius; (ii) the empirical
coverage of the resulting confidence ball; (iii) the power of the induced test
under fixed alternatives. The source covariates are generated from
\(\mathcal{N}(0, 4)\), while the target distribution is
\(\mathcal{N}(1,1)\). The true regression function is $f_*(x)=\exp(-0.5x^2)$, and the independent noise $\eps_i\sim \mathcal{N}(0,0.09)$. We use the Gaussian radial kernel to examine the performance of bootstrap radius calibration and coverage guarantee. For power evaluation, we compare three kernel functions with a fixed length scale \(\ell=1\). All regularization and truncation parameters are selected by validation as described below. These kernels are
\[
k_{\mathrm{G}}(x,x')=\exp(-r^2/2),\qquad
k_{3/2}(x,x')=(1+\sqrt{3}r)\exp(-\sqrt{3}r),\qquad
k_{1/2}(x,x')=\exp(-r),
\]
with $r=|x-x'|$, where the latter two are the Matérn kernels with smoothness parameters
\(\nu=3/2\) and \(\nu=1/2\), respectively. Note that $k_{1/2}(\cdot, \cdot)$ is also known as
the Laplacian kernel. All three kernels satisfy \(k(x,x)=1\).

The target-domain norm \(\|\cdot\|_{\splq}\) is approximated by Monte Carlo
integration using \(1200\) target samples. The IWKRR estimator uses normalized
importance weights and truncated weights \(w_n(x)=\min\{\rho(x),\tau_n\}\).

\emph{Tuning the regularization and truncation parameters.} For calibration, coverage at each sample size, and power, we independently draw a source-labelled tuning sample \(D^{\mathrm{tune}}\) and a source-labelled validation sample \(D^{\mathrm{val}}\), each of size \(n\), where \(n=300\) for calibration and power and \(n\in\{200,300\}\) for coverage. The parameters are chosen from
\[
\begin{aligned}
\lambda&\in n^{-0.6}
\{1/128,1/64,1/32,1/16,1/8,1/4,1/2,1,2,4\},\\
\tau&\in\left(\frac{n}{300}\right)^{0.06}
\{0.0625,0.125,0.25,0.5,0.75,1,1.25,1.5,1.75,2,2.5\}.
\end{aligned}
\]
For each pair $(\lambda, \tau)$, the candidate model is fitted on \(D^{\mathrm{tune}}\), whereas its target prediction
risk is estimated on \(D^{\mathrm{val}}\) using the self-normalized,
untruncated importance-weighted criterion
\[
\widehat R_{k}^{\mathrm{val}}(\lambda,\tau)
=
\frac{\sum_{i=1}^{n}\rho(X_{i}^{\mathrm{val}})
	\{Y_{i}^{\mathrm{val}}
	-\widehat f_{k,\lambda,\tau}(X_{i}^{\mathrm{val}})\}^{2}}
{\sum_{i=1}^{n}\rho(X_{i}^{\mathrm{val}})}.
\]
If numerical ties occur, we prefer larger $\lambda$ and smaller $\tau$. The selected pair is held fixed across all subsequent Monte Carlo repetitions: one pair is used for calibration, one pair for each \(n\) in coverage, and one pair for each kernel in power across all \(20\) independent experiments.

\subsection{Bootstrap calibration}

We use \(1000\) independent draws to simulate the bootstrap distribution.
Figure~\ref{fig:true-vs-bootstrap} compares the conditional bootstrap
distribution with the empirical distribution of the true scaled target error.
The true distribution is estimated using Monte Carlo simulation from \(1000\) independent training samples, each with a sample size \(n=300\). The bootstrap distribution is computed conditional on one representative fitted data set.

\begin{table}[htbp]
	\centering
	\caption{Bootstrap estimates of upper quantiles at $\alpha\in\{0.10, 0.05\}$ and \(n=300\).}
	\label{tab:bootstrap-calibration}
	\begin{tabular}{cccc}
\toprule
$\alpha$ & Sampling quantile & Conditional bootstrap quantile & Relative error \\
\midrule
0.10 & 1.109 & 1.071 & -0.034 \\
0.05 & 1.222 & 1.162 & -0.049 \\
\bottomrule
\end{tabular}

\end{table}

Table~\ref{tab:bootstrap-calibration} reports the calibration summary, whereas Figure~\ref{fig:true-vs-bootstrap} provides a smoother display of the corresponding empirical and conditional bootstrap distributions. The table's bootstrap quantiles are close to the empirical sampling quantiles at both levels, with absolute relative errors between three and five percent.

\begin{figure}[htbp]
	\centering
	\IfFileExists{simulation/numerical_experiments/figures/bootstrap_calibration.pdf}{%
		\includegraphics[width=0.72\textwidth]{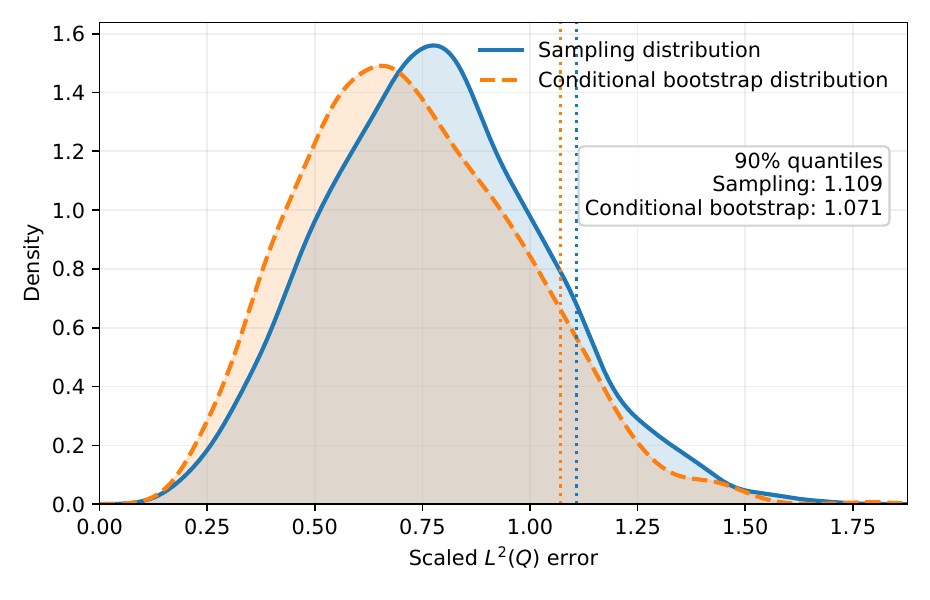}%
	}{%
		\fbox{\parbox{0.85\textwidth}{\centering Missing figure file:
				\texttt{simulation/numerical\_experiments/figures/bootstrap\_calibration.pdf}}}%
	}
	\caption{Bootstrap distribution of the $\splq$ error. The figure compares
		the true sampling distribution of
		\(\sqrt n\|\what f-f_*\|_{\splq}\) with the conditional bootstrap distribution
		of \(\|B\|_{\splq}\).}
	\label{fig:true-vs-bootstrap}
\end{figure}

\subsection{Empirical coverage and size}

We evaluate the empirical coverage of the feasible confidence ball
\(\cs_{\alpha}(\what f)\) with the critical value $\hat c_{\alpha, n}$ estimated from $250$ multiplier-bootstrap repetitions. Under the null hypothesis
\(f_0=f_*\), we record the relative frequency of rejecting the null hypothesis using the inverted test
$
    \mathbf{1}\{
    \|\what f-f_0\|_{\splq}>\hat c_{\alpha,n}
    \}.
$
For each sample size, the rejection probability is estimated from \(1000\)
independent repetitions. Table~\ref{tab:coverage-size} shows modest finite-sample over-rejection and under-coverage at \(n=200\), whereas at \(n=300\) the empirical size and coverage are close to their nominal levels.

\begin{table}[htbp]
\centering
\caption{Empirical size of the inverted test and coverage of the confidence
ball under \(H_0:f_0=f_*\).}
\label{tab:coverage-size}
\begin{tabular}{rrrr}
\toprule
$n$ & $\alpha$ & Empirical size & Empirical coverage \\
\midrule
200 & 0.10 & 0.138 & 0.862 \\
200 & 0.05 & 0.084 & 0.916 \\
300 & 0.10 & 0.099 & 0.901 \\
300 & 0.05 & 0.054 & 0.946 \\
\bottomrule
\end{tabular}

\end{table}

\subsection{Power performance}

We evaluate the empirical power of the induced test under fixed alternatives.  Let 
$
    f_{0,\delta}(x)
    =
    f_*(x)+\delta\, h(x)$ with
$
    h(x)=g(x)\|g\|_{\splq}^{-1},
$
where $g(x):=k_{\mathrm G}(x,1)-k_{\mathrm G}(x,0)=\exp\{-(x-1)^2/2\}-\exp(-x^2/2)$ and $\|g\|_{\splq}=3^{-1/4}(1-e^{-1/3})^{1/2}$.   
Since \(g\) is a difference of Gaussian kernel sections and is a Schwartz function, it belongs to each of the three RKHSs considered.   Similarly,   \(f_{0,\delta}\) belongs to all three RKHSs and
\(\|f_{0,\delta}-f_*\|_{\splq}=\delta\). 

The significance level is fixed at \(\alpha=0.10\) with \(n=300\).  The critical value $\hat c_{\alpha, n}$ is estimated using 300 bootstrap draws.  
For each separation $\delta$ and kernel function, the empirical rejection probability is estimated from \(100\) independent test repetitions, whose empirical variation is evaluated via \(20\) independent experiments. The results are displayed in Figure~\ref{fig:power-curve}.  
The rejection probabilities
increase with the target-domain separation for all three scenarios,  reaching
approximately one at \(\delta=0.12\). At \(\delta=0\), the average rejection
probabilities for the Gaussian, Matérn \(\nu=3/2\), and Matérn \(\nu=1/2\)
scenarios are \(0.1030\), \(0.0985\), and \(0.1075\), respectively. These values are empirical sizes, which are close to the nominal level \(0.10\).

\begin{figure}[htbp]
\centering
\IfFileExists{simulation/numerical_experiments/figures/power_curve.pdf}{%
    \includegraphics[width=0.98\textwidth]{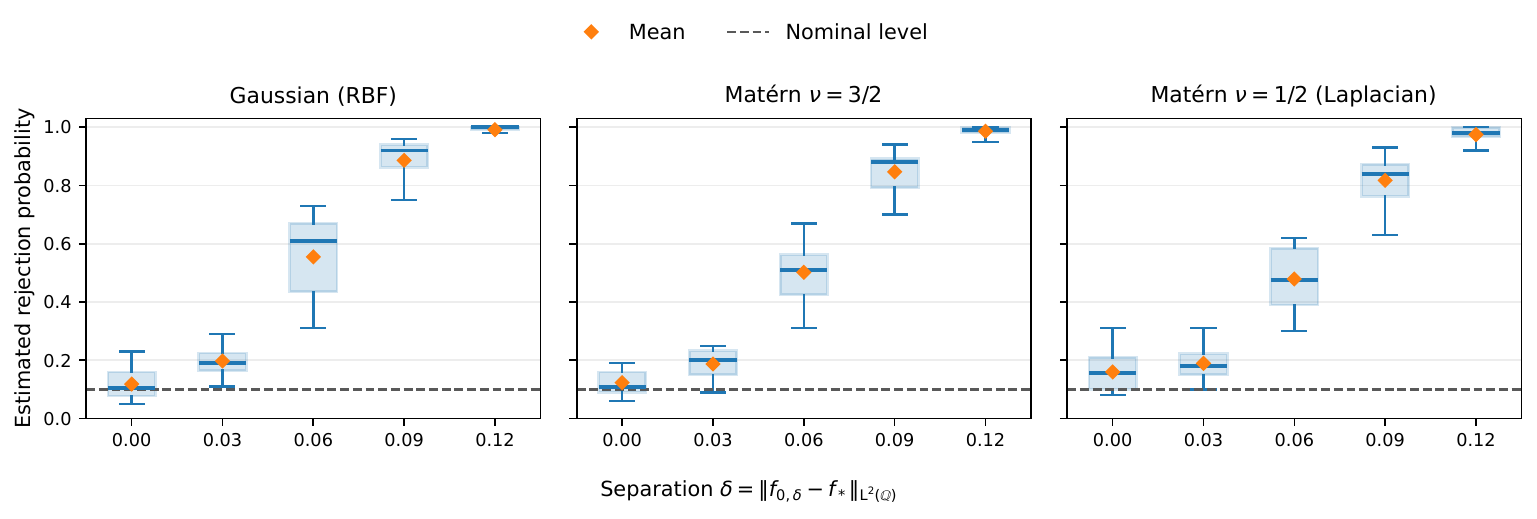}%
}{%
    \fbox{\parbox{0.85\textwidth}{\centering Missing figure file:
    \texttt{simulation/numerical\_experiments/figures/power\_curve.pdf}}}%
}
\caption{
Empirical power under fixed
alternatives \(\|f_{0,\delta}-f_*\|_{\splq}=\delta\) for three kernel functions.
 From left to right, the panels correspond to the Gaussian RBF, Matérn
\(\nu=3/2\), and Matérn \(\nu=1/2\) (Laplacian) kernels, all with length scale
\(\ell=1\). The parameters \((\lambda,\tau)\) are tuned
using independent source-labelled tuning and validation samples of size \(300\); the selected pair for each kernel is held fixed across all \(20\) simulation experiments. The rejection critical values are estimated using $300$ multiplier-bootstrap draws.
 The boxplots summarize variation across the \(20\) experiments, and the diamonds mark the corresponding means. The dashed horizontal line marks the nominal level \(\alpha=0.10\).
}
\label{fig:power-curve}
\end{figure}

\section{Application to Survey-Weighted Household Finance}
\label{sec:scf-application}

\subsection{Data and survey weights}
\label{subsec:scf-data-design}

The Survey of Consumer Finances (SCF) is a triennial cross-sectional survey of U.S. families that records
income, balance sheets, credit use, and demographic characteristics. Our experiment focuses on the implicate 1 of 2022 public summary extract released by \cite{federalreserve2023scf}, which contains $4,595$ records representing distinct families. Because the SCF records are not randomly sampled, each observation comes with a survey weight indicating how many U.S. families that observation represents. These survey weights are  normalized to have mean one, denoted by $\rho_i$ satisfying $n^{-1}\sum_{i=1}^n\rho_i=1$, which help convert averages over the unweighted records into averages for the full U.S. family population represented by the SCF.  We regard the unweighted records as a sample from the source distribution \(\PP\) and the survey-weighted records as a sample from the target distribution $\QQ$, under which the financial conclusions are sought.

Figure~\ref{fig:scf-weighting-summary} illustrates the empirical distribution of the normalized weights and the effect of reweighting on two specific covariates. Two measures are provided in panel A describing the variation of normalized weights, where the coefficient of variation ${\rm CV}:=\big(\sum_{i=1}^n(\rho_i-1)^2/n\big)^{1/2}$ and effective size size ${\rm ESS}=n^2/\big(\sum_{i=1}^{n}\rho_i^2\big)$. More summary statistics are reported in Appendix~\ref{app:scf-data-description}.

Our experiments primarily study four questions: (i) how does the recorded characteristics affect the household net worth; (ii) is the probability of holding stock equity nondecreasing in usual income; (iii) do households allocate the same share of wealth to stock equity; (iv) does the late-payment probability change across years? Due to page constraints, we present the two studies in this section and other results in the Appendix.

\begin{figure}[htbp]
\centering
\IfFileExists{applications/scf_household_finance/figures/scf_weighting_overview.pdf}{%
    \includegraphics[
        width=0.98\textwidth,
        trim=72bp 459bp 6bp 0bp,
        clip
    ]{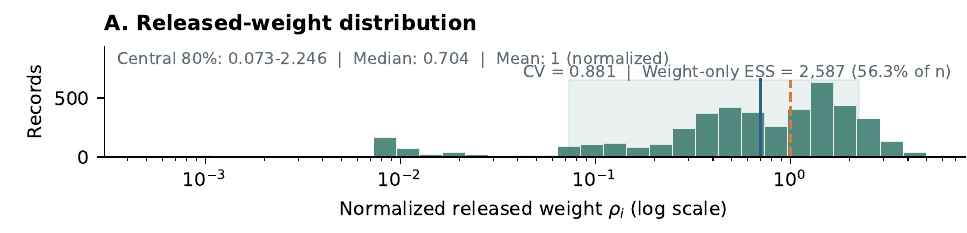}%
    \par\vspace{0.4em}
    \begin{minipage}[t]{0.48\textwidth}
    \centering
    \includegraphics[
        width=\textwidth,
        trim=14bp 179bp 268bp 258bp,
        clip
    ]{applications/scf_household_finance/figures/scf_weighting_overview.pdf}%
    \end{minipage}\hfill
    \begin{minipage}[t]{0.48\textwidth}
    \centering
    \includegraphics[
        width=\textwidth,
        trim=280bp 42bp 6bp 396bp,
        clip
    ]{applications/scf_household_finance/figures/scf_weighting_overview.pdf}%
    \end{minipage}
    \par\vspace{0.3em}
    \includegraphics[
        width=0.56\textwidth,
        trim=132bp 17bp 122bp 532bp,
        clip
    ]{applications/scf_household_finance/figures/scf_weighting_overview.pdf}%
}{%
    \fbox{\parbox{0.85\textwidth}{\centering
    The selected released-weight diagnostic is unavailable.}}%
}
\caption{Panel A shows the normalized released weights on a
logarithmic scale; the shaded region spans the 10th through 90th percentiles,
and the solid and dashed vertical lines mark the median and normalized mean,
respectively.  Panels D and G compare unweighted and weighted sample percentages for education and age bands, respectively.}
\label{fig:scf-weighting-summary}
\end{figure}

\subsection{Study 1: predicting household net worth}
\label{subsec:scf-frozen-benchmark}
The response is $Y={\rm asinh}\big(\texttt{NETWORTH}\times 10^{-5}\big)$ and target regression function $f_{\ast}(x)=\EE (Y\mid X=x)$. 
 Numerical variables are standardized, and categorical variables are encoded, yielding a $33$-dimensional covariate variable $X$.  The variable construction details are reported in Appendix. We propose four candidate functions for $f_0$ and test against $H_0: \|f_{\ast}-f_0\|_{\splq}=0$, among which the ordinary KRR $f_{0,{\rm KRR}}$ trained on $\PP$, i.e., unweighted household records, is of primary interest. 

The data is randomly split into three disjoint subsets with size: 
$
\lvert\mathcal D_0\rvert=1{,}400, 
\lvert\mathcal D_I\rvert=1{,}800, 
\lvert\mathcal D_Q\rvert=1{,}395. 
$
The subset $\calD_0$ is reserved for training the candidate function. Three constant-valued candidate functions are $f_{0,\mathrm{tgt}}(x)\equiv \sum_{i\in\mathcal D_0}\rho_iY_i/\big(\sum_{i\in\mathcal D_0}\rho_i\big)$, $f_{0,\mathrm{src}}(x)\equiv\sum_{i\in\mathcal D_0}Y_i/|\calD_0|$, and $f_{0,\mathrm{zero}}(x)\equiv 0$. For $f_{0,{\rm KRR}}$, we use the Gaussian RBF kernel with \(\gamma=1/33\) with $\lambda_n$ tuned by cross validation. 
The subset $\calD_I$ is used for fitting IWKRR $\what f$ and implementing the multiplier bootstrap. The tuning parameters $\lambda_n$ and $\tau_n$ are selected by cross validation using $\calD_0$.  The subset $\calD_Q$ is used for evaluation, based on which we define the empirical $\splq$-norm: $\norm{g}_{\widehat{\QQ}}^2
=\sum_{i\in\mathcal D_Q}\rho_i g^2(X_i)/\big(\sum_{i\in\mathcal D_Q}\rho_i\big)$.

By \(B=5{,}000\) multiplier-bootstrap draws, we get the $90\%$-percentile \(\hat c_{0.10}=0.1165\) as an estimate of the quantile for $\|\what f-f_{\ast}\|_{\splq}$. By the triangle inequality
$\big|
\norm{f_{\ast}-f_0}_{\widehat{\QQ}}
-
\norm{\what f-f_0}_{\widehat{\QQ}}
\big|
\leq
\norm{\what f-f_{\ast}}_{\widehat{\QQ}}, 
$
we construct the $90\%$ confidence interval for $\|f_{\ast}-f_0\|_{\splq}$ by $\|\what f-f_0\|_{\what\QQ}\pm \hat c_{0.10}$, where the left endpoint is rounded to zero if negative. The confidence interval is conservative by nature. 

The bootstrap distribution, test values $\|\what f - f_0\|_{\what\QQ}$ for four candidate functions, and the corresponding $90\%$ confidence intervals are illustrated in Figure~\ref{fig:scf-frozen-benchmark}. All four intervals exclude zero.  In particular, the source-only KRR candidate $f_{0,{\rm KRR}}$ is rejected at the nominal level.  Further details are provided in Appendix~\ref{app:scf-study1-details}.



\begin{figure}[htbp]
\centering
\IfFileExists{applications/scf_household_finance/figures/scf_study1_frozen_benchmark_comparison_v1.pdf}{%
    \includegraphics[width=\textwidth]{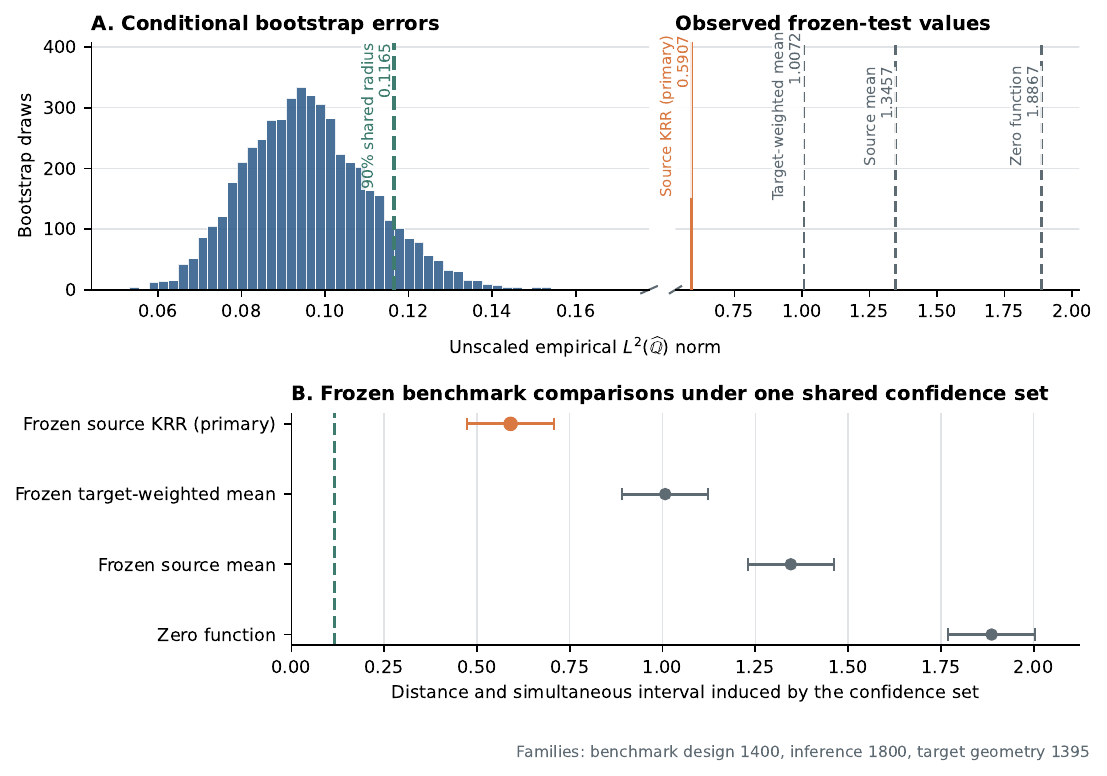}%
}{%
    \fbox{\parbox{0.85\textwidth}{\centering
    The frozen-benchmark comparison figure is unavailable.}}%
}
\caption{testing candidate functions for estimating household net worth based on IWKRR $\what f$. The four candidate functions include the source-only KRR $f_{0,{\rm KRR}}$, the target-weighted mean $f_{0,\mathrm{tgt}}$, the source mean $f_{0,\mathrm{src}}$, and the trivial function $f_{0,\mathrm{zero}}$.  Panel A
displays the bootstrap distribution over \(5{,}000\) draws,  \(90\%\) confidence-set radius $\hat c_{0.1}$, and the test values $\|\what f-f_0\|$ of four candidate functions.   Panel B
shows the $90\%$ CI for $\|f_{\ast}-f_0\|_{\splq}$ for four candidate functions, based on which all are rejected at the nominal level $\alpha=0.10$. }
\label{fig:scf-frozen-benchmark}
\end{figure}

\
\subsection{Study 2: does higher income increase the probability of holding equity?}
\label{subsec:scf-monotonicity}

We are interested in the relationship between household income and the probability of holding stock equity. Let the response variable be $Y=\mathbf 1\{\texttt{HEQUITY}=1\}$, the single covariate $Z=\log\{1+\texttt{NORMINC}\}$, and unknown true regression function $f_{\ast}(z)=\QQ(Y=1\mid Z=z)$, where \(\texttt{NORMINC}\) denotes usual or normal income. In particular, we want to test whether $f_{\ast}(z)$ is  monotonic in $z$.  Let \(\mathcal M_{\uparrow}\) be the class of
\([0,1]\)-valued nondecreasing functions on the income domain, where functions
equal \(\QQ\)-almost surely are identified. The class $\calM_{\uparrow}$ is a closed subset of \(\mathsf{L^2}(\QQ)\).
We aim to test against $
H_0: f_{\ast}\in\mathcal M_{\uparrow}$.  
This is a composite test. Alternatively, we test against $H_0: d_{\QQ}(f_{\ast}, \calM_{\uparrow})=0$, where the projection distance $d_{\QQ}(h, \calM_{\uparrow}):=\inf_{g\in\calM_{\uparrow}}\|h-g\|_{\splq}$.

We random split the data into the training set $\calD_t$ of size \(1{,}800\) and the evaluation set $\calD_e$ of size \(2{,}795\). The IWKRR estimator $\what f$ is learned using the Gaussian RBF kernel on the training set where the parameters $(\lambda_n, \tau_n)$ are selected by cross validation. The empirical target distribution $\what \QQ$ is computed using the evaluation set by $\what \QQ(Z\leq z):=\sum_{i\in\calD_e}\rho_i \mathbf{1}(Z_i\leq z)/\sum_{i\in\calD_e} \rho_i$. 
The empirical distance is $d(\what f,\mathcal M_{\uparrow})=\inf_{g\in\calM_{\uparrow}}\|\what f- g\|_{\sfL^2(\what \QQ)}=0.0151$, and the $90\%$-percentile based on $B=5,000$ multiplier bootstrap is $\hat c_{0.10}=0.0252$. Note that $\hat c_{0.10}$ is an estimate of the $90\%$-quantile of $\|\what f- f_{\ast}\|_{\splq}$. By the inequality $\big|d_{\QQ}(\what f,\calM_{\uparrow})-d_{\QQ}(f_{\ast},\calM_{\uparrow})\big|
\leq
\norm{\what f-f_{\ast}}_{\splq}$, we construct the $90\%$ confidence interval for $d_{\QQ}(f_{\ast},\calM_{\uparrow})$ by $d_{\QQ}(\what f, \calM_{\uparrow})\pm \hat c_{0.10}$, resulting into the interval $[0, 0.0403]$, where the left endpoint is rounded to zero if negative.  Since zero is included in the interval, we do not reject the null hypothesis at the significance level $\alpha=0.10$.
Figure~\ref{fig:scf-income-monotonicity} reports the multiplier-bootstrap distribution, the empirical distance $d(\what f, \calM_{\uparrow})$, and the critical value $\hat c_{0.1}$. More details can be found in the Appendix.

\begin{figure}[htbp]
\centering
\IfFileExists{applications/scf_household_finance/figures/scf_study3_bootstrap_error_monotonicity.pdf}{%
    \includegraphics[
        width=\textwidth,
        trim=0bp 0bp 0bp 16bp,
        clip
    ]{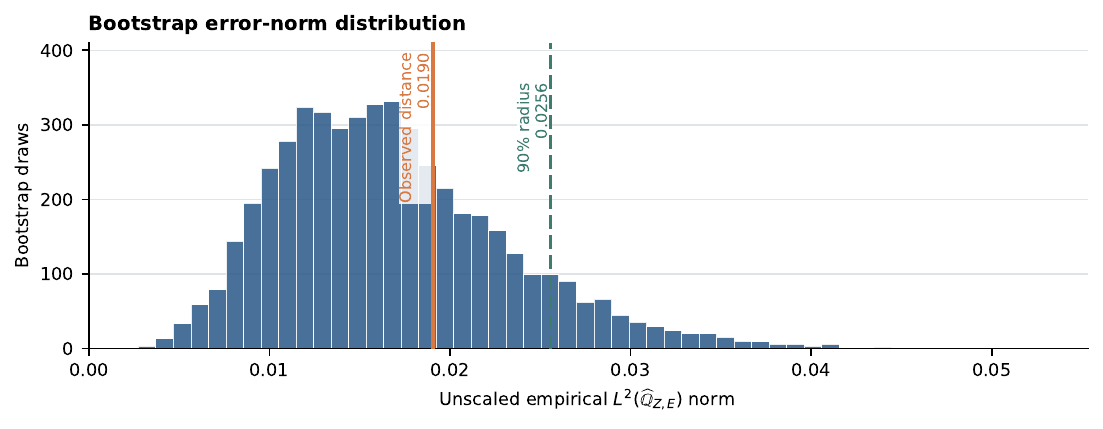}%
}{%
    \fbox{\parbox{0.85\textwidth}{\centering
    The income-monotonicity bootstrap-calibration figure is unavailable.}}%
}
\caption{testing whether the probability of holding stock equity is monotone with respect to the usual income.  The histogram plot shows the empirical bootstrap distribution based on  \(5{,}000\) draws.  The solid orange line is the
observed distance 
\(d(\what f,\mathcal M_{\uparrow})=0.0151\), and the dashed green line is the
\(90\%\) bootstrap percentile \(\hat c_{0.10}=0.0252\).  The null hypothesis is not rejected at level \(0.10\).}
\label{fig:scf-income-monotonicity}
\end{figure}

\section{Discussion}
\label{sec:discuss}

{\it Unknown density ratio.} For simplicity and technical convenience, we assumed known density ratio $\rho(x)$ throughout the paper. While this assumption may hold in semi-supervised learning where unlabelled datasets are abundant, it can be restrictive in many applications. There are several existing methods to estimate the density ratio if it is unknown. Suppose that a sample of source covariates $X_1,\cdots, X_n\stackrel{\rm i.i.d.}{\sim} \PP$ and target covariates $\wtilde X_1,\cdots, \wtilde X_N\stackrel{{\rm i.i.d.}}{\sim}\QQ$ are available. For example, the least squares density-ratio estimator \citep{kanamori2009least} approximates $\rho(x)$ by minimizing $\EE_{\PP_n}[\rho^2(X)]-2\EE_{\QQ_N}[\rho(X)]$; the Kullback-Leibler importance estimation procedure (KLIEP, \cite{sugiyama2008direct}) maximizes the target likelihood $\EE_{\QQ_N}\log \rho(X)$ subject to the normalization constraint $\EE_{\PP_n}\rho(X)=1$; the kernel mean matching (KMM, \cite{huang2007correcting,yu2012analysis}) approach aims to minimize the discrepancy between mean feature maps $\|\EE_{\PP_n}\rho(X)\phi(X)-\EE_{\QQ_N}\phi(X) \|_{\sfH}$ where $\phi$ denotes the canonical feature map. These empirical estimates of $\rho(x)$ are applicable to our procedures. However, their effect on inference remains unclear and challenging to characterize, because the main issue is not only consistency of $\what \rho$, but also its impact to the first-order approximation, bootstrap approximation, as well as the finite-sample coverage guarantees. We shall leave this in potential future work.

{\it Sub-Gaussian kernel function}.  The main theorems are established for the bounded kernel where $\sup_{x\in\frakX}k(x,x)\leq \kappa^2$.  Our methods and theories are still valid if the kernel evaluating function $k_X=k(X,\cdot)$ is sub-Gaussian as an $\sfH$-valued element in the source distribution.  In Appendix~\ref{app:ext-subG},  we provide the theoretical guarantees for the Gaussian and bootstrap approximation under sub-Gaussian kernel functions \citep{wang2026pseudo}.

{\it Density ratio moment condition and kernel spectral decay rate}.  The explicit error rates for the coverage probability established in Corollary ~\ref{cor:bd-density} - \ref{cor:moment-density} require a bounded high-order moment of the density ratio and sufficiently fast-decaying eigenvalues of the integral operator $\calL_Q$.  These conditions are stronger than the typical ones in existing literature on the consistency of (IW)KRR.  The stronger conditions arise because valid inference procedures require accurate distributional approximation such as the Bahadur representation involving $\sqrt{n}$-scale linearization,  Gaussian and bootstrap approximations, and anti-concentration for coverage guarantee.   Higher moments of the density ratio make the truncation-induced discrepancy between the target and weighted-source integral operators decay faster.  Moreover,  faster spectral decay reduces the effective dimension and the spectral-tail error in the infinite-dimensional Gaussian and bootstrap approximations.  Together, these conditions place the truncation level, effective dimension,  and spectral tail in a joint regime under which all inferential remainders are asymptotically negligible,  thereby yielding valid coverage.  It is possible that the moment order and spectral-decay requirements in those corollaries are improvable,  which we leave in potential future work.

\bibliography{reference}

@article{MINSKER2017111,
title = {On some extensions of Bernstein's inequality for self-adjoint operators},
journal = {Statistics \& Probability Letters},
volume = {127},
pages = {111-119},
year = {2017},
author = {Stanislav Minsker}
}

@article{singh2025kernelridgeregressioninference,
title={Kernel ridge regression inference},
  author={Singh, Rahul and Vijaykumar, Suhas},
  journal={arXiv preprint arXiv:2302.06578},
  year={2023}
  }

@book{vershynin2018high,
  title={High-dimensional probability: An introduction with applications in data science},
  author={Vershynin, Roman},
  volume={47},
  year={2018},
  publisher={Cambridge university press}
}

@article{kanamori2009least,
  title={A least-squares approach to direct importance estimation},
  author={Kanamori, Takafumi and Hido, Shohei and Sugiyama, Masashi},
  journal={The Journal of Machine Learning Research},
  volume={10},
  pages={1391--1445},
  year={2009},
  publisher={JMLR. org}
}

@article{yu2012analysis,
  title={Analysis of kernel mean matching under covariate shift},
  author={Yu, Yaoliang and Szepesv{\'a}ri, Csaba},
  journal={arXiv preprint arXiv:1206.4650},
  year={2012}
}

@book{hall2013bootstrap,
  title={The bootstrap and Edgeworth expansion},
  author={Hall, Peter},
  year={2013},
  publisher={Springer Science \& Business Media}
}

@article{xia2021statistical,
  title={Statistical inferences of linear forms for noisy matrix completion},
  author={Xia, Dong and Yuan, Ming},
  journal={Journal of the Royal Statistical Society Series B: Statistical Methodology},
  volume={83},
  number={1},
  pages={58--77},
  year={2021},
  publisher={Oxford University Press}
}

@book{efron1994introduction,
  title={An introduction to the bootstrap},
  author={Efron, Bradley and Tibshirani, Robert J},
  year={1994},
  publisher={Chapman and Hall/CRC}
}

@article{bickel1981some,
  title={Some asymptotic theory for the bootstrap},
  author={Bickel, Peter J and Freedman, David A},
  journal={The annals of statistics},
  volume={9},
  number={6},
  pages={1196--1217},
  year={1981},
  publisher={Institute of Mathematical Statistics}
}

@book{wahba1990spline,
  title={Spline models for observational data},
  author={Wahba, Grace},
  year={1990},
  publisher={SIAM}
}

@book{vapnik2013nature,
  title={The nature of statistical learning theory},
  author={Vapnik, Vladimir},
  year={2013},
  publisher={Springer science \& business media}
}

@book{scholkopf2002learning,
  title={Learning with kernels: support vector machines, regularization, optimization, and beyond},
  author={Sch{\"o}lkopf, Bernhard and Smola, Alexander J},
  year={2002},
  publisher={MIT press}
}

@article{mercer1909functions,
  title={Functions of positive and negative type, and their connection with the theory of integral equations},
  author={Mercer, James},
  journal={Philosophical Transactions of the Royal Society of London. Series A},
  volume={209},
  pages={415--446},
  year={1909}
}

@article{gogolashvili2023importance,
  title={When is Importance Weighting Correction Needed for Covariate Shift Adaptation?},
  author={Gogolashvili, Davit and Zecchin, Matteo and Kanagawa, Motonobu and Kountouris, Marios and Filippone, Maurizio},
  journal={arXiv preprint arXiv:2303.04020},
  year={2023}
}

@book{Cordes_1987, 
  place={Cambridge}, 
  series={London Mathematical Society Lecture Note Series}, 
  title={Spectral Theory of Linear Differential Operators and Comparison Algebras}, 
  publisher={Cambridge University Press}, author={Cordes, Heinz Otto}, 
  year={1987}, 
  collection={London Mathematical Society Lecture Note Series}
}

@article{caponnetto2007optimal,
  title={Optimal rates for the regularized least-squares algorithm},
  author={Caponnetto, Andrea and De Vito, Ernesto},
  journal={Foundations of Computational Mathematics},
  volume={7},
  number={3},
  pages={331--368},
  year={2007},
  publisher={Springer}
}

@article{shimodaira2000improving,
  title={Improving predictive inference under covariate shift by weighting the log-likelihood function},
  author={Shimodaira, Hidetoshi},
  journal={Journal of statistical planning and inference},
  volume={90},
  number={2},
  pages={227--244},
  year={2000},
  publisher={Elsevier}
}

@book{book_covariate_shift,
 author = {Masashi Sugiyama and Motoaki Kawanabe},
 publisher = {The MIT Press},
 title = {Machine Learning in Non-Stationary Environments: Introduction to Covariate Shift Adaptation},
 urldate = {2026-01-07},
 year = {2012}
}

@article{Guan_2022,
   title={Domain Adaptation for Medical Image Analysis: A Survey},
   volume={69},
   number={3},
   journal={IEEE Transactions on Biomedical Engineering},
   publisher={Institute of Electrical and Electronics Engineers (IEEE)},
   author={Guan, Hao and Liu, Mingxia},
   year={2022},
   month={March}, pages={1173--1185} }

@article{ma2023optimally,
  title={Optimally tackling covariate shift in RKHS-based nonparametric regression},
  author={Ma, Cong and Pathak, Reese and Wainwright, Martin J},
  journal={The Annals of Statistics},
  volume={51},
  number={2},
  pages={738--761},
  year={2023},
  publisher={Institute of Mathematical Statistics}
}

@article{cortes2010learning,
  title={Learning bounds for importance weighting},
  author={Cortes, Corinna and Mansour, Yishay and Mohri, Mehryar},
  journal={Advances in neural information processing systems},
  volume={23},
  year={2010}
}

@article{mansour2009domainadaptation,
      title={Domain adaptation: Learning bounds and algorithms},
  author={Mansour, Yishay and Mohri, Mehryar and Rostamizadeh, Afshin},
  journal={arXiv preprint arXiv:0902.3430},
  year={2009}
}

@inproceedings{kpotufe2018marginal,
  title={Marginal singularity, and the benefits of labels in covariate-shift},
  author={Kpotufe, Samory and Martinet, Guillaume},
  booktitle={Conference On Learning Theory},
  pages={1882--1886},
  year={2018},
  organization={PMLR}
}

@article{wang2024comprehensive,
title = {A comprehensive survey on deep active learning in medical image analysis},
journal = {Medical Image Analysis},
volume = {95},
pages = {103201},
year = {2024},
author = {Haoran Wang and Qiuye Jin and Shiman Li and Siyu Liu and Manning Wang and Zhijian Song},
}

@inproceedings{pathak2022new,
  title={A new similarity measure for covariate shift with applications to nonparametric regression},
  author={Pathak, Reese and Ma, Cong and Wainwright, Martin},
  booktitle={International Conference on Machine Learning},
  pages={17517--17530},
  year={2022},
  organization={PMLR}
}

@article{wang2026pseudo,
  title={Pseudo-labeling for kernel ridge regression under covariate shift},
  author={Wang, Kaizheng},
  journal={The Annals of Statistics},
  volume={54},
  number={1},
  pages={252--276},
  year={2026},
  publisher={Institute of Mathematical Statistics}
}

@article{schmidt2024local,
  title={Local convergence rates of the nonparametric least squares estimator with applications to transfer learning},
  author={Schmidt-Hieber, Johannes and Zamolodtchikov, Petr},
  journal={Bernoulli},
  volume={30},
  number={3},
  pages={1845--1877},
  year={2024},
  publisher={Bernoulli Society for Mathematical Statistics and Probability}
}

@article{Chernozhukov_2014,
   title={Anti-concentration and honest, adaptive confidence bands},
   volume={42},
   number={5},
   journal={The Annals of Statistics},
   publisher={Institute of Mathematical Statistics},
   author={Chernozhukov, Victor and Chetverikov, Denis and Kato, Kengo},
   year={2014},
   month={October}
    }

@article{chernozhukov2014gaussian,
author = {Victor Chernozhukov and Denis Chetverikov and Kengo Kato},
title = {{Gaussian approximation of suprema of empirical processes}},
volume = {42},
journal = {The Annals of Statistics},
number = {4},
publisher = {Institute of Mathematical Statistics},
pages = {1564 -- 1597},
year = {2014}
}

@article{chernozhukov2017central,
  author = {Victor Chernozhukov and Denis Chetverikov and Kengo Kato},
title = {{Central limit theorems and bootstrap in high dimensions}},
volume = {45},
journal = {The Annals of Probability},
number = {4},
publisher = {Institute of Mathematical Statistics},
pages = {2309 -- 2352},
year = {2017}
}

@inproceedings{gaussianapproximationsubgaussian,
  title={Strong Gaussian approximation for the sum of random vectors},
  author={Buzun, Nazar and Shvetsov, Nikolay and Dylov, Dmitry V},
  booktitle={Conference on Learning Theory},
  pages={1693--1715},
  year={2022},
  organization={PMLR}
}

@article{gaussianapproximationbounded,
author = {Zaitsev, A. Yu.},
title = {Estimates of the L{\'e}vy--Prokhorov Distance in the Multivariate Central Limit Theorem for Random Variables with Finite Exponential Moments},
journal = {Theory of Probability \& Its Applications},
volume = {31},
number = {2},
pages = {203-220},
year = {1987}
}

@article{yang2017frequentistcoveragesupnormconvergence,
        title={Frequentist coverage and sup-norm convergence rate in Gaussian process regression},
  author={Yang, Yun and Bhattacharya, Anirban and Pati, Debdeep},
  journal={arXiv preprint arXiv:1708.04753},
  year={2017}
}

@article{shang2025bootstrap,
  title={Bootstrap Nonparametric Inference under Data Integration},
  author={Shang, Zuofeng and Sang, Peijun and Jin, Chong},
  journal={arXiv preprint arXiv:2501.01610},
  year={2025}
}

@article{low1997nonparametric,
  title={On nonparametric confidence intervals},
  author={Low, Mark G},
  journal={The Annals of Statistics},
  volume={25},
  number={6},
  pages={2547--2554},
  year={1997},
  publisher={Institute of Mathematical Statistics}
}

@article{genovese2008adaptive,
author = {Christopher Genovese and Larry Wasserman},
title = {{Adaptive confidence bands}},
volume = {36},
journal = {The Annals of Statistics},
number = {2},
publisher = {Institute of Mathematical Statistics},
pages = {875 -- 905},
year = {2008}
}

@article{gine2010confidence,
  author={Gin{\'e}, Evarist and Nickl, Richard},
  title = {{Confidence bands in density estimation}},
volume = {38},
journal = {The Annals of Statistics},
number = {2},
publisher = {Institute of Mathematical Statistics},
pages = {1122 -- 1170},
year = {2010},
}

@article{cai2006adaptive,
  author = {T. Tony Cai and Mark G. Low},
title = {{Adaptive confidence balls}},
volume = {34},
journal = {The Annals of Statistics},
number = {1},
publisher = {Institute of Mathematical Statistics},
pages = {202 -- 228},
year = {2006}
}

@article{juditsky2003nonparametric,
  title={Nonparametric confidence set estimation},
  author={Juditsky, Anatoli and Lambert-Lacroix, Sophie},
  journal={Mathematical Methods of Statistics},
  volume={12},
  number={4},
  pages={410--428},
  year={2003},
  publisher={New York: Allerton Press, c1992-}
}

@article{li1989honest,
  title={Honest confidence regions for nonparametric regression},
  author={Li, Ker-Chau},
  journal={The Annals of Statistics},
  volume={17},
  number={3},
  pages={1001--1008},
  year={1989},
  publisher={Institute of Mathematical Statistics}
}

@article{spokoiny1996adaptive,
  title={Adaptive hypothesis testing using wavelets},
  author={Spokoiny, Vladimir G},
  journal={The Annals of Statistics},
  volume={24},
  number={6},
  pages={2477--2498},
  year={1996},
  publisher={Institute of Mathematical Statistics}
}

@article{robins2006adaptive,
 author = {James Robins and Aad van der Vaart},
title = {{Adaptive nonparametric confidence sets}},
volume = {34},
journal = {The Annals of Statistics},
number = {1},
publisher = {Institute of Mathematical Statistics},
pages = {229 -- 253},
year = {2006}
}

@article{gotze2019large,
  title={Large ball probabilities, Gaussian comparison and anti-concentration},
  author={G{\"o}tze, Friedrich and Naumov, Alexey and Spokoiny, Vladimir and Ulyanov, Vladimir},
  journal={Bernoulli},
  volume={25},
  number={4A},
  pages={2538--2563},
  year={2019},
  publisher={JSTOR}
}

@article{sugiyama2007covariate,
  title={Covariate Shift Adaptation by Importance Weighted Cross Validation},
  author={Sugiyama, Masashi and Krauledat, Matthias and M{\"u}ller, Klaus-Robert},
  journal={Journal of Machine Learning Research},
  volume={8},
  pages={985--1005},
  year={2007}
}

@inproceedings{kouw2019robust,
  title={Robust importance-weighted cross-validation under sample selection bias},
  author={Kouw, Wouter M and Krijthe, Jesse H and Loog, Marco},
  booktitle={2019 IEEE 29th International Workshop on Machine Learning for Signal Processing (MLSP)},
  pages={1--6},
  year={2019},
  organization={IEEE}
}

@article{sugiyama2008direct,
  title={Direct Importance Estimation for Covariate Shift Adaptation},
  author={Sugiyama, Masashi and Suzuki, Taiji and Nakajima, Shinichi and Kashima, Hisashi and von B{\"u}nau, Paul and Kawanabe, Motoaki},
  journal={Annals of the Institute of Statistical Mathematics},
  volume={60},
  number={4},
  pages={699--746},
  year={2008}
}

@inproceedings{huang2007correcting,
  title={Correcting Sample Selection Bias by Unlabeled Data},
  author={Huang, Jiayuan and Smola, Alexander J. and Gretton, Arthur and Borgwardt, Karsten M. and Sch{\"o}lkopf, Bernhard},
  booktitle={Advances in Neural Information Processing Systems},
  volume={19},
  pages={601--608},
  year={2007}
}

@inproceedings{cortes2008sample,
  title={Sample Selection Bias Correction Theory},
  author={Cortes, Corinna and Mohri, Mehryar and Riley, Michael and Rostamizadeh, Afshin},
  booktitle={Algorithmic Learning Theory},
  series={Lecture Notes in Computer Science},
  volume={5254},
  pages={38--53},
  year={2008},
  publisher={Springer},
  doi={10.1007/978-3-540-87987-9_8}
}

@inproceedings{tibshirani2019conformal,
  title={Conformal Prediction Under Covariate Shift},
  author={Tibshirani, Ryan J. and Barber, Rina Foygel and Candes, Emmanuel J. and Ramdas, Aaditya},
  booktitle={Advances in Neural Information Processing Systems},
  volume={32},
  year={2019}
}

@article{barber2023conformal,
  title={Conformal Prediction Beyond Exchangeability},
  author={Barber, Rina Foygel and Candes, Emmanuel J. and Ramdas, Aaditya and Tibshirani, Ryan J.},
  journal={The Annals of Statistics},
  volume={51},
  number={2},
  pages={816--845},
  year={2023},
  doi={10.1214/23-AOS2276}
}

@article{zamolodtchikov2026minimax,
  title={A Minimax Theory of Nonparametric Regression Under Covariate Shift},
  author={Zamolodtchikov, Petr},
  journal={arXiv preprint arXiv:2603.05897},
  year={2026}
}

@article{monrad1991nearby,
  title={Nearby Variables with Nearby Conditional Laws and a Strong Approximation Theorem for Hilbert Space Valued Martingales},
  author={Monrad, Ditlev and Philipp, Walter},
  journal={Probability Theory and Related Fields},
  volume={88},
  number={3},
  pages={381--404},
  year={1991},
  doi={10.1007/BF01418867}
}

@article{crump2008nonparametric,
  title={Nonparametric Tests for Treatment Effect Heterogeneity},
  author={Crump, Richard K. and Hotz, V. Joseph and Imbens, Guido W. and Mitnik, Oscar A.},
  journal={Review of Economics and Statistics},
  volume={90},
  number={3},
  pages={389--405},
  year={2008}
}

@article{ding2016randomization,
  title={Randomization Inference for Treatment Effect Variation},
  author={Ding, Peng and Feller, Avi and Miratrix, Luke},
  journal={Journal of the Royal Statistical Society: Series B (Statistical Methodology)},
  volume={78},
  number={3},
  pages={655--671},
  year={2016}
}

@inproceedings{gururangan2020dont,
  title={Don't Stop Pretraining: Adapt Language Models to Domains and Tasks},
  author={Gururangan, Suchin and Marasovi{\'c}, Ana and Swayamdipta, Swabha and Lo, Kyle and Beltagy, Iz and Downey, Doug and Smith, Noah A.},
  booktitle={Proceedings of the 58th Annual Meeting of the Association for Computational Linguistics},
  pages={8342--8360},
  year={2020}
}

@article{pinelis1986remarks,
  title={Remarks on Inequalities for Large Deviation Probabilities},
  author={Pinelis, I. F. and Sakhanenko, A. I.},
  journal={Theory of Probability and Its Applications},
  volume={30},
  number={1},
  pages={143--148},
  year={1986},
  doi={10.1137/1130013}
}

@misc{federalreserve2023scf,
  author={{Board of Governors of the Federal Reserve System}},
  title={{2022 Survey of Consumer Finances}},
  year={2023},
  howpublished={Data set},
  url={https://www.federalreserve.gov/econres/scfindex.htm}
}
\bibliographystyle{plainnat}

\newpage
\appendix

\section{Extension I: sub-Gaussian Kernel Function}\label{app:ext-subG}

In addition to the bounded-kernel condition assumed throughout the paper, this section considers instead the sub-Gaussian feature regularity and derives the
corresponding Gaussian and bootstrap coupling bounds and their specific spectral-rate consequences.

\begin{definition}[Sub-Gaussian in Hilbert Space]\label{def:sub-gaussian-hilbert-space}
  For a square-integrable random element $X$ in a Hilbert space $\mathsf{G}$,
  with its expectation understood as a Bochner mean, we say $X$ is
  $\mathsf{G}$-$b$-sub-Gaussian if
  \begin{equation*}
    \log \mathbb{E}_{\PP}\left[\exp \left\langle X-\mathbb EX, t\right\rangle_{\mathsf{G}}\right]
    \leq \frac{b^{2}}{2}
    \mathbb{E}_{\PP}\left(\left\langle X-\mathbb EX, t\right\rangle_{\mathsf{G}}^{2}\right),
    \quad \forall t \in \mathsf{G}.
  \end{equation*}
\end{definition}

\begin{assumption}[Sub-Gaussian Feature]\label{assump:sub-gaussian-feature}
For some \(b_0\geq1\), the feature element \(k_X=k(X,\cdot)\) is
\(\mathsf H\)-\(b_0\)-sub-Gaussian. Equivalently, for every
\(t\in\mathsf H\),
\begin{equation*}
    \begin{aligned}
        \log \mathbb{E}_{\PP}\left[
        \exp \left\langle
        k_X-\mathbb{E}_{\PP}k_X,t
        \right\rangle_{\mathsf H}
        \right]
        &=
        \log \mathbb{E}_{\PP}\left[
        \exp\{t(X)-\mathbb{E}_{\PP}[t(X)]\}
        \right] \\
        &\leq
        \frac{b_0^{2}}{2}\mathbb{E}_{\PP}\left[
        \{t(X)-\mathbb{E}_{\PP}[t(X)]\}^{2}
        \right].
    \end{aligned}
\end{equation*}
For \(\tau\geq\tau_0\), let
\(w_\tau(x):=\min\{\rho(x),\tau\}\). Since importance weighting is used, we additionally assume that, for every
\(\tau\geq\tau_0\) and \(t\in\mathsf H\),
\begin{equation*}
    \begin{aligned}
        \log \mathbb{E}_{\PP}\left[
        \exp \left\langle
        w_\tau(X)k_X-\mathbb{E}_{\PP}[w_\tau(X)k_X],t
        \right\rangle_{\mathsf H}
        \right]
        &=
        \log \mathbb{E}_{\PP}\left[
        \exp\{w_\tau(X)t(X)-\mathbb{E}_{\PP}[w_\tau(X)t(X)]\}
        \right] \\
        &\leq
        \frac{b_0^{2}}{2}\mathbb{E}_{\PP}\left[
        w_\tau(X)^2t(X)^2
        \right].
    \end{aligned}
\end{equation*}
The weighted-feature condition is imposed on the truncation family and does
not follow from sub-Gaussianity of \(k_X\) alone.
\end{assumption}

We first verify that the influence elements \(U_i\) in the Bahadur
representation~\eqref{eq:bahadur} satisfy the required sub-Gaussian condition.

\begin{lemma}\label{lem:sub-gaussianity-of-ui}
  Under Assumptions~\ref{assump:error}, \ref{assump:source-condition}, and
  \ref{assump:sub-gaussian-feature}, if \(\tau_n\geq\tau_0\), then \(U_i\) is
  \(\splq\)-\(b\)-sub-Gaussian with
  \begin{equation*}
    b = C_{sg}^3 \cdot \frac{(\kappa \|f_*\| + \overline{\sigma})}{\underline{\sigma}} b_0,
\end{equation*}
where $C_{sg}$ is a universal constant as in {\protect\cite{vershynin2018high}}.
\end{lemma}

The sub-Gaussian condition yields the following new Gaussian and bootstrap
coupling bounds, whereas the Bahadur representation and feasible-bootstrap
approximation remain as in the main analysis.

\begin{proposition}[Gaussian coupling under sub-Gaussian regularity]
\label{prop:gaussian-coupling-subgaussian}
Suppose \(U_1,\ldots,U_n\) are i.i.d.,
centered, $\splq$-$b$-sub-Gaussian random elements with covariance operator
\(\opv\). For any \(\eta\in(0,1)\) and any integer
\(m\in\{1,\ldots,n\}\), if
\(n \ge \max\{2,\log(2/\eta)\}\), then there exists
\(Z\sim\mathrm{N}_{\splq}(0,\opv)\) such that, with probability at least
\(1-\eta\),
\[ \left\| \left( \frac{1}{\sqrt{n}} \sum_{i=1}^n U_i \right) - Z \right\|_{\splq} \lesssim b \sigma_{\splq}(m, \opv) \sqrt{\log\!\left(\frac{6}{\eta}\right)} + \| \opv \|_{\splq}^{\frac{1}{2}} b^3 \left\{ \frac{m^{3/2} \log(n)^2}{\sqrt{n}} \right\} (3/\eta)^{1/\log(mn)}. \]
\end{proposition}

\begin{proposition}[Bootstrap coupling under sub-Gaussian regularity]
\label{prop:bootstrap-coupling-subgaussian}
Suppose \(U_1,\ldots,U_n\) are i.i.d.,
centered, \(\splq\)-\(b\)-sub-Gaussian random elements. Let \(g\) be the
isonormal Gaussian process in \(\splq\). Then there exists a random element
\(G\stackrel{{\rm d.}}{=}\what{\opv}_{\oracle}^{1/2}g\big|\calD\) such that,
conditioned on \(\calD\), with total probability at least \(1-6\eta\),
\begin{equation*}
  \left\| {\opv}^{1/2} g - G \right\|_{\splq}
  \lesssim
  \inf_{m \geq 1} \left[
    m^{1/4} \left\{
      \frac{b^4\sigma_{\splq}^4(0, \opv)}{n}
    \right\}^{1/4}
    + \sigma_{\splq}(m, \opv)
  \right]
  \left(\log\!\left(\frac{2}{\eta}\right)\right)^{3/2}.
\end{equation*}
where the covariance operator
\(\opv=\mathbb{E}[U_i\otimes_{\splq}U_i]\).
\end{proposition}

Combining these two coupling bounds with the unchanged bias, Bahadur,
feasible-bootstrap, and anti-concentration arguments from the main analysis
yields sub-Gaussian counterparts of
Corollaries~\ref{cor:bd-density}--\ref{cor:moment-density}.  Throughout the
following corollaries, fix \(\alpha\in(0,1)\), take
\(\eta\asymp n^{-1}\), and let \(n\to\infty\). Suppose
Assumptions~\ref{assump:error}, \ref{assump:source-condition},
\ref{assump:pseudo-dimn}, \ref{assump:lower-boundedness-of-covariance-operator},
and~\ref{assump:sub-gaussian-feature} hold, with source exponent
\(s\geq1/2\) and \(\tau_n\geq\tau_0\) for all sufficiently large \(n\).
In each corollary, the stated density-ratio and truncation conditions verify
Assumption~\ref{assump:covariate-shift-estimation} with
\(\mathfrak C_1\) bounded uniformly in \(n\). All other constants in these
assumptions are independent of \(n\).

\begin{corollary}[bounded density ratio under sub-Gaussian regularity]
\label{cor:bd-density-subgaussian}
Suppose that \(\rho(X)\leq\overline\rho\) \(\mathbb P\)-almost surely, and let
\(\tau_n=C_0\) for a sufficiently large fixed constant
\(C_0\geq\max\{\overline\rho,\tau_0\}\).
\begin{enumerate}
\item[(1)] If \(\mu_k\asymp k^{-\beta}\) for some \(\beta>3\) and all
\(k\geq1\), then, by choosing
\[
\lambda_n\asymp n^{-\frac{2\beta}{3\beta+2}},
\]
we obtain
\[
\big|\PP\big(f_{\ast}\notin\cs_{\alpha}(\what f)\big)-\alpha\big|
\lesssim
n^{-\frac{\beta-3}{2(3\beta+2)}}\log^3n.
\]
\item[(2)] If \(\mu_k\asymp\exp(-c k^\zeta)\) for some \(c,\zeta>0\) and
all \(k\geq1\), then, by choosing \(\lambda_n\asymp n^{-2/3}\), we obtain
\[
\big|\PP\big(f_{\ast}\notin\cs_{\alpha}(\what f)\big)-\alpha\big|
\lesssim
n^{-1/6}(\log n)^{3+1/\zeta}.
\]
\end{enumerate}
\end{corollary}

\begin{corollary}[sub-exponential density ratio under sub-Gaussian regularity]
\label{cor:subexp-density-subgaussian}
Suppose that the density ratio satisfies the sub-exponential tail condition
in Definition~\ref{def:bernstein-moment-condition-for-weight}.
\begin{enumerate}
\item[(1)] If \(\mu_k\asymp k^{-\beta}\) for some \(\beta>3\) and all
\(k\geq1\), then, by choosing
\[
\tau_n\asymp
n^{\frac{\beta-3}{6(3\beta+2)}},
\qquad
\lambda_n\asymp n^{-\frac{2\beta}{3\beta+2}},
\]
we obtain
\[
\big|\PP\big(f_{\ast}\notin\cs_{\alpha}(\what f)\big)-\alpha\big|
\lesssim
n^{-\frac{\beta-3}{4(3\beta+2)}}\log^3n.
\]
\item[(2)] If \(\mu_k\asymp\exp(-c k^\zeta)\) for some \(c,\zeta>0\) and
all \(k\geq1\), then, by choosing
\[
\tau_n\asymp n^{1/150},
\qquad
\lambda_n\asymp n^{-2/3},
\]
we obtain
\[
\big|\PP\big(f_{\ast}\notin\cs_{\alpha}(\what f)\big)-\alpha\big|
\lesssim
n^{-47/300}(\log n)^{3+1/\zeta}.
\]
\end{enumerate}
\end{corollary}

\begin{corollary}[bounded \(\theta\)-th moment under sub-Gaussian regularity]
\label{cor:moment-density-subgaussian}
Suppose that
\(\EE_{\PP}[\rho(X)^\theta]\leq\Omega_\theta<\infty\) for some
\(\theta>5\). If \(\mu_k\asymp\exp(-c k^\zeta)\) for some
\(c,\zeta>0\) and all \(k\geq1\), then, by choosing
\[
\lambda_n\asymp
n^{-\frac{8(\theta^2-1)}{3(4\theta^2+3\theta-3)}},
\qquad
\tau_n\asymp
n^{\frac{19\theta+1}{6(4\theta^2+3\theta-3)}},
\]
we obtain
\[
\big|\PP\big(f_{\ast}\notin\cs_{\alpha}(\what f)\big)-\alpha\big|
\lesssim
n^{-\frac{(\theta-5)(8\theta+1)}
{12(4\theta^2+3\theta-3)}}
(\log n)^{3+1/\zeta}.
\]
\end{corollary}

\section{Extension II: Bias-Corrected Multiplier Bootstrap}\label{app:ext-bCboot}

The asymptotic validity of the confidence ball constructed by Algorithm~\ref{alg:bootstrap-confidence-set} relies on undersmoothing to render the deterministic bias
\[
  \sqrt n(\fpw-f_*)
  =
  -\sqrt n\,\lambda_n(\oplpw+\lambda_n\opid)^{-1}f_*
\]
negligible in \(\splq\). The procedure developed here instead estimates the corresponding bias-correction term by replacing \((\oplpw,f_*)\) with \((\oplpnw,\what f)\), yielding, with \(\what\alpha:=(\mak+n\lambda_n\maw^{-1})^{-1}\vy\),
\begin{equation}
  \label{eq:bias-correction-empirical}
  \what b
  :=
  \sqrt n\,\lambda_n(\oplpnw+\lambda_n\opid)^{-1}\what f
  =
  \sqrt n\,\vkx^{\top}(\mak+n\lambda_n\maw^{-1})^{-1}
  n\lambda_n\maw^{-1}\hat\alpha .
\end{equation}
As \(\what b\) estimates the negative of \(\sqrt n(\fpw-f_*)\), the bias-adjusted analogue of \(\sqrt n(\what f-f_*)\) is \(B-\what b\). With \(\what d:=n\lambda_n\maw^{-1}\what\alpha\), equations~\eqref{eq:B-def} and \eqref{eq:bias-correction-empirical} give
\[
  B-\what b
  =
  \sqrt n\,\vkx^{\top}(\mak+n\lambda_n\maw^{-1})^{-1}
  \left\{
    \diag(\what\beps)\vz-\what d
  \right\}.
\]

Algorithm~\ref{alg:bias-corrected-bootstrap-confidence-set} summarizes this practical bias-corrected variant of Algorithm~\ref{alg:bootstrap-confidence-set}, which may improve finite-sample calibration accuracy.    

\begin{algorithm}[htbp]
\caption{Bias-corrected multiplier bootstrap for \(\splq\)-norm confidence balls}
\label{alg:bias-corrected-bootstrap-confidence-set}
\begin{algorithmic}[1]
\REQUIRE Data \(\calD=\{(X_i,Y_i)\}_{i=1}^n\), kernel \(k(\cdot,\cdot)\), density ratio \(\rho(x)=q(x)/p(x)\), weight \(w_n(x)=\min\{\rho(x),\tau_n\}\), confidence level \(1-\alpha\), bootstrap repetitions \(M\), regularization parameter \(\lambda_n\);
\STATE Form \(\mak\), \(\maw\), \(\vy\), \(\vkx\), and \(\mak_{\QQ}\) as in Algorithm~\ref{alg:bootstrap-confidence-set}.
\STATE Compute \(\what\alpha=(\mak+n\lambda_n\maw^{-1})^{-1}\vy\) and \(\what f=\vkx^{\top}\what\alpha\).
\STATE Compute the residual vector \(\what\beps=(Y_1-\what f(X_1),\dots,Y_n-\what f(X_n))^{\top}\).
\STATE Form the dual bias-correction vector \(\what d=n\lambda_n\maw^{-1}\what\alpha\).
\FOR{$b = 1,\dots,M$}
  \STATE Sample \(\vz^{(b)}\sim N(0,\maid_n-\vone\vone^{\top}/n)\).
  \STATE Set the bias-corrected coefficient vector
  \[
    \va_{\mathrm{bc}}^{(b)}
    =
    \sqrt n\,(\mak+n\lambda_n\maw^{-1})^{-1}
    \{\diag(\what\beps)\vz^{(b)}-\what d\}.
  \]
  \STATE Compute the bootstrap statistic
  \[
    T_{\mathrm{bc}}^{(b)}
    =
    \norm{\vkx^{\top}\va_{\mathrm{bc}}^{(b)}}_{\splq}
    =
    \sqrt{(\va_{\mathrm{bc}}^{(b)})^{\top}
    \mak_{\QQ}\va_{\mathrm{bc}}^{(b)}} .
  \]
\ENDFOR
\STATE Let \(\hat r_{\alpha}^{\mathrm{bc}}\) be the empirical \((1-\alpha)100\%\)-percentile of \(\{T_{\mathrm{bc}}^{(b)}\}_{b=1}^M\).
\STATE Define \( \cs_{\alpha}^{\mathrm{bc}}(\what f)=\{f\in\mathsf{H}:\sqrt n\,\norm{\what f-f}_{\splq}\le\hat r_{\alpha}^{\mathrm{bc}}\}\).
\ENSURE the IWKRR \(\what f\) and the bias-corrected critical value \(\hat r_{\alpha}^{\mathrm{bc}}\).
\end{algorithmic}
\end{algorithm}

\section{Extension III: Residual Bootstrap}\label{app:ext-residual-bootstrap}

The residual bootstrap resamples with replacement from the empirical
distribution of the centered fitted residuals.
Conditional on the observed covariates, it generates bootstrap responses
\(Y_{i,\mathrm{res}}^{(b)}=\what f(X_i)+
\what\eps_{i,\mathrm{res}}^{(b)}\), \(i=1,\dots,n\), and refits IWKRR with
the same \(w_n(x)\) and \(\lambda_n\).

\begin{algorithm}[htbp]
\caption{Residual bootstrap for $\splq$-norm confidence balls}
\label{alg:residual-bootstrap-confidence-set}
\begin{algorithmic}[1]
\REQUIRE Data \(\calD=\{(X_i,Y_i)\}_{i=1}^n\), kernel \(k(\cdot,\cdot)\), density ratio \(\rho(x)=q(x)/p(x)\), weight \(w_n(x)=\min\{\rho(x),\tau_n\}\), confidence level \(1-\alpha\), bootstrap repetitions \(M\), regularization parameter \(\lambda_n\).
\STATE Form \(\mak\), \(\maw\), \(\vy\), \(\vkx\), and \(\mak_{\QQ}\) as in Algorithm~\ref{alg:bootstrap-confidence-set}.
\STATE Compute \(\what\alpha=(\mak+n\lambda_n\maw^{-1})^{-1}\vy\) and \(\what f=\vkx^{\top}\what\alpha\).
\STATE Compute the residual vector \(\what\beps=(Y_1-\what f(X_1),\dots,Y_n-\what f(X_n))^{\top}\).
\STATE Center the fitted residuals by setting
\[
  \overline{\what\eps}
  :=
  \frac{1}{n}\sum_{j=1}^n\what\eps_j,
  \qquad
  \what\eps_j^\circ
  :=
  \what\eps_j-\overline{\what\eps}.
\]
\FOR{\(b = 1,\dots,M\)}
  \STATE Sample
  \(\what\eps_{1,\mathrm{res}}^{(b)},\dots,
  \what\eps_{n,\mathrm{res}}^{(b)}\) independently with replacement from
  \(\{\what\eps_1^\circ,\dots,\what\eps_n^\circ\}\).
  \STATE Form the bootstrap-response vector
  \[
    \vy_{\mathrm{res}}^{(b)}
    :=
    \big(
      \what f(X_1)+\what\eps_{1,\mathrm{res}}^{(b)},\dots,
      \what f(X_n)+\what\eps_{n,\mathrm{res}}^{(b)}
    \big)^{\top}.
  \]
  \STATE Refit IWKRR with the same regularized inverse:
  \[
    \what f_{\mathrm{res}}^{(b)}
    :=
    \vkx^{\top}
    (\mak+n\lambda_n\maw^{-1})^{-1}
    \vy_{\mathrm{res}}^{(b)}.
  \]
  \STATE Define the coefficient perturbation
  \[
    \va_{\mathrm{res}}^{(b)}
    :=
    (\mak+n\lambda_n\maw^{-1})^{-1}
    \big(\vy_{\mathrm{res}}^{(b)}-\vy\big),
    \qquad
    \what f_{\mathrm{res}}^{(b)}-\what f
    =
    \vkx^{\top}\va_{\mathrm{res}}^{(b)}.
  \]
  \STATE Compute the bootstrap statistic
  \[
    T_{\mathrm{res}}^{(b)}
    :=
    \norm{\what f_{\mathrm{res}}^{(b)}-\what f}_{\splq}
    =
    \norm{\vkx^{\top}\va_{\mathrm{res}}^{(b)}}_{\splq}
    =
    \left\{
      (\va_{\mathrm{res}}^{(b)})^{\top}
      \mak_{\QQ}\va_{\mathrm{res}}^{(b)}
    \right\}^{1/2}.
  \]
\ENDFOR
\STATE Let \(\hat c_{\alpha,n}^{\mathrm{res}}\) be the empirical
\((1-\alpha)100\%\)-percentile of
\(\{T_{\mathrm{res}}^{(b)}\}_{b=1}^M\).
\STATE Define \(\cs_{\alpha}^{\mathrm{res}}(\what f)\) by replacing
\(\hat c_{\alpha,n}\) with \(\hat c_{\alpha,n}^{\mathrm{res}}\) in
\eqref{eq:feasible-confidence-ball}.
\ENSURE The residual-bootstrap confidence ball
\(\cs_{\alpha}^{\mathrm{res}}(\what f)\).
\end{algorithmic}
\end{algorithm}

{\it Step 1: Gaussian and bootstrap approximation.} The residual-bootstrap coverage guarantee can be established using the same
proof framework as the multiplier bootstrap. Because the point
estimator is unchanged, it suffices to replace the bootstrap approximation in
\eqref{eq:theory-bootstrap-approximation} by its residual-bootstrap analogue: conditionally on \(\calD\),
\(\sqrt n(\what f_{\mathrm{res}}^{(b)}-\what f)\) can be approximated by the
same Gaussian element as in \eqref{eq:theory-gaussian-approximation}. If the
corresponding error is
\(\Delta_{B,\mathrm{res}}\), the proof of
Theorem~\ref{thm:main-coverage} yields the same
coverage bound with \(\Delta_B\) replaced by \(\Delta_{B,\mathrm{res}}\).

To identify the required bootstrap approximation, let
\[
  \psi_n(x)
  :=
  w_n(x)(\oplpw+\lambda_n\opid)^{-1}k_x,
  \qquad
  \what\psi_{i,n}
  :=
  w_n(X_i)(\oplpnw+\lambda_n\opid)^{-1}k_{X_i}.
\]
The leading noise component of the influence function is
\[
  U_{i,\eps}
  :=
  \eps_i\psi_n(X_i).
\]
The normal equations yield the exact decompositions
\begin{align}
  \what f-f_*
  &=
  \frac{1}{n}\sum_{i=1}^n
  \eps_i\what\psi_{i,n}
  -
  \lambda_n(\oplpnw+\lambda_n\opid)^{-1}f_*,
  \\
  \what f_{\mathrm{res}}^{(b)}-\what f
  &=
  \frac{1}{n}\sum_{i=1}^n
  \what\eps_{i,\mathrm{res}}^{(b)}\what\psi_{i,n}
  -
  \lambda_n(\oplpnw+\lambda_n\opid)^{-1}\what f .
\end{align}
The same Bahadur, undersmoothing, and operator-approximation methods used in
the main analysis can be applied to the exact decompositions above to obtain
the first-order representations
\begin{align*}
  \sqrt n(\what f-f_*)
  &\approx_{\splq}
  \frac{1}{\sqrt n}\sum_{i=1}^n
  \eps_i\psi_n(X_i),
  \\
  \sqrt n(\what f_{\mathrm{res}}^{(b)}-\what f)
  &\approx_{\splq}
  \frac{1}{\sqrt n}\sum_{i=1}^n
  \what\eps_{i,\mathrm{res}}^{(b)}\what\psi_{i,n},
\end{align*}
where we ignore the bias $f_{P^w}-f_{\ast}$ for simplicity. 
Starting from these representations, the bootstrap-approximation and coverage
arguments in the main text can be adapted to the residual bootstrap, subject
to the covariance-consistency condition below.

{\it Step 2: covariance comparison.}
The population covariance operator of
\(\frac{1}{\sqrt n}\sum_{i=1}^n\eps_i\psi_n(X_i)\) is
\begin{equation}
  \label{eq:residual-bootstrap-target-covariance}
  \opv_{\eps,n}
  :=
  \mathbb E\!\left[
    U_{i,\eps}\otimes_{\splq}U_{i,\eps}
  \right]
  =
  \mathbb E\!\left[
    \mathbb E[\eps^2\mid X]\,
    \psi_n(X)\otimes_{\splq}\psi_n(X)
  \right].
\end{equation}
Conditional on \(\calD\), the covariance operator of the leading term
generated by Algorithm~\ref{alg:residual-bootstrap-confidence-set} is
\begin{equation}
  \label{eq:residual-bootstrap-covariance}
  \what{\opv}_{\mathrm{res}}
  :=
  \what\sigma_{\mathrm{res}}^2\,
  \mathbb E_n\!\left[
    \what\psi_{i,n}\otimes_{\splq}\what\psi_{i,n}
  \right],
  \qquad
  \what\sigma_{\mathrm{res}}^2
  :=
  \mathbb E_n\!\left[(\what\eps_i^\circ)^2\right].
\end{equation}
Under conditional homoskedasticity,
\(\mathbb E[\eps^2\mid X]=\mathbb E[\eps^2]\) almost surely,
\eqref{eq:residual-bootstrap-target-covariance} factorizes as
\(\opv_{\eps,n}=\mathbb E[\eps^2]\,
\mathbb E[\psi_n(X)\otimes_{\splq}\psi_n(X)]\).
Write
\[
  \what{\mathcal A}_n
  :=
  \mathbb E_n\!\left[
    \what\psi_{i,n}\otimes_{\splq}\what\psi_{i,n}
  \right],
  \qquad
  \mathcal A_n
  :=
  \mathbb E\!\left[
    \psi_n(X)\otimes_{\splq}\psi_n(X)
  \right].
\]
Covariance consistency is therefore ensured by
\[
  \left|\what\sigma_{\mathrm{res}}^2-\mathbb E[\eps^2]\right|
  \left\|\what{\mathcal A}_n\right\|_{\mathrm{HS},\splq}
  +
  \mathbb E[\eps^2]
  \left\|\what{\mathcal A}_n-\mathcal A_n\right\|_{\mathrm{HS},\splq}
  \xrightarrow{\PP}0,\quad \textrm{as } n\to\infty. 
\]

{\it Step 3: comparison with the multiplier bootstrap.}
Notice that \(\what V_i=\what\eps_i\what\psi_{i,n}\). The conditional
covariance operators of the multiplier and residual bootstraps are,
respectively,
\[
  \what{\opv}
  =
  \mathbb E_n\!\left[
    \what\eps_i^2
    \bigl(
      \what\psi_{i,n}\otimes_{\splq}\what\psi_{i,n}
    \bigr)
  \right]
  -
  \mathbb E_n[\what V_i]
  \otimes_{\splq}
  \mathbb E_n[\what V_i],
  \qquad
  \what{\opv}_{\mathrm{res}}
  =
  \what\sigma_{\mathrm{res}}^2
  \mathbb E_n[
    \what\psi_{i,n}\otimes_{\splq}\what\psi_{i,n}
  ].
\]
Thus, the multiplier bootstrap weights each
\(\what\psi_{i,n}\otimes_{\splq}\what\psi_{i,n}\) by its paired squared
residual \(\what\eps_i^2\), whereas the residual bootstrap uses the common
factor \(\what\sigma_{\mathrm{res}}^2\) for all \(i\). The multiplier
bootstrap can therefore accommodate conditional heteroskedasticity, whereas
the residual bootstrap requires the covariance-matching condition above.

\section{Extension IV: mis-coverage rate for heavy-tailed density ratio}\label{app:moment-density-poly}

For ease of exposition,  the mis-coverage rate is presented in Corollary~\ref{cor:moment-density} only for the case when $\calL_Q$ has exponential-decay eigenvalues. The following corollary showcases the mis-coverage rate if $\calL_Q$ has polynomial-decay eigenvalues. 

\begin{corollary}[Polynomial spectral decay under a bounded
\(\theta\)-th moment]
\label{cor:moment-density-polynomial}
Under the assumptions of
Section~\ref{subsec:spectral-rate-consequences}, suppose that
\(\EE_{\PP}[\rho(X)^\theta]\leq\Omega_\theta<\infty\) for some
\(\theta>2\).  If \(\mu_k\asymp k^{-\beta}\) for some \(\beta>5\) and all
\(k\geq1\), assume that
\[
\theta-1>
\frac{4(\beta+1)^2(2\beta+5)}
{\beta(\beta-5)(\beta+10)}.
\]
Take \(d=1/\beta\) in Assumption~\ref{assump:pseudo-dimn} and choose
\[
 \delta=
 \frac{\beta(\theta-1)(\beta-5)(\beta+10)
 -4(\beta+1)^2(2\beta+5)}
 {3(2\beta+5)(\beta+10)
 \{3\beta(\theta-1)(\beta+3)+2(\beta+1)^2\}}.
\]
By choosing
\[
\begin{aligned}
 \lambda_n
 &\asymp
 n^{-\frac{2\beta(\theta-1)(\beta+1)}
 {3\beta(\theta-1)(\beta+3)+2(\beta+1)^2}},\\
 \tau_n
 &\asymp
 n^{\frac{\beta(\theta-1)(\beta-5)(\beta+10)
 +2(3\beta+28)(\beta+1)^2(2\beta+5)}
 {3(2\beta+5)(\beta+10)
 \{3\beta(\theta-1)(\beta+3)+2(\beta+1)^2\}}},
\end{aligned}
\]
we obtain
\[
\big|\PP\big(f_{\ast}\notin \cs_{\alpha}(\what f)\big)-\alpha\big|
\lesssim
n^{-\frac{(\beta+9)
\{\beta(\theta-1)(\beta-5)(\beta+10)
-4(\beta+1)^2(2\beta+5)\}}
{2(\beta+10)(2\beta+5)
\{3\beta(\theta-1)(\beta+3)+2(\beta+1)^2\}}}
\log^3 n.
\]
In particular, the displayed rate converges to zero under any of the
following sufficient regimes:
\[
 \theta\geq9,\ \beta\geq113;
 \qquad
 \theta>9,\ \beta>5+\frac{144}{5(\theta-9)};
 \qquad
 \theta\geq10,\ \beta\geq20.
\]
Another sufficient regime is \(\theta\geq11\) and \(\beta\geq15\).
\end{corollary}

\section{Additional Results of SCF Data Analysis}
\label{app:scf-additional-applications}

This appendix provides additional data, implementation, and diagnostic details
for Studies~1 and~2 reported in Section~\ref{sec:scf-application}, and presents additional results from 
Studies~3 and~4.

\subsection{Data description.}
\label{app:scf-data-description}

Studies~1--3 use the 2022 implicate-1 sample described in
Section~\ref{subsec:scf-data-design}. Study~4 additionally uses implicate~1 of
the 2019 public summary extract, comprising \(5{,}777\) distinct families. 

To illustrate how the released weights change covariate distribution,
we partition the \(33\) scalar coordinates into a collection \(\mathcal B\) of
\(12\) demographic and household covariate blocks.  For a coordinate \(Z\)
observed on the \(n\) records, define the unweighted moments
\(\widehat\mu_P(Z):=n^{-1}\sum_{i=1}^{n}Z_i\) and
\(\widehat\sigma_P^2(Z):=n^{-1}\sum_{i=1}^{n}
\{Z_i-\widehat\mu_P(Z)\}^2\), and the released-weight moments
\(\widehat\mu_Q(Z):=n^{-1}\sum_{i=1}^{n}\rho_iZ_i\) and
\(\widehat\sigma_Q^2(Z):=n^{-1}\sum_{i=1}^{n}\rho_i
\{Z_i-\widehat\mu_Q(Z)\}^2\).  The coordinate-level signed standardized
difference and the score of each block \(b\in\mathcal B\) are
\[
D(Z)
:=
\frac{\widehat\mu_Q(Z)-\widehat\mu_P(Z)}
{\{[\widehat\sigma_P^2(Z)+\widehat\sigma_Q^2(Z)]/2\}^{1/2}},
\qquad
S_b:=\max_{Z\in b}\lvert D(Z)\rvert.
\]
We rank the \(12\) block scores \(S_b\), and
Figure~\ref{fig:scf-weighting-overview} displays the six highest-ranked blocks.

\begin{figure}[htbp]
\centering
\IfFileExists{applications/scf_household_finance/figures/scf_weighting_overview.pdf}{%
    \includegraphics[width=0.88\textwidth]{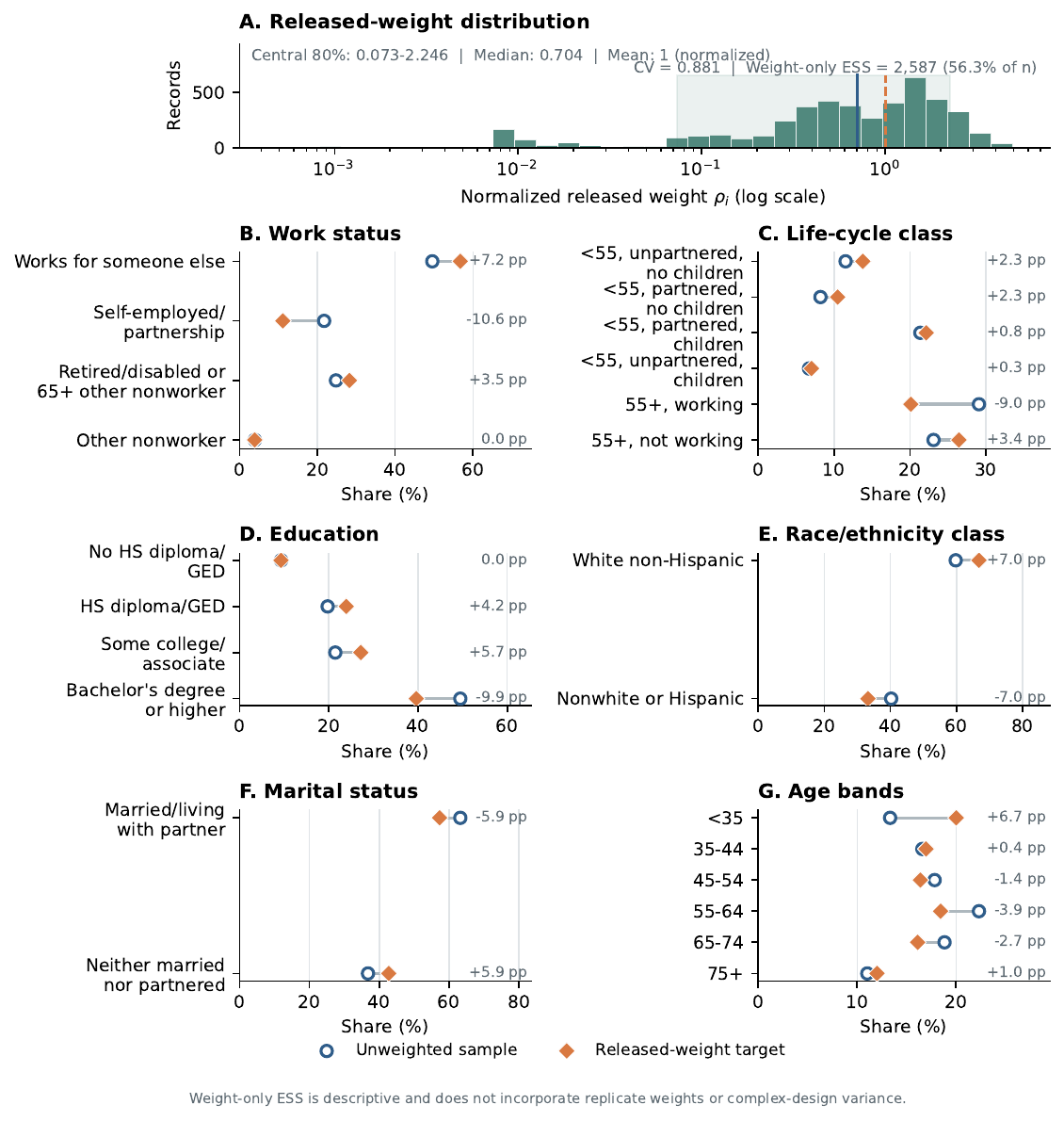}%
}{%
    \fbox{\parbox{0.85\textwidth}{\centering
    The released-weight diagnostic figure is unavailable.}}%
}
\caption{Distribution of survey weights and their impact in 2022
implicate-1 records.  Panel A summarizes the normalized weights; the shaded
band is the 10--90\% percentile range, and the vertical lines mark the
median and normalized mean.  Panels B--G compare unweighted and
weighted shares for the six blocks with the largest \(S_b\); labels
are target-minus-sample percentage-point differences. }
\label{fig:scf-weighting-overview}
\end{figure}

\subsection{Study 1: predicting household net worth}
\label{app:scf-study1-details}

{\it Variables and preprocessing.}
Table~\ref{tab:scf-study1-variables} describes the variables used in the
\(33\)-dimensional design summarized in
Section~\ref{subsec:scf-frozen-benchmark}. The work-status block is given by
the four-level public-file variable \texttt{OCCAT1}.

\begin{table}[htbp]
\centering
\small
\caption{Additional details of variables used for Study~1.}
\label{tab:scf-study1-variables}
\begin{tabular}{>{\raggedright\arraybackslash}m{0.15\textwidth}
                p{0.13\textwidth}
                p{0.27\textwidth}
                >{\raggedright\arraybackslash}m{0.34\textwidth}}
\toprule
Object & SCF variable & Meaning & Processing \\
\midrule
Response
& \texttt{NETWORTH}
& Total family net worth
& \(\operatorname{asinh}(\texttt{NETWORTH}/100000)\), truncated to
  zero  for negative net worth. \tabularnewline
\midrule
Target weight
& \texttt{WGT}
& \# of U.S. families represented 
&  normalized to mean one \tabularnewline
\midrule
\multirow{2}{0.15\textwidth}{Continuous coordinates}
& \texttt{AGE}
& Age of the reference person;
& \multirow{2}{0.34\textwidth}{standardized.} \tabularnewline
& \texttt{KIDS}
& Number of children
& \tabularnewline
\midrule
Ordered block
& \texttt{EDCL}
& Four levels of education 
&
  \(\mathbf 1\{\texttt{EDCL}\geq k\}\),
  \(k=2,3,4\). \tabularnewline
\midrule
\multirow{9}{0.15\textwidth}[0\baselineskip]{Unordered blocks}
& \texttt{HHSEX}
& Sex of the reference person;
& \multirow{9}{0.34\textwidth}[0\baselineskip]{one-hot encodings yielding \(28\)
  coordinates,
  three education indicators, and two continuous coordinates; total number: \(33\).} \tabularnewline
& \texttt{MARRIED}
& Married or not;
& \tabularnewline
& \texttt{LF}
& Working versus not working;
& \tabularnewline
& \texttt{LIFECL}
& Six categories of life cycle;
& \tabularnewline
& \texttt{FAMSTRUCT}
& Five-class family structure;
& \tabularnewline
& \texttt{RACECL}
& (non)White, (non)Hispanic;
& \tabularnewline
& \texttt{OCCAT1}
& Four-category work status;
& \tabularnewline
& \texttt{INDCAT}
& Three industry classes;
& \tabularnewline
& \texttt{HHOUSES}
& own residence or not
& \tabularnewline
\bottomrule
\end{tabular}
\end{table}

{\it KRR and parameters tuning on \(\mathcal D_0\).}
Within the \(980/420\) inner split of \(\mathcal D_0\), the \(980\) training
families determine the continuous-variable standardization constants and
categorical encoding.  This transformation is applied directly 
to the \(420\) validation families. Writing \(n_{\mathrm{tr}}=980\), the
ordinary KRR search uses
\[
\lambda_n(a)=a n_{\mathrm{tr}}^{-0.6},
\qquad
a\in\mathcal A_\lambda
:=\{0.0625,0.125,0.25,0.5,1,1.5,2,3,4\},
\]
 and are compared by unweighted
validation mean squared error.  The selected multiplier \(a=0.0625\) is evaluated
at \(n=\lvert\mathcal D_0\rvert=1{,}400\), giving
\(\lambda_{\mathrm{KRR}}=0.000809475\).  

{\it IWKRR  and multiplier bootstrap on \(\mathcal D_I\).}
The same \(980/420\) split of \(\mathcal D_0\) is used to tune the IWKRR parameters.  Writing
\(n_{\mathrm{tr}}=980\), the joint search uses
\[
\lambda_n(a)=a n_{\mathrm{tr}}^{-0.6},
\quad a\in\mathcal A_\lambda,
\qquad
\tau_n(b)=b n_{\mathrm{tr}}^{0.06},
\quad
b\in\mathcal A_\tau:=\{0.25,0.5,1,1.5,2\},
\]
together with \(\tau=10\) as the truncation cap.   After validation, the selected multipliers are fixed at \(a=0.25\) and
\(b=1\).  Substituting \(n=\lvert\mathcal D_I\rvert=1{,}800\) into the two
sequences gives
\(\lambda_{\mathrm{IW}}=0.002784688\) and \(\tau_n=1.567893\).   For the final IWKRR fit, the weights
are
\[
w_i=\rho_i\wedge\tau_n,
\qquad
\bar w_i
=\frac{\lvert\mathcal D_I\rvert w_i}
{\sum_{\ell\in\mathcal D_I}w_\ell},
\]
and the final inference-sample fit solves
\[
\widehat f_{\mathrm{IW}}
=\arg\min_{f\in\sfH}
\left[
  \frac{1}{\lvert\mathcal D_I\rvert}
  \sum_{i\in\mathcal D_I}
  \bar w_i\{Y_i-f(X_i)\}^2
  +\lambda_{\mathrm{IW}}\norm{f}_{\sfH}^{2}
\right].
\]
The multiplier bootstrap is then carried out using all \(1{,}800\) records in
\(\mathcal D_I\).

{\it Confidence set construction and \(p\)-values.}
The confidence-set radius in Section~\ref{subsec:scf-frozen-benchmark}
corresponds to the \(90\%\) critical value \(4.942249\) on the
\(\sqrt{1{,}800}\) scale.  Let
\(R_b\) denote the target-norm bootstrap error in draw \(b\).  For a candidate
\(f_0\), the corresponding plus-one bootstrap \(p\)-value is
\[
\widehat p(f_0)
=
\frac{1+\sum_{b=1}^{B}
\mathbf 1\!\left\{
R_b\geq\norm{\widehat f_{\mathrm{IW}}-f_0}_{\widehat{\QQ}}
\right\}}{B+1}.
\]
All bootstrap statistics are smaller than the four observed
distances, so \(\widehat p(f_0)=1/(B+1)\approx0.0002\) for all the four proposed candidate $f_0$.

{\it Response distribution.}
Figure~\ref{fig:scf-networth-overview} compares the empirical distribution of household net worth with and without re-weighting. 

\begin{figure}[htbp]
\centering
\IfFileExists{applications/scf_household_finance/figures/scf_networth_overview.pdf}{%
    \includegraphics[width=0.72\textwidth]{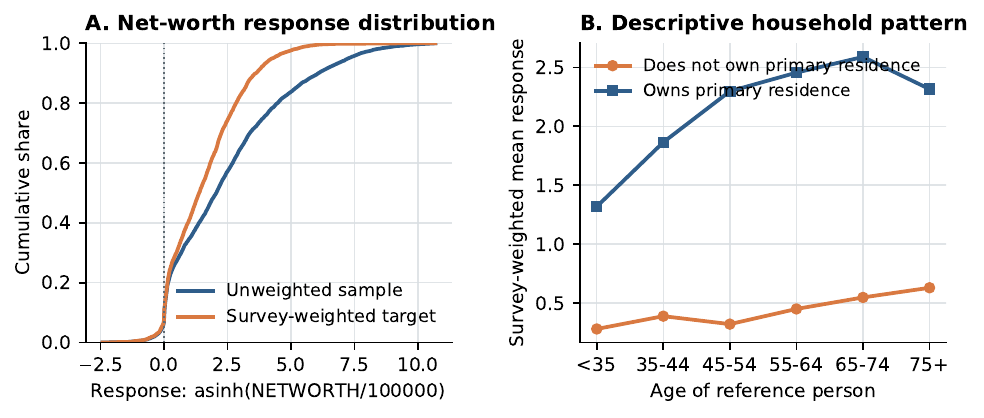}%
}{%
    \fbox{\parbox{0.85\textwidth}{\centering
    The transformed-net-worth diagnostic is unavailable.}}%
}
\caption{empirical distribution of household net worth with and without re-weighting in 
Study~1, where the full-sample means are \(2.497\) and \(1.587\), respectively.}
\label{fig:scf-networth-overview}
\end{figure}

{\it Covariate distributions before and after weighting.}
Figure~\ref{fig:scf-covariate-balance-blocks} reports the leading change of coordinates 
from each of the 12 blocks.  For example, the label ``Work status:
Self-employed/partnership'' means that, among the one-hot coordinates in the
work-status block, the self-employed/partnership indicator has the largest
\(\lvert D(Z)\rvert\).   
Figure~\ref{fig:scf-covariate-balance-components} reports $|D(Z)|$ for every 
coordinate.  These standardized differences illustrate the change of covariates' distribution under reweighting. 

\begin{figure}[htbp]
\centering
\IfFileExists{applications/scf_household_finance/figures/scf_covariate_balance_blocks.pdf}{%
    \includegraphics[width=0.88\textwidth]{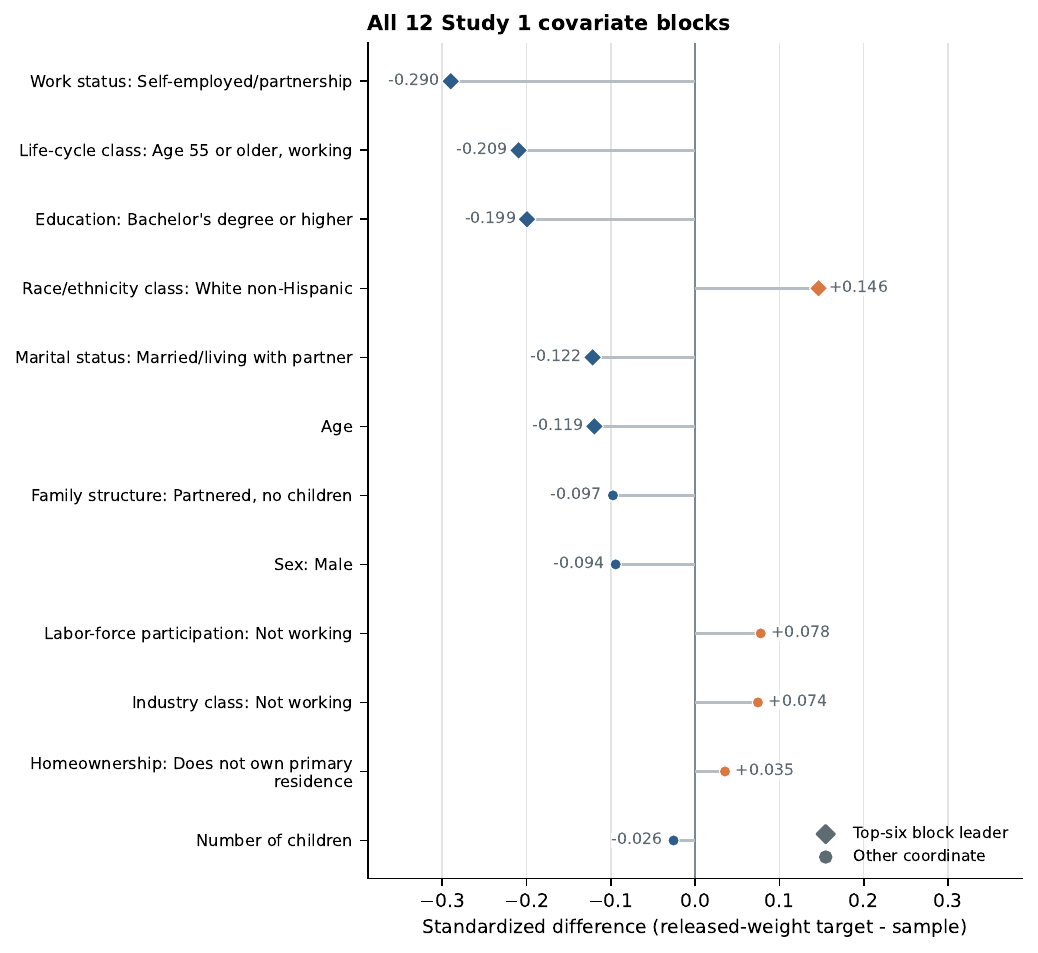}%
}{%
    \fbox{\parbox{0.85\textwidth}{\centering
    The block-level source--target diagnostic is unavailable.}}%
}
\caption{block-level source--target comparison of covariates used in Study~1.  The coordinate with the largest $|D(Z)|$ 
within each block is presented. The top-six covariates are already reported in 
Figure~\ref{fig:scf-weighting-overview}.}
\label{fig:scf-covariate-balance-blocks}
\end{figure}

\begin{figure}[htbp]
\centering
\IfFileExists{applications/scf_household_finance/figures/scf_covariate_balance_components.pdf}{%
    \includegraphics[width=\textwidth]{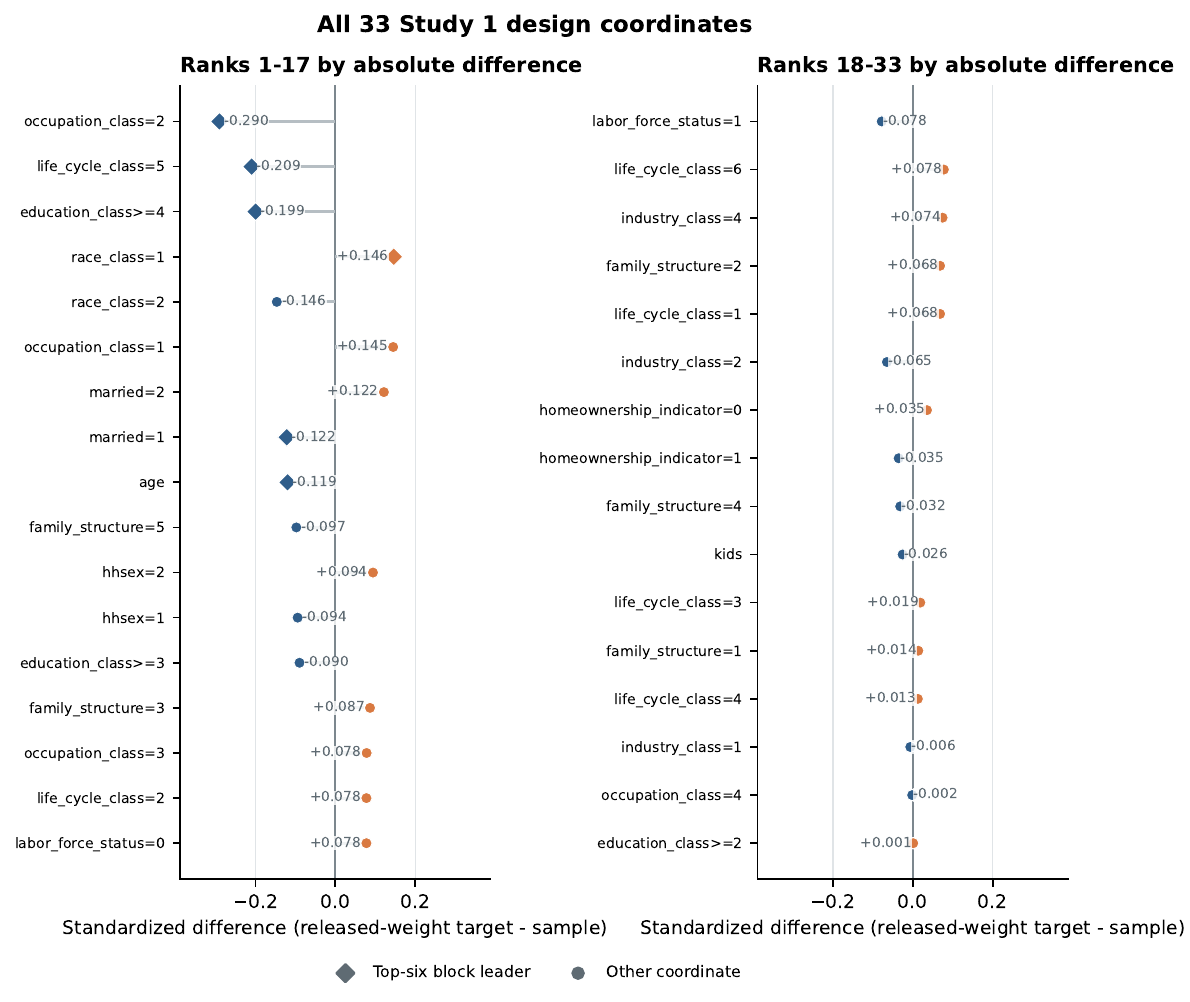}%
}{%
    \fbox{\parbox{0.85\textwidth}{\centering
    The coordinate-level source--target diagnostic is unavailable.}}%
}
\caption{Coordinate-level source--target comparison for all \(33\) covariates used in 
Study~1, ordered by $|D(Z)|$. The top-six covariates are already reported in 
Figure~\ref{fig:scf-weighting-overview}.}
\label{fig:scf-covariate-balance-components}
\end{figure}

\subsection{Study 2: does higher income increase the probability of holding equity?}
\label{app:scf-monotonicity-details}

The training set \(\calD_t\) is 
divided into a fitting subset \(\calD_{t,\mathrm{fit}}\) of size \(1{,}260\)
and a validation subset \(\calD_{t,\mathrm{val}}\) of size \(540\) for parameter tuning. The RBF kernel with \(\gamma=1\) is used after 
standardizing the income variable.

{\it Parameter tuning.} 
Writing \(n_t=|\calD_t|=1{,}800\), the validation search uses
\[
\lambda_n(a)=a n_t^{-0.6},
\quad
a\in\mathcal A_{\lambda}
:=\{0.015625,0.03125,0.0625,0.125,0.25,0.5,1,1.5,2,3,4\},
\]
\[
\tau_n(b)=b n_t^{0.06},
\quad
b\in\mathcal A_{\tau}
:=\{0.0625,0.125,0.25,0.5,1,1.5,2\},
\]
as well as the weak-truncation candidate \(\tau_n=10\).  The final selection is 
\((\lambda_n,\tau_n)=(0.005569,0.1960)\), corresponding to
\((a,b)=(0.5,0.125)\).  There are \(1{,}480\) 
records in \(\calD_t\) have \(\rho_i>\tau_n\),  which carry \(98.9\%\) of the total weight in \(\calD_t\) if no truncation.

{\it Projection and evaluation.} 
Let \(\QQ_Z\) denote the marginal law of \(Z\) under \(\QQ\), and let
\(\mathcal I_e\) denote the index set of the \(2{,}795\) records in
\(\calD_e\).  The empirical
target measure uses the untruncated \(\rho_i\) values, normalized to sum to one
over \(\mathcal I_e\):
\[
\pi_i
:=
\frac{\rho_i}{\sum_{j\in\mathcal I_e}\rho_j},
\quad i\in\mathcal I_e,
\qquad
\widehat{\QQ}
:=
\sum_{i\in\mathcal I_e}\pi_i\delta_{Z_i}.
\]
Let \(z_{(1)}<\cdots<z_{(K)}\) be the distinct values in
\(\{Z_i:i\in\mathcal I_e\}\); here \(K=802\).  For
\(I_k=\{i\in\mathcal I_e:Z_i=z_{(k)}\}\), define
\[
W_k=\sum_{i\in I_k}\pi_i,
\qquad
\bar f_k
=W_k^{-1}\sum_{i\in I_k}\pi_i\widehat f(Z_i).
\]
The pool-adjacent-violators algorithm (PAVA) computes the bounded weighted
projection onto the nondecreasing class:
\[
(\widehat g_1,\ldots,\widehat g_K)
=
\arg\min_{0\leq g_1\leq\cdots\leq g_K\leq1}
\sum_{k=1}^{K}W_k(\bar f_k-g_k)^2,
\]
and assigns \(\widehat g(Z_i)=\widehat g_k\) for every \(i\in I_k\).  

The fitted IWKRR and its monotonic projection are presented in 
Figure~\ref{fig:scf-income-monotonicity-curve}.   Among the \(5{,}000\) multiplier bootstrap draws, 
\(2{,}697\) exceed the observed distance, yielding the 
\(p\)-value \(0.5395\).

\begin{figure}[htbp]
\centering
\IfFileExists{applications/scf_household_finance/figures/stock_participation_income_monotonicity_2022_curve_projection.pdf}{%
    \includegraphics[width=\textwidth]{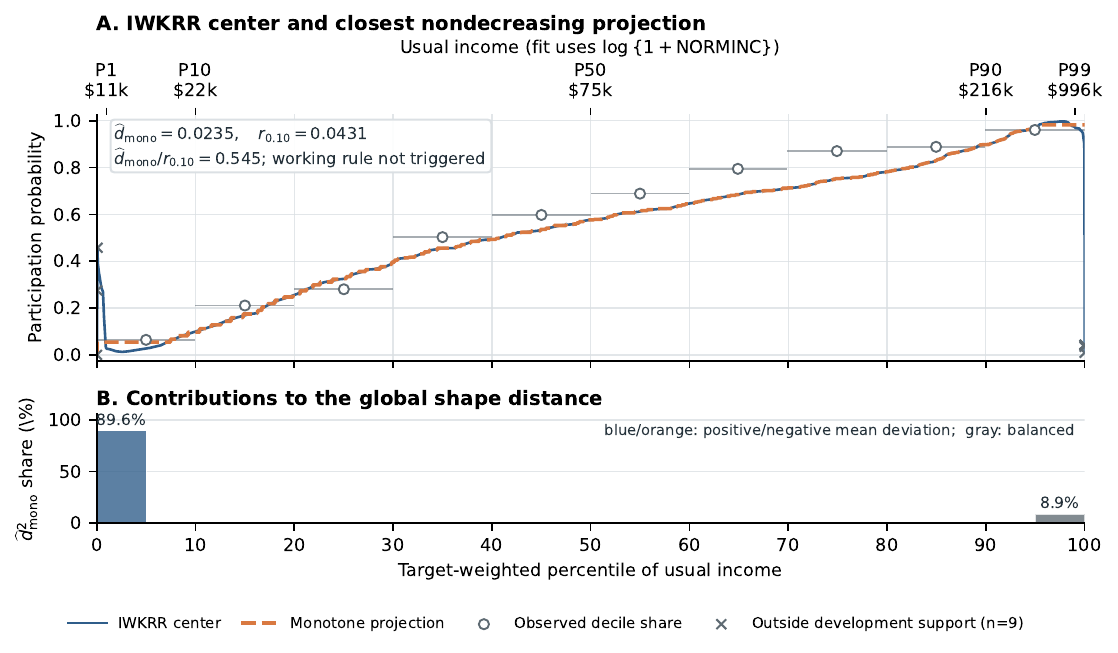}%
}{%
    \fbox{\parbox{0.85\textwidth}{\centering
    The income-monotonicity projection figure is unavailable.}}%
}
\caption{probability of holding stock-equity as a function of usual income.  Panel A shows the IWKRR center and its
closest nondecreasing projection.  Horizontal segments and open
circles give the income range and observed participation mean, respectively.  Panel B
partitions the difference between $\what f$ and its monotonic projection into \(20\) intervals, each
containing approximately \(5\%\) of the evaluation target weight. }
\label{fig:scf-income-monotonicity-curve}
\end{figure}

\subsection{Study 3: equity-share allocation over financial wealth}
\label{app:scf-flatness}

We are interested in whether, among families holding stock equity
(\(\texttt{HEQUITY}=1\)) and having positive financial assets
(\(\texttt{FIN}>0\)), the share of financial assets invested in stock equity
(\(\texttt{EQUITY}/\texttt{FIN}\)) is constant across levels of financial
wealth (\(\texttt{FIN}\)).  Define
\[
S=\mathbf 1\{\texttt{HEQUITY}=1,\ \texttt{FIN}>0\},
\qquad
\QQ^{\mathrm{alloc}}=\QQ(\,\cdot\mid S=1),
\]
\[
Y=\frac{\texttt{EQUITY}}{\texttt{FIN}},
\qquad
Z=\log\{1+\texttt{FIN}\},
\qquad
f^{\ast}(z)
=\EE_{\QQ^{\mathrm{alloc}}}(Y\mid Z=z).
\]
Let \(\QQ_Z^{\mathrm{alloc}}\) denote the marginal law of \(Z\) under
\(\QQ^{\mathrm{alloc}}\) and 
\[
\mathcal C
=\{g_c:g_c(z)\equiv c,\ c\in[0,1]\}
\]
be the class of constant-valued functions.  We test the composite
null
$
H_0:f^{\ast}\in\mathcal C, \QQ_Z^{\mathrm{alloc}} \textrm{-almost surely}.
$
Equivalently,
\[
H_0:d_{\QQ^{\mathrm{alloc}}}(f^{\ast},\mathcal C)=0,
\qquad
d_{\QQ^{\mathrm{alloc}}}(h,\mathcal C)
:=\inf_{g\in\mathcal C}
\norm{h-g}_{\mathsf{L^2}(\QQ_Z^{\mathrm{alloc}})}.
\]
There are \(2{,}895\) eligible records.   A random permutation  assigns
\(800\) records to the training set and the remaining  to
the held-out evaluation set.  After tuning, we select \(\lambda_n=0.001132\), \(\tau_n=10\), and refit the
 IWKRR on the full training set.

On the held-out survey-weighted evaluation set, let
\(\widehat{\QQ}^{\mathrm{alloc}}\) denote the normalized released-weight
empirical measure of \(Z\).  
Projecting the resulting IWKRR $\what f$ onto \(\mathcal C\) gives the nearest constant
\(0.4432\), and
\[
d(\widehat f,\mathcal C)=0.08703,
\qquad
\hat c_{0.10}=0.043620.
\]
The constructed confidence interval is 
$
d(f^{\ast},\mathcal C)\in[0.043406,0.130646],
$ and the  bootstrap \(p\)-value is
\(0.0002\), leading to the rejection of the null hypothesis. The empirical bootstrap distribution is displayed in
Figure~\ref{fig:scf-equity-flatness-bootstrap-errors}. The fitted IWKRR $\what f$ is illustrated in Figure~\ref{fig:scf-equity-flatness}.

\begin{figure}[htbp]
\centering
\IfFileExists{applications/scf_household_finance/figures/scf_study4_bootstrap_error_flatness.pdf}{%
    \includegraphics[width=\textwidth]{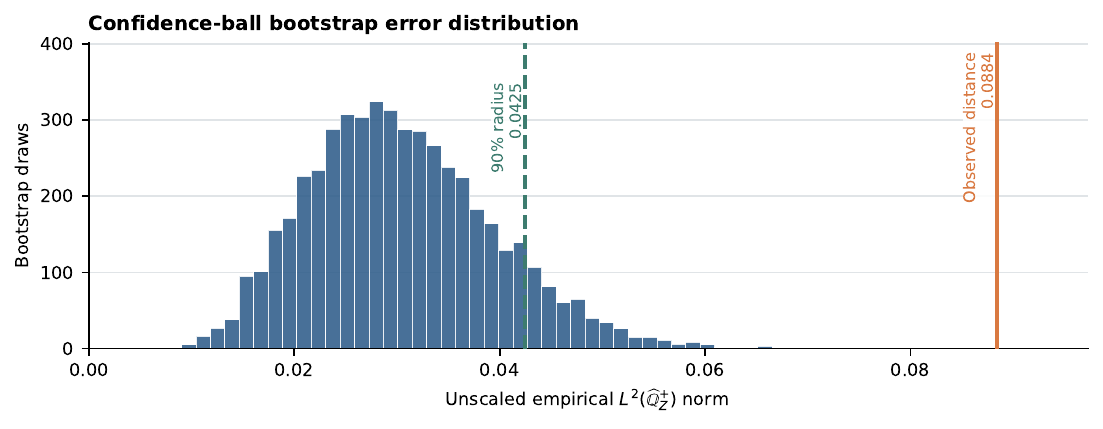}%
}{%
    \fbox{\parbox{0.85\textwidth}{\centering
    The Study~3 bootstrap-error figure is unavailable.}}%
}
\caption{empirical bootstrap distribution for Study 3 based on \(5{,}000\) bootstrap draws.  The dashed green line
is the 90\%  percentile \(\hat c_{0.10}=0.043620\), and the orange line is
the observed statistic \(d(\widehat f,\mathcal C)=0.087026\), yielding the plus-one \(p\)-value
\(0.0002\). We reject the null hypothesis that households allocate the same share of wealth to equity stocks.}
\label{fig:scf-equity-flatness-bootstrap-errors}
\end{figure}

\begin{figure}[htbp]
\centering
\IfFileExists{applications/scf_household_finance/figures/equity_share_wealth_flatness_2022_curve_projection.pdf}{%
    \includegraphics[width=0.96\textwidth]{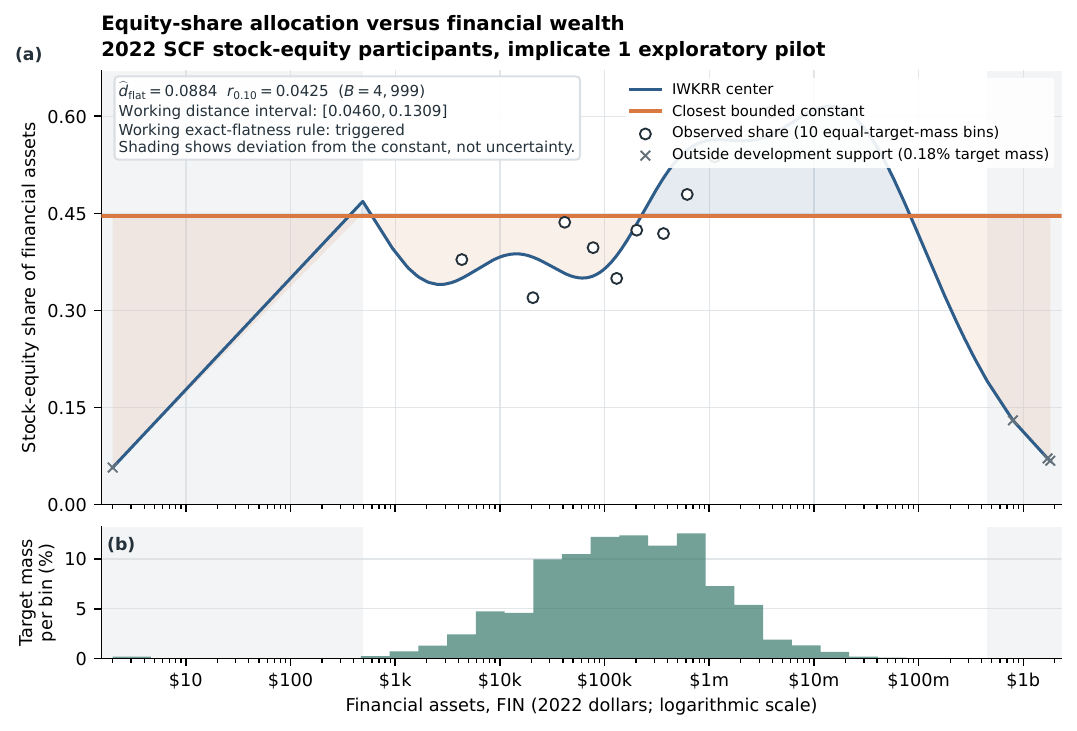}%
}{%
    \fbox{\parbox{0.85\textwidth}{\centering
    The equity-share flatness figure is unavailable.}}%
}
\caption{equity-share allocation over financial wealth.  The
upper panel shows the IWKRR $\what f$, its closest constant-valued function, and
descriptive equity shares in ten approximately equal-target-mass wealth
groups.  The lower panel shows the held-out target mass over financial wealth.}
\label{fig:scf-equity-flatness}
\end{figure}

\subsection{Study 4: does late-payment occurrence change across years?}
\label{app:scf-structural-change}

We study  whether late-payment occurrence in the full target population
changed between the 2019 and 2022 SCF waves.  Let \(\QQ_{2019}\) and
\(\QQ_{2022}\) denote the survey-weighted target distributions in the two
waves.  For each \(t\in\{2019,2022\}\), the SCF public summary extract contains
the Federal Reserve--constructed binary late-payment variable
\[
Y_t=\texttt{LATE}_t
=\mathbf 1\{\texttt{X3004}_t=5\}.
\]
Here, \(Y_t=1\) indicates that the family sometimes fell behind or missed
payments during the preceding year.  The value \(Y_t=0\) includes both
families with \(\texttt{X3004}_t=1\), indicating payments made on or ahead of
schedule, and families with \(\texttt{X3004}_t=0\), for whom the question was
inapplicable because they held none of the listed loan types.

For each \(t\in\{2019,2022\}\), let \(X_t\) denote the wave-\(t\) vector of the
same twelve SCF variables used in Study~1.  Let \(\QQ_{t,X}\) denote the marginal
distribution of \(X_t\) under \(\QQ_t\), and define the target regression
function
\[
f_t^{\ast}(x)=\QQ_t(Y_t=1\mid X_t=x).
\]
We test against the null hypothesis 
\[
H_0:
\norm{f_{2022}^{\ast}-f_{2019}^{\ast}}
_{\mathsf{L^2}(\QQ_{2022,X})}
=0.
\]
 Let
\(\widehat f_t\) be the wave-specific IWKRR estimator, and let
\(\widehat{\QQ}_{2022,E}\) be the normalized released-weight empirical
covariate measure on a held-out 2022 evaluation set.  Define
\[
\widehat D
=
\norm{\widehat f_{2022}-\widehat f_{2019}}
_{\mathsf{L^2}(\widehat{\QQ}_{2022,E})}.
\]
The triangle inequality gives
\[
\left|
\widehat D
-
\norm{f_{2022}^{\ast}-f_{2019}^{\ast}}
_{\mathsf{L^2}(\widehat{\QQ}_{2022,E})}
\right|
\leq
\norm{
(\widehat f_{2022}-f_{2022}^{\ast})
-(\widehat f_{2019}-f_{2019}^{\ast})
}_{\mathsf{L^2}(\widehat{\QQ}_{2022,E})}.
\]
For \(b=1,\ldots,B\), independently generate centered-Gaussian
residual-multiplier perturbations \(\widehat e_t^{(b)}\) of the two fitted
functions and compute
\[
R_b
=
\norm{\widehat e_{2022}^{(b)}-\widehat e_{2019}^{(b)}}
_{\mathsf{L^2}(\widehat{\QQ}_{2022,E})}.
\]
If \(\widehat c^{\mathrm{joint}}_{\alpha}\) is the 
\((1-\alpha)100\%\)-percentile of \(R_1,\ldots,R_B\), the corresponding 
interval for
\(\norm{f_{2022}^{\ast}-f_{2019}^{\ast}}
_{\mathsf{L^2}(\widehat{\QQ}_{2022,E})}\) is 
\[
\left[
\{\widehat D-\widehat c^{\mathrm{joint}}_{\alpha}\}_{+},
\widehat D+\widehat c^{\mathrm{joint}}_{\alpha}
\right],
\]
and its plus-one
upper-tail \(p\)-value is
\[
\widehat p
=
\frac{1+\sum_{b=1}^{B}\mathbf 1\{R_b\geq\widehat D\}}{B+1}.
\]

 In each wave, \(1{,}200\) records form the wave-specific
training set. The parameters \((\lambda_n,\tau_n)\) are tuned separately within each wave, using the Gaussian RBF kernel with \(\gamma=1/33\). The final selections are: 
\[
(\lambda_n,\tau_n)=(0.003551659,1.530209)
\quad\text{for 2019},
\qquad
(\lambda_n,\tau_n)=(0.028413272,0.765105)
\quad\text{for 2022}.
\]

The joint multiplier bootstrap uses \(B=5{,}000\) draws.  The fitted distance and 90th-percentile  are
\[
\widehat D=0.046569,
\qquad
\widehat c^{\mathrm{joint}}_{0.10}=0.049029.
\]
The corresponding interval is \([0,0.095598]\), and the plus-one
upper-tail \(p\)-value is \(0.1562\), leading to the non-rejection of the null hypothesis. 
Figure~\ref{fig:scf-cross-wave-change} showcases the survey-weighted distribution of the difference $\what f_{2022}-\what f_{2019}$ on the held-out 2022 evaluation set.

\begin{figure}[htbp]
\centering
\IfFileExists{applications/scf_household_finance/figures/scf_cross_wave_late_payment_structural_change_implicate1_pilot.pdf}{%
    \includegraphics[width=\textwidth]{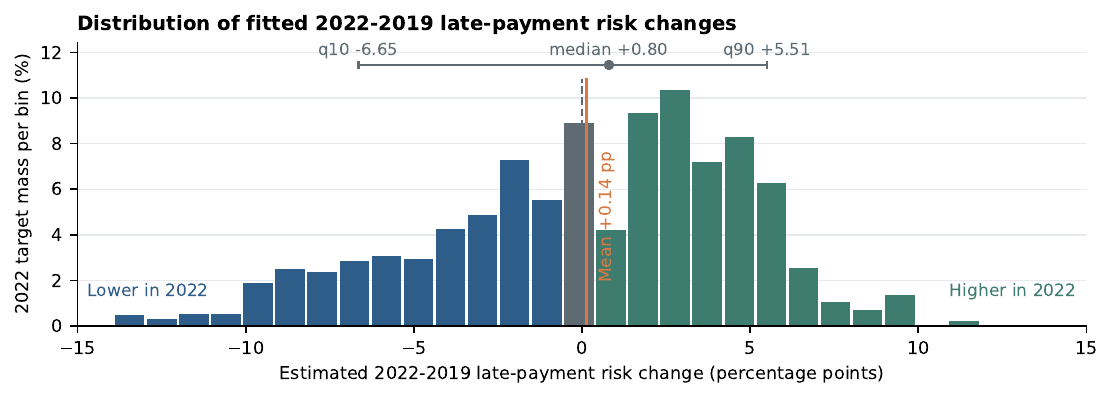}%
}{%
    \fbox{\parbox{0.85\textwidth}{\centering
    The cross-wave late-payment figure is unavailable.}}%
}
\caption{survey-weighted distribution of the difference $\what f_{2022}-\what f_{2019}$ on the held-out 2022 evaluation set.}
\label{fig:scf-cross-wave-change}
\end{figure}

\section{Proofs of Main Results}
\label{app:proofs-of-main-results}

We collect the proofs of main theorems and propositions that appear
in the main text and Appendix~\ref{app:ext-subG}. Proofs of auxiliary results
are collected in Appendix~\ref{app:proofs-of-technical-lemmas}.  For notational simplicity,  we may write $\lambda$ and $w(x)$ in short for $\lambda_n$ and $w_n(x)$,  respectively.  The notation $\|\cdot\|$ is reserved for $\|\cdot\|_{\sfH}$. 

\subsection{Proof of Lemma~\ref{lem:density-upper-bound-for-gaussian-in-splq}}

\begin{lemma}
  \label{lem:general-order-of-summation}
For fixed \(k\in\mathbb N_+\) and sufficiently small \(\lambda>0\),
\[
\sum_{i=1}^\infty
\left(\frac{\mu_i}{\mu_i+\lambda}\right)^k
\asymp
s^*
+\frac1{\lambda^k}\sum_{i=s^*+1}^\infty\mu_i^k,
\]
where \(s^*=\min\{s:\mu_s\leq\lambda\}\).
\end{lemma}

\begin{lemma}
  \label{lem:spectral-summation-rates}
  For fixed \(k\in\mathbb{N}_+\) and sufficiently small \(\lambda>0\), the following rates hold.
  If \(\mu_i\asymp i^{-\beta}\) for some \(\beta>1\), then
  \[
    \sum_{i=1}^{\infty}
    \left(\frac{\mu_i}{\mu_i+\lambda}\right)^k
    \asymp
    \lambda^{-1/\beta}.
  \]
  If \(\mu_i\asymp \exp(-\beta i^\gamma)\) for some
  \(\beta>0\) and \(\gamma>0\), then
  \[
    \sum_{i=1}^{\infty}
    \left(\frac{\mu_i}{\mu_i+\lambda}\right)^k
    \asymp
    \{\log(1/\lambda)\}^{1/\gamma}.
  \]
\end{lemma}

\begin{lemma}\label{lem:sigma-lower-operator-bound}
Under Assumption~\ref{assump:error}, for \(\lambda>0\),
\[
\opv \succeq_{\splq} \underline{\sigma}^{2} \mathbb{E}\left\{\left[\left(\oplpw+\lambda \opid\right)^{-1} w(X_{i}) k_{X_{i}}\right] \otimes_{\splq}\left[\left(\oplpw+\lambda \opid\right)^{-1} w(X_{i}) k_{X_{i}}\right]\right\}.
\]
\end{lemma}

\begin{lemma}
\label{lem:simplification-of-a-useful-operator}
\[
\mathbb{E}\!\left\{
w(X_i)\left[
(\oplpw+\lambda \opid)^{-1} k_{X_i}
\right]
\otimes_{\splq}
\left[
\bigl((\oplpw+\lambda \opid)^{-1} k_{X_i}\bigr)
\right]
\right\}
=
(\oplpw+\lambda \opid)^{-1} \oplpw (\oplpw+\lambda \opid)^{-1} \oplq
\]
\end{lemma}

\begin{lemma}
  \label{lem:same-spectrum-of-tlq-and-lq-half-tlq-half}
  Let $\opt: \mathsf{H} \to \mathsf{H}$ be a bounded positive self-adjoint
  operator in $\mathsf{H}$.
  If $\opt \oplq: \splq \to \splq$ is a compact operator in $\splq$ and $\oplq^{1/2} \opt \oplq^{1/2}: \mathsf{H} \to \mathsf{H}$ is a compact operator in $\mathsf{H}$,
  then their nonzero eigenvalues, counted with multiplicity, coincide.  In
  particular, after padding a finite nonzero spectrum by zeros,
  \[
    \mu_k(\opt \oplq, \splq)
    =
    \mu_k(\oplq^{1/2} \opt \oplq^{1/2}, \mathsf{H}),
    \qquad k\geq1.
  \]
\end{lemma}

\begin{lemma}
  \label{lem:bias-bound-in-h}
Under Assumption \ref{assump:source-condition}, when $\tau_n \ge \tau_0$, we have
\[
  \| \fpw - f_* \|
  \le
  \max\{\kappa^{2s-1},1\}
  \lambda^{\min\{s,1/2\}}
  \norm{g_{\ast}}.
\]
\end{lemma}

\begin{lemma}
  \label{lem:sigma-l2-trace-operator-norm-bound}
  Under Assumptions~\ref{assump:error}, \ref{assump:source-condition}, and
  \ref{assump:covariate-shift-estimation}, when $\tau_n \ge \tau_0$, we have
  \begin{equation*}
    \begin{aligned}
    \tr_{\splq}(\opv)
    &\le \tau_n \left(
      \overline{\sigma}^2
      + \kappa^2 \max\{\kappa^{2s-1},1\}^{2}
      \lambda^{2\min\{s,1/2\}} \norm{g_{\ast}}^2
    \right) \bigl(\mathfrak{C}_1\bigr)^2
    \tr_{\mathsf{H}}\bigl((\oplpw + \lambda \opid)^{-1} \oplpw\bigr) \\
    &\le \tau_n \left(
      \overline{\sigma}^2
      + \kappa^2 \max\{\kappa^{2s-1},1\}^{2}
      \lambda^{2\min\{s,1/2\}} \norm{g_{\ast}}^2
    \right) \bigl(\mathfrak{C}_1\bigr)^2\pd, \\
    \|\opv\|_{\splq}
    &\le \tau_n \left(
      \overline{\sigma}^2
      + \kappa^2 \max\{\kappa^{2s-1},1\}^{2}
      \lambda^{2\min\{s,1/2\}} \norm{g_{\ast}}^2
    \right) \bigl(\mathfrak{C}_1\bigr)^2.
    \end{aligned}
  \end{equation*}
\end{lemma}

\begin{lemma}
  \label{lem:anti-concentration-variance-lower-bound}
  Under Assumptions~\ref{assump:error}
  and~\ref{assump:lower-boundedness-of-covariance-operator}, writing
  \(\lambda=\lambda_n\), we have
  \[
  \opv \succeq_{\splq} 
  \bigl(\mathfrak{C}_2\bigr)^{-2}
  \underline{\rho}
  \underline{\sigma}^{2}
  (\oplq+\lambda \opid)^{-1} \oplq (\oplq+\lambda \opid)^{-1} \oplq.
  \]
\end{lemma}

\begin{lemma}[Theorem~2.6 of \protect\cite{gotze2019large}]
  \label{lem:density-upper-bound-for-gaussian-in-hilbert-space}
  Let \(\xi\sim\mathrm{N}_{\mathsf{G}}(0,\opv_\xi)\), where \(\mathsf{G}\) is a separable Hilbert space.
Denoting the probability density function of $\|\xi\|^2_{\mathsf{G}}$ by $p_{\norm{\xi}_{\mathsf{G}}^2}$, we have
$$\sup_{t \geq 0} p_{\norm{\xi}_{\mathsf{G}}^2}(t) \lesssim \varkappa(\opv_\xi) \lesssim(\Lambda_{1\xi}\Lambda_{2\xi})^{-1/2} $$
where
\[
    \varkappa(\opv_\xi)=
    \begin{cases}
      \Lambda_{1\xi}^{-1}, & \text{if } 3\lambda_{1\xi}^2 < \Lambda_{1\xi}^2; \\
      (\lambda_{1\xi}\Lambda_{2\xi})^{-1/2}, & 
      \text{if } 3\lambda_{1\xi}^2 \ge \Lambda_{1\xi}^2,\; 3\lambda_{2\xi}^2 < \Lambda_{2\xi}^2; \\
      (\lambda_{1\xi}\lambda_{2\xi})^{-1/2}, 
      & \text{if } 3\lambda_{1\xi}^2 \ge \Lambda_{1\xi}^2,\; 3\lambda_{2\xi}^2 \ge \Lambda_{2\xi}^2.
    \end{cases}
  \]
  with $\lambda_{i\xi}=\lambda_i(\opv_\xi,\mathsf{G})$ and
  $\Lambda_{k\xi}^{2}=\sum_{i=k}^\infty \lambda_{i\xi}^2$ for \(k=1,2\).
\end{lemma}

\begin{remark}
Lemma~\ref{lem:density-upper-bound-for-gaussian-in-hilbert-space} states a slightly corrected version of Theorem~2.6 in \cite{gotze2019large}, which does not affect any of the
\(\lesssim\)-order consequences used in this proof.
\end{remark}

\begin{proof}[Proof of Lemma~\ref{lem:density-upper-bound-for-gaussian-in-splq}]
  Let
  \((\nu_j,e_j)_{j\ge1}\) be the \(\splq\)-eigenpairs of \(\opv\), ordered so that
  \(\nu_1\geq\nu_2\geq\cdots\), and define
  \[
    \Lambda_{kZ}^{2}:=\sum_{j=k}^{\infty}\nu_j^2,
    \qquad k=1,2.
  \]
  We first prove the general bound
\begin{equation*}
\begin{aligned}
\sup_{t \ge 0} p_{\norm{Z}_{\splq}} (t)
&\lesssim 
\sqrt{\tr_{\splq}(\opv)}
(\Lambda_{1Z}\Lambda_{2Z})^{-1/2}.
\end{aligned}
\end{equation*}
Let
\[
  T:=\norm{Z}_{\splq}.
\]
By Lemma~\ref{lem:density-upper-bound-for-gaussian-in-hilbert-space}, applied with
\(\mathsf{G}=\splq\) and covariance \(\opv\), the random variable \(T^2\) has a density satisfying
  \[
  \sup_{u\ge0}p_{T^2}(u)
  \lesssim
  (\Lambda_{1Z}\Lambda_{2Z})^{-1/2}.
  \]
For \(t>0\), the density of \(T\) satisfies
\[
  p_T(t)=2t\,p_{T^2}(t^2).
\]
  It remains to turn this pointwise bound into a bound that is uniform in \(t\). For a finite-rank truncation
\[
  Z_m=\sum_{j=1}^m \sqrt{\nu_j}g_j e_j,
\]
the distribution function of \(T_m=\norm{Z_m}_{\splq}\) can be written as
\[
  \mathbb{P}(T_m\le t)=\gamma_m(tK_m),
\]
  where \(\gamma_m\) is a centered Gaussian measure on
  \(\operatorname{span}\{e_1,\ldots,e_m\}\) and \(K_m\) is the unit ball in that subspace. Since Gaussian measures are log-concave and
  \[
    a(tK_m)+(1-a)(sK_m)=\{at+(1-a)s\}K_m,
    \qquad a\in[0,1],
  \]
  the distribution function \(F_{T_m}(t)=\mathbb{P}(T_m\leq t)\) is log-concave on \([0,\infty)\). Moreover, \(T_m\to T\) almost surely. Since \(T\) has a density, \(F_{T_m}(t)\to F_T(t)\) for every \(t\), and hence \(F_T\) is also log-concave.

  Put
  \[
    A_Z:=(\Lambda_{1Z}\Lambda_{2Z})^{-1/2},
    \qquad
    m_Z:=\{\tr_{\splq}(\opv)\}^{1/2}.
  \]
  The Gaussian fourth-moment identity gives
  \[
    \mathbb{E}T^4
    =\{\tr_{\splq}(\opv)\}^2+2\norm{\opv}_{HS,\splq}^2
    \leq 3m_Z^4.
  \]
  Applying the Paley--Zygmund inequality to \(T^2\) therefore yields
  \[
    \mathbb{P}\left(T\geq \frac{m_Z}{\sqrt2}\right)
    \geq \frac1{12},
    \qquad
    \mathbb{E}T\geq\frac{m_Z}{12\sqrt2}.
  \]
  If \(0<t\leq4\mathbb{E}T\), the squared-norm density bound gives
  \[
    p_T(t)=2t p_{T^2}(t^2)\lesssim A_Zm_Z.
  \]
  If \(t>4\mathbb{E}T\), Markov's inequality gives \(F_T(t/2)\geq1/2\). At every differentiability point of \(\log F_T\), its concavity implies
  \[
    \frac{p_T(t)}{F_T(t)}
    =\{\log F_T\}'(t)
    \leq
    \frac{\log F_T(t)-\log F_T(t/2)}{t/2}
    \leq\frac{2\log2}{t}.
  \]
  Thus \(p_T(t)\lesssim m_Z^{-1}\) for almost every \(t\) in this second regime. Finally,
  \[
    m_Z^2=\tr_{\splq}(\opv)
    \geq\Lambda_{1Z}
    \geq(\Lambda_{1Z}\Lambda_{2Z})^{1/2},
  \]
  so that \(m_Z^{-1}\leq A_Zm_Z\). Choosing the density version satisfying these almost-everywhere bounds and combining the two regimes proves
  \[
    \sup_{t\geq0}p_{\norm{Z}_{\splq}}(t)
    \lesssim
    \sqrt{\tr_{\splq}(\opv)}
    (\Lambda_{1Z}\Lambda_{2Z})^{-1/2}.
  \]

We now bound the two covariance quantities in the preceding display. By Lemma~\ref{lem:sigma-l2-trace-operator-norm-bound} and its proof,
\[
  \tr_{\splq}(\opv)
  \lesssim
  \tau_n \tr_{\mathsf H}\!\left((\oplq+\lambda\opid)^{-1}\oplq\right)
  =
  \tau_n\pd
  =
  \tau_n\sum_{i=1}^{\infty}\frac{\mu_i}{\mu_i+\lambda},
\]
where fixed constants depending on the noise, source condition, and compatibility constant are absorbed into \(\lesssim\). On the other hand, Lemma~\ref{lem:anti-concentration-variance-lower-bound} gives
\[
  \opv
  \succeq_{\splq}
  \bigl(\mathfrak{C}_2\bigr)^{-2}
  \underline{\rho}\underline{\sigma}^{2}
  (\oplq+\lambda\opid)^{-1}\oplq(\oplq+\lambda\opid)^{-1}\oplq.
\]
By monotonicity of eigenvalues under the Loewner order for positive compact operators,
the ordered eigenvalues satisfy
  \[
  \nu_i
  \gtrsim
  \left(\frac{\mu_i}{\mu_i+\lambda}\right)^2,
  \qquad i\geq1,
  \]
because the eigenvalues of
  \((\oplq+\lambda\opid)^{-1}\oplq(\oplq+\lambda\opid)^{-1}\oplq\)
  are \((\mu_i/(\mu_i+\lambda))^2\). Consequently, for \(k=1,2\),
  \[
  \Lambda_{kZ}^{2}
  \gtrsim
  \sum_{i=k}^{\infty}
  \left(\frac{\mu_i}{\mu_i+\lambda}\right)^4.
  \]
Under either two-sided spectral condition, Lemma~\ref{lem:spectral-summation-rates}
shows that the full sum on the right-hand side diverges as \(\lambda\to0\).
Since each summand is at most one, deleting its first term does not change its
order for all sufficiently small \(\lambda\). On any remaining compact subinterval
of \((0,1]\), the ratio of the tail sum to the full sum is continuous and strictly
positive. Thus, after enlarging the implicit constant, for \(k=1,2\),
  \[
  \Lambda_{kZ}
  \gtrsim
  \left\{
  \sum_{i=1}^{\infty}
  \left(\frac{\mu_i}{\mu_i+\lambda}\right)^4
  \right\}^{1/2}.
  \]
Combining the preceding bounds yields
\[
\sup_{t \ge 0} p_{\norm{Z}_{\splq}}(t)
\lesssim
\sqrt{\tau_n}
\left\{
\frac{\sum_{i=1}^{\infty} \mu_i/(\mu_i+\lambda)}
{\sum_{i=1}^{\infty} \{\mu_i/(\mu_i+\lambda)\}^{4}}
\right\}^{1/2}.
\]
Under either the two-sided polynomial or two-sided exponential spectral decay regime, Lemma~\ref{lem:spectral-summation-rates}, after the same compact-interval adjustment over \(0<\lambda\leq1\),
\[
  \sum_{i=1}^{\infty} \frac{\mu_i}{\mu_i+\lambda}
  \asymp
  \sum_{i=1}^{\infty}
  \left(\frac{\mu_i}{\mu_i+\lambda}\right)^4;
\]
as \(\lambda\to0\), both sums are of order \(\lambda^{-1/\beta}\) in the polynomial case and
\(\{\log(1/\lambda)\}^{1/\gamma}\) in the exponential case, up to constants depending only on the spectral parameters. Therefore
\[
  \sup_{t\ge0}p_{\norm{Z}_{\splq}}(t)
  \lesssim \sqrt{\tau_n}.
\]
\end{proof}

\subsection{Proof of Theorem~\ref{thm:main-coverage}}

\begin{proof}[Proof of Theorem \ref{thm:main-coverage}]
Let
\[
  T:=\sqrt n\norm{\hat f-f_*}_{\splq},
  \qquad
  S:=\norm{Z}_{\splq},
  \qquad
  a_n:=\sup_{t\geq0}p_{\norm{Z}_{\splq}}(t),
\]
and let \(t_\beta\) denote the deterministic upper \(\beta\)-quantile of
\(S\). By the Gaussian approximation and the bias bound, for every
deterministic \(t\),
\begin{align*}
  P(T\geq t)
  &\leq
  P(S\geq t-\Delta_G-\Delta_{\mathrm{bias}})+\eta,\\
  P(T\geq t)
  &\geq
  P(S\geq t+\Delta_G+\Delta_{\mathrm{bias}})-\eta.
\end{align*}
Let \(\mathcal E_B\) be the data event on which the conditional bootstrap
approximation~\eqref{eq:theory-bootstrap-approximation} holds. By the statement
preceding~\eqref{eq:theory-bootstrap-approximation},
\(P(\mathcal E_B)\geq1-\eta\). On \(\mathcal E_B\),
the conditional coupling of \(B\) and \(Z'\) gives, for every \(s\geq0\),
\begin{align*}
  P(S\geq s+\Delta_B)
  &\leq
  P\!\left(\|B\|_{\splq}\geq s\mid\mathcal D\right)+\eta,\\
  P\!\left(\|B\|_{\splq}\geq s+\Delta_B\mid\mathcal D\right)
  &\leq
  P(S\geq s)+\eta,
\end{align*}
because \(Z'\mid\mathcal D\) has the same law as \(Z\). Taking
\(s=\hat t_\alpha\) in the first inequality and
\(s=t_{\alpha-\eta}\) in the second yields
\[
  t_{\alpha+\eta}-\Delta_B
  \leq
  \hat t_\alpha
  \leq
  t_{\alpha-\eta}+\Delta_B.
\]
Consequently,
\begin{align*}
P(T\geq\hat t_\alpha)
&\leq
P(\mathcal E_B^c)
+
P(T\geq t_{\alpha+\eta}-\Delta_B)\\
&\leq
\alpha+3\eta+a_n\Delta,
\end{align*}
where the last step uses the density bound to compare the Gaussian upper-tail
probabilities at \(t_{\alpha+\eta}-\Delta\) and
\(t_{\alpha+\eta}\). Similarly,
\begin{align*}
P(T\geq\hat t_\alpha)
&\geq
P(T\geq t_{\alpha-\eta}+\Delta_B)-P(\mathcal E_B^c)\\
&\geq
\alpha-3\eta-a_n\Delta.
\end{align*}
Lemma~\ref{lem:density-upper-bound-for-gaussian-in-splq} gives
\(a_n\lesssim\sqrt{\tau_n}\), which proves the result.
\end{proof}

\begin{corollary}[Asymptotic coverage]\label{cor:asymptotic-coverage}
Under the conditions of Theorem~\ref{thm:main-coverage}, let the tuning
parameters and probability tolerance depend on \(n\). If
\[
  \eta_n\to0
  \qquad\text{and}\qquad
  \sqrt{\tau_n}\,
  \{\Delta_{G,n}+\Delta_{B,n}+\Delta_{\mathrm{bias},n}\}
  \to0,
\]
then, for every fixed \(\alpha\in(0,1)\),
\[
  P\!\left(f_*\in\cs_{\alpha}(\hat f)\right)\to1-\alpha.
\]
\end{corollary}

Corollary~\ref{cor:asymptotic-coverage} follows immediately from Theorem~\ref{thm:main-coverage}.

\subsection{Proof of Proposition~\ref{prop:main-bias}}

\begin{lemma}
  \label{lem:source-condition-bound-1}
  Under Assumption \ref{assump:source-condition}, when $\tau_n \ge \tau_0$, we have
\[ \| (\oplpw + \lambda \opid)^{-\frac{1}{2}} f_* \| \le \max\{\kappa^{2s-1}, 1\}\,\lambda^{\min\{s-\frac{1}{2}, 0\}}\,\|g_{\ast}\|. \]
\end{lemma}

\begin{proof}[Proof of Proposition~\ref{prop:main-bias}]
Applying Lemma~\ref{lem:source-condition-bound-1} with
\(\lambda=\lambda_n\), we obtain
\begin{align*}
\sqrt n\,\left\|\fpw-f_*\right\|_{\splq}
&=
\sqrt n\,\lambda_n
\left\|
\oplq^{\frac12}
\left(\oplpw+\lambda_n\opid\right)^{-1}f_*
\right\| \\
&\leq
\mathfrak C_1\sqrt n\,\lambda_n
\left\|
\left(\oplpw+\lambda_n\opid\right)^{-\frac12}f_*
\right\| \\
&\leq
\max\left\{\kappa^{2s-1},1\right\}
\mathfrak C_1\sqrt n\,
\lambda_n^{\min\left\{s+\frac12,1\right\}}
\norm{g_{\ast}}_{\sfH}
=
\Delta_{\mathrm{bias}}.
\end{align*}
This proves the proposition.
\end{proof}

\subsection{Proof of Proposition~\ref{prop:main-bahadur}}

Let \(\beps=(\eps_1,\ldots,\eps_n)^\top\).  
The following extension of Bernstein's inequality for Hilbert--Schmidt operators is useful.

\begin{lemma}[\protect\cite{MINSKER2017111}, Theorem~3.1]\label{lem:hs-inequality}
Let $\mathsf G$ be a separable Hilbert space, let $R > 0$, and let
$\{ \xi_i \}_{i=1}^n$ be independent, self-adjoint Hilbert--Schmidt
random operators on $\mathsf G$ such that
\[
\mathbb{E}[\xi_i] = 0
\quad\text{and}\quad
\|\xi_i\| \le R
\quad\text{a.s.}
\]

Suppose there exist \(V,W>0\) such that
\[
\left\|
\mathbb{E}\Big[
\Big(\sum_{i=1}^n \xi_i \Big)^2
\Big]
\right\|
\le V
\]
\[
\operatorname{Tr}\!\left(
\mathbb{E}\Big[
\Big(\sum_{i=1}^n \xi_i \Big)^2
\Big]
\right)
\le V W
\]

Then $\forall\, t \ge V^{1/2} + R/3$,
\[
\mathbb{P}\left(
\left\|
\sum_{i=1}^n \xi_i
\right\|
\ge t
\right)
\le
4 W \exp\!\left(
-\frac{t^2}{2\big(V + R t / 3\big)}
\right).
\]
\end{lemma}

\begin{lemma}\label{lem:bound-bahadur-1}
  For \(0<\eta<1\) and \(0<\lambda\leq1\), with probability at least
  \(1-\eta\),
\begin{equation*}
    \begin{aligned}
  &\left\|
    \left(\oplpw + \lambda \opid\right)^{-\frac{1}{2}}
    \left(\oplpnw - \oplpw\right)
    \left(\oplpw + \lambda \opid\right)^{-\frac{1}{2}}
  \right\|\\
  <&
  C_1
  \left[
    \frac{\bigl(\tau_n \kappa^2 + 1\bigr)\log(e+\pd)}{n \lambda}
    +
    \sqrt{
      \frac{\bigl(\tau_n \kappa^2 + 1\bigr)\log(e+\pd)}{n \lambda}
    }
  \right]
  \log \frac{2}{\eta}.
    \end{aligned}
\end{equation*}
If further $\tau_n \to \infty$ and $\lambda \to 0$ as $n \to \infty$, and
\begin{equation*}
  \frac{\bigl(\tau_n \kappa^2 + 1\bigr)\log(e+\pd)}
       {n \lambda}
  < 1,
\end{equation*}
then with probability at least $1 - \eta$,
\begin{equation*}
\begin{aligned}
  &\left\|
    \left(\oplpw + \lambda \opid\right)^{-\frac{1}{2}}
    \left(\oplpnw - \oplpw\right)
    \left(\oplpw + \lambda \opid\right)^{-\frac{1}{2}}
  \right\|
  \lesssim&
    \sqrt{
      \frac{\bigl(\tau_n \kappa^2 + 1\bigr)\log(e+\pd)}
           {n \lambda}
    }
  \log \frac{2}{\eta}.
    \end{aligned}
\end{equation*}
\end{lemma}

\begin{lemma}
  \label{lem:bound-bahadur-2}
  For \(0<\eta<1\) and \(0<\lambda\leq1\), if
  \begin{equation*}
  \frac{(\tau_n \kappa^2 + 1)\log(e+\pd)}{n \lambda}
  <
  C_2:=
  \min \left\{
    1,\,
    \frac{1}{\bigl(4 C_1 \log \tfrac{2}{\eta}\bigr)^2}
  \right\}
\end{equation*}
where $C_1$ is the constant in Lemma \ref{lem:bound-bahadur-1}, then on the same $1 - \eta$ probability event of Lemma \ref{lem:bound-bahadur-1},
\begin{equation*}
  \bigl\|
    (\oplpw + \lambda \opid)^{\frac{1}{2}}
    (\oplpnw + \lambda \opid)^{-1}
    (\oplpw + \lambda \opid)^{\frac{1}{2}}
  \bigr\|
  \le 2.
\end{equation*}
\end{lemma}

\begin{lemma}
  \label{lem:bound-bahadur-3}
  Under Assumption~\ref{assump:error}, for \(0<\eta<1\) and
  \(\lambda>0\), with probability at least \(1-\eta\),
\[ \| (\oplpw+\lambda \opid)^{-\frac{1}{2}} \left(\frac{1}{n}\vkx^{\top}\bW\beps\right) \| \lesssim 2 \overline{\sigma} \log\!\left(\frac{2}{\eta}\right) \max \{ \sqrt{\frac{\tau_n \pd}{n}}, \frac{2 \tau_n \kappa}{n\sqrt{\lambda}} \} \]
\end{lemma}

\begin{proof}[Proof of Proposition~\ref{prop:main-bahadur}]
For brevity, write \(\lambda=\lambda_n\), \(w=w_n\), and \(\pd=\pdn\)
throughout this proof.

\begin{equation*}
\begin{aligned}
\hatf - \fpw
&= (\oplpnw + \lambda \opid)^{-1}  \left(\frac{1}{n}\vkx^{\top} \bW\vy\right)
 - (\oplpw + \lambda \opid)^{-1} \oplpw f_*
\\
&= (\oplpnw + \lambda \opid)^{-1}  \left(\frac{1}{n}\vkx^{\top}\bW\beps\right)
\\
&\quad
+ (\oplpnw + \lambda \opid)^{-1} \oplpnw f_*
 - (\oplpw + \lambda \opid)^{-1} \oplpw f_*
\\
&= (\oplpnw + \lambda \opid)^{-1}  \left(\frac{1}{n}\vkx^{\top}\bW\beps\right)
\\
&\quad
+(\oplpnw + \lambda \opid)^{-1} \oplpnw f_*
 - (\oplpnw + \lambda \opid)^{-1} \oplpw f_*
 \\
&\quad
+\left[ (\oplpnw + \lambda \opid)^{-1} - (\oplpw + \lambda \opid)^{-1}  \right]  \left( \oplpw f_* \right)
\\
&= (\oplpnw + \lambda \opid)^{-1}  \left(\frac{1}{n}\vkx^{\top}\bW\beps\right)
\\
&\quad
+(\oplpnw + \lambda \opid)^{-1} (\oplpnw
 - \oplpw) f_*
 \\
&\quad
- (\oplpnw + \lambda \opid)^{-1} (\oplpnw - \oplpw) \fpw
\\
&= (\oplpw + \lambda \opid)^{-1} \left[ \left(\frac{1}{n}\vkx^{\top}\bW\beps\right)
+(\oplpnw - \oplpw)(f_* - \fpw)\right]
\\
&\quad
+ \Big[ (\oplpnw + \lambda \opid)^{-1} - (\oplpw + \lambda \opid)^{-1} \Big]
  \big[ \left(\frac{1}{n}\vkx^{\top}\bW\beps\right) + (\oplpnw - \oplpw)(f_* - \fpw) \big]
\end{aligned}
\end{equation*}

Denote
\[
u = \Big[ (\oplpnw + \lambda \opid)^{-1} - (\oplpw + \lambda \opid)^{-1} \Big]
    \big[ \left(\frac{1}{n}\vkx^{\top}\bW\beps\right) + (\oplpnw - \oplpw)(f_* - \fpw) \big].
\]
Since
\[
\left(\frac1n\vkx^{\top}\bW\beps\right)
=\frac1n\sum_{i=1}^n w(X_i)\eps_i k_{X_i},
\qquad
\oplpnw-\oplpw
=\frac1n\sum_{i=1}^n
\bigl(w(X_i)k_{X_i}\otimes k_{X_i}-\oplpw\bigr),
\]
the first term in the preceding decomposition is \(n^{-1}\sum_{i=1}^n U_i\).
Thus
\[
\hatf-\fpw=\frac1n\sum_{i=1}^n U_i+u.
\]

With the bound
\[
\big\| \oplq^{\tfrac12} (\oplpw + \lambda \opid)^{-\tfrac12} \big\|
\le \mathfrak{C}_1,
\]
we can bound $\|u\|_{\splq}$

\begin{equation*}
\begin{aligned}
\|u\|_{\splq}
&= \big\| \oplq^{\tfrac12} u \big\|
\\
&\le \mathfrak{C}_1 \cdot
   \big\| (\oplpw + \lambda \opid)^{\tfrac12} u \big\|.
\end{aligned}
\end{equation*}

\begin{equation*}
\begin{aligned}
(\oplpnw + \lambda \opid)^{-1} - (\oplpw + \lambda \opid)^{-1}
&= \Big[ (\oplpnw + \lambda \opid)^{-1} (\oplpw + \lambda \opid) - \opid \Big] (\oplpw + \lambda \opid)^{-1}
\\
&= (\oplpnw + \lambda \opid)^{-1} (\oplpw - \oplpnw) (\oplpw + \lambda \opid)^{-1}
\\
&= (\oplpw + \lambda \opid)^{-1} (\oplpw - \oplpnw) (\oplpnw + \lambda \opid)^{-1}.
\end{aligned}
\end{equation*}

Combining this identity with
\begin{equation*}
\begin{aligned}
f_* - \fpw &= f_* - (\oplpw + \lambda \opid)^{-1} \oplpw f_* \\
&= \lambda (\oplpw + \lambda \opid)^{-1} f_*,
\end{aligned}
\end{equation*}
rewrite
\begin{equation*}
\begin{aligned}
(\oplpw + \lambda \opid)^{\frac{1}{2}} u
&= (\oplpw + \lambda \opid)^{-\frac{1}{2}}
    (\oplpw - \oplpnw)
    (\oplpnw + \lambda \opid)^{-1}
    \left(\frac{1}{n}\vkx^{\top}\bW\beps\right) \\
&\quad
  - (\oplpw + \lambda \opid)^{-\frac{1}{2}}
    (\oplpnw - \oplpw)
    (\oplpnw + \lambda \opid)^{-1}
    (\oplpnw - \oplpw)(f_* - \fpw) \\[0.7em]
&=
  \Bigl[(\oplpw + \lambda \opid)^{-\frac{1}{2}}
        (\oplpw - \oplpnw)
        (\oplpw + \lambda \opid)^{-\frac{1}{2}}\Bigr] \\
&\qquad\quad \cdot
  \Bigl[(\oplpw + \lambda \opid)^{\frac{1}{2}}
        (\oplpnw + \lambda \opid)^{-1}
        (\oplpw + \lambda \opid)^{\frac{1}{2}}\Bigr]\\
&\qquad\quad \cdot
  \Bigl[(\oplpw + \lambda \opid)^{-\frac{1}{2}}
  \left(\frac{1}{n}\vkx^{\top}\bW\beps\right)\Bigr]\\[0.7em]
&\quad
  - \lambda
  \Bigl[(\oplpw + \lambda \opid)^{-\frac{1}{2}}
        (\oplpnw - \oplpw)
        (\oplpw + \lambda \opid)^{-\frac{1}{2}}\Bigr] \\
&\qquad\quad \cdot
  \Bigl[(\oplpw + \lambda \opid)^{\frac{1}{2}}
        (\oplpnw + \lambda \opid)^{-1}
        (\oplpw + \lambda \opid)^{\frac{1}{2}}\Bigr] \\
&\qquad\quad \cdot
  \Bigl[(\oplpw + \lambda \opid)^{-\frac{1}{2}}
        (\oplpnw - \oplpw)
        (\oplpw + \lambda \opid)^{-\frac{1}{2}}\Bigr] \\
&\qquad\quad \cdot
  \Bigl[(\oplpw + \lambda \opid)^{-\frac{1}{2}} f_* \Bigr].
\end{aligned}
\end{equation*}
Apply the stochastic bounds below with suitably split failure probabilities;
this changes their logarithmic factors only by universal constants. Combining
Lemmas~\ref{lem:bound-bahadur-1}, \ref{lem:bound-bahadur-2},
\ref{lem:bound-bahadur-3}, and \ref{lem:source-condition-bound-1}, and then
using both sample-size conditions in Proposition~\ref{prop:main-bahadur},
gives, with probability at least \(1-\eta\),
\[
\norm{u}_{\splq}
\lesssim
\mathfrak C_1\tau_n\overline\sigma\,\kappa^3
\bigl(1+\norm{g_{\ast}}_{\sfH}\bigr)
\frac1n\sqrt{\frac{\pdn\log(e+\pdn)}{\lambda_n}}
\log^2\!\left(\frac2\eta\right).
\]
Consequently,
\[
\begin{aligned}
\left\|
\sqrt n(\hatf-\fpw)-\frac1{\sqrt n}\sum_{i=1}^n U_i
\right\|_{\splq}
&=\sqrt n\,\norm{u}_{\splq}\\
&\lesssim
\mathfrak C_1\tau_n\overline\sigma\,\kappa^3
\bigl(1+\norm{g_{\ast}}_{\sfH}\bigr)
\sqrt{\frac{\pdn\log(e+\pdn)}{\lambda_n n}}
\log^2\!\left(\frac2\eta\right)
=\Delta_{G,\mathrm{Bah}}.
\end{aligned}
\]
Finally, using the definition of \(U_i\), the transfer bound, and
\(\norm{(\oplpw+\lambda_n\opid)^{-1/2}}\leq\lambda_n^{-1/2}\), we obtain
\[
\begin{aligned}
\norm{U_i}_{\splq}
&=\norm{\oplq^{1/2}U_i}\\
&\leq
\frac{\tau_n\mathfrak C_1}{\sqrt{\lambda_n}}
\left(\kappa^2\norm{f_*}_{\sfH}+\kappa\overline\sigma\right)
=\bar U,
\end{aligned}
\]
which is the asserted uniform bound.
\end{proof}

\subsection{Proof of Proposition~\ref{prop:main-gaussian-coupling}}

Let \(\{(\nu_j,e_j)\}_{j\geq1}\) be an eigensystem of \(\opv\) in
\(\splq\), ordered so that \(\nu_1\geq\nu_2\geq\cdots\), and define
\[
  \Pi_m:=\sum_{j=1}^m e_j\otimes_{\splq}e_j.
\]

\begin{proof}[Proof of Proposition~\ref{prop:main-gaussian-coupling}]
It suffices to prove the bound inside the braces for an arbitrary fixed \(m\), with a Gaussian element \(Z_m\) that may depend on \(m\). The displayed right-hand side is attained over \(m\ge1\); indeed, if \(\bar U>0\), the finite-dimensional coupling term diverges as \(m\to\infty\), while if \(\bar U=0\), any \(m\) is admissible. Applying the resulting construction at a minimizer gives the stated infimum. The proof follows the same projection-coupling strategy as in \cite{singh2025kernelridgeregressioninference}, with all projections and norms taken in \(\splq\). Let
\[
  A_m:\splq\to\mathbb{R}^{m},
  \qquad
  A_m f=(\langle f,e_1\rangle_{\splq},\ldots,\langle f,e_m\rangle_{\splq}),
\]
so that \(A_m^*z=\sum_{s=1}^m z_s e_s\) and \(A_m^*A_m=\Pi_m\). The covariance of \(A_mU_i\) is the diagonal matrix \(\operatorname{diag}(\nu_1,\ldots,\nu_m)\), and \(\|A_mU_i\|_{\mathbb{R}^m}\le \|U_i\|_{\splq}\le \bar U\).

By the finite-dimensional Gaussian coupling theorem of \cite{gaussianapproximationbounded}, there exist independent Gaussian vectors \(\zeta_i\in\mathbb{R}^m\), each having the same covariance as \(A_mU_i\), such that
\[
  \left\|
  \frac1{\sqrt n}\sum_{i=1}^n\{\Pi_mU_i-A_m^*\zeta_i\}
  \right\|_{\splq}
  \lesssim
  \frac{\bar U m^2\log\!\left(\frac{m^2}{\eta}\right)}{\sqrt n}
\]
with probability at least \(1-\eta/3\). Define
\[
  Z_i=A_m^*\zeta_i+\sum_{s=m+1}^{\infty}\sqrt{\nu_s}\,h_{i,s}e_s .
\]
Then the \(Z_i\)'s are independent centered Gaussian elements in \(\splq\) with covariance \(\opv\).
Set \(Z=n^{-1/2}\sum_{i=1}^n Z_i\). Then
\(Z\sim\mathrm{N}_{\splq}(0,\opv)\).

It remains to control the two tails. Since \(\mathbb{E}U_i=0\),
\[
  \mathbb{E}\left\|
  \frac1{\sqrt n}\sum_{i=1}^n (I-\Pi_m)U_i
  \right\|_{\splq}^2
  =
  \operatorname{tr}_{\splq}\{(I-\Pi_m)\opv\}
  =
  \sigma_{\splq}^2(m,\opv).
\]
The boundedness of \(U_i\) and the Hilbert-space Bernstein inequality \cite[Corollary~1]{pinelis1986remarks} imply, with probability at least \(1-\eta/3\),
\[
  \left\|
  \frac1{\sqrt n}\sum_{i=1}^n (I-\Pi_m)U_i
  \right\|_{\splq}
  \lesssim
  \sigma_{\splq}(m,\opv)\sqrt{\log\!\left(\frac{6}{\eta}\right)}
  +\frac{\bar U\log\!\left(\frac{6}{\eta}\right)}{\sqrt n}.
\]
Similarly, by Borell's inequality for the Gaussian element
\(\frac1{\sqrt n}\sum_{i=1}^n (I-\Pi_m)Z_i\), whose second moment is
\(\sigma_{\splq}^2(m,\opv)\), with probability at least \(1-\eta/3\),
\[
  \left\|
  \frac1{\sqrt n}\sum_{i=1}^n (I-\Pi_m)Z_i
  \right\|_{\splq}
  \lesssim
  \sigma_{\splq}(m,\opv)\sqrt{\log\!\left(\frac{6}{\eta}\right)}.
\]
Combining these three bounds by the triangle inequality and a union bound gives the claim, since the term \(\bar U\log\!\left(\frac{6}{\eta}\right)/\sqrt n\) is absorbed by the displayed finite-dimensional coupling term for \(m\ge1\).
\end{proof}

\subsection{Proof of Proposition~\ref{prop:main-feasible-bootstrap}}

Let
\begin{equation*}
\begin{aligned}
V_i &= w(X_i)\bigl(Y_i - \fpw(X_i)\bigr)\,(\oplpw+\lambda \opid)^{-1} k_{X_i}, \\
\hat{V}_i &= w(X_i)\bigl(Y_i - \hatf(X_i)\bigr)\,(\oplpnw+\lambda \opid)^{-1} k_{X_i} .
\end{aligned}
\end{equation*}

Let
\begin{equation*}
\begin{aligned}
B
&= \frac{1}{n} \sum_{i=1}^n \sum_{j=1}^n h_{ij}\,\frac{\hat{V}_i-\hat{V}_j}{\sqrt{2}}, \\
B_{\mathrm{oracle}}
&= \frac{1}{n} \sum_{i=1}^n \sum_{j=1}^n h_{ij}\,\frac{V_i-V_j}{\sqrt{2}} .
\end{aligned}
\end{equation*}

For notational convenience, write
\begin{equation*}
\begin{aligned}
\oplpiw &:= w(X_i)\,k_{X_i}\otimes k_{X_i}, \\
\oplpwl &:= \oplpw+\lambda \opid, \\
\oplpnwl &:= \oplpnw+\lambda \opid,
\\
\epsilon_i^{P^w} &:= Y_i - \fpw(X_i)\\
\hat{\epsilon}_i &:= Y_i - \hatf(X_i).
\end{aligned}
\end{equation*}

Then
\begin{equation*}
\begin{aligned}
  V_i &= w(X_i)\,\epsilon_i^{P^w}\,\oplpwl^{-1} k_{X_i}\\
\hat{V}_i &= w(X_i)\,\hat{\epsilon}_i\,\oplpnwl^{-1} k_{X_i}.
\end{aligned}
\end{equation*}

Let $\mathbb{E}_h[\cdot]$ be the expectation with respect to the random variables $\{h_{ij}\}_{i,j=1}^n$ under the conditional law given \(\mathcal{D}=\{(X_i,Y_i)\}_{i=1}^n\). Note that the randomness comes from data and $\{h_{ij}\}_{i,j=1}^n$. We first find a bound under \(\mathbb{P}(\cdot\mid\mathcal{D})\), and then bound it further.

\begin{lemma}\label{lem:feasible-bootstrap-decomposition}
  If we denote
\begin{equation*}
\begin{aligned}
\alpha_{ij}
&= \frac{1}{\sqrt{2}}
\Bigl(
w(X_i)\,\epsilon_i^{P^w}\,k_{X_i}
-
w(X_j)\,\epsilon_j^{P^w}\,k_{X_j}
\Bigr), \\
\hat{\alpha}_{ij}
&= \frac{1}{\sqrt{2}}
\Bigl(
w(X_i)\,\what{\eps}_i\,k_{X_i}
-
w(X_j)\,\what{\eps}_j\,k_{X_j}
\Bigr),
\end{aligned}
\end{equation*}
then
\begin{equation*}
\begin{aligned}
B - B_{\oracle} = \Delta_{\mathrm{inv}} + \Delta_{\mathrm{res}},
\end{aligned}
\end{equation*}
where
\begin{equation*}
\begin{aligned}
\Delta_{\mathrm{inv}}
&= \bigl(\oplpnwl^{-1} - \oplpwl^{-1}\bigr)
\left(
\frac{1}{n}\sum_{i=1}^n \sum_{j=1}^n h_{ij}\,\hat{\alpha}_{ij}
\right), \\
\Delta_{\mathrm{res}}
&= \oplpwl^{-1}
\left(
\frac{1}{n}\sum_{i=1}^n \sum_{j=1}^n h_{ij}\,(\hat{\alpha}_{ij}-\alpha_{ij})
\right).
\end{aligned}
\end{equation*}
\end{lemma}

\begin{lemma}
  \label{lem:estimation-error-bound-in-h}
  For \(0<\eta<1\), under Assumptions~\ref{assump:error}
  and~\ref{assump:source-condition}, suppose
  \(\tau_n\geq\tau_0\) for all sufficiently large \(n\),
  \(\tau_n\to\infty\), and \(\lambda=\lambda_n\to0\). If, for all
  sufficiently large \(n\),
\begin{equation*}
  n
  >
  \max
  \left\{
  \frac{(\tau_n \kappa^2 + 1)\log(e+\pd)}{ C_2\lambda},
  \frac{\tau_n}{\lambda \pd}
  \right\}
\end{equation*}
where $C_2$ is the constant in Lemma \ref{lem:bound-bahadur-2} after replacing
its failure probability by the fixed fraction of \(\eta\) used in the union
bound, then, for all sufficiently large \(n\), with probability at least
\(1-\eta\),
\begin{equation*}
  \begin{aligned}
  \|\hatf - \fpw\|
  &\lesssim
  \sqrt{\frac{\tau_n \pd}{n\lambda}}
  \overline{\sigma}\,\kappa^{2s+1}
  \bigl(1+\norm{g_{\ast}}_{\sfH}\bigr)
  \log(\frac{2}{\eta}), \\
  \|\hatf - f_* \|
  &\lesssim
  \left[\sqrt{\frac{\tau_n \pd}{n\lambda}}
  +\lambda^{\min\{s, \frac{1}{2}\}}\right]
  \overline{\sigma}\,\kappa^{2s+1}
  \bigl(1+\norm{g_{\ast}}_{\sfH}\bigr)
  \log(\frac{2}{\eta}).
  \end{aligned}
\end{equation*}
\end{lemma}

\begin{lemma}
    \label{lem:feasible-bootstrap-inverse-bound}
    For \(0<\eta<1\), under Assumptions~\ref{assump:error},
    \ref{assump:source-condition}, and
    \ref{assump:covariate-shift-estimation}, suppose
    \(\tau_n\geq\tau_0\) for all sufficiently large \(n\),
    \(\tau_n\to\infty\), and \(\lambda=\lambda_n\to0\). If, for all
    sufficiently large \(n\),
\begin{equation*}
  n
  >
  \max
  \left\{
  \frac{(\tau_n \kappa^2 + 1)\log(e+\pd)}{ C_2\lambda},
  \frac{\tau_n}{\lambda \pd}
  \right\}
\end{equation*}
where $C_2$ is the constant in Lemma \ref{lem:bound-bahadur-2} after replacing
its failure probability by the fixed fraction of \(\eta\) used in the union
bound, then, for all sufficiently large \(n\), with probability at least
\(1-\eta\) over the data, conditionally on the data with
probability at least \(1-\eta\),
\[
\norm{\Delta_{\mathrm{inv}}}_{\splq} \lesssim
\mathfrak{C}_1
\overline{\sigma}\,\kappa^{2s+4}
\bigl(1+\norm{g_{\ast}}_{\sfH}\bigr)
\frac{\tau_n\pd}{\sqrt{n\lambda}}
\left(\log\!\left(\frac{2}{\eta}\right)\right)^3.
\]
\end{lemma}

\begin{lemma}
    \label{lem:feasible-bootstrap-residual-bound}
    For \(0<\eta<1\), under Assumptions~\ref{assump:error},
    \ref{assump:source-condition}, and
    \ref{assump:covariate-shift-estimation}, suppose
    \(\tau_n\geq\tau_0\) for all sufficiently large \(n\),
    \(\tau_n\to\infty\), and \(\lambda=\lambda_n\to0\). If, for all
    sufficiently large \(n\),
\begin{equation*}
  n
  >
  \max
  \left\{
  \frac{(\tau_n \kappa^2 + 1)\log(e+\pd)}{ C_2\lambda},
  \frac{\tau_n}{\lambda \pd}
  \right\}
\end{equation*}
where $C_2$ is the constant in Lemma \ref{lem:bound-bahadur-2} after replacing
its failure probability by the fixed fraction of \(\eta\) used in the union
bound, then, for all sufficiently large \(n\), with probability at least
\(1-\eta\) over the data, conditionally on the data with
probability at least \(1-\eta\),
\[
\norm{\Delta_{\mathrm{res}}}_{\splq} \lesssim
\mathfrak{C}_1
\overline{\sigma}\,\kappa^{2s+3}
\bigl(1+\norm{g_{\ast}}_{\sfH}\bigr)
\frac{\tau_n \pd}{\sqrt{n\lambda}}
\left(\log\!\left(\frac{2}{\eta}\right)\right)^2.
\]
\end{lemma}

\begin{proof}[Proof of Proposition~\ref{prop:main-feasible-bootstrap}]
Conditional on \(\mathcal D\), the processes \(B_{\oracle}\) and \(B\)
are centered Gaussian elements in \(\splq\), because the multipliers \(h_{ij}\)
are independent standard Gaussian variables. The anti-symmetric covariance
calculation gives the covariance operators
\(\hat{\opv}_{\mathrm{oracle}}\) and \(\hat{\opv}_{\mathrm{feas}}\) displayed in
Proposition~\ref{prop:main-feasible-bootstrap}. The high-probability bound for
\(\norm{B-B_{\oracle}}_{\splq}\) follows by combining
Lemmas~\ref{lem:feasible-bootstrap-decomposition},
\ref{lem:feasible-bootstrap-inverse-bound}, and
\ref{lem:feasible-bootstrap-residual-bound}. In these lemmas, the
data-dependent events are intersected first, and the conditional multiplier
events are then intersected on that data event. Replacing \(\eta\) in each
constituent bound by a fixed fraction of \(\eta\) and applying the union bound
changes only the universal constants inside the logarithms.
\end{proof}

\subsection{Proof of Proposition~\ref{prop:main-bootstrap-coupling}}

\begin{lemma}[Conditional Strassen--Dudley coupling lemma]\label{lem:strassen-coupling}
The following simplified form is a direct consequence of the conditional
Strassen--Dudley theorem of \citet[Theorem~4]{monrad1991nearby}. Let
\((\mathsf S,d)\) be a complete separable metric space, let
\(\mathcal A\) be a countably generated sub-\(\sigma\)-field, and suppose the
underlying probability space supports a \([0,1]\)-uniform random variable
independent of \(\mathcal A\vee\sigma(X)\). If \(X,X',Y'\) are
\(\mathsf S\)-valued random elements such that
\[
  X\mid\mathcal A \stackrel{d}{=} X'\mid\mathcal A,
  \qquad
  P\{d(X',Y')>r\}\leq\varepsilon
\]
for some \(r,\varepsilon\ge0\), then there exists an \(\mathsf S\)-valued
random element \(Y\) such that
\[
  Y\mid\mathcal A
  \stackrel{d}{=}
  Y'\mid\mathcal A
  \qquad\text{and}\qquad
  P\{d(X,Y)>r\}\leq\varepsilon .
\]
\end{lemma}

Let \(\{(\nu_j,e_j)\}_{j\geq1}\) be the eigensystem of \(\opv\) introduced
above, and retain the notation
\(\Pi_m=\sum_{j=1}^m e_j\otimes_{\splq}e_j\).

\begin{proof}[Proof of Proposition~\ref{prop:main-bootstrap-coupling}]
The proof uses the covariance-comparison and anti-symmetric multiplier-bootstrap techniques of \cite{singh2025kernelridgeregressioninference}, but all covariance operators, projections, Hilbert-Schmidt norms, and Gaussian elements below are taken in \(\splq\).

First, the anti-symmetric multiplier process has the desired empirical covariance. Since \(V_i=U_i+\mu\), where
\[
  \mu=(\oplpw+\lambda\opid)^{-1}\oplpw(f_*-\fpw),
\]
we have \(V_i-V_j=U_i-U_j\). Conditional on the data,
\[
  \frac1n\sum_{i=1}^n\sum_{j=1}^n
  h_{ij}\frac{V_i-V_j}{\sqrt2}
\]
is a centered Gaussian element in \(\splq\). Its conditional covariance is
\begin{equation*}
\begin{aligned}
  \mathbb{E}_h\left[
  \left(\frac1n\sum_{i,j}h_{ij}\frac{U_i-U_j}{\sqrt2}\right)
  \otimes_{\splq}
  \left(\frac1n\sum_{k,\ell}h_{k\ell}\frac{U_k-U_\ell}{\sqrt2}\right)
  \right]
  &=
  \mathbb{E}_n[U_i\otimes_{\splq}U_i]
  -
  \mathbb{E}_n[U_i]\otimes_{\splq}\mathbb{E}_n[U_i]  \\
  &=\what{\opv}_{\oracle}.
\end{aligned}
  \end{equation*}
Thus the conditional law of the anti-symmetric multiplier process is the same as that of \(\what{\opv}_{\oracle}^{1/2}g\).

Since \(\opv\) is trace class, writing \(g_i:=g(e_i)\) for the coordinates of
the isonormal process in Proposition~\ref{prop:main-bootstrap-coupling} gives
\[
  \opv^{1/2}g
  =\sum_{i=1}^{\infty}\sqrt{\nu_i}\,g_i e_i
  \in\splq
  \quad\text{a.s.},
  \qquad
  \mathbb E\!\left[
    \bigl(\opv^{1/2}g\bigr)\otimes_{\splq}
    \bigl(\opv^{1/2}g\bigr)
  \right]
  =\opv.
\]

We next compare the Gaussian elements with covariances \(\opv\) and \(\what{\opv}_{\oracle}\). Fix \(m\ge1\), and write \(\Pi_m^\perp=I-\Pi_m\). Define
\[
  \Delta_{\mathrm{HS}}=\|\what{\opv}_{\oracle}-\opv\|_{\mathrm{HS},\splq},
  \qquad
  \Delta_{\mathrm{tail}}=
  \left|\tr_{\splq}\{\Pi_m^\perp(\what{\opv}_{\oracle}-\opv)\Pi_m^\perp\}\right|.
\]
The covariance-comparison argument gives a random element \(G\), with the same conditional distribution as \(\what{\opv}_{\oracle}^{1/2}g\), such that conditional on the data, with probability at least \(1-3\eta\),
\begin{equation}
\label{eq:abstract-bootstrap-coupling-l2q}
  \|\opv^{1/2}g-G\|_{\splq}
  \le
  \{1+\sqrt{2\log\!\left(\frac{1}{\eta}\right)}\}
  \left\{
    m^{1/4}\Delta_{\mathrm{HS}}^{1/2}
    +\Delta_{\mathrm{tail}}^{1/2}
    +2\sigma_{\splq}(m,\opv)
  \right\}.
\end{equation}
Indeed, decompose \(\opv^{1/2}g\) into its \(\Pi_m\) and \(\Pi_m^\perp\) parts. On the leading \(m\)-dimensional subspace, the Powers--St{\o}rmer inequality gives
\[
  \|(\Pi_m\opv\Pi_m)^{1/2}-(\Pi_m\what{\opv}_{\oracle}\Pi_m)^{1/2}\|_{\mathrm{HS},\splq}
  \le
  m^{1/4}\|\opv-\what{\opv}_{\oracle}\|_{\mathrm{HS},\splq}^{1/2}.
\]
Borell's inequality then controls the corresponding Gaussian norm. The
\(\Pi_m^\perp\opv\Pi_m^\perp\) tail contributes \(\sigma_{\splq}(m,\opv)\).
The remaining \(\Pi_m^\perp\what{\opv}_{\oracle}\Pi_m^\perp\) tail is coupled to the full
\(\what{\opv}_{\oracle}^{1/2}g\) by Lemma~\ref{lem:strassen-coupling} and is bounded by
Borell's inequality through
\[
  \|(\Pi_m^\perp\what{\opv}_{\oracle}\Pi_m^\perp)^{1/2}\|_{\mathrm{HS},\splq}
  \le
  \Delta_{\mathrm{tail}}^{1/2}+\sigma_{\splq}(m,\opv).
\]
This proves \eqref{eq:abstract-bootstrap-coupling-l2q}.

It remains to bound \(\Delta_{\mathrm{HS}}\) and \(\Delta_{\mathrm{tail}}\) under the target-norm
bound on \(U_i\). Hilbert-space Bernstein bounds (Lemma~\ref{lem:hs-inequality}) for the empirical covariance
operator imply, with probability at least \(1-3\eta\),
\[
  \Delta_{\mathrm{HS}}
  \lesssim
  \left(\log\!\left(\frac{2}{\eta}\right)\right)^2
  \left\{
    \left(\frac{\bar U^2\sigma_{\splq}^2(0,\opv)}{n}\right)^{1/2}
    +\frac{\bar U^2}{n}
  \right\},
\]
and
\[
  \Delta_{\mathrm{tail}}
  \lesssim
  \log\!\left(\frac{2}{\eta}\right)
  \left\{
    \left(\frac{\bar U^2\sigma_{\splq}^2(m,\opv)}{n}\right)^{1/2}
    +\frac{\bar U^2}{n}
  \right\}.
\]
Substituting these bounds into \eqref{eq:abstract-bootstrap-coupling-l2q}, using
\[
  \Delta_{\mathrm{tail}}^{1/2}
  \lesssim
  \left(\log\!\left(\frac{2}{\eta}\right)\right)^{1/2}
  \left[
    \sigma_{\splq}(m,\opv)
    +
    \left\{
      \frac{\bar U^2\sigma_{\splq}^2(0,\opv)}{n}
      +
      \frac{\bar U^4}{n^2}
    \right\}^{1/4}
  \right],
\]
and absorbing logarithmic factors into \(\left(\log\!\left(\frac{2}{\eta}\right)\right)^{3/2}\), gives the claim after taking the infimum over \(m\ge1\). The total probability is at least \(1-6\eta\) by a union bound over the covariance-estimation event and the conditional Gaussian-coupling event.
\end{proof}

\subsection{Proof of Corollary~\ref{cor:main-approximate-bootstrap}}

\begin{proof}[Proof of Corollary~\ref{cor:main-approximate-bootstrap}]
Fix \(\eta\in(0,1)\), and let \(\eta_{\mathrm{tol}}=c\eta^2\), where \(c>0\) is a sufficiently
small universal constant. In this proof, \(\Delta_{B,\mathrm{feas}}\) and
\(\Delta_{B,\mathrm{coup}}\) are evaluated with tolerance \(\eta_{\mathrm{tol}}\). Write
\[
  W=B_{\oracle},
  \qquad
  W'=B,
  \qquad
  r_{\mathrm{feas}}=\Delta_{B,\mathrm{feas}},
  \qquad
  r_{\mathrm{coup}}=\Delta_{B,\mathrm{coup}}.
\]
Let \(\mathcal E_{\mathrm{feas}}\) be the data event from
Proposition~\ref{prop:main-feasible-bootstrap}, evaluated with tolerance
\(\eta_{\mathrm{tol}}\). Then
\(P(\mathcal E_{\mathrm{feas}}^c)\leq\eta_{\mathrm{tol}}\), and, on
\(\mathcal E_{\mathrm{feas}}\), the conditional failure probability is at most
\(\eta_{\mathrm{tol}}\). Consequently,
\[
\begin{aligned}
P\!\left(\norm{W'-W}_{\splq}>r_{\mathrm{feas}}\right)
&\leq
P(\mathcal E_{\mathrm{feas}}^c)
+
\mathbb E\!\left[
\mathbf 1_{\mathcal E_{\mathrm{feas}}}
P\!\left(
\norm{W'-W}_{\splq}>r_{\mathrm{feas}}
\,\middle|\,\mathcal D
\right)
\right] \\
&\leq 2\eta_{\mathrm{tol}}.
\end{aligned}
\]
Next, the identity \(V_i=U_i+\mu\), with
\[
  \mu=(\oplpw+\lambda\opid)^{-1}\oplpw(f_*-\fpw),
\]
implies \(V_i-V_j=U_i-U_j\). Hence, conditional on \(\mathcal D\), the oracle
anti-symmetric bootstrap \(W\) is centered Gaussian with covariance
\[
  \what{\opv}_{\oracle}
  =
  \mathbb E_n[U_i\otimes_{\splq}U_i]
  -
  \mathbb E_n[U_i]\otimes_{\splq}\mathbb E_n[U_i].
\]
Thus \(W\mid\mathcal D\) has the same law as \(\what{\opv}_{\oracle}^{1/2}g\). By
Proposition~\ref{prop:main-bootstrap-coupling}, applied with tolerance
\(\eta_{\mathrm{tol}}\), there is a Gaussian element \(Z_0\sim\mathrm{N}_{\splq}(0,\opv)\)
and a random element \(G\), with \(G\mid\mathcal D\) distributed as
\(W\mid\mathcal D\), such that
\[
  P\!\left(\norm{Z_0-G}_{\splq}>r_{\mathrm{coup}}\right)
  \leq 6\eta_{\mathrm{tol}} .
\]
By Lemma~\ref{lem:strassen-coupling}, applied to the regular conditional laws given
\(\mathcal D\), we may replace the copy \(G\) by the actual oracle bootstrap
\(W\): there exists a copy \(Z\) with
\[
  Z\mid\mathcal D\sim\mathrm{N}_{\splq}(0,\opv)
\]
such that
\[
  P\!\left(\norm{Z-W}_{\splq}>r_{\mathrm{coup}}\right)\leq 6\eta_{\mathrm{tol}} .
\]
Gluing this coupling with the original pair \((W,W')\) and using the triangle
inequality gives
\[
  P\!\left(
    \norm{Z-W'}_{\splq}
    >
    r_{\mathrm{coup}}+r_{\mathrm{feas}}
  \right)
  \leq 8\eta_{\mathrm{tol}} .
\]
Let
\[
  A(\mathcal D)
  =
  P\!\left(
    \norm{Z-W'}_{\splq}
    >
    r_{\mathrm{coup}}+r_{\mathrm{feas}}
    \,\middle|\,\mathcal D
  \right).
\]
Then \(\mathbb E A(\mathcal D)\leq8\eta_{\mathrm{tol}}\). Markov's inequality gives
\[
  P\{A(\mathcal D)>\eta\}
  \leq
  \frac{8\eta_{\mathrm{tol}}}{\eta}.
\]
Choosing \(c\leq 1/8\) yields \(8\eta_{\mathrm{tol}}/\eta\leq\eta\). Therefore, with
probability at least \(1-\eta\) over the data,
\[
  P\!\left(
    \norm{B-Z}_{\splq}
    \leq
    \Delta_{B,\mathrm{feas}}
    +
    \Delta_{B,\mathrm{coup}}
    \,\middle|\,\mathcal D
  \right)
  \geq 1-\eta .
\]
Renaming \(Z\) as \(Z'\) proves the corollary.
\end{proof}

\subsection{Proof of Corollary~\ref{cor:bd-density}}

\begin{lemma}
  \label{lem:effective-dimension-bound-polynomial-decay}
Suppose that the eigenvalues of \(\oplq\) satisfy
\(\mu_k\leq\alpha k^{-\beta}\) for all \(k\geq1\), where
\(\alpha>0\) and \(\beta>1\). Then, for every \(\lambda>0\),
\[
\pd
\leq
\alpha^{1/\beta}\frac{\beta}{\beta-1}\lambda^{-1/\beta}.
\]
\end{lemma}

\begin{lemma}
  \label{lem:sigma-bound-polynomial-decay}
Suppose that the eigenvalues of \(\oplq\) satisfy
\(\mu_k\leq\alpha k^{-\beta}\) for all \(k\geq1\), where
\(\alpha>0\) and \(\beta>1\). Then, for every \(m\in\mathbb N_+\),
\[
\sigma(m,\oplq)
\leq
\left(\frac{\alpha}{\beta-1}\right)^{1/2}
m^{1/2-\beta/2}.
\]
\end{lemma}

\begin{lemma}
  \label{lem:tail-bound-exponential-decay-integral}
For \(\beta>0\) and \(\gamma>0\),
\[
\lim_{a\to\infty}
\frac{\int_a^\infty\exp(-\beta x^\gamma)\,dx}
{\exp(-\beta a^\gamma)/(\beta\gamma a^{\gamma-1})}
=1.
\]
\end{lemma}

\begin{lemma}
  \label{lem:effective-dimension-bound-exponential-decay}
Suppose that the eigenvalues of \(\oplq\) satisfy
\(\mu_k\leq\alpha\exp(-\beta k^\gamma)\) for all \(k\geq1\), where
\(\alpha>0\), \(\beta>0\), and \(\gamma>0\). Then there exist
\(\lambda_0\in(0,1)\) and \(C_{\alpha,\beta,\gamma}>0\) such that, for every
\(0<\lambda\leq\lambda_0\),
\[
\pd
\leq
C_{\alpha,\beta,\gamma}\{\log(1/\lambda)\}^{1/\gamma}.
\]
\end{lemma}

\begin{lemma}
  \label{lem:sigma-bound-exponential-decay}
Suppose that the eigenvalues of \(\oplq\) satisfy
\(\mu_k\leq\alpha\exp(-\beta k^\gamma)\) for all \(k\geq1\), where
\(\alpha>0\), \(\beta>0\), and \(\gamma>0\). Then there exists
\(C_{\alpha,\beta,\gamma}>0\) such that, for every \(m\in\mathbb N_+\),
\[
\sigma(m,\oplq)
\leq
C_{\alpha,\beta,\gamma}
m^{1/2-\gamma/2}\exp(-\beta m^\gamma/2).
\]
\end{lemma}

\begin{lemma}
  \label{lem:tail-sum-sigma-bound}
Under Assumptions~\ref{assump:error} and \ref{assump:source-condition}, when
\(\tau_n\geq\tau_0\), for every \(m\in\mathbb N_+\) and \(\lambda>0\),
we have
\[
\sigma^2_{\splq}(m,\opv)
\leq
\frac{\tau_n}{\lambda}
\left(
  \overline{\sigma}^2
  +\kappa^2\max\{\kappa^{2s-1},1\}^{2}
  \lambda^{2\min\{s,1/2\}}\|g_{\ast}\|^2
\right)
\sigma^2(m,\oplq).
\]
\end{lemma}

\begin{proposition}
  \label{prop:gaussian-approximation-polynomial-boundedness}
Suppose \(\mu_k\leq\alpha k^{-\beta}\) for some \(\alpha>0\) and
\(\beta>1\). Under Assumptions~\ref{assump:error},
\ref{assump:source-condition}, and
\ref{assump:covariate-shift-estimation}, when
\(\tau_n\geq\tau_0\), if \(\tau_n\to\infty\), \(\lambda\to0\), and
\[
n>
\max\left\{
\frac{(\tau_n\kappa^2+1)\log(e+\pd)}{C_2\lambda},
\frac{\tau_n}{\lambda\pd}
\right\},
\]
where \(C_2\) is the constant in Lemma~\ref{lem:bound-bahadur-2} after
replacing its failure probability by the fixed fraction of \(\eta\) used in
the union bound, then with probability at least \(1-\eta\), there exist
i.i.d. Gaussian elements \(Z_i\sim\mathrm{N}_{\splq}(0,\opv)\) such that
\[
\begin{aligned}
&\left\|
\sqrt n(\hatf-\fpw)-\frac1{\sqrt n}\sum_{i=1}^n Z_i
\right\|_{\splq}\\
&\quad\lesssim
\frac{\tau_n}{\sqrt n}
\lambda^{-1/2-1/(2\beta)}
\log\!\left(\frac1\lambda\right)
\left(\log\!\left(\frac2\eta\right)\right)^2
+\sqrt{\frac{\tau_n}{\lambda}}
\left(\frac{\tau_n}{n}\right)^{\frac{\beta-1}{2(\beta+3)}}
\log(n)\log\!\left(\frac2\eta\right).
\end{aligned}
\]
\end{proposition}

\begin{proof}[Proof of Proposition~\ref{prop:gaussian-approximation-polynomial-boundedness}]
Apply Propositions~\ref{prop:main-bahadur}
and~\ref{prop:main-gaussian-coupling} with a fixed fraction of the tolerance
\(\eta\). By Lemma~\ref{lem:effective-dimension-bound-polynomial-decay},
\[
  \pd\lesssim\lambda^{-1/\beta}.
\]
Consequently, the Bahadur remainder satisfies
\[
  \Delta_{G,\mathrm{Bah}}
  \lesssim
  \frac{\tau_n}{\sqrt n}
  \lambda^{-\frac12-\frac1{2\beta}}
  \log\!\left(\frac1\lambda\right)
  \left(\log\!\left(\frac2\eta\right)\right)^2.
\]
Lemmas~\ref{lem:tail-sum-sigma-bound}
and~\ref{lem:sigma-bound-polynomial-decay} give
\[
  \sigma_{\splq}(m,\opv)
  \lesssim
  \sqrt{\frac{\tau_n}{\lambda}}\,
  m^{-\frac{\beta-1}{2}}.
\]
Moreover, Proposition~\ref{prop:main-bahadur} gives
\(\bar U\lesssim\tau_n/\sqrt\lambda\). Hence
Proposition~\ref{prop:main-gaussian-coupling} yields
\[
\begin{aligned}
  \Delta_{G,\mathrm{coup}}
  \lesssim\inf_{m\geq1}\bigg\{&
  \sqrt{\frac{\tau_n}{\lambda}}\,
  m^{-\frac{\beta-1}{2}}
  \left(\log\!\left(\frac2\eta\right)\right)^{1/2}
  +\frac{\tau_n}{\sqrt{n\lambda}}m^2
  \log\!\left(\frac{m^2}{\eta}\right)
  \bigg\}.
\end{aligned}
\]
Take
\[
  m=\left\lceil
  \left(\frac n{\tau_n}\right)^{1/(\beta+3)}
  \right\rceil.
\]
The sample-size condition implies \(n/\tau_n\to\infty\), and substitution
gives
\[
  \Delta_{G,\mathrm{coup}}
  \lesssim
  \sqrt{\frac{\tau_n}{\lambda}}
  \left(\frac{\tau_n}{n}\right)^{\frac{\beta-1}{2(\beta+3)}}
  \log(n)\log\!\left(\frac2\eta\right).
\]
The triangle inequality and a union bound combine the Bahadur and coupling
events. Finally, the coupled Gaussian element can be represented as
\(n^{-1/2}\sum_{i=1}^n Z_i\) by adjoining an independent Gaussian bridge,
with the \(Z_i\)'s i.i.d. with covariance \(\opv\). This proves the claim.
\end{proof}

\begin{proposition}
  \label{prop:bootstrap-approximation-polynomial-boundedness}
Suppose \(\mu_k\leq\alpha k^{-\beta}\) for some \(\alpha>0\) and
\(\beta>1\). Under Assumptions~\ref{assump:error},
\ref{assump:source-condition}, and
\ref{assump:covariate-shift-estimation}, together with
Definition~\ref{def:bernstein-moment-condition-for-weight}, when
\(\tau_n\geq\tau_0\), if \(\tau_n\to\infty\), \(\lambda\to0\), and
\[
n>
\max\left\{
\frac{(\tau_n\kappa^2+1)\log(e+\pd)}{C_2\lambda},
\frac{\tau_n}{\lambda\pd}
\right\},
\]
where \(C_2\) is the constant in Lemma~\ref{lem:bound-bahadur-2} evaluated at
the component tolerance \(\eta_{\mathrm{tol}}=c\eta^2\) used in
Corollary~\ref{cor:main-approximate-bootstrap}, then, on a data event of
probability at least \(1-\eta\), there exists a coupling of \(B\) and \(Z\),
conditionally on \(\mathcal D\), such that
\(Z\mid\mathcal D\sim\mathrm{N}_{\splq}(0,\opv)\). On this event, let
\[
\begin{aligned}
R_{\mathrm{B,poly}}
:={}&
\left(\log\!\left(\frac2\eta\right)\right)^3
\frac{\tau_n}{\sqrt n}\lambda^{-1/2-1/\beta}\\
&+
\left(\log\!\left(\frac2\eta\right)\right)^{3/2}
\sqrt{\frac{\tau_n}{\lambda}}
\left(
\frac{\tau_n\lambda^{1-1/\beta}}{n}
\right)^{\frac{\beta-1}{2(2\beta-1)}}.
\end{aligned}
\]
Then
\[
\mathbb P\!\left\{
\|Z-B\|_{\splq}\lesssim R_{\mathrm{B,poly}}
\;\middle|\;\mathcal D
\right\}
\geq1-\eta.
\]
\end{proposition}

\begin{proof}[Proof of Proposition~\ref{prop:bootstrap-approximation-polynomial-boundedness}]
Apply Corollary~\ref{cor:main-approximate-bootstrap}. By
Proposition~\ref{prop:main-feasible-bootstrap} and
Lemma~\ref{lem:effective-dimension-bound-polynomial-decay},
\[
  \Delta_{B,\mathrm{feas}}
  \lesssim
  \left(\log\!\left(\frac2\eta\right)\right)^3
  \frac{\tau_n}{\sqrt n}\lambda^{-\frac12-\frac1\beta}.
\]
Also, Proposition~\ref{prop:main-bahadur} and
Lemma~\ref{lem:sigma-l2-trace-operator-norm-bound} imply
\[
  \frac{\bar U^2\tr_{\splq}(\opv)}{n}
  +\frac{\bar U^4}{n^2}
  \lesssim
  \frac{\tau_n^3\pd}{n\lambda}
  +\frac{\tau_n^4}{n^2\lambda^2}.
\]
The condition \(n>\tau_n/(\lambda\pd)\) ensures that the second term is
bounded by the first. Therefore, Proposition~\ref{prop:main-bootstrap-coupling}
and Lemmas~\ref{lem:tail-sum-sigma-bound},
\ref{lem:effective-dimension-bound-polynomial-decay}, and
\ref{lem:sigma-bound-polynomial-decay} give
\[
\begin{aligned}
  \Delta_{B,\mathrm{coup}}
  \lesssim
  \left(\log\!\left(\frac2\eta\right)\right)^{3/2}
  \inf_{m\geq1}\bigg\{&
  m^{1/4}
  \left(\frac{\tau_n^3\lambda^{-1/\beta}}{n\lambda}\right)^{1/4}
  +\sqrt{\frac{\tau_n}{\lambda}}\,
  m^{-\frac{\beta-1}{2}}
  \bigg\}.
\end{aligned}
\]
Choose
\[
  m=\left\lceil
  \left(\frac{n}{\tau_n\lambda^{1-1/\beta}}\right)^{1/(2\beta-1)}
  \right\rceil.
\]
Both terms inside the braces are then bounded by
\[
  \sqrt{\frac{\tau_n}{\lambda}}
  \left(
  \frac{\tau_n\lambda^{1-1/\beta}}{n}
  \right)^{\frac{\beta-1}{2(2\beta-1)}}.
\]
Combining the feasible and coupling bounds in
Corollary~\ref{cor:main-approximate-bootstrap}, with tolerances adjusted by
universal constants, proves the stated outer and conditional probability
claim.
\end{proof}

\begin{proposition}
  \label{prop:gaussian-approximation-exponential-boundedness}
Suppose \(\mu_k\leq\alpha\exp(-\beta k^\gamma)\) for some
\(\alpha,\beta,\gamma>0\). Under Assumptions~\ref{assump:error},
\ref{assump:source-condition}, and
\ref{assump:covariate-shift-estimation}, when \(\tau_n\geq\tau_0\), if
\(\tau_n\to\infty\), \(\lambda\to0\), and
\[
n>
\max\left\{
\frac{(\tau_n\kappa^2+1)\log(e+\pd)}{C_2\lambda},
\frac{\tau_n}{\lambda\pd}
\right\},
\]
where \(C_2\) is the constant in Lemma~\ref{lem:bound-bahadur-2} after
replacing its failure probability by the fixed fraction of \(\eta\) used in
the union bound, then with probability at least \(1-\eta\), there exist
i.i.d. Gaussian elements \(Z_i\sim\mathrm{N}_{\splq}(0,\opv)\) such that
\[
\begin{aligned}
&\left\|
\sqrt n(\hatf-\fpw)-\frac1{\sqrt n}\sum_{i=1}^n Z_i
\right\|_{\splq}\\
&\quad\lesssim
\frac{\tau_n}{\sqrt{n\lambda}}
\log^2\!\left(\frac2\eta\right)
\Biggl[
\log^{1/(2\gamma)}\!\left(\frac1\lambda\right)
\log\log\!\left(\frac1\lambda\right)
+
\log^{2/\gamma}\!\left(\frac n{\tau_n}\right)
\log\log\!\left(\frac n{\tau_n}\right)
\Biggr].
\end{aligned}
\]
\end{proposition}

\begin{proof}[Proof of Proposition~\ref{prop:gaussian-approximation-exponential-boundedness}]
Apply Propositions~\ref{prop:main-bahadur}
and~\ref{prop:main-gaussian-coupling} with a fixed fraction of \(\eta\).
Lemma~\ref{lem:effective-dimension-bound-exponential-decay} gives
\[
  \pd\lesssim
  \left(\log\!\left(\frac1\lambda\right)\right)^{1/\gamma},
\]
and hence
\[
  \Delta_{G,\mathrm{Bah}}
  \lesssim
  \left(\log\!\left(\frac2\eta\right)\right)^2
  \frac{\tau_n}{\sqrt{n\lambda}}
  \left(\log\!\left(\frac1\lambda\right)\right)^{1/(2\gamma)}
  \log\log\!\left(\frac1\lambda\right).
\]
Lemmas~\ref{lem:tail-sum-sigma-bound}
and~\ref{lem:sigma-bound-exponential-decay} imply
\[
  \sigma_{\splq}(m,\opv)
  \lesssim
  \sqrt{\frac{\tau_n}{\lambda}}\,
  m^{\frac12-\frac\gamma2}
  \exp\!\left(-\frac{\beta m^\gamma}{2}\right).
\]
Together with \(\bar U\lesssim\tau_n/\sqrt\lambda\), this yields
\[
\begin{aligned}
  \Delta_{G,\mathrm{coup}}
  \lesssim\inf_{m\geq1}\bigg\{&
  \sqrt{\frac{\tau_n}{\lambda}}\,
  m^{\frac12-\frac\gamma2}e^{-\beta m^\gamma/2}
  +\frac{\tau_n}{\sqrt{n\lambda}}m^2
  \log\!\left(\frac{m^2}{\eta}\right)
  \bigg\}\log\!\left(\frac2\eta\right).
\end{aligned}
\]
Take
\[
  m=\left\lceil
  \left\{\frac1\beta\log\!\left(\frac n{\tau_n}\right)\right\}^{1/\gamma}
  \right\rceil.
\]
The sample-size condition makes this choice well defined for all sufficiently
large \(n\). Since \(m^{1/2-\gamma/2}\leq1+m^2\), substitution gives
\[
  \Delta_{G,\mathrm{coup}}
  \lesssim
  \frac{\tau_n}{\sqrt{n\lambda}}
  \left(\log\!\left(\frac n{\tau_n}\right)\right)^{2/\gamma}
  \log\log\!\left(\frac n{\tau_n}\right)
  \log\!\left(\frac2\eta\right).
\]
The triangle inequality, a union bound, and the Gaussian-bridge
representation used in the polynomial case complete the proof.
\end{proof}

\begin{proposition}
  \label{prop:bootstrap-approximation-exponential-boundedness}
Suppose \(\mu_k\leq\alpha\exp(-\beta k^\gamma)\) for some
\(\alpha,\beta,\gamma>0\). Under Assumptions~\ref{assump:error},
\ref{assump:source-condition}, and
\ref{assump:covariate-shift-estimation}, together with
Definition~\ref{def:bernstein-moment-condition-for-weight}, when
\(\tau_n\geq\tau_0\), if \(\tau_n\to\infty\), \(\lambda\to0\), and
\[
n>
\max\left\{
\frac{(\tau_n\kappa^2+1)\log(e+\pd)}{C_2\lambda},
\frac{\tau_n}{\lambda\pd}
\right\},
\]
where \(C_2\) is the constant in Lemma~\ref{lem:bound-bahadur-2} evaluated at
the component tolerance \(\eta_{\mathrm{tol}}=c\eta^2\) used in
Corollary~\ref{cor:main-approximate-bootstrap}, then, on a data event of
probability at least \(1-\eta\), there exists a coupling of \(B\) and \(Z\),
conditionally on \(\mathcal D\), such that
\(Z\mid\mathcal D\sim\mathrm{N}_{\splq}(0,\opv)\) and
\[
\begin{aligned}
\mathbb P\Biggl\{
\|Z-B\|_{\splq}
\lesssim{}&
\left(\log\!\left(\frac2\eta\right)\right)^3
\frac{\tau_n}{\sqrt{n\lambda}}
\left(\log\!\left(\frac1\lambda\right)\right)^{1/\gamma}\\
&+
\left(\log\!\left(\frac2\eta\right)\right)^{3/2}
\left(\frac{\tau_n^3}{n\lambda}\right)^{1/4}
\left(\log\!\left(\frac1\lambda\right)\right)^{1/(4\gamma)}
\left(\log\!\left(\frac n{\tau_n\lambda}\right)\right)^{1/(4\gamma)}
\;\Bigm|\;\mathcal D
\Biggr\}
\geq1-\eta.
\end{aligned}
\]
\end{proposition}

\begin{proof}[Proof of Proposition~\ref{prop:bootstrap-approximation-exponential-boundedness}]
Apply Corollary~\ref{cor:main-approximate-bootstrap}. Proposition
\ref{prop:main-feasible-bootstrap} and
Lemma~\ref{lem:effective-dimension-bound-exponential-decay} give
\[
  \Delta_{B,\mathrm{feas}}
  \lesssim
  \left(\log\!\left(\frac2\eta\right)\right)^3
  \frac{\tau_n}{\sqrt{n\lambda}}
  \left(\log\!\left(\frac1\lambda\right)\right)^{1/\gamma}.
\]
As in the polynomial case, the sample-size condition and
Lemma~\ref{lem:sigma-l2-trace-operator-norm-bound} imply
\[
  \frac{\bar U^2\tr_{\splq}(\opv)}{n}
  +\frac{\bar U^4}{n^2}
  \lesssim
  \frac{\tau_n^3\pd}{n\lambda}.
\]
Thus Proposition~\ref{prop:main-bootstrap-coupling} and
Lemmas~\ref{lem:tail-sum-sigma-bound},
\ref{lem:effective-dimension-bound-exponential-decay}, and
\ref{lem:sigma-bound-exponential-decay} yield
\[
\begin{aligned}
  \Delta_{B,\mathrm{coup}}
  \lesssim
  \left(\log\!\left(\frac2\eta\right)\right)^{3/2}
  \inf_{m\geq1}\bigg\{&
  m^{1/4}
  \left(\frac{\tau_n^3}{n\lambda}\right)^{1/4}
  \left(\log\!\left(\frac1\lambda\right)\right)^{1/(4\gamma)}\\
  &+\sqrt{\frac{\tau_n}{\lambda}}\,
  m^{\frac12-\frac\gamma2}e^{-\beta m^\gamma/2}
  \bigg\}.
\end{aligned}
\]
For a sufficiently large fixed constant \(C>1/(2\beta)\), take
\[
  m=\left\lceil
  \left\{C\log\!\left(\frac n{\tau_n\lambda}\right)\right\}^{1/\gamma}
  \right\rceil.
\]
The first term in braces is bounded by
\[
  \left(\frac{\tau_n^3}{n\lambda}\right)^{1/4}
  \left(\log\!\left(\frac1\lambda\right)\right)^{1/(4\gamma)}
  \left(\log\!\left(\frac n{\tau_n\lambda}\right)\right)^{1/(4\gamma)}.
\]
Because \(C\beta/2>1/4\), the exponential factor in the spectral-tail
term absorbs its polynomial factor, so that term is bounded by the same
display for all sufficiently large \(n\). Combining the coupling and
feasible-bootstrap bounds proves the result.
\end{proof}

\begin{proof}[Proof of Corollary~\ref{cor:bd-density}]
Let \(\bar\rho<\infty\) be an almost-sure upper bound for the density
ratio and choose \(C_0\geq\max\{\bar\rho,\tau_0\}\).  Then
\(w_n=\rho\), \(\oplpw=\oplq\), and
Assumption~\ref{assump:covariate-shift-estimation} holds with
\(\mathfrak C_1=1\).  A bounded density ratio also satisfies
Definition~\ref{def:bernstein-moment-condition-for-weight}: one may take
\(\gamma=1\) and
\(W=\sigma^2=\max\{\bar\rho,1\}\).

We next justify the use of the spectral approximation bounds with a fixed
envelope.  Under either two-sided spectral condition in the corollary,
\(\pd\to\infty\) as \(\lambda\to0\).  Since \(\tau_n=C_0\), the first,
fully nonasymptotic inequality in Lemma~\ref{lem:bound-bahadur-1} implies
the same simplified Bahadur bound used in
Proposition~\ref{prop:main-bahadur}: the linear term is absorbed by the square-root
term under the stated sample-size condition, while
\(\tau_n\kappa^2+1\lesssim C_0\) and
\(\log(e+\pd)\asymp\log\pd\) for all sufficiently large \(n\).
The remaining Gaussian and bootstrap arguments use only this bound, the
fixed truncated-weight envelope, and their displayed sample-size conditions.
Thus their nonasymptotic conclusions, and hence the specialized spectral
bounds, remain valid for \(\tau_n=C_0\).

First consider polynomial spectral decay.  At
\(\eta\asymp n^{-1}\), Proposition~\ref{prop:main-bias} and
Propositions~\ref{prop:gaussian-approximation-polynomial-boundedness}
and~\ref{prop:bootstrap-approximation-polynomial-boundedness} give
\begin{equation}
\label{eq:poly-total-approximation-envelope}
\Delta
\lesssim
\log^3 n\,
\bigl(\Delta_0+\Delta_1+\Delta_2+\Delta_3\bigr),
\end{equation}
where
\[
\begin{aligned}
\Delta_0&=\sqrt n\,\lambda,\\
\Delta_1&=
n^{-\frac{\beta-1}{2(\beta+3)}}
\lambda^{-\frac12}
\tau_n^{\frac{\beta+1}{\beta+3}},\\
\Delta_2&=
n^{-\frac12}
\lambda^{-\frac12-\frac1\beta}\tau_n,\\
\Delta_3&=
n^{-\frac{\beta-1}{2(2\beta-1)}}
\lambda^{-\frac{\beta^2+\beta-1}{2\beta(2\beta-1)}}
\tau_n^{\frac{3\beta-2}{2(2\beta-1)}}.
\end{aligned}
\]
Here the Gaussian Bahadur term
\(\tau_n n^{-1/2}\lambda^{-1/2-1/(2\beta)}\) is absorbed by
\(\Delta_2\), since their ratio is
\(\lambda^{1/(2\beta)}\to0\).  The factor \(\log^3n\) in
\eqref{eq:poly-total-approximation-envelope} is obtained directly from the
component bounds: the Gaussian Bahadur and coupling terms contribute at most
\(\log^3n\) and \(\log^2n\), respectively, whereas the feasible-bootstrap and
bootstrap-coupling terms contribute at most \(\log^3n\) and
\(\log^{3/2}n\), respectively.

Write \(\lambda=n^{-a}\).  Since \(\tau_n=C_0\), the four powers of \(n\) in
\eqref{eq:poly-total-approximation-envelope} are
\[
 \delta_0=\frac12-a,
 \quad
 \delta_1=-\frac{\beta-1}{2(\beta+3)}+\frac a2,
 \quad
 \delta_2=-\frac12+a\left(\frac12+\frac1\beta\right),
 \quad
 \delta_3=-\frac{\beta-1}{2(2\beta-1)}
 +a\frac{\beta^2+\beta-1}{2\beta(2\beta-1)}.
\]
The exponent \(\delta_0\) is decreasing in \(a\), whereas
\(\delta_1,\delta_2,\delta_3\) are increasing.  The balance
\(\delta_0=\delta_2\) gives \(a=2\beta/(3\beta+2)\), but at this value
\(\delta_1>\delta_0\) and \(\delta_3>\delta_0\), so it cannot be optimal.

Balancing \(\delta_0=\delta_1\) gives
\[
a_{01}=\frac{2(\beta+1)}{3(\beta+3)},
\qquad
A_{01}=\frac{-\beta+5}{6(\beta+3)}.
\]
At this value, \(\delta_2<A_{01}\), and
\[
\delta_1-\delta_3
=-\frac{(\beta-1)(\beta^2-12\beta+2)}
{6\beta(2\beta-1)(\beta+3)}.
\]
Balancing \(\delta_0=\delta_3\) gives
\[
a_{03}=\frac{\beta(3\beta-2)}{5\beta^2-\beta-1},
\qquad
A_{03}=\frac{-\beta^2+3\beta-1}{2(5\beta^2-\beta-1)}.
\]
At this value, \(\delta_2<A_{03}\), and
\[
\delta_1-\delta_0
=
\frac{(\beta-1)(-\beta^2+12\beta-2)}
{2(\beta+3)(5\beta^2-\beta-1)}.
\]
Consequently, with \(\beta_0=6+\sqrt{34}\), the sharp optimizer is
\[
 (a_\beta^\sharp,r_\beta^\sharp)
 =
 \begin{cases}
 \left(
 \dfrac{2(\beta+1)}{3(\beta+3)},
 \dfrac{\beta-5}{6(\beta+3)}
 \right),&5<\beta\leq\beta_0,\\[3mm]
 \left(
 \dfrac{\beta(3\beta-2)}{5\beta^2-\beta-1},
 \dfrac{\beta^2-3\beta+1}{2(5\beta^2-\beta-1)}
 \right),&\beta>\beta_0.
 \end{cases}
\]
The switch occurs when \(\beta^2-12\beta+2=0\).  Thus the sharper,
piecewise choice \(a=a_\beta^\sharp\) yields
\[
\Delta\lesssim n^{-r_\beta^\sharp}\log^3 n
\]
and, by Theorem~\ref{thm:main-coverage},
\[
\left|
\PP\{f_*\notin\cs_\alpha(\what f)\}-\alpha
\right|
\lesssim
n^{-1}+n^{-r_\beta^\sharp}\log^3n.
\]

For the non-piecewise statement in the corollary, instead set
\[
 \widetilde a_\beta
 =\frac{2(\beta+1)}{3(\beta+3)},
 \qquad
 \widetilde r_\beta
 =\frac{(\beta-5)(\beta+10)}{6(\beta+3)(2\beta+5)}.
\]
At \(a=\widetilde a_\beta\), the four exponents simplify to
\[
 \delta_0=\delta_1
 =-\frac{\beta-5}{6(\beta+3)},
 \qquad
 \delta_2
 =-\frac{(\beta-1)(\beta+4)}{6\beta(\beta+3)},
\]
and
\[
 \delta_3
 =-\frac{(\beta-2)(\beta^2+4\beta-1)}
 {6\beta(\beta+3)(2\beta-1)}.
\]
Writing \(x=\beta-5>0\), direct calculation gives
\[
 -\delta_0-\widetilde r_\beta
 =-\delta_1-\widetilde r_\beta
 =\frac{(\beta-5)^2}{6(\beta+3)(2\beta+5)}\geq0,
\]
\[
 -\delta_2-\widetilde r_\beta
 =\frac{x^3+21x^2+192x+540}
 {6\beta(\beta+3)(2\beta+5)}>0,
\]
and
\[
 -\delta_3-\widetilde r_\beta
 =\frac{97x^2+879x+1980}
 {6\beta(\beta+3)(2\beta-1)(2\beta+5)}>0.
\]
Consequently, \(\delta_j\leq-\widetilde r_\beta\) for
\(j=0,1,2,3\), and hence
\[
 \Delta
 \lesssim
 n^{-\widetilde r_\beta}\log^3 n.
\]
Because \(\sqrt{\tau_n}=O(1)\), Theorem~\ref{thm:main-coverage} with
\(\eta\asymp n^{-1}\) proves part~(1).

Under exponential spectral decay, Proposition~\ref{prop:main-bias} and
Propositions~\ref{prop:gaussian-approximation-exponential-boundedness}
and~\ref{prop:bootstrap-approximation-exponential-boundedness} give, after
suppressing only their displayed logarithmic and tolerance factors, the
polynomial-order envelope
\begin{equation}
\label{eq:exp-total-approximation-envelope}
\Delta
\lesssim
\sqrt n\,\lambda
+\frac{\tau_n}{\sqrt{n\lambda}}
+\left(\frac{\tau_n^3}{n\lambda}\right)^{1/4}.
\end{equation}
Since \(\tau_n=C_0\), its three polynomial-order contributions are
\[
 \sqrt n\lambda,
 \qquad
 \frac{1}{\sqrt{n\lambda}},
 \qquad
 \left(\frac{1}{n\lambda}\right)^{1/4}.
\]
The active balance is between the first and third terms, giving
\(\lambda\asymp n^{-3/5}\) and the power \(n^{-1/10}\).  At
\(\eta\asymp n^{-1}\), the logarithmic multiplier of the active bootstrap
coupling term in
Proposition~\ref{prop:bootstrap-approximation-exponential-boundedness} is
\[
 (\log n)^{3/2}
 (\log(1/\lambda))^{1/(4\zeta)}
 (\log(n/\lambda))^{1/(4\zeta)}
 \lesssim
 (\log n)^{\frac32+\frac{1}{2\zeta}}.
\]
The bias term has the same power without this logarithmic multiplier, and all
remaining stochastic terms have a strictly smaller power of \(n\).  A final
application of Theorem~\ref{thm:main-coverage} proves part~(2).
\end{proof}

\subsection{Proof of Corollary~\ref{cor:subexp-density}}

\begin{lemma}[\textbf{Cordes inequality; see }\protect\cite{Cordes_1987}]
  \label{lem:cordes-inequality}
Let \(A\) and \(B\) be bounded, self-adjoint, positive semidefinite operators
on a Hilbert space \(\mathsf G\). Then, for every \(s\in[0,1]\),
\[
\|A^sB^s\|_{\mathsf G}\leq\|AB\|_{\mathsf G}^s.
\]
The positive-semidefinite case follows by applying the strictly positive
version to \(A+\varepsilon I\) and \(B+\varepsilon I\) and letting
\(\varepsilon\downarrow0\).
\end{lemma}

\begin{lemma}[\protect\cite{gogolashvili2023importance}, Lemma~21]
  \label{lem:weight-mse-bound-under-infinite-moment}
Under Definition~\ref{def:bernstein-moment-condition-for-weight}, for every
integer \(m\geq2\) and \(\tau_n>0\),
\[
\mathbb E_{\mathbb Q}
\left(1-\frac{w(X)}{\rho(X)}\right)^2
\leq
\left\{
\frac12\tau_n^{-(m-1)}m!W^{m-2}\sigma^2
\right\}^{1/\gamma}.
\]
\end{lemma}

\begin{lemma}[\protect\cite{gogolashvili2023importance}, Lemma~22]
  \label{lem:lq-lw-bound-under-infinite-moment}
Under Definition~\ref{def:bernstein-moment-condition-for-weight}, for every
integer \(m\geq2\), \(\tau_n>0\), and \(\lambda>0\),
\[
\left\|
(\oplq-\oplpw)(\oplq+\lambda\opid)^{-1}
\right\|
\leq
\kappa\lambda^{-1/2}\pd^{1/2}
\left\{
\frac12\tau_n^{-(m-1)}m!W^{m-2}\sigma^2
\right\}^{1/(2\gamma)}.
\]
\end{lemma}

\begin{lemma}[\protect\cite{gogolashvili2023importance}, Lemma~23]
  \label{lem:bound-a-under-infinite-moment}
Under Assumption~\ref{assump:pseudo-dimn} and
Definition~\ref{def:bernstein-moment-condition-for-weight}, if, for some
integer \(m_0\geq2\) and \(0<\lambda\leq1\),
\[
\tau_n
\geq
\left(
2^{2\gamma-1}\kappa^{2\gamma}E_d^{2\gamma}
m_0!W^{m_0-2}\sigma^2
\right)^{1/(m_0-1)}
\lambda^{-\frac{(1+d)\gamma}{m_0-1}},
\]
then, as operators on \(\mathsf H\),
\begin{enumerate}
\item
\(\|(\oplq-\oplpw)(\oplq+\lambda\opid)^{-1}\|\leq1/2\);
\item
\(\|\oplq(\oplpw+\lambda\opid)^{-1}\|\leq2\);
\item
\(\|\oplq^{1/2}(\oplpw+\lambda\opid)^{-1/2}\|\leq\sqrt2\).
\end{enumerate}
\end{lemma}

\begin{proof}[Proof of Corollary~\ref{cor:subexp-density}]
Let \(\tau_n=n^\delta\) for a fixed \(\delta>0\), and consider a tuning
parameter \(\lambda=n^{-a}\).  The polynomial spectral condition verifies
Assumption~\ref{assump:pseudo-dimn} with \(d=1/\beta\), while the exponential
condition does so with any fixed \(d\in(0,1]\).  Choose a fixed integer
\(m_0\geq2\) such
that
\[
 m_0-1>\frac{a(1+d)\gamma}{\delta}.
\]
It follows that
\[
 \lambda^{-\frac{(1+d)\gamma}{m_0-1}}
 =n^{\frac{a(1+d)\gamma}{m_0-1}}
 =o(n^\delta).
\]
Consequently, for all sufficiently large \(n\),
Lemma~\ref{lem:bound-a-under-infinite-moment} verifies
Assumption~\ref{assump:covariate-shift-estimation} with
\(\mathfrak C_1\leq\sqrt2\).  We may also write
\(\tau_n=\lambda^{-\varrho}\), where \(\varrho=\delta/a\).  Because the
Bernstein-type condition holds for every integer \(m_0\geq2\), this argument
is available for every fixed \(\delta>0\).  Taking \(\delta\) sufficiently
small also ensures all sample-size restrictions in the relevant spectral
approximation propositions, including \(a+\delta<1\).

For polynomial spectral decay, take
\[
 a=\widetilde a_\beta
 =\frac{2(\beta+1)}{3(\beta+3)},
 \qquad
 \widetilde r_\beta
 =\frac{(\beta-5)(\beta+10)}{6(\beta+3)(2\beta+5)}.
\]
  Restoring \(\tau_n=n^\delta\) in the four terms of
  \eqref{eq:poly-total-approximation-envelope}, their powers of \(n\) are
\[
 \delta_0,
 \qquad
 \delta_1+\delta\frac{\beta+1}{\beta+3},
 \qquad
 \delta_2+\delta,
 \qquad
 \delta_3+\delta\frac{3\beta-2}{2(2\beta-1)},
\]
where the preceding proof shows that
\(\delta_j\leq-\widetilde r_\beta\) for \(j=0,1,2,3\).  Both fractional
truncation coefficients in the display are smaller than one.  Hence
\[
 \Delta\lesssim
 n^{-\widetilde r_\beta+\delta}\log^3 n.
\]
The factor \(n^\delta\) is a uniform upper bound for the truncation loss of every
term; the active terms have a strictly smaller truncation power.  Thus any
remaining fixed logarithmic factors are absorbed by this polynomial slack.
Theorem~\ref{thm:main-coverage} now yields
\[
 \big|\PP\{f_*\notin\cs_\alpha(\what f)\}-\alpha\big|
 \lesssim
 n^{-1}+n^{-\widetilde r_\beta+\frac32\delta}\log^3 n.
\]
For part~(1), choose
\[
 \delta
 =\frac{\beta-5}{9(\beta+3)(2\beta+5)},
 \qquad
 \tau_n
 =n^{\frac{\beta-5}{9(\beta+3)(2\beta+5)}}.
\]
Then
\[
 \widetilde r_\beta-\frac32\delta
 =\frac{(\beta-5)(\beta+9)}{6(\beta+3)(2\beta+5)},
\]
and hence the preceding display gives the bound in part~(1).

  For exponential spectral decay, take \(a=3/5\).  Substituting
  \(\lambda=n^{-3/5}\) and \(\tau_n=n^\delta\) into
  \eqref{eq:exp-total-approximation-envelope}, and then multiplying the
  approximation error by the anti-concentration factor
  \(\sqrt{\tau_n}\), gives the conservative bound
\[
 \big|\PP\{f_*\notin\cs_\alpha(\what f)\}-\alpha\big|
 \lesssim
 n^{-1}+n^{-\frac1{10}+\frac32\delta}\log^3 n.
\]
Indeed, after multiplication by \(\sqrt{\tau_n}\), the active coupling term
has truncation power \(5\delta/4\), leaving a polynomial margin of \(\delta/4\) that
absorbs the fixed \(\zeta\)-dependent logarithmic factors.  Substituting
\[
 \delta=\frac1{150},
 \qquad
 \tau_n=n^{1/150},
\]
and using
\[
 \frac1{10}-\frac32\delta=\frac9{100}
\]
proves part~(2).  The two selected values of \(\delta\) satisfy the sufficient
positivity conditions
\[
 0<\delta<\frac{(\beta-5)(\beta+10)}
 {9(\beta+3)(2\beta+5)}
 \qquad\text{and}\qquad
 0<\delta<\frac1{15},
\]
respectively.  These upper bounds are also smaller than \(1-a\) for the
corresponding choices of \(a\), and hence imply \(a+\delta<1\).
\end{proof}

\subsection{Proof of Corollary~\ref{cor:moment-density-polynomial}}

\begin{definition}[Bounded \(\vartheta\)-th moment condition]
  \label{def:vartheta-moment-condition-for-weight}
There exist \(\vartheta>2\) and \(\Omega_\vartheta<\infty\) such that
\[
\mathbb E_{\mathbb P}\!\left[\rho(X)^\vartheta\right]
=
\mathbb E_{\mathbb Q}\!\left[\rho(X)^{\vartheta-1}\right]
\leq\Omega_\vartheta.
\]
\end{definition}

\begin{lemma}
  \label{lem:lq-lw-bound-under-vartheta-moment}
Under Definition~\ref{def:vartheta-moment-condition-for-weight}, for
\(\tau_n>0\) and \(\lambda>0\),
\[
\left\|
(\oplq-\oplpw)(\oplq+\lambda\opid)^{-1}
\right\|
\leq
\kappa\lambda^{-1/2}\pd^{1/2}
\tau_n^{-(\vartheta-1)/2}\Omega_\vartheta^{1/2}.
\]
\end{lemma}

\begin{lemma}
  \label{lem:bound-a-under-vartheta-moment}
Under Assumption~\ref{assump:pseudo-dimn} and
Definition~\ref{def:vartheta-moment-condition-for-weight}, for
\(0<\lambda\leq1\), if
\begin{equation}
\label{eq:tn-condition-vartheta-moment}
\tau_n
\geq
\left(
4\kappa^2\Omega_\vartheta E_d^2\lambda^{-(1+d)}
\right)^{1/(\vartheta-1)},
\end{equation}
then, as operators on \(\mathsf H\),
\begin{enumerate}
\item
\(\|(\oplq-\oplpw)(\oplq+\lambda\opid)^{-1}\|\leq1/2\);
\item
\(\|\oplq(\oplpw+\lambda\opid)^{-1}\|\leq2\);
\item
\(\|\oplq^{1/2}(\oplpw+\lambda\opid)^{-1/2}\|\leq\sqrt2\).
\end{enumerate}
\end{lemma}

The following compatibility reduction is common to
Corollaries~\ref{cor:moment-density-polynomial}
and~\ref{cor:moment-density}.  Set \(q_\theta=\theta-1\).  Then \(q_\theta>1\).
Since \(d\mathbb Q=\rho\,d\mathbb P\), the moment condition in either corollary is
Definition~\ref{def:vartheta-moment-condition-for-weight} with
\(\vartheta=\theta\).  Hence, by
Lemma~\ref{lem:lq-lw-bound-under-vartheta-moment},
\[
 \left\|(\oplq-\oplpw)(\oplq+\lambda\opid)^{-1}\right\|
 \lesssim
 \lambda^{-1/2}\pd^{1/2}\tau_n^{-q_\theta/2}.
\]
Thus the perturbation proof of
Lemma~\ref{lem:bound-a-under-vartheta-moment} verifies
Assumption~\ref{assump:covariate-shift-estimation} whenever
\begin{equation}
 \tau_n\gtrsim (\lambda^{-1}\pd)^{1/q_\theta}.
 \label{eq:moment-corollary-compatibility}
\end{equation}
Once this compatibility condition holds, the proofs of the Gaussian and
bootstrap approximation propositions use only the bounded truncated-weight
envelope and the standing assumptions.  In particular, the Bernstein-type
condition appearing in the statements of the specialized bootstrap
propositions is not used in their proofs.  Repeating those proofs from
Propositions~\ref{prop:main-feasible-bootstrap} and
\ref{prop:main-bootstrap-coupling}, together with the corresponding
effective-dimension and spectral-tail lemmas, yields the same spectral
envelopes under \eqref{eq:moment-corollary-compatibility}.

\begin{proof}[Proof of Corollary~\ref{cor:moment-density-polynomial}]
Suppose that \(\mu_k\asymp k^{-\beta}\).  Lemma
\ref{lem:effective-dimension-bound-polynomial-decay} gives
\(\pd\lesssim\lambda^{-1/\beta}\), so we may take \(d=1/\beta\) in
Assumption~\ref{assump:pseudo-dimn}.  Set
\[
 v_{\beta,\theta}
 =\frac{\beta+1}{\beta q_\theta}.
\]
For \(\lambda=n^{-a}\), because \(n^\delta\to\infty\), the truncation level
\[
 \tau_n
 \asymp n^{av_{\beta,\theta}+\delta}
\]
satisfies \eqref{eq:moment-corollary-compatibility} for all sufficiently
large \(n\).  Setting \(\delta=0\), after multiplying the four terms in
\eqref{eq:poly-total-approximation-envelope} by the
anti-concentration factor \(\sqrt{\tau_n}\), their powers of \(n\) are
\[
\begin{aligned}
 e_0^{(0)}(a)
 &=\frac12-a+\frac{av_{\beta,\theta}}2,\\
 e_1^{(0)}(a)
 &=-\frac{\beta-1}{2(\beta+3)}+\frac a2
 +av_{\beta,\theta}
 \left(\frac{\beta+1}{\beta+3}+\frac12\right),\\
 e_2^{(0)}(a)
 &=-\frac12+a\left(\frac12+\frac1\beta\right)
 +\frac32av_{\beta,\theta},\\
 e_3^{(0)}(a)
 &=-\frac{\beta-1}{2(2\beta-1)}
 +a\frac{\beta^2+\beta-1}{2\beta(2\beta-1)}
 +av_{\beta,\theta}
 \left(\frac{3\beta-2}{2(2\beta-1)}+\frac12\right).
\end{aligned}
\]

For reference, the exact optimizer of these four exponents remains
piecewise.  Let
\[
 \beta_0=6+\sqrt{34},\qquad
 D_1=3\beta q_\theta(\beta+3)+2(\beta+1)^2,
\]
\[
 D_2=q_\theta(5\beta^2-\beta-1)+(\beta+1)(3\beta-2).
\]
Then
\[
 a_{\beta,\theta}^{\sharp}
 =
 \begin{cases}
 \dfrac{2\beta q_\theta(\beta+1)}{D_1},
 &5<\beta\leq\beta_0,\\[2mm]
 \dfrac{\beta q_\theta(3\beta-2)}{D_2},
 &\beta>\beta_0,
 \end{cases}
\]
\[
 t_{\beta,\theta}^{\sharp}
 =a_{\beta,\theta}^{\sharp}v_{\beta,\theta}
 =
 \begin{cases}
 \dfrac{2(\beta+1)^2}{D_1},
 &5<\beta\leq\beta_0,\\[2mm]
 \dfrac{(\beta+1)(3\beta-2)}{D_2},
 &\beta>\beta_0,
 \end{cases}
\]
\[
 r_{\beta,\theta}^{\sharp}
 =
 \begin{cases}
 \dfrac{\beta q_\theta(\beta-5)-4(\beta+1)^2}{2D_1},
 &5<\beta\leq\beta_0,\\[2mm]
 \dfrac{q_\theta(\beta^2-3\beta+1)
 -2(\beta+1)(3\beta-2)}{2D_2},
 &\beta>\beta_0.
 \end{cases}
\]
Indeed, the first branch balances \(e_0^{(0)}=e_1^{(0)}\), the second
balances \(e_0^{(0)}=e_3^{(0)}\), and the crossover equation is
\(\beta^2-12\beta+2=0\).  These sharper piecewise quantities are retained
only for comparison and are not needed in
Corollary~\ref{cor:moment-density-polynomial}.

For the non-piecewise statement, put
\[
 D=3\beta q_\theta(\beta+3)+2(\beta+1)^2,
 \qquad
 a_{\mathrm M}=\frac{2\beta q_\theta(\beta+1)}{D}.
\]
Thus
\[
 \lambda=n^{-a_{\mathrm M}},
 \qquad
 \tau_n\asymp n^{u_{\mathrm p}},
 \qquad
 u_{\mathrm p}
 =a_{\mathrm M}v_{\beta,\theta}+\delta
 =\frac{2(\beta+1)^2}{D}+\delta,
\]
which is the truncation power used in
Corollary~\ref{cor:moment-density-polynomial}.  At
\(\delta=0\), define
\[
 r_{01}
 =\frac{\beta q_\theta(\beta-5)-4(\beta+1)^2}{2D},
\]
\[
 r_{03}
 =\frac{q_\theta(\beta^3+2\beta^2-9\beta+2)
 -4(\beta+1)^2(2\beta-1)}
 {2(2\beta-1)D}
\]
and
\[
 r_{\mathrm M}
 =\frac{\beta q_\theta(\beta-5)(\beta+10)
 -4(\beta+1)^2(2\beta+5)}
 {2(2\beta+5)D}.
\]
Direct substitution gives
\[
 e_0^{(0)}(a_{\mathrm M})
 =e_1^{(0)}(a_{\mathrm M})=-r_{01},
 \qquad
 e_2^{(0)}(a_{\mathrm M})-e_0^{(0)}(a_{\mathrm M})
 =-\frac{2q_\theta(2\beta-1)}{D}<0,
\]
and \(e_3^{(0)}(a_{\mathrm M})=-r_{03}\).  Moreover,
\[
 r_{01}-r_{\mathrm M}
 =\frac{q_\theta\beta(\beta-5)^2}
 {2(2\beta+5)D}>0,
\]
\[
 r_{03}-r_{\mathrm M}
 =\frac{q_\theta(97\beta^2-91\beta+10)}
 {2(2\beta-1)(2\beta+5)D}>0.
\]
Hence every exponent is at most \(-r_{\mathrm M}\).  Restoring the factor
\(n^\delta\) in \(\tau_n\) adds, respectively,
\[
 \frac{\delta}{2},\qquad
 \frac{3\beta+5}{2(\beta+3)}\delta,\qquad
 \frac{3\delta}{2},\qquad
 \frac{5\beta-3}{2(2\beta-1)}\delta
\]
to the four coverage exponents.  All four coefficients are at most
\(3/2\), so Theorem~\ref{thm:main-coverage} gives
\[
 \big|\PP\{f_*\notin\cs_\alpha(\what f)\}-\alpha\big|
 \lesssim
 n^{-1}+n^{-\{r_{\mathrm M}-3\delta/2\}}\log^3n.
\]
The displayed expression for \(r_{\mathrm M}\) is the polynomial rate before
fixing the truncation slack.  Its positivity is equivalent to
\[
 q_\theta>
 \frac{4(\beta+1)^2(2\beta+5)}
 {\beta(\beta-5)(\beta+10)},
\]
which is assumed in Corollary~\ref{cor:moment-density-polynomial}.  Now choose
\[
 \delta
 =\frac{\beta q_\theta(\beta-5)(\beta+10)
 -4(\beta+1)^2(2\beta+5)}
 {3(2\beta+5)(\beta+10)D}
 =\frac{2r_{\mathrm M}}{3(\beta+10)}.
\]
Since \(\beta+10>1\), this choice lies strictly between zero and the
admissible upper bound \(2r_{\mathrm M}/3\).  Substitution gives
\[
 u_{\mathrm p}
 =\frac{\beta q_\theta(\beta-5)(\beta+10)
 +2(3\beta+28)(\beta+1)^2(2\beta+5)}
 {3(2\beta+5)(\beta+10)D}
\]
and, writing the resulting coverage exponent as \(r_{\mathrm p}\),
\[
 r_{\mathrm p}
 :=r_{\mathrm M}-\frac{3\delta}{2}
 =\frac{\beta+9}{\beta+10}r_{\mathrm M}
 =\frac{(\beta+9)
 \{\beta q_\theta(\beta-5)(\beta+10)
 -4(\beta+1)^2(2\beta+5)\}}
 {2(\beta+10)(2\beta+5)D},
\]
which are the truncation power and coverage exponent stated in
Corollary~\ref{cor:moment-density-polynomial}.
Moreover, the assumed moment threshold and the definition of \(D\) imply
\[
 0<r_{\mathrm p}
 <\frac{(\beta+9)(\beta-5)}
 {6(2\beta+5)(\beta+3)}
 <1.
\]
Consequently, the \(n^{-1}\) term in the preceding coverage bound is
absorbed by \(n^{-r_{\mathrm p}}\log^3 n\).

It remains to verify the listed transparent sufficient regimes.  The
numerator of \(r_{\mathrm M}\) can be written as
\[
 (\theta-9)\beta(\beta-5)(\beta+10)
 +4(\beta^2-112\beta-5).
\]
It is therefore positive when \(\theta\geq9\) and \(\beta\geq113\).
For the flexible regime, put
\[
 R(\beta)=
 \frac{4(\beta+1)^2(2\beta+5)}
 {\beta(\beta-5)(\beta+10)}.
\]
For every \(\beta>5\),
\[
 (\beta-5)\{R(\beta)-8\}
 =\frac{4(-\beta^2+112\beta+5)}{\beta(\beta+10)}
 \leq\frac{144}{5};
\]
the rational function in the middle is decreasing on \((5,\infty)\).
Consequently,
\(\theta>9\) and
\(\beta>5+144/\{5(\theta-9)\}\) imply
\(R(\beta)<\theta-1\).  Finally, the same numerator is increasing in
both parameters on each of
\(\{\theta\geq10,\beta\geq20\}\) and
\(\{\theta\geq11,\beta\geq15\}\), and its values at the respective
lower corners are \(1620\) and \(1660\).  This proves the remaining two
sufficient regimes.

Finally, direct substitution of the selected polynomial slack gives
\[
1-a_{\mathrm M}-u_{\mathrm p}
 =\frac{2\beta q_\theta(\beta+10)(3\beta^2+28\beta+55)
 +4(\beta+1)^2(2\beta+5)}
 {3(2\beta+5)(\beta+10)D}>0.
\]
Under the assumed two-sided polynomial spectral decay,
Lemma~\ref{lem:spectral-summation-rates} gives
\(\pd\asymp\lambda^{-1/\beta}\).  Together with the strict margin above,
this implies \(\tau_n/(n\lambda\pd)\to0\); the remaining fixed logarithmic
factors are absorbed by the same margin.  Thus all sample-size conditions in
the relevant approximation propositions hold for sufficiently large \(n\),
completing the proof.
\end{proof}

\subsection{Proof of Corollary~\ref{cor:moment-density}}

\begin{proof}
Suppose that \(\mu_k\asymp\exp(-c k^\zeta)\), and write
\[
 q_\theta=\theta-1,
 \qquad
 a=\frac{3q_\theta}{5q_\theta+3},
 \qquad
 \lambda=n^{-a}.
\]
By Lemma~\ref{lem:effective-dimension-bound-exponential-decay},
\[
 \pd\lesssim\{\log(1/\lambda)\}^{1/\zeta}.
\]
Choose a sufficiently large fixed constant \(C>0\) and set
\[
 \tau_n
 =C\left[
 \lambda^{-1}\{\log(1/\lambda)\}^{1/\zeta}
 \right]^{1/q_\theta}.
\]
Then \eqref{eq:moment-corollary-compatibility} holds and
\[
 \tau_n
 \asymp
 n^{\frac{3}{5\theta-2}}
 (\log n)^{\frac{1}{\zeta(\theta-1)}},
\]
as stated in the corollary.

After multiplying the three terms in
\eqref{eq:exp-total-approximation-envelope} by the anti-concentration factor
\(\sqrt{\tau_n}\), their polynomial powers of \(n\) are
\[
 e_0=\frac12-a+\frac{a}{2q_\theta},
 \qquad
 e_1=-\frac12+\frac a2+\frac{3a}{2q_\theta},
 \qquad
 e_2=\frac{a-1}{4}+\frac{5a}{4q_\theta}.
\]
Define
\[
 r_\theta
 =\frac{q_\theta-6}{2(5q_\theta+3)}
 =\frac{\theta-7}{2(5\theta-2)}.
\]
Direct substitution gives
\[
 e_0=e_2=-r_\theta,
 \qquad
 e_2-e_1=\frac{q_\theta}{2(5q_\theta+3)}>0.
\]
Thus the bias and bootstrap-coupling terms have the same polynomial order,
while the middle term is polynomially smaller.  The bias term carries the
factor \((\log n)^{1/\{2\zeta q_\theta\}}\).  The bootstrap-coupling term
carries
\[
 (\log n)^{\frac32+\frac{1}{2\zeta}
 +\frac{5}{4\zeta q_\theta}},
\]
where the final exponent is contributed by the factor
\(\tau_n^{5/4}\).  Hence the bootstrap-coupling term dominates
logarithmically, and the fixed logarithmic factors of the polynomially
smaller terms are absorbed.
Theorem~\ref{thm:main-coverage} consequently yields
\[
 \big|\PP\{f_*\notin\cs_\alpha(\what f)\}-\alpha\big|
 \lesssim
 n^{-1}
 n^{-r_\theta}
 (\log n)^{\frac32+\frac{1}{2\zeta}
 +\frac{5}{4\zeta q_\theta}}
 \lesssim
 n^{-r_\theta}
 (\log n)^{\frac32+\frac{1}{2\zeta}
 +\frac{5}{4\zeta q_\theta}},
\]
where \(0<r_\theta<1\) because \(\theta>7\).

Finally,
\[
 1-a-\frac{a}{q_\theta}
 =\frac{2q_\theta}{5q_\theta+3}>0.
\]
Hence both sample-size restrictions in the relevant approximation
propositions hold for all sufficiently large \(n\); the strict polynomial
margin absorbs the additional logarithmic factors in \(\tau_n\) and in the
failure-probability tolerances.  This completes the proof.
\end{proof}

\subsection{Proof of Theorem~\ref{thm:power}}

\begin{proof}[Proof of Theorem \ref{thm:power}]
The type II error probability is
\[
P\left(\sqrt{n}\|\hat f-f_0\|_{\splq}<\hat t_{\alpha}\right)
=
P\left(\|\hat f-f_0\|_{\splq}<\frac{\hat t_{\alpha}}{\sqrt n}\right).
\]
By the triangle inequality,
\[
\|f_*-f_0\|_{\splq}
\le
\|\hat f-f_0\|_{\splq}
+
\|\hat f-f_*\|_{\splq}.
\]
Hence
\[
\left\{\|\hat f-f_0\|_{\splq}<\frac{\hat t_{\alpha}}{\sqrt n}\right\}
\subseteq
\left\{
\delta
<
\|\hat f-f_*\|_{\splq}
+
\frac{\hat t_{\alpha}}{\sqrt n}
\right\}.
\]
Since \(\delta>0\) is fixed and
\[
\|\hat f-f_*\|_{\splq}
+
\frac{\hat t_{\alpha}}{\sqrt n}
\xrightarrow{P}0,
\]
the probability of the event on the right-hand side tends to zero. Therefore
\[
P\left(\sqrt{n}\|\hat f-f_0\|_{\splq}<\hat t_{\alpha}\right)\to 0,
\]
which is equivalent to
\[
P\left(\sqrt{n}\|\hat f-f_0\|_{\splq}>\hat t_{\alpha}\right)\to 1.
\]
\end{proof}

\subsection{Proof of Proposition~\ref{prop:gaussian-coupling-subgaussian}}

\begin{proof}
Fix \(\eta\in(0,1)\) and an integer \(m\in\{1,\ldots,n\}\).
Again let \(A_m:\splq\to\mathbb{R}^m\) be the coordinate projection onto the
first \(m\) eigenfunctions of \(\opv\). Since \(U_i\) is
\(\splq\)-\(b\)-sub-Gaussian, \(A_mU_i\) is \(b\)-sub-Gaussian in
\(\mathbb{R}^m\), and its covariance is bounded in operator norm by
\(\|\opv\|_{\splq}\). If this covariance is singular, apply the
finite-dimensional sub-Gaussian coupling theorem of
\cite{gaussianapproximationsubgaussian} on the range of
\(\Pi_m\opv\Pi_m\), whose dimension is at most \(m\) and whose covariance is
nonsingular. Every zero-variance coordinate of the centered vector
\(A_mU_i\) vanishes almost surely, so the resulting Gaussian vector can be
extended by zero on the orthogonal complement. Thus, in either case, the
theorem gives a Gaussian vector \(\zeta\in\mathbb{R}^m\) with covariance
\(\operatorname{diag}(\nu_1,\ldots,\nu_m)\) such that, with probability at
least \(1-\eta/3\),
\[
  \left\|
  \frac1{\sqrt n}\sum_{i=1}^n \Pi_mU_i-A_m^*\zeta
  \right\|_{\splq}
  \lesssim
  \|\opv\|_{\splq}^{1/2} b^3
  \frac{m^{3/2}\log(n)^2}{\sqrt n}
  (3/\eta)^{1/\log(mn)}.
\]
Let
\[
  Z=A_m^*\zeta+\sum_{s=m+1}^{\infty}\sqrt{\nu_s}\,h_s e_s,
\]
where \((h_s)_{s>m}\) are independent standard Gaussian variables,
independent of the data and of \(\zeta\). Then
\(Z\sim\mathrm{N}_{\splq}(0,\opv)\).

For the empirical tail, the \(\splq\)-sub-Gaussianity of \(U_i\) gives the
moment bound
\[
  \left\{
  \mathbb{E}\|(I-\Pi_m)U_i\|_{\splq}^{p}
  \right\}^{1/p}
  \lesssim b\sqrt p\,\sigma_{\splq}(m,\opv),
  \qquad p\ge2.
\]
The Hilbert-space Bernstein inequality \cite[Corollary~1]{pinelis1986remarks} therefore yields, with probability at
least \(1-\eta/3\),
\[
  \left\|
  \frac1{\sqrt n}\sum_{i=1}^n (I-\Pi_m)U_i
  \right\|_{\splq}
  \lesssim
  b\sigma_{\splq}(m,\opv)\sqrt{\log\!\left(\frac{6}{\eta}\right)}
  +\frac{b\sigma_{\splq}(m,\opv)\log\!\left(\frac{6}{\eta}\right)}{\sqrt n}.
\]
The last term is dominated by the first when
\(n\ge\log\!\left(\frac{2}{\eta}\right)\). Finally, Borell's inequality
applied to \((I-\Pi_m)Z\), whose second moment is
\(\sigma_{\splq}^2(m,\opv)\), gives
\[
  \|(I-\Pi_m)Z\|_{\splq}
  \lesssim
  \sigma_{\splq}(m,\opv)\sqrt{\log\!\left(\frac{6}{\eta}\right)}
\]
with probability at least \(1-\eta/3\). A union bound and the triangle
inequality complete the proof.
\end{proof}

\subsection{Proof of Proposition~\ref{prop:bootstrap-coupling-subgaussian}}

\begin{proof}
The anti-symmetric multiplier identity gives the oracle empirical covariance
\(\what{\opv}_{\oracle}\), and the covariance-comparison bound
\eqref{eq:abstract-bootstrap-coupling-l2q} applies with this empirical
covariance. Since adding a common deterministic element does not change an
empirical covariance,
\[
  \what{\opv}_{\oracle}-\opv
  =
  \frac1n\sum_{i=1}^n(U_i\otimes_{\splq}U_i-\opv)
  -
  \left(\frac1n\sum_{i=1}^nU_i\right)
  \otimes_{\splq}
  \left(\frac1n\sum_{i=1}^nU_i\right).
\]

We first make the covariance-estimation step explicit. For every orthogonal
projection \(P\) on \(\splq\) and every \(p\geq2\), the standard moment bound
for a Hilbert-space sub-Gaussian random element gives
\[
  \left\{\mathbb E\|PU_i\|_{\splq}^p\right\}^{1/p}
  \lesssim
  b\sqrt p\,
  \left\{\tr_{\splq}(P\opv P)\right\}^{1/2}.
\]
This follows by applying the sub-Gaussian increment bound to finite-dimensional
projections and then passing to the limit by monotone convergence. Indeed, on
each finite-dimensional range, Gaussian integration and
Definition~\ref{def:sub-gaussian-hilbert-space} give
\[
  \mathbb E\exp\{\theta\|PU_i\|_{\splq}^2\}
  \leq
  \mathbb E_g\exp\{\theta b^2
  \langle g,P\opv P g\rangle_{\splq}\},
\]
for every
\(0<\theta<\{2b^2\|P\opv P\|_{\splq}\}^{-1}\), where \(g\) is standard
Gaussian on that range. The Gaussian quadratic-form bound yields the displayed
moment estimate, and increasing finite-dimensional projections give the
general case.
Because
\(\|u\otimes_{\splq}u\|_{\mathrm{HS},\splq}=\|u\|_{\splq}^2\) and
\(\|\opv\|_{\mathrm{HS},\splq}\leq\tr_{\splq}(\opv)\), it follows that
\[
  \left\{
  \mathbb E\|U_i\otimes_{\splq}U_i-\opv\|_{\mathrm{HS},\splq}^p
  \right\}^{1/p}
  \lesssim
  b^2p\,\sigma_{\splq}^2(0,\opv).
\]
The preceding moment estimate satisfies the Bernstein moment condition with
variance proxy of order \(b^4\sigma_{\splq}^4(0,\opv)\) and scale of order
\(b^2\sigma_{\splq}^2(0,\opv)\). Hence the Hilbert-space Bernstein inequality
\cite[Corollary~1]{pinelis1986remarks}, applied in the Hilbert space of
Hilbert--Schmidt operators, yields
\[
  \left\|
  \frac1n\sum_{i=1}^n(U_i\otimes_{\splq}U_i-\opv)
  \right\|_{\mathrm{HS},\splq}
  \lesssim
  b^2\sigma_{\splq}^2(0,\opv)
  \left\{
    \sqrt{\frac{\log(2/\eta)}n}
    +\frac{\log(2/\eta)}n
  \right\}.
\]
Independence and Definition~\ref{def:sub-gaussian-hilbert-space} also show
directly that \(n^{-1}\sum_{i=1}^nU_i\) is \(\splq\)-\(b\)-sub-Gaussian with
covariance \(\opv/n\). Applying the preceding norm-moment bound with
\(p=\max\{2,\log(2/\eta)\}\) and then Markov's inequality gives
\[
  \left\|\frac1n\sum_{i=1}^nU_i\right\|_{\splq}
  \lesssim
  b\sigma_{\splq}(0,\opv)
  \sqrt{\frac{\log(2/\eta)}n}.
\]
Combining the last three displays and enlarging the logarithmic factor gives
\[
  \Delta_{\mathrm{HS}}
  \lesssim
  \left(\log\!\left(\frac{2}{\eta}\right)\right)^2
  \frac{b^2\sigma_{\splq}^2(0,\opv)}{\sqrt n}.
\]

For the tail, take \(P=\Pi_m^\perp\). The centering identity above gives
\[
\begin{aligned}
  \tr_{\splq}\{P(\what{\opv}_{\oracle}-\opv)P\}
  ={}&
  \frac1n\sum_{i=1}^n
  \left\{
    \|PU_i\|_{\splq}^2-
    \mathbb E\|PU_i\|_{\splq}^2
  \right\}\\
  &-
  \left\|P\left(\frac1n\sum_{i=1}^nU_i\right)\right\|_{\splq}^2.
\end{aligned}
\]
Moreover,
\[
  \left\{
  \mathbb E\left|
    \|PU_i\|_{\splq}^2-\mathbb E\|PU_i\|_{\splq}^2
  \right|^p
  \right\}^{1/p}
  \lesssim
  b^2p\,\sigma_{\splq}^2(m,\opv).
\]
Scalar Bernstein controls the first term in the tail identity. The same
sub-Gaussian closure argument applied to
\(P(n^{-1}\sum_{i=1}^nU_i)\), whose covariance trace is
\(\sigma_{\splq}^2(m,\opv)/n\), controls the second term by
\(b^2\sigma_{\splq}^2(m,\opv)\log(2/\eta)/n\).
After replacing \(\eta\) by a fixed fraction and taking a union bound, these
bounds hold simultaneously with probability at least \(1-3\eta\), and
\[
  \Delta_{\mathrm{tail}}
  \lesssim
  \log\!\left(\frac{2}{\eta}\right)
  \frac{b^2\sigma_{\splq}^2(m,\opv)}{\sqrt n}.
\]
Substitution into \eqref{eq:abstract-bootstrap-coupling-l2q} yields
\[
  \|\opv^{1/2}g-G\|_{\splq}
  \lesssim
  \left(\log\!\left(\frac{2}{\eta}\right)\right)^{3/2}
  \left[
    m^{1/4}
    \left\{
      \frac{b^4\sigma_{\splq}^4(0,\opv)}{n}
    \right\}^{1/4}
    +
    \sigma_{\splq}(m,\opv)
  \right].
\]
Here the empirical-tail contribution
\(b\,n^{-1/4}\sigma_{\splq}(m,\opv)\) is absorbed by the first term because
\(m\geq1\) and
\(\sigma_{\splq}(m,\opv)\leq\sigma_{\splq}(0,\opv)\).
Taking the infimum over \(m\ge1\) proves the claim.
\end{proof}

\subsection{Proof of Corollary~\ref{cor:bd-density-subgaussian}}

\paragraph{Polynomial spectral decay.}

\begin{proposition}
  \label{prop:gaussian-approximation-polynomial-sub-gaussianity}
  Suppose \(\mu_k\leq\alpha k^{-\beta}\) for some \(\alpha>0\) and
  \(\beta>1\). Under Assumptions~\ref{assump:error},
  \ref{assump:source-condition}, \ref{assump:covariate-shift-estimation},
  and~\ref{assump:sub-gaussian-feature}, when $\tau_n \ge \tau_0$, if $\tau_n \to \infty$ and $\lambda \to 0$ and
\begin{equation*}
  n
  >
  \max
  \left\{
  \frac{(\tau_n \kappa^2 + 1)\log(e+\pd)}{ C_2\lambda},
  \frac{\tau_n}{\lambda \pd}
  \right\}
\end{equation*}
where $C_2$ is the constant in Lemma~\ref{lem:bound-bahadur-2} after replacing
its failure probability by the fixed fraction of \(\eta\) used in the union
bound, then with probability at least \(1-\eta\), there exists
\(Z\sim\mathrm{N}_{\splq}(0,\opv)\) such that
\begin{equation*}
\begin{aligned}
\left\| \sqrt{n} (\hatf - \fpw) - Z \right\|_{\splq}
\lesssim
\frac{\tau_n}{\sqrt{n}}  \lambda^{-\frac{1}{2} - \frac{1}{2\beta}}  \log\!\left(\frac{1}{\lambda}\right)
\left(\log\!\left(\frac{2}{\eta}\right)\right)^2
+ \sqrt{\frac{\tau_n}{n}}
 \left(\frac{n}{\lambda}\right)^{\frac{3}{2(2+\beta)}} \log (n)^{2} \left(\frac{3}{\eta}\right)^{\frac{1}{\log ( n)}}.
\end{aligned}
\end{equation*}
\end{proposition}

\begin{proof}[Proof of Proposition~\ref{prop:gaussian-approximation-polynomial-sub-gaussianity}]
Lemmas~\ref{lem:sigma-l2-trace-operator-norm-bound},
\ref{lem:tail-sum-sigma-bound},
\ref{lem:effective-dimension-bound-polynomial-decay}, and
\ref{lem:sigma-bound-polynomial-decay} give
\[
\tr_{\splq}(\opv)\lesssim\tau_n\lambda^{-1/\beta},
\qquad
\|\opv\|_{\splq}\lesssim\tau_n,
\qquad
\sigma_{\splq}(m,\opv)
\lesssim\sqrt{\frac{\tau_n}{\lambda}}m^{-(\beta-1)/2}.
\]
By Lemma~\ref{lem:sub-gaussianity-of-ui}, the influence elements are
uniformly \(\splq\)-\(b\)-sub-Gaussian. Apply
Proposition~\ref{prop:gaussian-coupling-subgaussian} with
\[
m=\min\left\{n,
\left\lceil\left(\frac n\lambda\right)^{1/(\beta+2)}\right\rceil
\right\}.
\]
Both coupling terms are then bounded, up to their displayed logarithmic and
tolerance factors, by
\[
\sqrt{\frac{\tau_n}{n}}
\left(\frac n\lambda\right)^{\frac3{2(\beta+2)}}.
\]
The sample-size condition and \(\tau_n\to\infty\) imply \(n\lambda\to\infty\);
hence the untruncated ceiling is at most \(n\) for all sufficiently large
\(n\). Moreover,
\((3/\eta)^{1/\log(mn)}\leq(3/\eta)^{1/\log n}\).

Proposition~\ref{prop:main-bahadur} and
\(\pd\lesssim\lambda^{-1/\beta}\) give
\[
\Delta_{G,\mathrm{Bah}}
\lesssim
\frac{\tau_n}{\sqrt n}\lambda^{-\frac12-\frac1{2\beta}}
\log\!\left(\frac1\lambda\right)
\left(\log\!\left(\frac2\eta\right)\right)^2.
\]
Combining the Bahadur and coupling events after dividing \(\eta\) by a fixed
constant proves the result.
\end{proof}

\begin{proposition}
  \label{prop:bootstrap-approximation-polynomial-sub-gaussianity}
  Suppose \(\mu_k\leq\alpha k^{-\beta}\) for some \(\alpha>0\) and
  \(\beta>1\). Under Assumptions~\ref{assump:error}, \ref{assump:source-condition},
  \ref{assump:covariate-shift-estimation}, and
  \ref{assump:sub-gaussian-feature}, when
  $\tau_n \ge \tau_0$, if $\tau_n \to \infty$ and $\lambda \to 0$ and
\begin{equation*}
  n
  >
  \max
  \left\{
  \frac{(\tau_n \kappa^2 + 1)\log(e+\pd)}{ C_2\lambda},
  \frac{\tau_n}{\lambda \pd}
  \right\}
\end{equation*}
where $C_2$ is the constant in Lemma~\ref{lem:bound-bahadur-2} evaluated at
the component tolerance \(\eta_{\mathrm{tol}}=c\eta^2\) used in
Corollary~\ref{cor:main-approximate-bootstrap}, then, on a data event of
probability at least \(1-\eta\), there exists a coupling of \(B\) and \(Z\),
conditionally on \(\mathcal D\), such that
\(Z\mid\mathcal D\sim\mathrm{N}_{\splq}(0,\opv)\) and
\begin{equation*}
\begin{aligned}
\mathbb{P}\Biggl\{
\|Z-B\|_{\splq}
\lesssim\;&
\left(\log\!\left(\frac{2}{\eta}\right)\right)^{3}
\frac{\tau_n}{\sqrt{n}}
\lambda^{-\frac{1}{2}-\frac{1}{\beta}}
\\
&+\left(\log\!\left(\frac{2}{\eta}\right)\right)^{\frac{3}{2}}
\left(\frac{\tau_{n}}{\lambda}\right)^{\frac{1}{2}}
\left(
\frac{\tau_{n} \lambda^{1-\frac{2}{\beta}}}{\sqrt{n}}
\right)^{\frac{\beta-1}{2 \beta -1}}
\;\Bigm|\; \mathcal{D}
\Biggr\}
\ge 1-\eta.
\end{aligned}
\end{equation*}
\end{proposition}

\begin{proof}[Proof of Proposition~\ref{prop:bootstrap-approximation-polynomial-sub-gaussianity}]
The spectral bounds used in the preceding proof imply
\[
\sigma_{\splq}^2(0,\opv)\lesssim\tau_n\lambda^{-1/\beta},
\qquad
\sigma_{\splq}(m,\opv)
\lesssim\sqrt{\frac{\tau_n}{\lambda}}m^{-(\beta-1)/2}.
\]
In Proposition~\ref{prop:bootstrap-coupling-subgaussian}, take
\[
m=\left\lceil
\max\left\{1,
\left(
\frac{\sqrt n}{\tau_n\lambda^{1-2/\beta}}
\right)^{\frac{2}{2\beta-1}}
\right\}
\right\rceil.
\]
The spectral-tail term is bounded by
\[
\sqrt{\frac{\tau_n}{\lambda}}
\left(
\frac{\tau_n\lambda^{1-2/\beta}}{\sqrt n}
\right)^{\frac{\beta-1}{2\beta-1}}.
\]
The finite-dimensional term obeys the same bound because the ratio of that
term to the preceding display is at most
\(\tau_n^{-1/2}\lambda^{1/(2\beta)}\leq1\). If the expression inside the
maximum is smaller than one, the choice \(m=1\) gives the same conclusion.

On the other hand, Proposition~\ref{prop:main-feasible-bootstrap} and
\(\pd\lesssim\lambda^{-1/\beta}\) give
\[
\Delta_{B,\mathrm{feas}}
\lesssim
\left(\log\!\left(\frac2\eta\right)\right)^3
\frac{\tau_n}{\sqrt n}\lambda^{-\frac12-\frac1\beta}.
\]
The conditional Strassen--Dudley argument in
Lemma~\ref{lem:strassen-coupling}, followed by the same two-layer union bound
as in Corollary~\ref{cor:main-approximate-bootstrap}, couples the oracle
bootstrap to \(Z\) and then combines it with the feasible-bootstrap bound.
The triangle inequality proves the assertion.
\end{proof}

\paragraph{Exponential spectral decay.}

\begin{proposition}
  \label{prop:gaussian-approximation-exponential-sub-gaussianity}
  Suppose \(\mu_k\leq\alpha\exp(-\beta k^\gamma)\) for some
  \(\alpha,\beta,\gamma>0\). Under Assumptions~\ref{assump:error},
  \ref{assump:source-condition}, \ref{assump:covariate-shift-estimation},
  and~\ref{assump:sub-gaussian-feature}, when $\tau_n \ge \tau_0$, if $\tau_n \to \infty$ and $\lambda \to 0$ and
\begin{equation*}
  n
  >
  \max
  \left\{
  \frac{(\tau_n \kappa^2 + 1)\log(e+\pd)}{ C_2\lambda},
  \frac{\tau_n}{\lambda \pd}
  \right\}
\end{equation*}
where $C_2$ is the constant in Lemma~\ref{lem:bound-bahadur-2} after replacing
its failure probability by the fixed fraction of \(\eta\) used in the union
bound, then with probability at least \(1-\eta\), there exists
\(Z\sim\mathrm{N}_{\splq}(0,\opv)\) such that
\begin{equation*}
\begin{aligned}
\left\| \sqrt{n} (\hatf - \fpw) - Z \right\|_{\splq}
&\lesssim
\left(\log \left(\frac{2}{\eta}\right)\right)^{2} \frac{\tau_{n}}{\sqrt{n\lambda} }\left(\log \left(\frac{1}{\lambda}\right)\right)^{\frac{1}{2 \gamma}} \log \log \left(\frac{1}{\lambda}\right)\\
&\hphantom{xxx}
+\left(\frac{1}{\eta}\right)^{\frac{1}{\log n}}\sqrt{\frac{\tau_n}{n}} \left[ \frac{1}{\beta} \log\left(\frac{n}{\lambda}\right) \right]^{\frac{3}{2\gamma}} \left(\log(n)\right)^{2} .
\end{aligned}
\end{equation*}
\end{proposition}

\begin{proof}[Proof of Proposition~\ref{prop:gaussian-approximation-exponential-sub-gaussianity}]
Lemmas~\ref{lem:sigma-l2-trace-operator-norm-bound},
\ref{lem:tail-sum-sigma-bound},
\ref{lem:effective-dimension-bound-exponential-decay}, and
\ref{lem:sigma-bound-exponential-decay} imply
\[
\tr_{\splq}(\opv)
\lesssim\tau_n\left(\log\!\left(\frac1\lambda\right)\right)^{1/\gamma},
\qquad
\|\opv\|_{\splq}\lesssim\tau_n,
\]
and
\[
\sigma_{\splq}(m,\opv)
\lesssim
\sqrt{\frac{\tau_n}{\lambda}}m^{(1-\gamma)/2}
\exp(-\beta m^\gamma/2).
\]
Apply Proposition~\ref{prop:gaussian-coupling-subgaussian}, using
Lemma~\ref{lem:sub-gaussianity-of-ui}, with
\[
m=\min\left\{n,
\left\lceil
\left\{\frac1\beta\log\!\left(\frac n\lambda\right)\right\}^{1/\gamma}
\right\rceil
\right\}.
\]
Since \(e^{-\beta m^\gamma/2}\lesssim(n/\lambda)^{-1/2}\) whenever the
ceiling is at most \(n\), and
\(m^{(1-\gamma)/2}\leq1+m^{3/2}\), both coupling terms are bounded by
\[
\sqrt{\frac{\tau_n}{n}}
\left\{\frac1\beta\log\!\left(\frac n\lambda\right)\right\}^{3/(2\gamma)}
\log^2(n)
\left(\frac1\eta\right)^{1/\log n}.
\]
The sample-size condition and \(\tau_n\to\infty\) imply \(n\lambda\to\infty\);
hence the untruncated ceiling is at most \(n\) for all sufficiently large
\(n\).

Finally, Proposition~\ref{prop:main-bahadur} and the exponential
effective-dimension bound give
\[
\Delta_{G,\mathrm{Bah}}
\lesssim
\left(\log\!\left(\frac2\eta\right)\right)^2
\frac{\tau_n}{\sqrt{n\lambda}}
\left(\log\!\left(\frac1\lambda\right)\right)^{1/(2\gamma)}
\log\log\!\left(\frac1\lambda\right).
\]
A union bound completes the proof.
\end{proof}

\begin{proposition}
  \label{prop:bootstrap-approximation-exponential-sub-gaussianity}
  Suppose \(\mu_k\leq\alpha\exp(-\beta k^\gamma)\) for some
  \(\alpha,\beta,\gamma>0\). Under Assumptions~\ref{assump:error}, \ref{assump:source-condition},
  \ref{assump:covariate-shift-estimation}, and
  \ref{assump:sub-gaussian-feature}, when
  $\tau_n \ge \tau_0$, if $\tau_n \to \infty$ and $\lambda \to 0$ and
\begin{equation*}
  n
  >
  \max
  \left\{
  \frac{(\tau_n \kappa^2 + 1)\log(e+\pd)}{ C_2\lambda},
  \frac{\tau_n}{\lambda \pd}
  \right\}
\end{equation*}
where $C_2$ is the constant in Lemma~\ref{lem:bound-bahadur-2} evaluated at
the component tolerance \(\eta_{\mathrm{tol}}=c\eta^2\) used in
Corollary~\ref{cor:main-approximate-bootstrap}, then, on a data event of
probability at least \(1-\eta\), there exists a coupling of \(B\) and \(Z\),
conditionally on \(\mathcal D\), such that
\(Z\mid\mathcal D\sim\mathrm{N}_{\splq}(0,\opv)\) and
\begin{equation*}
\begin{aligned}
\mathbb{P}\Biggl\{
\|Z-B\|_{\splq}
\lesssim\;&
\left(\log \left(\frac{2}{\eta}\right)\right)^{3} \frac{\tau_{n}}{\sqrt{n \lambda}}
\left(\log \left(\frac{1}{\lambda}\right)\right)^{\frac{1}{\gamma}}\\
& +
\left(\log\!\left(\frac{2}{\eta}\right)\right)^{\frac{3}{2}}
\frac{\sqrt{\tau_n}}{n^{\frac{1}{4}}}
\left(\log\left(\frac{1}{\lambda}\right)\right)^{\frac{1}{2\gamma}}
\left(\log\left(\frac{n}{\lambda}\right)\right)^{\frac{1}{4\gamma}}
\;\Bigm|\; \mathcal{D}
\Biggr\}
\ge 1-\eta.
\end{aligned}
\end{equation*}
\end{proposition}

\begin{proof}[Proof of Proposition~\ref{prop:bootstrap-approximation-exponential-sub-gaussianity}]
The exponential spectral bounds give
\[
\sigma_{\splq}^2(0,\opv)
\lesssim
\tau_n\left(\log\!\left(\frac1\lambda\right)\right)^{1/\gamma}
\]
and
\[
\sigma_{\splq}(m,\opv)
\lesssim
\sqrt{\frac{\tau_n}{\lambda}}m^{(1-\gamma)/2}
e^{-\beta m^\gamma/2}.
\]
In Proposition~\ref{prop:bootstrap-coupling-subgaussian}, take
\[
m=\left\lceil
\left\{\frac1\beta\log\!\left(\frac n\lambda\right)\right\}^{1/\gamma}
\right\rceil.
\]
The finite-dimensional term is then bounded by
\[
\frac{\sqrt{\tau_n}}{n^{1/4}}
\left(\log\!\left(\frac1\lambda\right)\right)^{1/(2\gamma)}
\left(\log\!\left(\frac n\lambda\right)\right)^{1/(4\gamma)}.
\]
The spectral-tail term is at most
\[
\sqrt{\frac{\tau_n}{n}}
\left(\log\!\left(\frac n\lambda\right)\right)^{(1-\gamma)/(2\gamma)},
\]
which is smaller than the preceding display for all sufficiently large
\(n\), because a polynomial factor \(n^{-1/4}\) dominates every fixed power
of the logarithm.

Proposition~\ref{prop:main-feasible-bootstrap} and the exponential
effective-dimension bound give
\[
\Delta_{B,\mathrm{feas}}
\lesssim
\left(\log\!\left(\frac2\eta\right)\right)^3
\frac{\tau_n}{\sqrt{n\lambda}}
\left(\log\!\left(\frac1\lambda\right)\right)^{1/\gamma}.
\]
Combining the sub-Gaussian oracle coupling with the feasible bootstrap through
Lemma~\ref{lem:strassen-coupling}, and applying the two-layer union bound from
Corollary~\ref{cor:main-approximate-bootstrap}, proves the result.
\end{proof}

\begin{proof}[Proof of Corollary~\ref{cor:bd-density-subgaussian}]
Choose \(C_0\geq\max\{\overline\rho,\tau_0\}\). Then \(w_n=\rho\),
\(\oplpw=\oplq\), and
Assumption~\ref{assump:covariate-shift-estimation} holds with
\(\mathfrak C_1=1\). We first establish the required spectral approximation
bounds directly for the fixed envelope \(\tau_n=C_0\). Because
\(w_n=\rho\leq\overline\rho\leq C_0\), every occurrence of the truncation
envelope in the nonasymptotic Bahadur, covariance, and feasible-bootstrap
bounds is controlled by \(C_0\). Repeating the proofs of the four
sub-Gaussian spectral approximation propositions with this fixed envelope,
rather than invoking their asymptotic statements, gives the same displayed
approximation bounds under the corresponding sample-size conditions. Thus
\(\tau_n\to\infty\) is not required in the bounded-ratio case.
The corresponding sample-size conditions hold for all sufficiently large
\(n\) under each tuning choice below.

We next establish the polynomial result. Since \(s\geq1/2\),
Proposition~\ref{prop:main-bias} gives
\(\Delta_{\mathrm{bias}}\lesssim\sqrt n\lambda\). At
\(\eta\asymp n^{-1}\), the fixed-envelope polynomial bounds just established
yield
\begin{equation}
\label{eq:poly-subgaussian-total-envelope}
\Delta
\lesssim
\log^3n\,
\bigl(\Delta^{\mathrm{sg}}_0+\Delta^{\mathrm{sg}}_1
+\Delta^{\mathrm{sg}}_2+\Delta^{\mathrm{sg}}_3\bigr),
\end{equation}
where
\[
\begin{aligned}
\Delta^{\mathrm{sg}}_0
&=\sqrt n\lambda,\\
\Delta^{\mathrm{sg}}_1
&=\sqrt{\frac{\tau_n}{n}}
\left(\frac n\lambda\right)^{\frac3{2(\beta+2)}},\\
\Delta^{\mathrm{sg}}_2
&=\frac{\tau_n}{\sqrt n}\lambda^{-\frac12-\frac1\beta},\\
\Delta^{\mathrm{sg}}_3
&=\sqrt{\frac{\tau_n}{\lambda}}
\left(
\frac{\tau_n\lambda^{1-2/\beta}}{\sqrt n}
\right)^{\frac{\beta-1}{2\beta-1}}.
\end{aligned}
\]
Here the Bahadur term is absorbed by
\(\Delta^{\mathrm{sg}}_2\), and
\((3/\eta)^{1/\log n}=O(1)\).

Before deriving the simpler rate stated in the corollary, we record the
sharper optimization underlying \eqref{eq:poly-subgaussian-total-envelope}.
With \(\tau_n=C_0\) and \(\lambda=n^{-a}\), the four powers of \(n\) are
\[
\frac12-a,\qquad
-\frac12+\frac{3(1+a)}{2(\beta+2)},\qquad
-\frac12+a\left(\frac12+\frac1\beta\right),\qquad
-\frac{\beta-1}{2(2\beta-1)}
+a\frac{5\beta-4}{2\beta(2\beta-1)}.
\]
For \(\beta>3\), their sharp balance is
\begin{equation}
\label{eq:poly-subgaussian-sharp-optimizer}
(a_\beta^\sharp,r_\beta^\sharp)
=
\begin{cases}
\left(
\dfrac{\beta(3\beta-2)}{4\beta^2+3\beta-4},
\dfrac{2\beta^2-7\beta+4}{2(4\beta^2+3\beta-4)}
\right),
&3<\beta\leq3+\sqrt5,\\[4mm]
\left(
\dfrac{2\beta}{3\beta+2},
\dfrac{\beta-2}{2(3\beta+2)}
\right),
&\beta>3+\sqrt5.
\end{cases}
\end{equation}
Indeed, the first branch balances
\(\Delta^{\mathrm{sg}}_0\) and \(\Delta^{\mathrm{sg}}_3\), whereas the second
balances \(\Delta^{\mathrm{sg}}_0\) and
\(\Delta^{\mathrm{sg}}_2\); the two formulas meet at
\(\beta=3+\sqrt5\), and direct substitution shows that the remaining terms
are smaller. Consequently,
\[
\Delta\lesssim n^{-r_\beta^\sharp}\log^3n.
\]
More generally, if \(\tau_n=n^t\) and the same \(a_\beta^\sharp\) is used,
the four powers of \(\tau_n\) in
\(\Delta^{\mathrm{sg}}_0,\ldots,\Delta^{\mathrm{sg}}_3\) are
\[
0,\qquad \frac12,\qquad 1,\qquad
\frac{4\beta-3}{2(2\beta-1)}<1.
\]
Thus, whenever \(a_\beta^\sharp+t<1\),
\[
\Delta\lesssim n^{-r_\beta^\sharp+t}\log^3n,
\qquad
\sqrt{\tau_n}\Delta
\lesssim n^{-r_\beta^\sharp+3t/2}\log^3n.
\]

For the non-piecewise statement of the corollary, set
\[
\widetilde a_\beta=\frac{2\beta}{3\beta+2},
\qquad
\widetilde r_\beta=\frac{\beta-3}{2(3\beta+2)}.
\]
At \(a=\widetilde a_\beta\), the negatives of the four powers of \(n\) are
\[
\frac{\beta-2}{2(3\beta+2)},\qquad
\frac{3\beta^2-7\beta-2}{2(\beta+2)(3\beta+2)},\qquad
\frac{\beta-2}{2(3\beta+2)},\qquad
\frac{(\beta-3)(3\beta-2)}
{2(3\beta+2)(2\beta-1)}.
\]
Each is at least \(\widetilde r_\beta\) when \(\beta>3\). Therefore
\[
\Delta\lesssim n^{-\widetilde r_\beta}\log^3n.
\]
Since \(\sqrt{\tau_n}=O(1)\), Theorem~\ref{thm:main-coverage} proves
part~(1).

Under exponential spectral decay, the fixed-envelope exponential bounds just
established give the polynomial-order envelope
\begin{equation}
\label{eq:exp-subgaussian-total-envelope}
\Delta
\lesssim_{\mathrm{poly}}
\sqrt n\lambda
+\frac{\tau_n}{\sqrt{n\lambda}}
+\sqrt{\frac{\tau_n}{n}}
+\frac{\sqrt{\tau_n}}{n^{1/4}},
\end{equation}
where \(\lesssim_{\mathrm{poly}}\) suppresses only the displayed logarithmic
and tolerance factors. With fixed \(\tau_n\) and \(\lambda=n^{-a}\), the
potentially active powers are
\[
\frac12-a,\qquad \frac{a-1}{2},\qquad -\frac14.
\]
The first two are balanced at the sharp polynomial-order choice \(a=2/3\),
where they equal \(-1/6\), while the last equals \(-1/4\). More generally,
with \(\lambda=n^{-2/3}\) and \(\tau_n=n^t\),
\[
\Delta
\lesssim
n^{-1/6+t}(\log n)^{3+1/\zeta},
\qquad
\sqrt{\tau_n}\Delta
\lesssim
n^{-1/6+3t/2}(\log n)^{3+1/\zeta}.
\]
For \(t=0\), the active stochastic term is the feasible-bootstrap term, whose
logarithmic factor is \((\log n)^{3+1/\zeta}\); all other stochastic terms
have strictly smaller powers of \(n\), so their fixed logarithmic factors are
absorbed. Theorem~\ref{thm:main-coverage} now proves part~(2).
\end{proof}

\subsection{Proof of Corollary~\ref{cor:subexp-density-subgaussian}}

\begin{proof}
Let \(\tau_n=n^t\) and \(\lambda=n^{-a}\). As in the proof of
Corollary~\ref{cor:subexp-density}, polynomial spectral decay verifies
Assumption~\ref{assump:pseudo-dimn} with \(d=1/\beta\), while exponential
spectral decay does so with any fixed \(d\in(0,1]\). Under
Definition~\ref{def:bernstein-moment-condition-for-weight}, we may choose a
fixed integer \(m_0\) sufficiently large that
\[
m_0-1>\frac{a(1+d)\gamma}{t}.
\]
Lemma~\ref{lem:bound-a-under-infinite-moment} then verifies
Assumption~\ref{assump:covariate-shift-estimation} with
\(\mathfrak C_1\leq\sqrt2\) for all sufficiently large \(n\).

For polynomial spectral decay, use
\[
a=\widetilde a_\beta=\frac{2\beta}{3\beta+2},
\qquad
t=\frac{\beta-3}{6(3\beta+2)}
=\frac{\widetilde r_\beta}{3}.
\]
By the calculation following
\eqref{eq:poly-subgaussian-sharp-optimizer}, restoring \(\tau_n=n^t\)
enlarges the approximation error by at most \(n^t\), and
Theorem~\ref{thm:main-coverage} contributes one further factor \(n^{t/2}\).
Consequently,
\[
\big|\PP\{f_*\notin\cs_\alpha(\what f)\}-\alpha\big|
\lesssim
n^{-1}
+n^{-\widetilde r_\beta+3t/2}\log^3n
\lesssim
n^{-\frac{\beta-3}{4(3\beta+2)}}\log^3n.
\]
The selected powers satisfy \(a+t<1\), so the sample-size conditions in the
polynomial approximation propositions hold for all sufficiently large \(n\).
This proves part~(1).

For exponential spectral decay, take \(a=2/3\) and \(t=1/150\). The
calculation following \eqref{eq:exp-subgaussian-total-envelope} gives
\[
\Delta\lesssim
n^{-1/6+t}(\log n)^{3+1/\zeta}
=n^{-4/25}(\log n)^{3+1/\zeta}.
\]
After multiplication by the anti-concentration factor \(\sqrt{\tau_n}\),
\[
\big|\PP\{f_*\notin\cs_\alpha(\what f)\}-\alpha\big|
\lesssim
n^{-1}
+n^{-1/6+3t/2}(\log n)^{3+1/\zeta}
\lesssim
n^{-47/300}(\log n)^{3+1/\zeta}.
\]
Here \(a+t<1\), and the strictly smaller polynomial powers of the remaining
terms absorb their fixed logarithmic factors. This proves part~(2).
\end{proof}

\subsection{Proof of Corollary~\ref{cor:moment-density-subgaussian}}

\begin{proof}
By the common moment-compatibility reduction in
\eqref{eq:moment-corollary-compatibility}, it is sufficient that
\[
\tau_n\gtrsim(\lambda^{-1}\pd)^{1/(\theta-1)}.
\]
The proof of
Proposition~\ref{prop:bootstrap-approximation-exponential-sub-gaussianity}
uses only this compatibility condition, the bounded truncated-weight
envelope, and the standing assumptions; the Bernstein-type condition is not
used.

Write \(\lambda=n^{-a}\) and \(\tau_n=n^t\). Under exponential spectral
decay,
\(\pd\lesssim\{\log(1/\lambda)\}^{1/\zeta}\). With
\[
a=\frac{8(\theta^2-1)}{3(4\theta^2+3\theta-3)},\qquad
t=\frac{19\theta+1}{6(4\theta^2+3\theta-3)},
\]
we have
\[
t-\frac{a}{\theta-1}
=\frac{\theta-5}{2(4\theta^2+3\theta-3)}>0.
\]
Thus the strict polynomial slack absorbs the logarithmic effective-dimension
factor and verifies \eqref{eq:moment-corollary-compatibility}.

After suppressing logarithmic factors, the four approximation terms in
\eqref{eq:exp-subgaussian-total-envelope} have powers
\[
\frac12-a,\qquad
-\frac12+\frac a2+t,\qquad
-\frac12+\frac t2,\qquad
-\frac14+\frac t2.
\]
The feasible-bootstrap term is the largest. Indeed, its excess over the bias
term equals
\[
-1+\frac{3a}{2}+t
=\frac{\theta-5}{6(4\theta^2+3\theta-3)}>0,
\]
and its excesses over the final two terms are, respectively,
\((a+t)/2>0\) and
\(-1/4+(a+t)/2>0\). Its negative exponent is
\[
\frac12-\frac a2-t
=\frac{2\theta^2-5\theta-1}
{3(4\theta^2+3\theta-3)}.
\]
Therefore
\[
\Delta
\lesssim
n^{-\frac{2\theta^2-5\theta-1}
{3(4\theta^2+3\theta-3)}}
(\log n)^{3+1/\zeta}.
\]
Multiplying by \(\sqrt{\tau_n}\) gives
\[
\frac{2\theta^2-5\theta-1}
{3(4\theta^2+3\theta-3)}-\frac t2
=
\frac{(\theta-5)(8\theta+1)}
{12(4\theta^2+3\theta-3)},
\]
which is positive for \(\theta>5\). The active feasible-bootstrap term has
logarithmic factor \((\log n)^{3+1/\zeta}\), while all remaining terms have
strictly smaller powers of \(n\).

Finally,
\[
1-a-t
=\frac{8\theta^2-\theta-3}
{6(4\theta^2+3\theta-3)}>0,
\]
so both sample-size restrictions hold for all sufficiently large \(n\).
Since the displayed coverage exponent is smaller than one, the \(n^{-1}\)
term in Theorem~\ref{thm:main-coverage} is absorbed. This proves the
corollary.
\end{proof}

\section{Proofs of Technical Lemmas}
\label{app:proofs-of-technical-lemmas}

This appendix collects auxiliary proofs and the derivation of the bootstrap
matrix representation.

\subsection{Additional lemmas}

\begin{lemma}
  \label{lem:tq-self-adjoint-in-l2q}
  Let $\opt: \mathsf{H} \to \mathsf{H}$ be a bounded linear self-adjoint
  operator on $\mathsf{H}$.
Then $\opt \oplq: \splq \to \splq$ is a self-adjoint operator in $\splq$.
\end{lemma}

\begin{lemma}
  \label{lem:bound-bahadur-2-2}
  For \(0<\eta<1\) and \(0<\lambda\leq1\), if
  \begin{equation*}
  \frac{(\tau_n \kappa^2 + 1)\log(e+\pd)}{n \lambda}
  <
  C_2
\end{equation*}
where $C_2$ is the constant in Lemma \ref{lem:bound-bahadur-2}, then, on the same event of probability at least $1 - \eta$ as in Lemma \ref{lem:bound-bahadur-1},
\begin{equation*}
  \bigl\|
    (\oplpnw + \lambda \opid)^{-\frac{1}{2}}
    (\oplpw + \lambda \opid)^{\frac{1}{2}}
  \bigr\|
  \le \sqrt{2}.
\end{equation*}
\end{lemma}

\subsection{Proof of Lemma~\ref{lem:LQ_half_unitary_isomorphism}}

\begin{proof}
  Let \(f=\sum_{\mu_j>0} a_j\phi_j\in\overline{\ran(\oplq)}\). Then
  \[
  \oplq^{1/2}f=\sum_{\mu_j>0}\mu_j^{1/2}a_j\phi_j.
  \]
  By Lemma~\ref{lem:eigendecomposition_LQ}, this element belongs to \(\mathsf{H}\), and for
  \(g=\sum_{\mu_j>0} b_j\phi_j\in\overline{\ran(\oplq)}\),
  \[
  \left\langle \oplq^{1/2}f,\oplq^{1/2}g\right\rangle_{\mathsf{H}}
  =
  \sum_{\mu_j>0} a_j b_j
  =
  \langle f,g\rangle_{\splq}.
  \]
  Thus \(\oplq^{1/2}\) preserves inner products on \(\overline{\ran(\oplq)}\). It remains to note that it is onto \(\mathsf{H}\): if \(h=\sum_{\mu_j>0}h_j\phi_j\in\mathsf{H}\), then
  \[
  f_h:=\sum_{\mu_j>0}\frac{h_j}{\mu_j^{1/2}}\phi_j
  \]
  belongs to \(\overline{\ran(\oplq)}\), since
  \(\|f_h\|_{\splq}^2=\sum_{\mu_j>0}h_j^2/\mu_j=\|h\|_{\mathsf{H}}^2<\infty\), and
  \(\oplq^{1/2}f_h=h\). Hence \(\oplq^{1/2}:\overline{\ran(\oplq)}\to\mathsf{H}\) is a unitary isomorphism.
\end{proof}

\subsection{Proof of Lemma~\ref{lem:general-order-of-summation}}

\begin{proof}
  Since the eigenvalues are ordered non-increasingly, the definition of
  \(s^*\) implies that \(\mu_i>\lambda\) for \(i<s^*\) and
  \(\mu_i\le\lambda\) for \(i\ge s^*\). For sufficiently small
  \(\lambda\), \(s^*\ge2\). Hence, for \(i<s^*\),
  \[
    2^{-k}
    \le
    \left(\frac{\mu_i}{\mu_i+\lambda}\right)^k
    \le
    1,
  \]
  and therefore
  \[
    \sum_{i=1}^{s^*-1}
    \left(\frac{\mu_i}{\mu_i+\lambda}\right)^k
    \asymp s^* .
  \]
  On the other hand, for \(i>s^*\) we have \(\mu_i\le\lambda\), so
  \[
    2^{-k}\frac{\mu_i^k}{\lambda^k}
    \le
    \left(\frac{\mu_i}{\mu_i+\lambda}\right)^k
    \le
    \frac{\mu_i^k}{\lambda^k}.
  \]
  Consequently,
  \[
    \sum_{i=s^*+1}^{\infty}
    \left(\frac{\mu_i}{\mu_i+\lambda}\right)^k
    \asymp
    \frac{1}{\lambda^k}
    \sum_{i=s^*+1}^{\infty}\mu_i^k .
  \]
  The remaining term \(i=s^*\) is bounded above by one and below by
  zero, and is therefore absorbed by the \(s^*\) term. Combining the
  preceding two displays gives the claimed equivalence.
\end{proof}

\subsection{Proof of Lemma~\ref{lem:spectral-summation-rates}}

\begin{proof}
  In the polynomial case, \(s^*\asymp\lambda^{-1/\beta}\) and
  \[
    \lambda^{-k}\sum_{i=s^*+1}^{\infty}\mu_i^k
    \asymp
    \lambda^{-k}(s^*)^{1-k\beta}
    \asymp
    \lambda^{-1/\beta}.
  \]
  Lemma~\ref{lem:general-order-of-summation} gives the first rate.
  In the exponential case, \(s^*\asymp\{\log(1/\lambda)\}^{1/\gamma}\), while
  an integral comparison gives
  \[
    \lambda^{-k}\sum_{i=s^*+1}^{\infty}\mu_i^k
    \lesssim
    s^* .
  \]
  Lemma~\ref{lem:general-order-of-summation} then gives the second rate,
  since the leading \(s^*\) term also provides the matching lower bound.
\end{proof}

\subsection{Proof of Lemma~\ref{lem:sigma-lower-operator-bound}}

\begin{proof}
Decompose the influence element \(U_i\) in \eqref{eq:bahadur} as
\(U_{i,1}+U_{i,2}\), where
\begin{equation*}
  U_{i,1}
  = \bigl(\oplpw + \lambda \opid\bigr)^{-1} w(X_i)\,\eps_i k_{X_i},
  \qquad
  U_{i,2}
  = \bigl(\oplpw + \lambda \opid\bigr)^{-1}
    \bigl(w(X_i) k_{X_i} \otimes k_{X_i} - \oplpw\bigr)\bigl(f_* - \fpw\bigr).
\end{equation*}

As \( \mathbb{E}[\eps_i \mid X_i] = 0 \), we have
\begin{equation*}
  \begin{aligned}
  \opv
  &= \mathbb{E}\bigl[ U_i \otimes_{\splq} U_i \bigr]\\
  &= \mathbb{E}\bigl[ U_{i,1} \otimes_{\splq} U_{i,1} \bigr]
    + \mathbb{E}\bigl[ U_{i,2} \otimes_{\splq} U_{i,2} \bigr]\\
  &\succeq_{\splq} \mathbb{E}\bigl[ U_{i,1} \otimes_{\splq} U_{i,1} \bigr]\\
  &\succeq_{\splq} \underline{\sigma}^2 \,
      \mathbb{E}\Bigl\{
        \bigl[\bigl(\oplpw + \lambda \opid\bigr)^{-1} w(X_i) k_{X_i}\bigr]
        \otimes_{\splq}
        \bigl[\bigl(\oplpw + \lambda \opid\bigr)^{-1} w(X_i) k_{X_i}\bigr]
      \Bigr\}.
      \end{aligned}
\end{equation*}
\end{proof}

\subsection{Proof of Lemma~\ref{lem:simplification-of-a-useful-operator}}

\begin{proof}
For any \(f \in \splq\),
\begin{equation*}
\begin{aligned}
&\mathbb{E}\!\left\{
w(X_i)\left[
(\oplpw+\lambda \opid)^{-1} k_{X_i}
\right]
\otimes_{\splq}
\left[
\bigl((\oplpw+\lambda \opid)^{-1} k_{X_i}\bigr)
\right]
\right\} f\\
=& \mathbb{E}\!\left[
w(X_i)\,
\left\langle (\oplpw+\lambda \opid)^{-1} k_{X_i}, f \right\rangle_{\splq}\,
(\oplpw+\lambda \opid)^{-1} k_{X_i}
\right] \\
=& \mathbb{E}\!\left[
w(X_i)\,
\left\langle (\oplpw+\lambda \opid)^{-1} k_{X_i}, \oplq f \right\rangle\,
(\oplpw+\lambda \opid)^{-1} k_{X_i}
\right] \\
=& \mathbb{E}\!\left[
w(X_i)\,
(\oplpw+\lambda \opid)^{-1}
\left\langle k_{X_i}, (\oplpw+\lambda \opid)^{-1} \oplq f \right\rangle\,
k_{X_i}
\right] \\
=& (\oplpw+\lambda \opid)^{-1} \oplpw (\oplpw+\lambda \opid)^{-1} \oplq f.
\end{aligned}
\end{equation*}
\end{proof}

\subsection{Proof of Lemma~\ref{lem:same-spectrum-of-tlq-and-lq-half-tlq-half}}

\begin{proof}
  Let
  \[
    \mathsf H_Q:=\overline{\ran(\oplq)}
    \quad\text{and}\quad
    U:=\oplq^{1/2}:\mathsf H_Q\to\mathsf H.
  \]
  By Lemma~\ref{lem:LQ_half_unitary_isomorphism}, \(U\) is a unitary
  isomorphism.  Moreover,
  \(\splq=\mathsf H_Q\oplus\ker(\oplq)\), and
  \(\opt\oplq\) vanishes on \(\ker(\oplq)\).  For \(f,g\in\mathsf H_Q\),
  \[
  \begin{aligned}
    \langle f,\opt\oplq g\rangle_{\splq}
    &=
    \left\langle
      Uf,\,
      \oplq^{1/2}\opt\oplq^{1/2}Ug
    \right\rangle_{\mathsf H}.
  \end{aligned}
  \]
  Since \(\opt\) is positive self-adjoint,
  \(\oplq^{1/2}\opt\oplq^{1/2}\) is positive self-adjoint.  The preceding
  identity therefore shows that the restriction of \(\opt\oplq\) to
  \(\mathsf H_Q\) is also positive self-adjoint.  Its nonzero eigenvalues are
  precisely those of \(\opt\oplq\) on \(\splq\).

  We may now apply the min--max theorem and
  Lemma~\ref{lem:LQ_half_unitary_isomorphism}:
  \begin{equation*}
    \begin{aligned}
    \mu_k(\opt \oplq, \splq)
    &= \max_{\substack{V \subset \mathsf H_Q \\ \dim(V) = k}}
    \min_{\substack{f \in V \\ \|f\|_{\splq} = 1}}
    \langle f, \opt \oplq f \rangle_{\splq} \\
    &= \max_{\substack{V \subset \mathsf H_Q \\ \dim(V) = k}}
    \min_{\substack{f \in V \\ \|\oplq^{1/2} f\|_{\mathsf{H}} = 1}}
    \langle \oplq^{1/2} f, \oplq^{1/2} \opt \oplq^{1/2} \oplq^{1/2} f \rangle_{\mathsf{H}} \\
    &= \max_{\substack{W \subset \mathsf{H} \\ \dim(W) = k}}
    \min_{\substack{g \in W \\ \|g\|_{\mathsf{H}} = 1}}
    \langle g, \oplq^{1/2} \opt \oplq^{1/2} g \rangle_{\mathsf{H}} \\
    &= \mu_k(\oplq^{1/2} \opt \oplq^{1/2}, \mathsf{H}).
    \end{aligned}
  \end{equation*}
\end{proof}

\subsection{Proof of Lemma~\ref{lem:source-condition-bound-1}}

\begin{proof}
  If $0\le s\le \frac{1}{2}$ in Assumption \ref{assump:source-condition},
\begin{equation*}
    \begin{aligned}
  \left\| \bigl(\oplpw + \lambda \opid\bigr)^{-\frac{1}{2}} f_* \right\|
  =& \left\| \bigl(\oplpw + \lambda \opid\bigr)^{-\frac{1}{2}} \oplpwo^s g_{\ast} \right\| \\
  \le& \left\| \bigl(\oplpwo + \lambda \opid\bigr)^{-\frac{1}{2}} \oplpwo^s g_{\ast} \right\|\\
  \le& \lambda^{-\frac{1}{2} + s} \| g_{\ast} \|.
    \end{aligned}
\end{equation*}
If $\frac{1}{2} < s \le 1$ in Assumption \ref{assump:source-condition},
\begin{equation*}
\begin{aligned}
\left\|(\oplpw+\lambda \opid)^{-\frac{1}{2}} f_*\right\|
&\le \left\|\oplpwo^{s-\frac{1}{2}} g_{\ast}\right\| \\
&\le \left\|\oplq^{s-\frac{1}{2}} g_{\ast}\right\| \\
&\le \kappa^{2s-1}\left\|g_{\ast}\right\|.
\end{aligned}
\end{equation*}
In conclusion,
\begin{equation*}
\left\| (\oplpw + \lambda \opid)^{-\frac{1}{2}} f_* \right\|
\le \max\{\kappa^{2s-1}, 1\}\,\lambda^{\min\{s-\frac{1}{2}, 0\}}\,\|g_{\ast}\|.
\end{equation*}
\end{proof}

\subsection{Proof of Lemma~\ref{lem:bias-bound-in-h}}

\begin{proof}
Since
\[
  \fpw-f_*
  =
  -\lambda(\oplpw+\lambda\opid)^{-1}f_*,
\]
we have
\begin{align*}
  \|\fpw-f_*\|
  &=
  \lambda
  \left\|
    (\oplpw+\lambda\opid)^{-1}f_*
  \right\| \\
  &\le
  \lambda^{1/2}
  \left\|
    (\oplpw+\lambda\opid)^{-1/2}f_*
  \right\| \\
  &\le
  \max\{\kappa^{2s-1},1\}
  \lambda^{1/2+\min\{s-1/2,0\}}
  \|g_{\ast}\| \\
  &=
  \max\{\kappa^{2s-1},1\}
  \lambda^{\min\{s,1/2\}}
  \|g_{\ast}\|.
\end{align*}
The first inequality follows from
\[
  \left\|
    \lambda^{1/2}(\oplpw+\lambda\opid)^{-1/2}
  \right\|
  \le 1,
\]
and the second follows from
Lemma~\ref{lem:source-condition-bound-1}.
\end{proof}

\subsection{Proof of Lemma~\ref{lem:sigma-l2-trace-operator-norm-bound}}

\begin{proof}
  With the same notation as in the proof of Lemma \ref{lem:sigma-lower-operator-bound},
  \begin{equation*}
    \begin{aligned}
      \opv &= \mathbb{E} [U_{i,1} \otimes_{\splq} U_{i,1}]+ \mathbb{E} [U_{i,2} \otimes_{\splq} U_{i,2}] \\
      &=: \opv_1 + \opv_2
    \end{aligned}
  \end{equation*}

  \begin{equation*}
    \begin{aligned}
      \opv_1 &= \mathbb{E} \left\{ \left[(\oplpw + \lambda \opid)^{-1} w(X_i) \eps_i k_{X_i}\right] \otimes_{\splq} \left[(\oplpw + \lambda \opid)^{-1} w(X_i) \eps_i k_{X_i}\right] \right\} \\
&\preceq_{\splq} \tau_n \overline{\sigma}^2 \mathbb{E} \left\{ w(X_i) [(\oplpw + \lambda \opid)^{-1} k_{X_i}] \otimes_{\splq} [(\oplpw + \lambda \opid)^{-1} k_{X_i}] \right\}
    \end{aligned}
  \end{equation*}

By Lemma \ref{lem:simplification-of-a-useful-operator},
\begin{equation*}
\opv_1 \preceq_{\splq}\tau_n\overline{\sigma}^2\cdot
(\oplpw+\lambda \opid)^{-1} \oplpw (\oplpw+\lambda \opid)^{-1} \oplq .
\end{equation*}
Denote $\opt_1 = (\oplpw+\lambda \opid)^{-1} \oplpw (\oplpw+\lambda \opid)^{-1} \oplq$, $\opt_2 = \oplq^{\frac{1}{2}} (\oplpw+\lambda \opid)^{-1} \oplpw (\oplpw+\lambda \opid)^{-1} \oplq^{\frac{1}{2}}$.

With Lemma \ref{lem:same-spectrum-of-tlq-and-lq-half-tlq-half},
\begin{equation}\label{eq:t1-l2q-operator-norm-bound}
\begin{aligned}
\|\opt_1\|_{\splq} = \|\opt_2\|_\h \le \bigl(\mathfrak{C}_1\bigr)^2 .
\end{aligned}
\end{equation}
Then,
\begin{equation*}
\begin{aligned}
\|\opv_1\|_{\splq}
\le \tau_n \overline{\sigma}^2 \|\opt_1\|_{\splq}
\le \tau_n \overline{\sigma}^2 \bigl(\mathfrak{C}_1\bigr)^2 .
\end{aligned}
\end{equation*}

$\forall f \in \splq$,
\begin{equation*}
  \begin{aligned}
\left\langle f, \opv_1 f \right\rangle_{\splq}
\le \left\langle
f,\,
\tau_n \overline{\sigma}^2 \opt_1 f
\right\rangle_{\splq}
=\tau_n \overline{\sigma}^2
\left\langle
\oplq^{1/2} f,\,
\opt_2 \oplq^{1/2} f
\right\rangle .
  \end{aligned}
\end{equation*}

Similarly with Lemma \ref{lem:same-spectrum-of-tlq-and-lq-half-tlq-half},
\begin{equation}\label{eq:t1-l2-trace-bound}
\begin{aligned}
\operatorname{tr}_{\splq}(\opv_1)
&= \sum_{k=1}^\infty \mu_k(\opv_1 , \splq)\\
&\le \tau_n \overline{\sigma}^2 \sum_{k=1}^\infty
\mu_k(\opt_1,\splq) \\
&= \tau_n \overline{\sigma}^2 \sum_{k=1}^\infty
\mu_k(\opt_2, \h) \\
&= \tau_n \overline{\sigma}^2 \tr_{\mathsf{H}}(\opt_2) \\
&\leq \tau_n \overline{\sigma}^2
\bigl(\mathfrak{C}_1\bigr)^2
\tr_{\mathsf H}\bigl((\oplpw+\lambda\opid)^{-1}\oplpw\bigr) \\
&\leq \tau_n \overline{\sigma}^2
\bigl(\mathfrak{C}_1\bigr)^2 \pd.
\end{aligned}
\end{equation}
The last inequality follows from the domination of the truncated weighted-source
operator by the target operator.  Indeed, since \(w\le\rho\), for every
\(h\in\mathsf H\),
\[
  \langle h,(\oplq-\oplpw)h\rangle_{\mathsf H}
  =
  \int_{\mathfrak X}(\rho(x)-w(x))|h(x)|^2\,d\mathbb P(x)
  \ge0,
\]
and hence \(\oplpw\preceq_{\mathsf H}\oplq\).  The min--max principle and the
monotonicity of \(t\mapsto t/(t+\lambda)\) on \([0,\infty)\) therefore give
\[
  \tr_{\mathsf H}\bigl((\oplpw+\lambda\opid)^{-1}\oplpw\bigr)
  \le
  \tr_{\mathsf H}\bigl((\oplq+\lambda\opid)^{-1}\oplq\bigr)
  =
  \pd .
\]

\begin{equation*}
\begin{aligned}
\opv_2
&= \mathbb{E}\Bigl\{
\Bigl[
(\oplpw+\lambda \opid)^{-1}\bigl(w(X_i)\,k_{X_i}\otimes k_{X_i}-\oplpw\bigr)\,(f_*-\fpw)
\Bigr]\otimes_{\splq}
\\
&\hphantom{=\mathbb{E}\Bigl\{====}
\Bigl[
(\oplpw+\lambda \opid)^{-1}\bigl(w(X_i)\,k_{X_i}\otimes k_{X_i}-\oplpw\bigr)\,(f_*-\fpw)
\Bigr]
\Bigr\} \\
&\preceq_{\splq} \mathbb{E}\Bigl\{
\Bigl[
(\oplpw+\lambda \opid)^{-1}\bigl(w(X_i)\,k_{X_i}\otimes k_{X_i}\bigr)\,(f_*-\fpw)
\Bigr]\otimes_{\splq}
\\
&\hphantom{\le \mathbb{E}\Bigl\{======}
\Bigl[
(\oplpw+\lambda \opid)^{-1}\bigl(w(X_i)\,k_{X_i}\otimes k_{X_i}\bigr)\,(f_*-\fpw)
\Bigr]
\Bigr\} \\
&= \mathbb{E}\Bigl\{
\Bigl[w(X_i)\bigl(f_*(X_i)-\fpw(X_i)\bigr)\,(\oplpw+\lambda \opid)^{-1}k_{X_i}\Bigr]\otimes_{\splq}
\\
&\hphantom{=\mathbb{E}\Bigl\{======} 
\Bigl[w(X_i)\bigl(f_*(X_i)-\fpw(X_i)\bigr)\,(\oplpw+\lambda \opid)^{-1}k_{X_i}\Bigr]
\Bigr\} \\
&= \mathbb{E}\Bigl\{
w(X_i)\bigl(f_*(X_i)-\fpw(X_i)\bigr)^2 w(X_i)
\Bigl[(\oplpw+\lambda \opid)^{-1}k_{X_i}\Bigr]\otimes_{\splq}
\Bigl[(\oplpw+\lambda \opid)^{-1}k_{X_i}\Bigr]
\Bigr\} \\
&\preceq_{\splq}
\tau_n \kappa^2
\max\{\kappa^{2s-1},1\}^{2}
\lambda^{2\min\{s,1/2\}}
\norm{g_{\ast}}^2 \\
&\qquad {}\times
\mathbb{E}\Bigl\{
w(X_i)
\Bigl[(\oplpw+\lambda \opid)^{-1}k_{X_i}\Bigr]\otimes_{\splq}
\Bigl[(\oplpw+\lambda \opid)^{-1}k_{X_i}\Bigr]
\Bigr\} \\
&=
\tau_n \kappa^2
\max\{\kappa^{2s-1},1\}^{2}
\lambda^{2\min\{s,1/2\}} \\
&\qquad {}\times
\norm{g_{\ast}}^2\, \opt_1 .
\end{aligned}
\end{equation*}
where we use Lemma \ref{lem:bias-bound-in-h} to get the inequality
\begin{equation*}
\begin{aligned}
\bigl|f_*(X_i)-\fpw(X_i)\bigr|
= \bigl|\left\langle f_*-\fpw,\, k_{X_i}\right\rangle\bigr| 
\le \kappa \|f_*-\fpw\|
\le
\kappa
\max\{\kappa^{2s-1},1\}
\lambda^{\min\{s,1/2\}}
\|g_{\ast}\|.
\end{aligned}
\end{equation*}
Combining this bound with \eqref{eq:t1-l2q-operator-norm-bound} and
\eqref{eq:t1-l2-trace-bound} proves the result.
\end{proof}

\subsection{Proof of Lemma~\ref{lem:anti-concentration-variance-lower-bound}}

\begin{proof}
  With Lemma \ref{lem:sigma-lower-operator-bound} and \ref{lem:simplification-of-a-useful-operator},
  \begin{equation*}
    \begin{aligned}
\opv
&\succeq_{\splq} \underline{\sigma}^2 \,\mathbb{E}\Bigl\{
\Bigl[(\oplpw + \lambda \opid)^{-1} w(X_i) k_{X_i}\Bigr] \otimes_{\splq} \Bigl[(\oplpw + \lambda \opid)^{-1} w(X_i) k_{X_i}\Bigr]
\Bigl\} \\
&\succeq_{\splq} \underline{\rho} \underline{\sigma}^2 (\oplpw + \lambda \opid)^{-1} \oplpw (\oplpw + \lambda \opid)^{-1} \oplq .
    \end{aligned}
  \end{equation*}
  For any $f \in \splq$,
  \begin{equation*}
    \begin{aligned}
\langle f, \opv f \rangle_{\splq}
&\ge \underline{\rho} \underline{\sigma}^2 \langle f, (\oplpw + \lambda \opid)^{-1} \oplpw (\oplpw + \lambda \opid)^{-1} \oplq f \rangle_{\splq} \\
&= \underline{\rho} \underline{\sigma}^2 \langle \oplq f, (\oplpw + \lambda \opid)^{-1} \oplpw (\oplpw + \lambda \opid)^{-1} \oplq f \rangle_\h \\
&\ge \bigl(\mathfrak{C}_2\bigr)^{-2}
\underline{\rho} \underline{\sigma}^2
\langle \oplq f, (\oplq + \lambda \opid)^{-1} \oplq (\oplq + \lambda \opid)^{-1} \oplq f \rangle_{\h} \\
&= \bigl(\mathfrak{C}_2\bigr)^{-2}
\underline{\rho} \underline{\sigma}^2
\langle f, (\oplq + \lambda \opid)^{-1} \oplq (\oplq + \lambda \opid)^{-1} \oplq f \rangle_{\splq}.
    \end{aligned}
  \end{equation*}
\end{proof}

\subsection{Proof of Lemma~\ref{lem:bound-bahadur-1}}

\begin{proof}
Let
\[
\eta(x) = (\oplpw + \lambda \opid)^{-\frac12}\, w(x)\,(k_x \otimes k_x)(\oplpw + \lambda \opid)^{-\frac12}
\]

$\forall h \in \mathsf{H}$,
\begin{equation*}
\begin{aligned}
\eta(x)h
&= w(x)\,(\oplpw + \lambda \opid)^{-\frac12} 
   \langle k_x, (\oplpw + \lambda \opid)^{-\frac12} h \rangle \, k_x \\
\left[\eta(x)\right]^2 h
&= w^2(x)\,(\oplpw + \lambda \opid)^{-\frac12} 
   \langle k_x, (\oplpw + \lambda \opid)^{-1} \langle k_x, (\oplpw + \lambda \opid)^{-\frac12} h \rangle k_x \rangle \, k_x \\
&= w^2(x)\,\langle k_x, (\oplpw + \lambda \opid)^{-1} k_x \rangle\,
   (\oplpw + \lambda \opid)^{-\frac12} (k_x \otimes k_x)(\oplpw + \lambda \opid)^{-\frac12} h .
\end{aligned}
\end{equation*}

Thus,
\begin{equation*}
\begin{aligned}
\left[\eta(x)\right]^2 
&= w^2(x)\,\langle k_x, (\oplpw + \lambda \opid)^{-1} k_x \rangle (\oplpw + \lambda \opid)^{-\frac12} (k_x \otimes k_x)(\oplpw + \lambda \opid)^{-\frac12}.
\end{aligned}
\end{equation*}

As
\begin{equation*}
\begin{aligned}
 \bigl|\langle k_x , (\oplpw + \lambda \opid)^{-1} k_x \rangle\bigr|
&= \bigl|\operatorname{tr}\bigl((\oplpw + \lambda \opid)^{-1} k_x \otimes k_x \bigr)\bigr| \\
&\leq \lambda^{-1} \kappa^2
\end{aligned}
\end{equation*}

\begin{equation*}
\begin{aligned}
\bigl\| \mathbb{E}[(\eta(x))^2] \bigr\|
&\leq \frac{\tau_n\, \kappa^2}{\lambda}\,
   \bigl\| \mathbb{E}\bigl[ w(x) (\oplpw + \lambda \opid)^{-\frac12}
      (k_x \otimes k_x) (\oplpw + \lambda \opid)^{-\frac12} \bigr] \bigr\| \\
&= \frac{\tau_n\, \kappa^2}{\lambda}\,
   \bigl\| (\oplpw + \lambda \opid)^{-\frac12} \oplpw (\oplpw + \lambda \opid)^{-\frac12} \bigr\| \\
&\leq \frac{\tau_n\, \kappa^2}{\lambda}\, .
\end{aligned}
\end{equation*}

Since \(w_n\leq\rho\), we have \(\oplpw\preceq\oplq\).  The map
\(A\mapsto A(A+\lambda\opid)^{-1}\) is operator monotone on the cone of
positive operators.  Therefore,
\[
\operatorname{Tr}\!\left\{
(\oplpw+\lambda\opid)^{-1}\oplpw
\right\}
\leq
\operatorname{Tr}\!\left\{
(\oplq+\lambda\opid)^{-1}\oplq
\right\}
=\pd,
\]
and hence
\begin{equation*}
    \operatorname{Tr}\Bigl( \mathbb{E}\bigl[(\eta(x))^2\bigr] \Bigr)
\leq \tau_n\, \kappa^2 \lambda^{-1}\, \pd.
\end{equation*}

Let
\[
\xi(x) = \eta(x) - \mathbb{E}[\eta(x)].
\]

\[
\|\xi(x)\|
\leq \|\eta(x)\| + \|\mathbb{E}[\eta(x)]\|
\leq \frac{\tau_n \kappa^2}{\lambda} + 1.
\]

\begin{equation*}
\begin{aligned}
\bigl\| \mathbb{E}\Bigl[\Bigl(\sum_{i=1}^n \xi(X_i)\Bigr)^2\Bigr] \bigr\|
&= \Bigl\| \sum_{i=1}^n \mathbb{E}[\xi(X_i)^2] \Bigr\|
= n \bigl\|\mathbb{E}\eta^2 - (\mathbb{E}\eta)^2\bigr\| \\
&\leq n \bigl\|\mathbb{E}\eta^2\bigr\| + n \bigl\|(\mathbb{E}\eta)^2\bigr\|
\leq n\Bigl(\frac{\tau_n \kappa^2}{\lambda} + 1\Bigr).
\end{aligned}
\end{equation*}

\begin{equation*}
\begin{aligned}
\operatorname{Tr}\Bigl( \mathbb{E}\Bigl[\Bigl(\sum_{i=1}^n \xi(X_i)\Bigr)^2\Bigr] \Bigr)
&= n\, \operatorname{Tr}\Bigl( \mathbb{E}\bigl[(\eta)^2\bigr] - (\mathbb{E}\eta)^2 \Bigr)\\
&\leq  n\operatorname{Tr}\bigl(\mathbb{E}[(\eta)^2]\bigr)\\
&\leq n \tau_n\, \kappa^2 \lambda^{-1} \pd \\
&\leq n\Bigl(\frac{\tau_n \kappa^2}{\lambda} + 1\Bigr) \pd.
\end{aligned}
\end{equation*}

Consider \(0<\lambda\leq1\).

Applying Lemma \ref{lem:hs-inequality} with
\[
R = \frac{\tau_n \kappa^2 + 1}{\lambda},
\qquad
V = n \frac{\tau_n \kappa^2 + 1}{\lambda},
\qquad
W = e+\pd
\]

\[
\forall\, t \ge 
\sqrt{\frac{\tau_n \kappa^2 + 1}{n\lambda}}
+ \frac{\tau_n \kappa^2 + 1}{3 n \lambda}
\]

\begin{equation*}
    \begin{aligned}
&\mathbb{P}\Bigl(
  \bigl\|
    (\oplpw + \lambda \opid)^{-\frac12}
    (\oplpnw - \oplpw)
    (\oplpw + \lambda \opid)^{-\frac12}
  \bigr\|
  \ge t
\Bigr)\\
=&
\mathbb{P}\Bigl(
  \Bigl\|
    \frac{1}{n} \sum_{i=1}^n \xi(X_i)
  \Bigr\|
  \ge t
\Bigr)\\
\le&
4(e+\pd)
\exp\Bigl(
  -\frac{
    \lambda\, n\, t^2
  }{
    2\, (\tau_n \kappa^2 + 1)
    + \dfrac{2t(\tau_n \kappa^2 + 1)}{3}
  }
\Bigr).
    \end{aligned}
\end{equation*}

Let \(\epsilon>0\) be the positive solution of
\[
\frac{n \lambda\, \epsilon^2}
{2(\tau_n \kappa^2 + 1)\bigl(1 + \frac{\epsilon}{3}\bigr)}
=
\log\!\left(\frac{4(e+\pd)}{\eta}\right).
\]
Solving this quadratic equation and using
\(\sqrt{x^2+y}\leq x+\sqrt{y}\) for \(x,y\geq0\) give
\[
\begin{aligned}
\epsilon
={}&
\frac{\tau_n\kappa^2+1}{3n\lambda}
\log\!\left(\frac{4(e+\pd)}{\eta}\right)\\
&\quad+
\sqrt{
\left\{
\frac{\tau_n\kappa^2+1}{3n\lambda}
\log\!\left(\frac{4(e+\pd)}{\eta}\right)
\right\}^{2}
+
\frac{2(\tau_n\kappa^2+1)}{n\lambda}
\log\!\left(\frac{4(e+\pd)}{\eta}\right)
}\\
\leq{}&
\frac{2(\tau_n\kappa^2+1)}{3n\lambda}
\log\!\left(\frac{4(e+\pd)}{\eta}\right)
+
\sqrt{
\frac{2(\tau_n\kappa^2+1)}{n\lambda}
\log\!\left(\frac{4(e+\pd)}{\eta}\right)
}.
\end{aligned}
\]
Since \(\log\{4(e+\pd)/\eta\}>1\), this choice also satisfies the
applicability threshold of Lemma~\ref{lem:hs-inequality}; indeed,
\[
\epsilon
\geq
\sqrt{\frac{\tau_n\kappa^2+1}{n\lambda}}
+
\frac{\tau_n\kappa^2+1}{3n\lambda}.
\]
Consequently,
\[
\begin{aligned}
&\mathbb{P}\Bigl(
  \bigl\|
    (\oplpw + \lambda \opid)^{-\frac12}
    (\oplpnw - \oplpw)
    (\oplpw + \lambda \opid)^{-\frac12}
  \bigr\|
  \geq \epsilon
\Bigr)\\
&\qquad\leq
4(e+\pd)
\exp\!\left\{
-\log\!\left(\frac{4(e+\pd)}{\eta}\right)
\right\}
=\eta.
\end{aligned}
\]
Moreover, because \(\log(e+\pd)\geq1\) and
\(\log(2/\eta)\geq\log2\),
\[
\log\!\left(\frac{4(e+\pd)}{\eta}\right)
=
\log(e+\pd)+\log\!\left(\frac{4}{\eta}\right)
\lesssim
\log(e+\pd)\log\!\left(\frac{2}{\eta}\right).
\]
Together with
\(\{\log(2/\eta)\}^{1/2}\lesssim\log(2/\eta)\), substituting this comparison
into the preceding upper bound for \(\epsilon\)
proves the first assertion after increasing \(C_1\) by a universal constant.
Under the additional condition in the second assertion, the nonnegative ratio
satisfies
\[
\frac{\bigl(\tau_n \kappa^2 + 1\bigr)\log(e+\pd)}
     {n\lambda}
\leq
\sqrt{
  \frac{\bigl(\tau_n \kappa^2 + 1\bigr)\log(e+\pd)}
       {n\lambda}
}.
\]
Thus, the second assertion follows from the first after absorbing the factor
\(2C_1\) into the implicit constant.
\end{proof}

\subsection{Proof of Lemma~\ref{lem:bound-bahadur-2}}

\begin{proof}
With Lemma \ref{lem:bound-bahadur-1}, when
\begin{equation*}
  \frac{(\tau_n \kappa^2 + 1)\log(e+\pd)}{n \lambda}
  <
  \min \left\{
    1,\,
    \frac{1}{\bigl(4 C_1 \log \tfrac{2}{\eta}\bigr)^2}
  \right\},
\end{equation*}
\begin{equation*}
  \bigl\|
    (\oplpw + \lambda \opid)^{-\frac{1}{2}}
    (\oplpw - \oplpnw)
    (\oplpw + \lambda \opid)^{-\frac{1}{2}}
  \bigr\|
  \le \frac{1}{2},
\end{equation*}
then
\begin{equation*}
\begin{aligned}
  &(\oplpw + \lambda \opid)^{\frac{1}{2}}
    (\oplpnw + \lambda \opid)^{-1}
    (\oplpw + \lambda \opid)^{\frac{1}{2}} \\
  =&
    (\oplpw + \lambda \opid)^{\frac{1}{2}}
    \Big[
      (\oplpnw + \lambda \opid)^{-1}
      - (\oplpw + \lambda \opid)^{-1}
    \Big]
    (\oplpw + \lambda \opid)^{\frac{1}{2}}
    + \opid \\
  =&
    I
    +
    (\oplpw + \lambda \opid)^{-\frac{1}{2}}
    (\oplpw - \oplpnw)
    (\oplpw + \lambda \opid)^{-\frac{1}{2}}
    (\oplpw + \lambda \opid)^{\frac{1}{2}}
    (\oplpnw + \lambda \opid)^{-1}
    (\oplpw + \lambda \opid)^{\frac{1}{2}}.
\end{aligned}
\end{equation*}
Thus,
\begin{equation*}
\begin{aligned}
  \bigl\|
    &(\oplpw + \lambda \opid)^{\frac{1}{2}}
    (\oplpnw + \lambda \opid)^{-1}
    (\oplpw + \lambda \opid)^{\frac{1}{2}}
  \bigr\|\\
  \le&
  1
  +
  \frac{1}{2}
  \bigl\|
    (\oplpw + \lambda \opid)^{\frac{1}{2}}
    (\oplpnw + \lambda \opid)^{-1}
    (\oplpw + \lambda \opid)^{\frac{1}{2}}
  \bigr\|,
\end{aligned}
\end{equation*}
and hence
\begin{equation*}
  \bigl\|
    (\oplpw + \lambda \opid)^{\frac{1}{2}}
    (\oplpnw + \lambda \opid)^{-1}
    (\oplpw + \lambda \opid)^{\frac{1}{2}}
  \bigr\|
  \le 2.
\end{equation*}
\end{proof}


\subsection{Proof of Lemma~\ref{lem:bound-bahadur-3}}

\begin{proof}
\[ (\oplpw+\lambda \opid)^{-\frac{1}{2}}  \left(\frac{1}{n}\vkx^{\top}\bW\beps\right) = \frac{1}{n} \sum_{i=1}^n (\oplpw+\lambda \opid)^{-\frac{1}{2}} w(X_i) \eps_i k_{X_i} \]
\[ \mathbb E[ (\oplpw+\lambda \opid)^{-\frac{1}{2}} w(X_i) \eps_i k_{X_i} ] = 0. \]
\[ \| (\oplpw+\lambda \opid)^{-\frac{1}{2}} w(X_i) \eps_i k_{X_i} \| \le \tau_n \frac{\kappa \cdot \overline{\sigma}}{\sqrt{\lambda}} \]

\begin{equation*}
    \begin{aligned}
    & \mathbb{E} \| (\oplpw+\lambda \opid)^{-\frac{1}{2}} w(X_i) \eps_i k_{X_i} \|^2 \\
    & \le \overline{\sigma}^2 \cdot \mathbb{E} \| (\oplpw+\lambda \opid)^{-\frac{1}{2}} w(X_i) k_{X_i} \|^2 \\
    & = \overline{\sigma}^2 \cdot \mathbb{E} \tr ([ (\oplpw+\lambda \opid)^{-\frac{1}{2}} w(X_i) k_{X_i} ] \otimes [ (\oplpw+\lambda \opid)^{-\frac{1}{2}} w(X_i) k_{X_i} ]) \\
    & \le \tau_n \cdot \overline{\sigma}^2 \cdot \pd
    \end{aligned}
\end{equation*}

For
\[
Z_i
:=
(\oplpw+\lambda\opid)^{-\frac12}
w(X_i)\eps_i k_{X_i},
\]
the preceding bounds imply
\[
\sum_{i=1}^n\mathbb E\|Z_i\|^2
\leq
n\tau_n\overline{\sigma}^2\pd,
\qquad
\|Z_i\|
\leq
\frac{\tau_n\kappa\overline{\sigma}}{\sqrt{\lambda}}
\quad\text{a.s.}
\]
Therefore, the Hilbert-space Bernstein inequality
\cite[Corollary~1]{pinelis1986remarks} yields, with probability at least
\(1-\eta\),
\[
\left\|
\frac1n\sum_{i=1}^n Z_i
\right\|
\lesssim
\overline{\sigma}
\left[
\sqrt{
\frac{\tau_n\pd}{n}
\log\!\left(\frac{2}{\eta}\right)
}
+
\frac{\tau_n\kappa}{n\sqrt{\lambda}}
\log\!\left(\frac{2}{\eta}\right)
\right].
\]
Since
\(\sqrt{\log(2/\eta)}\lesssim\log(2/\eta)\) and
\(a+b\leq2\max\{a,b\}\), the asserted bound follows.
\end{proof}

\subsection{Proof of Lemma~\ref{lem:feasible-bootstrap-decomposition}}

\begin{proof}
\begin{equation*}
\begin{aligned}
B - B_{\oracle}
&= \frac{1}{n}\sum_{i=1}^n \sum_{j=1}^n h_{ij}\,\frac{\what V_i-\what V_j}{\sqrt{2}}
 - \frac{1}{n}\sum_{i=1}^n \sum_{j=1}^n h_{ij}\,\frac{V_i-V_j}{\sqrt{2}} \\
&= \frac{1}{n}\sum_{i=1}^n \sum_{j=1}^n h_{ij}\,\oplpnwl^{-1}\hat{\alpha}_{ij}
 - \frac{1}{n}\sum_{i=1}^n \sum_{j=1}^n h_{ij}\,\oplpwl^{-1}\alpha_{ij} \\
&= \frac{1}{n}\sum_{i=1}^n \sum_{j=1}^n h_{ij}\,
\bigl(\oplpnwl^{-1}-\oplpwl^{-1}\bigr)\hat{\alpha}_{ij}
+ \frac{1}{n}\sum_{i=1}^n \sum_{j=1}^n h_{ij}\,
\oplpwl^{-1}\bigl(\hat{\alpha}_{ij}-\alpha_{ij}\bigr) .
\end{aligned}
\end{equation*}
\end{proof}

\subsection{Proof of Lemma~\ref{lem:estimation-error-bound-in-h}}

\begin{lemma}
  \label{lem:pre-estimation-error-bound-in-h}
  Under Assumption \ref{assump:source-condition}, when
  \(\tau_n\geq\tau_0\), if
  \[
  \left\|(\oplpw + \lambda \opid)^{-\frac{1}{2}} (\oplpnw - \oplpw)(\oplpw + \lambda \opid)^{-\frac{1}{2}}\right\|
  < \delta \le \frac{1}{2}
  \]
and
  \[
  \left\|(\oplpw + \lambda \opid)^{-\frac{1}{2}} \left(\frac{1}{n}\vkx^{\top}\bW\beps\right)\right\|
  < \gamma,
  \]
then
\begin{equation*}
  \begin{aligned}
  \|\hatf - \fpw\|
  &\le \dfrac{2}{\sqrt{\lambda}}\left(\gamma + \delta \cdot \max\{\kappa^{2s-1}, 1\} \cdot \lambda^{\min\{s+\frac{1}{2}, 1\}}\cdot\|g_{\ast}\|\right), \\
  \|\hatf - f_*\|
  &\le \dfrac{2}{\sqrt{\lambda}}\left(\gamma +  \max\{\kappa^{2s-1}, 1\} \cdot \lambda^{\min\{s+\frac{1}{2}, 1\}}\cdot\|g_{\ast}\|\right).
  \end{aligned}
\end{equation*}
\end{lemma}

\begin{proof}[Proof of Lemma~\ref{lem:pre-estimation-error-bound-in-h}]
The Bahadur representation gives
\begin{equation*}
\begin{aligned}
\hatf - \fpw
&= (\oplpw + \lambda \opid)^{-1} \left[  \left(\frac{1}{n}\vkx^{\top}\bW\beps\right)
   +  (\oplpnw - \oplpw)(f_* - \fpw)\right] + u \\
&= \oplpwl^{-1} 
\left\{
  \mathbb{E}_n \bigl[w(X_i) k_{X_i} \eps_i\bigr]
   +  (\oplpnw - \oplpw)(f_* - \fpw)
\right\} + u .
\end{aligned}
\end{equation*}

\text{Here $u$ satisfies}
\begin{equation*}
\|(\oplpw+\lambda \opid)^{\frac{1}{2}} u\|
\le  
2\delta\left(\gamma + \delta \cdot \max\{\kappa^{2s-1}, 1\}
\cdot \lambda^{\min\{s+\frac{1}{2}, 1\}}\|g_{\ast}\|\right).
\end{equation*}

Then,
\begin{align*}
\left\|(\oplpw+\lambda \opid)^{\frac{1}{2}}(\hatf-\fpw)\right\|
&\le \gamma + \delta \cdot \left\|(\oplpw+\lambda \opid)^{\frac{1}{2}}(f_*-\fpw)\right\| \\
&\quad + 2\delta\left(\gamma + \delta \cdot \left\|(\oplpw+\lambda \opid)^{\frac{1}{2}}(f_*-\fpw)\right\|\right) \\
&\le 2\left(\gamma + \delta \cdot \max\{\kappa^{2s-1}, 1\} \cdot \lambda^{\min\{s+\frac{1}{2}, 1\}}\cdot\|g_{\ast}\|\right).
\end{align*}

\begin{align*}
\left\| (\oplpw+\lambda \opid)^{\frac{1}{2}}(\hatf-f_*)\right\|
&\le 2\left(\gamma + \delta \cdot \max\{\kappa^{2s-1}, 1\} \cdot \lambda^{\min\{s+\frac{1}{2}, 1\}}\cdot\|g_{\ast}\|\right) \\
&\quad + \max\{\kappa^{2s-1}, 1\} \cdot \lambda^{\min\{s+\frac{1}{2}, 1\}} \cdot \|g_{\ast}\| \\
&\le 2\left(\gamma + \max\{\kappa^{2s-1}, 1\} \cdot \lambda^{\min\{s+\frac{1}{2}, 1\}} \cdot \|g_{\ast}\|\right).
\end{align*}
\end{proof}

\begin{proof}[Proof of Lemma~\ref{lem:estimation-error-bound-in-h}]
Apply Lemmas~\ref{lem:bound-bahadur-1}, \ref{lem:bound-bahadur-2}, and
\ref{lem:bound-bahadur-3} with suitably split failure probabilities.  The
second sample-size condition implies
\[
\max\left\{
\sqrt{\frac{\tau_n\pd}{n}},
\frac{2\tau_n\kappa}{n\sqrt{\lambda}}
\right\}
\lesssim
\kappa\sqrt{\frac{\tau_n\pd}{n}}.
\]
The sample-size condition is meaningful only if \(\pd>0\), and hence
\(\calL_Q\neq0\).  Let \(\mu_1>0\) be a positive eigenvalue of \(\calL_Q\).
Since \(\lambda\to0\), for all sufficiently large \(n\),
\[
\pd
\geq
\frac{\mu_1}{\mu_1+\lambda}
\geq\frac12.
\]
It follows that \(\log(e+\pd)\lesssim\pd\).  Hence
Lemma~\ref{lem:bound-bahadur-1} and the first
sample-size condition give
\[
\left\|(\oplpw+\lambda\opid)^{-1/2}
(\oplpnw-\oplpw)
(\oplpw+\lambda\opid)^{-1/2}\right\|
\lesssim
\kappa\sqrt{\frac{\tau_n\pd}{n\lambda}}
\log\!\left(\frac{2}{\eta}\right),
\]
whereas Lemma~\ref{lem:bound-bahadur-3} gives
\[
\frac{1}{\sqrt{\lambda}}
\left\|(\oplpw+\lambda\opid)^{-1/2}
\left(\frac{1}{n}\vkx^{\top}\bW\beps\right)\right\|
\lesssim
\overline{\sigma}\kappa
\sqrt{\frac{\tau_n\pd}{n\lambda}}
\log\!\left(\frac{2}{\eta}\right).
\]
Substituting these two bounds into
Lemma~\ref{lem:pre-estimation-error-bound-in-h}, and using
\(\max\{\kappa^{2s-1},1\}\leq\kappa^{2s}\),
\(0<\lambda\leq1\), and
\(\|g_{\ast}\|_{\sfH}\leq1+\|g_{\ast}\|_{\sfH}\), proves both assertions.
\end{proof}

\subsection{Proof of Lemma~\ref{lem:feasible-bootstrap-inverse-bound}}

Denote \(\Delta_{\mathrm{inv}} = (\oplpnwl^{-1} - \oplpwl^{-1})u\),
where
\begin{equation*}
\begin{aligned}
u = \frac{1}{n}\sum_{i=1}^n \sum_{j=1}^n h_{ij}\,\hat{\alpha}_{ij}.
\end{aligned}
\end{equation*}

\begin{lemma}
  \label{lem:feasible-bootstrap-inverse-pre-bound}
Under Assumption~\ref{assump:covariate-shift-estimation}, writing
\(\lambda=\lambda_n\), we have
\begin{equation*}
\begin{aligned}
\|\Delta_{\mathrm{inv}}\|_{\splq}
&\le 
\mathfrak{C}_1 \|\oplpwl^{1/2} \oplpnwl^{-1} \oplpwl^{1/2}\| \|\oplpwl^{-1/2} (\oplpw - \oplpnw) \oplpwl^{-1/2}\| \|\oplpwl^{-1/2} u\|.
\end{aligned}
\end{equation*}
\end{lemma}

\begin{proof}[Proof of Lemma~\ref{lem:feasible-bootstrap-inverse-pre-bound}]

\begin{equation*}
\begin{aligned}
\|\Delta_{\mathrm{inv}}\|_{\splq}
&= \bigl\|\oplq^{1/2}\Delta_{\mathrm{inv}}\bigr\| \\
&= \bigl\|\oplq^{1/2}(\oplpnwl^{-1}-\oplpwl^{-1})u\bigr\| \\
&= \bigl\|\oplq^{1/2}\oplpnwl^{-1}(\oplpw-\oplpnw)\oplpwl^{-1}u\bigr\| \\
&= \Bigl\|\Bigl(\oplq^{1/2}\oplpwl^{-1/2}\Bigr)
\Bigl(\oplpwl^{1/2}\oplpnwl^{-1}\oplpwl^{1/2}\Bigr)
\Bigl[\oplpwl^{-1/2}(\oplpw-\oplpnw)\oplpwl^{-1/2}\Bigr]
\oplpwl^{-1/2}u\Bigr\| \\
&\le \mathfrak{C}_1 \|\oplpwl^{1/2} \oplpnwl^{-1} \oplpwl^{1/2}\| \|\oplpwl^{-1/2} (\oplpw - \oplpnw) \oplpwl^{-1/2}\| \|\oplpwl^{-1/2} u\|.
\end{aligned}
\end{equation*}
\end{proof}

\begin{lemma}
  \label{lem:bound-of-oplpwlu-conditional-on-data}
  For \(0<\eta<1\) and \(\lambda>0\), conditionally on \(\mathcal D\),
  with probability at least \(1-\eta\),
\begin{equation*}
\begin{aligned}
\bigl\|\oplpwl^{-1/2}u\bigr\|
&\le 
\Bigl\{1+\sqrt{2\log\!\left(\frac{1}{\eta}\right)}\Bigr\}
\sqrt{\frac{4}{n}\sum_{i=1}^n
\bigl\|w(X_i)\,\eps_i\,\oplpwl^{-1/2}k_{X_i}\bigr\|^2
+ \frac{4}{n}\sum_{i=1}^n
\bigl\|\oplpwl^{-1/2}\oplpiw(f_*-\hatf)\bigr\|^2
 } .
\end{aligned}
\end{equation*}
\end{lemma}

\begin{proof}[Proof of Lemma~\ref{lem:bound-of-oplpwlu-conditional-on-data}]
\(\oplpwl^{-1/2}u\mid\mathcal{D}\) is Gaussian.
By Borell's inequality under \(\mathbb{P}(\cdot\mid\mathcal{D})\), with probability \(1-\eta\),
\begin{equation*}
\begin{aligned}
\bigl\|\oplpwl^{-1/2}u\bigr\|
&\le \Bigl\{1+\sqrt{2\log\!\left(\frac{1}{\eta}\right)}\Bigr\}
\sqrt{\mathbb{E}_h\bigl\|\oplpwl^{-1/2}u\bigr\|^2}.
\end{aligned}
\end{equation*}

\begin{equation*}
\begin{aligned}
\mathbb{E}_h\bigl\|\oplpwl^{-1/2}u\bigr\|^2
&= \frac{1}{n^2}\sum_{i=1}^n\sum_{j=1}^n
\bigl\|\oplpwl^{-1/2}\hat{\alpha}_{ij}\bigr\|^2 .
\end{aligned}
\end{equation*}

\begin{equation*}
\begin{aligned}
\bigl\|\oplpwl^{-1/2}\hat{\alpha}_{ij}\bigr\|^2
&= \frac{1}{2}
\bigl\|w(X_i)\what{\eps}_i\,\oplpwl^{-1/2}k_{X_i}
      - w(X_j)\what{\eps}_j\,\oplpwl^{-1/2}k_{X_j}\bigr\|^2 \\
&\le \bigl\|w(X_i)\what{\eps}_i\,\oplpwl^{-1/2}k_{X_i}\bigr\|^2
   + \bigl\|w(X_j)\what{\eps}_j\,\oplpwl^{-1/2}k_{X_j}\bigr\|^2 .
\end{aligned}
\end{equation*}

\begin{equation*}
\begin{aligned}
w(X_i)\what{\eps}_i\,k_{X_i}
&= w(X_i)\bigl[Y_i-\hatf(X_i)\bigr]\,k_{X_i} \\
&= w(X_i)\Bigl[(Y_i-f_*(X_i)) + (f_*(X_i)-\hatf(X_i))\Bigr]\,k_{X_i} \\
&= w(X_i)\eps_i\,k_{X_i} + w(X_i)\bigl(f_*(X_i)-\hatf(X_i)\bigr)\,k_{X_i} \\
&= w(X_i)\eps_i\,k_{X_i} + \oplpiw(f_*-\hatf) .
\end{aligned}
\end{equation*}
Then,
\begin{equation*}
\begin{aligned}
\mathbb{E}_h \bigl\|\oplpwl^{-1/2}u\bigr\|^2
&\le \frac{2}{n}\sum_{i=1}^n
\bigl\|w(X_i)\what{\eps}_i\,\oplpwl^{-1/2}k_{X_i}\bigr\|^2 \\
&\le \frac{4}{n}\sum_{i=1}^n
\bigl\|w(X_i)\eps_i\,\oplpwl^{-1/2}k_{X_i}\bigr\|^2
+ \frac{4}{n}\sum_{i=1}^n
\bigl\|\oplpwl^{-1/2}\oplpiw(f_*-\hatf)\bigr\|^2 .
\end{aligned}
\end{equation*}
\end{proof}

\begin{lemma}
  \label{lem:supporting-lemma-1-for-feasible-bootstrap}
  Under Assumption~\ref{assump:error}, for \(0<\eta<1\) and
  \(\lambda>0\), with probability at least \(1-\eta\),
\begin{equation*}
\begin{aligned}
\frac{1}{n}\sum_{i=1}^n
\bigl\|w(X_i)\eps_i\,\oplpwl^{-1/2}k_{X_i}\bigr\|^2
&\le \tau_n \overline{\sigma}^2 \pd
+ 2\log\!\left(\frac{2}{\eta}\right)\left\{
 \frac{\overline{\sigma}^2 \kappa^2\tau_n^2}{n\lambda}
+ \sqrt{\frac{ \overline{\sigma}^4 \kappa^2\tau_n^3 \pd}{n\lambda}}
\right\}.
\end{aligned}
\end{equation*}
\end{lemma}

\begin{proof}[Proof of Lemma~\ref{lem:supporting-lemma-1-for-feasible-bootstrap}]
Let
\begin{equation*}
\begin{aligned}
\xi_i := \bigl\| w(X_i)\eps_i\,\oplpwl^{-1/2} k_{X_i} \bigr\|^2 .
\end{aligned}
\end{equation*}

Then
\begin{equation*}
\begin{aligned}
\xi_i
&\le \frac{\tau_n^2 \,\overline{\sigma}^2\, \kappa^2}{\lambda}.
\end{aligned}
\end{equation*}

Moreover,
\begin{equation*}
\begin{aligned}
\mathbb{E} \xi_i
&= \mathbb{E}\bigl\| w(X_i)\eps_i\,\oplpwl^{-1/2} k_{X_i} \bigr\|^2 \\
&\le \tau_n \overline{\sigma}^2\,\mathbb{E}\!\left[w(X_i)\,\bigl\|\oplpwl^{-1/2}k_{X_i}\bigr\|^2\right] \\
&= \tau_n \overline{\sigma}^2\,\mathbb{E}\!\left[w(X_i)\,
\bigl\langle \oplpwl^{-1/2}k_{X_i},\,\oplpwl^{-1/2}k_{X_i}\bigr\rangle\right] \\
&= \tau_n \overline{\sigma}^2\,\mathrm{tr}\!\left(\oplpwl^{-1}\oplpw\right) \\
&\le \tau_n \overline{\sigma}^2\,\pd.
\end{aligned}
\end{equation*}
The last inequality follows from \(w\leq\rho\), which implies
\(\oplpw\preceq_{\sfH}\oplq\), and the operator monotonicity of
\(x\mapsto x/(x+\lambda)\).

Then,
\begin{equation*}
\begin{aligned}
\mathbb{E} \xi_i^2
&\le \frac{\tau_n^3 \overline{\sigma}^4 \kappa^2}{\lambda}\,\pd.
\end{aligned}
\end{equation*}
Therefore, by Bernstein's inequality, with probability \(1-\eta\),
\begin{equation*}
\begin{aligned}
\mathbb{E}_n \xi_i
&\le \mathbb{E} \xi_i
+ 2\log\!\left(\frac{2}{\eta}\right)\left\{
 \frac{\tau_n^2\overline{\sigma}^2 \kappa^2}{n\lambda}
+ \sqrt{\frac{\tau_n^3 \overline{\sigma}^4 \kappa^2 \pd}{n\lambda}}
\right\}.
\end{aligned}
\end{equation*}
\end{proof}

\begin{lemma}
  \label{lem:supporting-lemma-2-for-feasible-bootstrap}
For \(0<\eta<1\) and \(\lambda>0\), with probability at least
\(1-\eta\),
\begin{equation*}
\begin{aligned}
\frac{1}{n}\sum_{i=1}^n \bigl\|\oplpwl^{-1/2}\oplpiw\bigr\|_{HS}^2
\le\;& \tau_n \kappa^2 \pd + 2\log\!\left(\frac{2}{\eta}\right)\left\{
 \frac{2\kappa^4\tau_n^2}{n\lambda}
+ \sqrt{ \frac{\kappa^6 \tau_n^3\pd}{n\lambda}}
\right\}.
\end{aligned}
\end{equation*}
\end{lemma}

\begin{proof}[Proof of Lemma~\ref{lem:supporting-lemma-2-for-feasible-bootstrap}]
Let
\begin{equation*}
\begin{aligned}
\xi_i
&:= \bigl\|\oplpwl^{-1/2} \oplpiw\bigr\|_{HS}^2 \\
&= \bigl\|(\oplpw+\lambda \opid)^{-1/2}\, w(X_i)\, k_{X_i}\otimes k_{X_i}\bigr\|_{HS}^2 .
\end{aligned}
\end{equation*}
Then
\begin{equation*}
\begin{aligned}
0 \le \xi_i \le \frac{1}{\lambda}\,\tau_n^2 \kappa^4 .
\end{aligned}
\end{equation*}

Moreover,
\begin{equation*}
\begin{aligned}
\mathbb{E}\xi_i
&\le \tau_n\,\mathbb{E}\!\left[
w(X_i)\,\bigl\|(\oplpw+\lambda \opid)^{-1/2}(k_{X_i}\otimes k_{X_i})\bigr\|_{HS}^2
\right] \\
&= \tau_n\,\mathbb{E}\!\left[
w(X_i)\,\sum_{j=1}^{\infty}
\bigl\|(\oplpw+\lambda \opid)^{-1/2}(k_{X_i}\otimes k_{X_i}) e_j\bigr\|^2
\right] \\
&= \tau_n\,\mathbb{E}\!\left[
w(X_i)\,\sum_{j=1}^{\infty}
\bigl\|\langle k_{X_i},e_j\rangle\,(\oplpw+\lambda \opid)^{-1/2}k_{X_i}\bigr\|^2
\right] \\
&= \tau_n\,\mathbb{E}\!\left[
w(X_i)\,\bigl\|(\oplpw+\lambda \opid)^{-1/2}k_{X_i}\bigr\|^2
\sum_{j=1}^{\infty}\langle k_{X_i},e_j\rangle^2
\right] \\
&= \tau_n\,\mathbb{E}\!\left[
w(X_i)\,\bigl\|(\oplpw+\lambda \opid)^{-1/2}k_{X_i}\bigr\|^2\,\|k_{X_i}\|^2
\right] \\
&\le \tau_n \kappa^2\,\pd.
\end{aligned}
\end{equation*}
Thus,
\begin{equation*}
\begin{aligned}
\mathbb{E}\,\xi_i^2
\le \frac{1}{\lambda}\,\tau_n^3\,\kappa^6\,\pd.
\end{aligned}
\end{equation*}
By Bernstein's inequality, with probability \(1-\eta\),
\begin{equation*}
\begin{aligned}
\mathbb{E}_n \xi_i
\le \mathbb{E}\xi_i
+ 2\log\!\left(\frac{2}{\eta}\right)\left\{
\tau_n^2 \frac{2\kappa^4}{n\lambda}
+ \sqrt{\tau_n^3 \frac{\kappa^6 \pd}{n\lambda}}
\right\}.
\end{aligned}
\end{equation*}
\end{proof}

\begin{lemma}
    \label{lem:bound-of-oplpwlu}
    For \(0<\eta<1\), under Assumptions~\ref{assump:error}
    and~\ref{assump:source-condition}, suppose
    \(\tau_n\geq\tau_0\) for all sufficiently large \(n\),
    \(\tau_n\to\infty\), and \(\lambda=\lambda_n\to0\). If, for all
    sufficiently large \(n\),
\begin{equation*}
  n
  >
  \max
  \left\{
  \frac{(\tau_n \kappa^2 + 1)\log(e+\pd)}{ C_2\lambda},
  \frac{\tau_n}{\lambda \pd}
  \right\}
\end{equation*}
where $C_2$ is the constant in Lemma \ref{lem:bound-bahadur-2} after replacing
its failure probability by the fixed fraction of \(\eta\) used in the union
bound, then, for all sufficiently large \(n\), with probability at least
\(1-\eta\) over the data, conditionally on the data with
probability at least \(1-\eta\),
\[
\norm{\oplpwl^{-1/2}u} \lesssim
\overline{\sigma}\,\kappa^{2s+3}
\bigl(1+\norm{g_{\ast}}_{\sfH}\bigr)
\sqrt{\tau_n \pd}
\left(1+\sqrt{\frac{\tau_n\pd}{n\lambda}}\right)
\left(\log\!\left(\frac{2}{\eta}\right)\right)^2.
\]
\end{lemma}

\begin{proof}[Proof of Lemma~\ref{lem:bound-of-oplpwlu}]
Apply Lemmas~\ref{lem:supporting-lemma-1-for-feasible-bootstrap},
\ref{lem:supporting-lemma-2-for-feasible-bootstrap}, and
\ref{lem:estimation-error-bound-in-h} with suitably split failure
probabilities.  The second sample-size condition implies
\[
\begin{aligned}
\frac{1}{n}\sum_{i=1}^n
\norm{w(X_i)\eps_i\oplpwl^{-1/2}k_{X_i}}^2
&\lesssim
\overline\sigma^2\kappa^2\tau_n\pd
\log\!\left(\frac{2}{\eta}\right),\\
\frac{1}{n}\sum_{i=1}^n
\norm{\oplpwl^{-1/2}\oplpiw}_{HS}^2
&\lesssim
\kappa^4\tau_n\pd
\log\!\left(\frac{2}{\eta}\right).
\end{aligned}
\]
Furthermore,
\[
\frac{1}{n}\sum_{i=1}^n
\norm{\oplpwl^{-1/2}\oplpiw(f_*-\hatf)}^2
\leq
\norm{f_*-\hatf}^2\,
\frac{1}{n}\sum_{i=1}^n
\norm{\oplpwl^{-1/2}\oplpiw}_{HS}^2.
\]
Substituting these bounds into
Lemma~\ref{lem:bound-of-oplpwlu-conditional-on-data}, followed by the second
bound in Lemma~\ref{lem:estimation-error-bound-in-h}, yields
\[
\begin{aligned}
\norm{\oplpwl^{-1/2}u}
\lesssim{}&
\overline\sigma\,\kappa^{2s+3}
\bigl(1+\norm{g_{\ast}}_{\sfH}\bigr)\sqrt{\tau_n\pd}\\
&\quad{}\times
\left[
1+\sqrt{\frac{\tau_n\pd}{n\lambda}}
+\lambda^{\min\{s,1/2\}}
\right]
\left(\log\!\left(\frac{2}{\eta}\right)\right)^2.
\end{aligned}
\]
Since \(0<\lambda\leq1\) for all sufficiently large \(n\), the last
\(\lambda\)-term is absorbed by \(1\), proving the stated bound.  Replacing
\(\eta\) by fixed fractions of \(\eta\) in the data and conditional
multiplier bounds changes only universal constants inside the logarithm.
\end{proof}

\begin{proof}[Proof of Lemma~\ref{lem:feasible-bootstrap-inverse-bound}]
Apply Lemmas~\ref{lem:bound-bahadur-1},
\ref{lem:bound-bahadur-2},
\ref{lem:feasible-bootstrap-inverse-pre-bound}, and
\ref{lem:bound-of-oplpwlu}, with suitably split failure probabilities.  This
gives
\[
\begin{aligned}
\norm{\Delta_{\mathrm{inv}}}_{\splq}
\lesssim{}&
\mathfrak C_1\overline\sigma\,\kappa^{2s+4}
\bigl(1+\norm{g_{\ast}}_{\sfH}\bigr)
\tau_n\sqrt{\frac{\pd\log(e+\pd)}{n\lambda}}\\
&\quad{}\times
\left(1+\sqrt{\frac{\tau_n\pd}{n\lambda}}\right)
\left(\log\!\left(\frac{2}{\eta}\right)\right)^3.
\end{aligned}
\]
By the argument in the proof of
Lemma~\ref{lem:estimation-error-bound-in-h},
\(\log(e+\pd)\lesssim\pd\) for all sufficiently large \(n\).  Hence
\[
\sqrt{\frac{\log(e+\pd)}{\pd}}\lesssim1.
\]
The first sample-size condition also gives
\[
\sqrt{\frac{\log(e+\pd)}{\pd}}\,
\sqrt{\frac{\tau_n\pd}{n\lambda}}
=
\sqrt{\frac{\tau_n\log(e+\pd)}{n\lambda}}
\leq
\frac{\sqrt{C_2}}{\kappa}
\leq1.
\]
Both terms are therefore bounded by
\(\tau_n\pd/\sqrt{n\lambda}\), which proves the claim.
\end{proof}

\subsection{Proof of Lemma~\ref{lem:feasible-bootstrap-residual-bound}}

\begin{lemma}
  \label{lem:feasible-bootstrap-residual-conditional-bound}
  Under Assumption~\ref{assump:covariate-shift-estimation}, writing
  \(\lambda=\lambda_n\), for \(0<\eta<1\), conditionally on \(\mathcal D\),
  with probability at least \(1-\eta\),
\begin{equation*}
\|\Delta_{\mathrm{res}}\|_{\splq}
\le 
\mathfrak{C}_1\,\|\fpw - \hatf\|\,
\bigl(1 + \sqrt{2 \log\!\left(\frac{1}{\eta}\right)}\bigr)
\sqrt{\frac{2}{n} \sum_{i=1}^n \bigl\|\oplpwl^{-\frac{1}{2}} \oplpiw\bigr\|_{HS}^{2} } .
\end{equation*}
\end{lemma}

\begin{proof}[Proof of Lemma~\ref{lem:feasible-bootstrap-residual-conditional-bound}]
\begin{equation*}
\begin{aligned}
\hat{\alpha}_{ij} - \alpha_{ij}
&= \frac{1}{\sqrt{2}}
\Bigl(w(X_i)\what{\eps}_i k_{X_i} - w(X_j)\what{\eps}_j k_{X_j}\Bigr) - \frac{1}{\sqrt{2}}
\Bigl(w(X_i)\epsilon_i^{P^w} k_{X_i} - w(X_j)\epsilon_j^{P^w} k_{X_j}\Bigr) \\
&= \frac{1}{\sqrt{2}}
\Bigl[
w(X_i)\bigl(\fpw(X_i) - \hatf(X_i)\bigr) k_{X_i}
- w(X_j)\bigl(\fpw(X_j) - \hatf(X_j)\bigr) k_{X_j}
\Bigr] \\
&= \frac{1}{\sqrt{2}}
\Bigl[\oplpiw(\fpw - \hatf) - \oplpjw(\fpw - \hatf)\Bigr] \\
&= \frac{1}{\sqrt{2}} \bigl(\oplpiw - \oplpjw\bigr)(\fpw - \hatf).
\end{aligned}
\end{equation*}
\text{Hence,}
\begin{equation*}
\Delta_{\mathrm{res}}
= \oplpwl^{-1}
\left[
 \frac{1}{n} \sum_{i=1}^n \sum_{j=1}^n
h_{ij} \frac{1}{\sqrt{2}} \bigl(\oplpiw - \oplpjw\bigr)(\fpw - \hatf)
\right].
\end{equation*}
\begin{equation*}
\|\Delta_{\mathrm{res}}\|_{\splq}
\le \mathfrak{C}_1\,\|\fpw - \hatf\|\,\bigl\|\oplpwl^{-\frac{1}{2}} u\bigr\|_{HS},
\end{equation*}
\text{where}
\begin{equation*}
u
= \frac{1}{n} \sum_{i=1}^n \sum_{j=1}^n
h_{ij} \frac{1}{\sqrt{2}} \bigl(\oplpiw - \oplpjw\bigr).
\end{equation*}
\text{By Borell's inequality under \(\mathbb{P}(\cdot\mid\mathcal{D})\), with probability at least $1-\eta$,}
\begin{equation*}
\bigl\|\oplpwl^{-\frac{1}{2}} u\bigr\|_{HS}
\le \bigl(1 + \sqrt{2 \log\!\left(\frac{1}{\eta}\right)}\bigr)
\sqrt{\mathbb{E}_h \bigl\|\oplpwl^{-\frac{1}{2}} u\bigr\|_{HS}^2 } .
\end{equation*}

\text{And}
\begin{equation*}
\begin{aligned}
\mathbb{E}_h \bigl\|\oplpwl^{-\frac{1}{2}} u\bigr\|_{HS}^2
&= \frac{1}{n^2} \sum_{i=1}^n \sum_{j=1}^n \frac{1}{2}
\bigl\| \oplpwl^{-\frac{1}{2}} \oplpiw - \oplpwl^{-\frac{1}{2}} \oplpjw \bigr\|_{HS}^2 \\
&\le \frac{1}{n^2} \sum_{i=1}^n \sum_{j=1}^n
\left[
\bigl\|\oplpwl^{-\frac{1}{2}} \oplpiw\bigr\|_{HS}^2
+ \bigl\|\oplpwl^{-\frac{1}{2}} \oplpjw\bigr\|_{HS}^2
\right] \\
&= \frac{2}{n} \sum_{i=1}^n \bigl\|\oplpwl^{-\frac{1}{2}} \oplpiw\bigr\|_{HS}^2 .
\end{aligned}
\end{equation*}
\end{proof}

\begin{proof}[Proof of Lemma~\ref{lem:feasible-bootstrap-residual-bound}]
On the intersection of the data events in
Lemmas~\ref{lem:estimation-error-bound-in-h} and
\ref{lem:supporting-lemma-2-for-feasible-bootstrap}, the second sample-size
condition gives
\[
\sqrt{\frac{2}{n}\sum_{i=1}^n
\norm{\oplpwl^{-1/2}\oplpiw}_{HS}^2}
\lesssim
\kappa^2\sqrt{\tau_n\pd}
\sqrt{\log\!\left(\frac{2}{\eta}\right)}.
\]
Conditionally on the data,
Lemma~\ref{lem:feasible-bootstrap-residual-conditional-bound} and the first
bound in Lemma~\ref{lem:estimation-error-bound-in-h} therefore imply
\[
\norm{\Delta_{\mathrm{res}}}_{\splq}
\lesssim
\mathfrak C_1\overline\sigma\,\kappa^{2s+3}
\bigl(1+\norm{g_{\ast}}_{\sfH}\bigr)
\frac{\tau_n\pd}{\sqrt{n\lambda}}
\left(\log\!\left(\frac{2}{\eta}\right)\right)^2.
\]
Using fixed fractions of \(\eta\) for the two data events and the conditional
multiplier event gives the asserted two-layer probability statement without
changing the displayed logarithmic order.
\end{proof}

\subsection{Proof of Lemma~\ref{lem:strassen-coupling}}

\begin{proof}
Let \(H(\cdot\mid\mathcal A)\) be a regular conditional distribution of the
pair \((X',Y')\) given \(\mathcal A\). Define its two marginal kernels by
\[
  F(C\mid\mathcal A)=H(C\times\mathsf S\mid\mathcal A),
  \qquad
  G(C\mid\mathcal A)=H(\mathsf S\times C\mid\mathcal A),
\]
for Borel \(C\subset\mathsf S\). Then \(G(\cdot\mid\mathcal A)\) is a regular
conditional distribution of \(Y'\) given \(\mathcal A\), and, because
\(X\mid\mathcal A\stackrel d=X'\mid\mathcal A\),
\(F(\cdot\mid\mathcal A)\) is also a version of the regular conditional
distribution of \(X\) given \(\mathcal A\).

The conditional Strassen--Dudley theorem of
\citet[Theorem~4]{monrad1991nearby} gives a random element \(Y\), with
conditional law \(G(\cdot\mid\mathcal A)\) given \(\mathcal A\), provided that
\[
  \mathbb E
  \sup_{C}
  \left[
    F(C\mid\mathcal A)
    -
    G(C^r\mid\mathcal A)
  \right]
  \leq \varepsilon ,
\]
where the supremum is over Borel sets \(C\subset\mathsf S\) and
\[
  C^r=\{s\in\mathsf S:\inf_{t\in C}d(s,t)\le r\}.
\]
The extra uniform random variable in the assumptions is the randomization
needed in that theorem to realize the resulting conditional transition kernel
on the original probability space.

It remains to check this displayed condition. Let
\[
  D_r=\{(s,t)\in\mathsf S\times\mathsf S:d(s,t)\le r\}.
\]
For any Borel set \(C\subset\mathsf S\), the set inclusion
\[
  (C\times\mathsf S)\cap D_r
  \subset
  \mathsf S\times C^r
\]
holds by the definition of \(C^r\). Hence
\[
  C\times\mathsf S
  \subset
  (\mathsf S\times C^r)\cup D_r^c .
\]
Applying the conditional probability kernel \(H(\cdot\mid\mathcal A)\) to this
inclusion gives, outside a fixed null set in the underlying probability space
and simultaneously for all Borel \(C\),
\[
\begin{aligned}
  F(C\mid\mathcal A)
  -
  G(C^r\mid\mathcal A)
  &\le
  H(D_r^c\mid\mathcal A) \\
  &=
  P\{d(X',Y')>r\mid\mathcal A\}.
\end{aligned}
\]
The right-hand side does not depend on \(C\), and therefore the same bound holds
after taking the supremum over Borel \(C\). Hence
\[
\begin{aligned}
  \mathbb E
  \sup_C
  \left[
    F(C\mid\mathcal A)-G(C^r\mid\mathcal A)
  \right]
  &\le
  \mathbb E P\{d(X',Y')>r\mid\mathcal A\} \\
  &=
  P\{d(X',Y')>r\}
  \le \varepsilon .
\end{aligned}
\]
The conditional Strassen--Dudley theorem applies and yields an
\(\mathsf S\)-valued random element \(Y\) satisfying
\(Y\mid\mathcal A\stackrel d=Y'\mid\mathcal A\) and
\(P\{d(X,Y)>r\}\le\varepsilon\).
\end{proof}

\subsection{Proof of Lemma~\ref{lem:weight-mse-bound-under-infinite-moment}}

\begin{proof}
Because \(w(X)=\min\{\rho(X),\tau_n\}\),
\[
1-\frac{w(X)}{\rho(X)}
=
\left(1-\frac{\tau_n}{\rho(X)}\right)\mathbf{1}\{\rho(X)\ge \tau_n\}.
\]
In particular,
\[
\left(1-\frac{w(X)}{\rho(X)}\right)^2
\le \mathbf{1}\{\rho(X)\ge \tau_n\}.
\]
Let \(p_m=(m-1)/\gamma\). Since \(\gamma\geq1\), Lyapunov's inequality
and Definition~\ref{def:bernstein-moment-condition-for-weight} give
\[
\mathbb E_{\mathbb Q}\rho(X)^{p_m}
\leq
\left\{
\mathbb E_{\mathbb Q}\rho(X)^{\gamma(m-1)}
\right\}^{1/\gamma^2}
\leq
\left(\frac12m!W^{m-2}\sigma^2\right)^{1/\gamma}.
\]
Consequently, Markov's inequality yields
\begin{equation*}
\begin{aligned}
\mathbb{E}_{\mathbb Q}\left(1-\frac{w(X)}{\rho(X)}\right)^2
&\le
\mathbb{Q}\{\rho(X)\ge \tau_n\}\\
&\le
\tau_n^{-p_m}\mathbb{E}_{\mathbb Q}\rho(X)^{p_m}\\
&\le
\tau_n^{-\frac{m-1}{\gamma}}
\left(\frac{1}{2}m!W^{m-2}\sigma^2\right)^{\frac{1}{\gamma}}\\
&=
\left(
\frac{1}{2}\tau_n^{-(m-1)}m!W^{m-2}\sigma^2
\right)^{\frac{1}{\gamma}}.
\end{aligned}
\end{equation*}
\end{proof}

\subsection{Proof of Lemma~\ref{lem:lq-lw-bound-under-infinite-moment}}

\begin{proof}
Since \(q=\rho p\) and \(w(x)=\min\{\rho(x),\tau_n\}\),
\[
\oplpw f
=
\int w(x)k_x f(x)\,p(x)\,dx
=
\int \frac{w(x)}{\rho(x)}k_x f(x)\,q(x)\,dx.
\]
Hence, with \(a(x):=1-w(x)/\rho(x)\),
\[
\oplq-\oplpw
=
\int a(x)\,k_x\otimes_{\h}k_x\,\mathbb{Q}(dx).
\]
For \(A_\lambda=(\oplq+\lambda\opid)^{-1}\), self-adjointness gives
\[
(k_x\otimes_{\h}k_x)A_\lambda
=
k_x\otimes_{\h} A_\lambda k_x.
\]
Thus
\[
(\oplq-\oplpw)(\oplq+\lambda\opid)^{-1}
=
\int a(x)\,k_x\otimes_{\h} A_\lambda k_x\,\mathbb{Q}(dx).
\]
Using \(\|T\|\le \|T\|_{HS}\), Jensen's inequality in Hilbert--Schmidt
norm, and Cauchy--Schwarz,
\begin{equation*}
\begin{aligned}
\|(\oplq-\oplpw)(\oplq+\lambda\opid)^{-1}\|
&\le
\int |a(x)|\,\|k_x\otimes_{\h}A_\lambda k_x\|_{HS}\,\mathbb{Q}(dx)\\
&\le
\left(\mathbb{E}_{\mathbb Q}a(X)^2\right)^{1/2}
\left(\mathbb{E}_{\mathbb Q}
\|k_X\otimes_{\h}A_\lambda k_X\|_{HS}^{2}
\right)^{1/2}.
\end{aligned}
\end{equation*}
For a rank-one operator,
\(\|u\otimes v\|_{HS}=\|u\|\,\|v\|\), and
\(\|k_X\|\le \kappa\). Therefore,
\begin{equation*}
\begin{aligned}
\mathbb{E}_{\mathbb Q}
\|k_X\otimes_{\h}A_\lambda k_X\|_{HS}^{2}
&\le
\kappa^2\,
\mathbb{E}_{\mathbb Q}\|A_\lambda k_X\|^2\\
&=
\kappa^2\,
\tr_{\h}\!\left(A_\lambda^2\oplq\right)\\
&\le
\kappa^2\lambda^{-1}
\tr_{\h}\!\left(A_\lambda\oplq\right)
=
\kappa^2\lambda^{-1}\pd,
\end{aligned}
\end{equation*}
because \(A_\lambda\preceq \lambda^{-1}\opid\). Combining this with
Lemma~\ref{lem:weight-mse-bound-under-infinite-moment} yields
\[
\bigl\|(\oplq-\oplpw)(\oplq+\lambda\opid)^{-1}\bigr\|
\le
\kappa\lambda^{-1/2}\pd^{1/2}
\left(
\frac{1}{2}\tau_n^{-(m-1)}m!W^{m-2}\sigma^2
\right)^{\frac{1}{2\gamma}} .
\]
\end{proof}

\subsection{Proof of Lemma~\ref{lem:bound-a-under-infinite-moment}}

\begin{proof}
Let
\[
  \Delta_\lambda
  :=(\oplq-\oplpw)(\oplq+\lambda \opid)^{-1}.
\]
By Lemma~\ref{lem:lq-lw-bound-under-infinite-moment} with \(m=m_0\),
\begin{equation*}
\begin{aligned}
\|\Delta_\lambda\|
&\le
\kappa \lambda^{-1/2} \pd^{1/2}
\left(
\frac{1}{2}\tau_n^{-(m_0-1)}m_0!W^{m_0-2}\sigma^2
\right)^{\frac{1}{2\gamma}}  \\
&\le
\kappa E_d \lambda^{-\frac{1+d}{2}}
\left(
\frac{1}{2}\tau_n^{-(m_0-1)}m_0!W^{m_0-2}\sigma^2
\right)^{\frac{1}{2\gamma}},
\end{aligned}
\end{equation*}
where the second inequality follows from
\(\pd\le E_d^2\lambda^{-d}\). The lower bound imposed on \(\tau_n\)
is exactly the condition that the last display is at most \(1/2\).
This proves the first claim.

For the second claim, use the factorization
\begin{equation*}
\begin{aligned}
\oplpw+\lambda \opid
&= \oplq+\lambda \opid-(\oplq-\oplpw) \\
&= \bigl(\opid-\Delta_\lambda\bigr)(\oplq+\lambda \opid).
\end{aligned}
\end{equation*}
Since \(\|\Delta_\lambda\|\le 1/2\), the operator
\(\opid-\Delta_\lambda\) is invertible and
\[
  \|(\opid-\Delta_\lambda)^{-1}\|
  \le \frac{1}{1-\|\Delta_\lambda\|}
  \le 2 .
\]
Consequently,
\begin{equation*}
\begin{aligned}
\|\oplq(\oplpw+\lambda \opid)^{-1}\|
&=
\bigl\|\oplq(\oplq+\lambda \opid)^{-1}
(\opid-\Delta_\lambda)^{-1}\bigr\| \\
&\le
\|\oplq(\oplq+\lambda \opid)^{-1}\|
\|(\opid-\Delta_\lambda)^{-1}\|
\le 2,
\end{aligned}
\end{equation*}
because \(0\preceq \oplq(\oplq+\lambda \opid)^{-1}\preceq \opid\).

It remains to prove the third claim. By the Cordes inequality
in Lemma~\ref{lem:cordes-inequality}, applied with
\(A=\oplq\), \(B=(\oplpw+\lambda \opid)^{-1}\), and \(s=1/2\),
\[
\|\oplq^{1/2}(\oplpw+\lambda \opid)^{-1/2}\|
\le
\|\oplq(\oplpw+\lambda \opid)^{-1}\|^{1/2}
\le \sqrt{2}.
\]
\end{proof}

\subsection{Proof of Lemma~\ref{lem:weight-mse-bound-under-second-moment}}

\begin{lemma}
  \label{lem:weight-mse-bound-under-second-moment}
Under Definition~\ref{def:second-moment-condition-for-weight} with
\(\theta=2\), for \(\tau_n>0\),
\[
\mathbb E_{\mathbb Q}
\left(1-\frac{w(X)}{\rho(X)}\right)^2
\leq\tau_n^{-1}\Omega.
\]
\end{lemma}

\begin{proof}
  \begin{align*}
\mathbb{E}_{\mathbb{Q}}\!\left(1-\frac{w(X)}{\rho(X)}\right)^{2}
&= \int \left(1-\frac{w(x)}{\rho(x)}\right)^{2}\,\mathbb{Q}(dx) \\
&= \int_{\rho \ge \tau_{n}} \left(1-\frac{\tau_{n}}{\rho(x)}\right)^{2}\,\mathbb{Q}(dx) \\
&\le \int_{\rho \ge \tau_{n}} 1\,\mathbb{Q}(dx) \\
&= \mathbb{Q}\!\left(\rho(X) \ge \tau_{n}\right) \\
&\le \tau_{n}^{-1}\mathbb{E}_{\mathbb Q}\rho(X)
\le \tau_{n}^{-1}\Omega .
\end{align*}
\end{proof}

\subsection{Proof of Lemma~\ref{lem:lq-lw-bound-under-second-moment}}

\begin{lemma}
  \label{lem:lq-lw-bound-under-second-moment}
Under Definition~\ref{def:second-moment-condition-for-weight} with
\(\theta=2\), for \(\tau_n>0\) and \(\lambda>0\),
\[
\left\|
(\oplq-\oplpw)(\oplq+\lambda\opid)^{-1}
\right\|
\leq
\kappa\lambda^{-1/2}\pd^{1/2}\tau_n^{-1/2}\Omega^{1/2}.
\]
\end{lemma}

\begin{proof}
The proof is identical to the Hilbert--Schmidt argument in
Lemma~\ref{lem:lq-lw-bound-under-infinite-moment}; the only input that
changes is the moment bound for the truncated density ratio. Indeed, the
same calculation gives
\[
\left\|
\left(\oplq-\oplpw\right)\left(\oplq+\lambda \opid\right)^{-1}
\right\|
\le
\kappa\lambda^{-1/2}\pd^{1/2}
\left[
\mathbb{E}_{\mathbb Q}
\left(1-\frac{w(X)}{\rho(X)}\right)^2
\right]^{1/2}.
\]
Applying Lemma~\ref{lem:weight-mse-bound-under-second-moment} yields
\[
\left\|
\left(\oplq-\oplpw\right)\left(\oplq+\lambda \opid\right)^{-1}
\right\|
\le
\kappa\lambda^{-1/2}\pd^{1/2}\tau_n^{-1/2}\Omega^{1/2}.
\]
\end{proof}

\subsection{Proof of Lemma~\ref{lem:bound-a-under-second-moment}}

\begin{lemma}
  \label{lem:bound-a-under-second-moment}
Under Assumption~\ref{assump:pseudo-dimn} and
Definition~\ref{def:second-moment-condition-for-weight} with \(\theta=2\),
for \(0<\lambda\leq1\), if
\begin{equation}
\label{eq:tn-condition}
\tau_n\geq4\kappa^2\Omega E_d^2\lambda^{-(1+d)},
\end{equation}
then, as operators on \(\mathsf H\),
\begin{enumerate}
\item
\(\|(\oplq-\oplpw)(\oplq+\lambda\opid)^{-1}\|\leq1/2\);
\item
\(\|\oplq(\oplpw+\lambda\opid)^{-1}\|\leq2\);
\item
\(\|\oplq^{1/2}(\oplpw+\lambda\opid)^{-1/2}\|\leq\sqrt2\).
\end{enumerate}
\end{lemma}

\begin{proof}
Let again
\[
  \Delta_\lambda
  :=(\oplq-\oplpw)(\oplq+\lambda \opid)^{-1}.
\]
By Lemma~\ref{lem:lq-lw-bound-under-second-moment} and
\(\pd\le E_d^2\lambda^{-d}\),
\begin{equation*}
\begin{aligned}
\|\Delta_\lambda\|
&\le
\kappa \lambda^{-1/2}\pd^{1/2}\tau_n^{-1/2}\Omega^{1/2} \\
&\le
\kappa E_d\Omega^{1/2}\lambda^{-\frac{1+d}{2}}\tau_n^{-1/2}.
\end{aligned}
\end{equation*}
Condition~\eqref{eq:tn-condition} implies that the right-hand side is
at most \(1/2\), proving the first claim.

The remaining two claims follow from the same deterministic perturbation
argument as in Lemma~\ref{lem:bound-a-under-infinite-moment}. Namely,
\[
\oplpw+\lambda \opid
=
(\opid-\Delta_\lambda)(\oplq+\lambda \opid),
\qquad
\|(\opid-\Delta_\lambda)^{-1}\|\le 2.
\]
Therefore,
\[
\|\oplq(\oplpw+\lambda \opid)^{-1}\|
\le
\|\oplq(\oplq+\lambda \opid)^{-1}\|
\|(\opid-\Delta_\lambda)^{-1}\|
\le 2,
\]
and Lemma~\ref{lem:cordes-inequality} gives
\[
\|\oplq^{1/2}(\oplpw+\lambda \opid)^{-1/2}\|
\le
\|\oplq(\oplpw+\lambda \opid)^{-1}\|^{1/2}
\le \sqrt{2}.
\]
\end{proof}

\subsection{Proof of Lemma~\ref{lem:weight-mse-bound-under-vartheta-moment}}

\begin{lemma}
  \label{lem:weight-mse-bound-under-vartheta-moment}
Under Definition~\ref{def:vartheta-moment-condition-for-weight}, for
\(\tau_n>0\),
\[
\mathbb E_{\mathbb Q}
\left(1-\frac{w(X)}{\rho(X)}\right)^2
\leq
\tau_n^{-(\vartheta-1)}\Omega_\vartheta.
\]
\end{lemma}

\begin{proof}
Since \(w(X)=\min\{\rho(X),\tau_n\}\),
\begin{align*}
\mathbb{E}_{\mathbb{Q}}\!\left(1-\frac{w(X)}{\rho(X)}\right)^2
&\leq
\mathbb{Q}\!\left\{\rho(X)\geq\tau_n\right\} \\
&\leq
\tau_n^{-(\vartheta-1)}
\mathbb{E}_{\mathbb{Q}}\!\left[\rho(X)^{\vartheta-1}\right] \\
&\leq
\tau_n^{-(\vartheta-1)}\Omega_{\vartheta},
\end{align*}
where the second inequality follows from Markov's inequality.
\end{proof}

\subsection{Proof of Lemma~\ref{lem:lq-lw-bound-under-vartheta-moment}}

\begin{proof}
The Hilbert--Schmidt argument used in
Lemma~\ref{lem:lq-lw-bound-under-infinite-moment} gives
\begin{equation*}
\left\|
  (\oplq-\oplpw)(\oplq+\lambda\opid)^{-1}
\right\|
\leq
\kappa\lambda^{-1/2}\pd^{1/2}
\left[
  \mathbb{E}_{\mathbb{Q}}\!\left(1-\frac{w(X)}{\rho(X)}\right)^2
\right]^{1/2}.
\end{equation*}
The claim follows by applying
Lemma~\ref{lem:weight-mse-bound-under-vartheta-moment}.
\end{proof}

\subsection{Proof of Lemma~\ref{lem:bound-a-under-vartheta-moment}}

\begin{proof}
Let
\[
  \Delta_\lambda
  :=(\oplq-\oplpw)(\oplq+\lambda\opid)^{-1}.
\]
By Lemma~\ref{lem:lq-lw-bound-under-vartheta-moment} and
\(\pd\leq E_d^2\lambda^{-d}\),
\begin{align*}
\|\Delta_\lambda\|
&\leq
\kappa\lambda^{-1/2}\pd^{1/2}
\tau_n^{-\frac{\vartheta-1}{2}}\Omega_{\vartheta}^{1/2} \\
&\leq
\kappa E_d\Omega_{\vartheta}^{1/2}
\lambda^{-\frac{1+d}{2}}
\tau_n^{-\frac{\vartheta-1}{2}}.
\end{align*}
Condition~\eqref{eq:tn-condition-vartheta-moment} makes the last display at
most \(1/2\), proving the first claim.

The remaining claims follow from the deterministic perturbation argument in
Lemma~\ref{lem:bound-a-under-infinite-moment}. In particular,
\[
  \oplpw+\lambda\opid
  =
  (\opid-\Delta_\lambda)(\oplq+\lambda\opid),
  \qquad
  \|(\opid-\Delta_\lambda)^{-1}\|\leq2.
\]
Consequently,
\[
  \|\oplq(\oplpw+\lambda\opid)^{-1}\|
  \leq2,
\]
and Lemma~\ref{lem:cordes-inequality} yields
\[
  \|\oplq^{1/2}(\oplpw+\lambda\opid)^{-1/2}\|
  \leq
  \|\oplq(\oplpw+\lambda\opid)^{-1}\|^{1/2}
  \leq\sqrt2.
\]
\end{proof}

\subsection{Proof of Lemma~\ref{lem:effective-dimension-bound-polynomial-decay}}

\begin{proof}
\begin{equation*}
\begin{aligned}
\pd
&= \sum_{i=1}^{\infty} \frac{\mu_i}{\mu_i+\lambda} \\
&\le \sum_{i=1}^{\infty} \frac{\alpha i^{-\beta}}{\alpha i^{-\beta}+\lambda} \\
&= \sum_{i=1}^{\infty} \frac{\alpha}{\alpha + \lambda i^{\beta}} \\
&\le \int_0^\infty \frac{\alpha}{\alpha + \lambda x^\beta}\, dx \\
\end{aligned}
\end{equation*}

Let \(t=(\lambda/\alpha)^{1/\beta}x\). Then
\begin{equation*}
  \begin{aligned}
\pd
&\le
\left(\frac{\alpha}{\lambda}\right)^{1/\beta}
\left[
\int_0^1 \frac{1}{1+t^\beta}\, dt
+ \int_1^\infty \frac{1}{1+t^\beta}\, dt
\right] \\
&\le
\left(\frac{\alpha}{\lambda}\right)^{1/\beta}
\left[ 1 + \frac{1}{\beta-1} \right] \\
&=
\alpha^{1/\beta}\frac{\beta}{\beta-1}\lambda^{-1/\beta}.
\end{aligned}
\end{equation*}
\end{proof}

\subsection{Proof of Lemma~\ref{lem:sigma-bound-polynomial-decay}}

\begin{proof}
  \begin{equation*}
\begin{aligned}
\sigma^2(m, \oplq)
&= \sum_{i=m+1}^{\infty} \mu_i\\
&\le \int_m^{\infty} \alpha \cdot t^{-\beta}\, dt \\
&= \alpha \left[ \frac{1}{1-\beta} t^{1-\beta} \right]_m^{\infty} \\
&= \frac{\alpha}{\beta-1} m^{1-\beta}.
\end{aligned}
\end{equation*}
Taking square roots proves the claim.
\end{proof}

\subsection{Proof of Lemma~\ref{lem:tail-bound-exponential-decay-integral}}

\begin{proof}
Let \(u=\beta x^\gamma\).  Then
\[
\int_a^\infty e^{-\beta x^\gamma}\,dx
=
\frac{1}{\gamma\beta^{1/\gamma}}
\int_{\beta a^\gamma}^{\infty} u^{1/\gamma-1}e^{-u}\,du .
\]
The standard incomplete-gamma tail estimate
\[
\int_t^\infty u^{q-1}e^{-u}\,du
\sim
t^{q-1}e^{-t},
\qquad t\to\infty,
\]
with \(q=1/\gamma\), yields
\[
\int_a^\infty e^{-\beta x^\gamma}\,dx
\sim
\frac{1}{\gamma\beta^{1/\gamma}}
(\beta a^\gamma)^{1/\gamma-1}e^{-\beta a^\gamma}
=
\frac{e^{-\beta a^\gamma}}{\beta\gamma a^{\gamma-1}}.
\]
\end{proof}

\subsection{Proof of Lemma~\ref{lem:effective-dimension-bound-exponential-decay}}

\begin{proof}
\begin{equation*}
\begin{aligned}
\pd
&= \sum_{i=1}^{\infty} \frac{\mu_i}{\mu_i+\lambda} \\
&\le \sum_{i=1}^{\infty} \frac{\alpha \exp(-\beta i^{\gamma})}{\alpha \exp(-\beta i^{\gamma})+\lambda} \\
&= \sum_{i=1}^{\infty} \frac{\alpha}{\alpha + \lambda \exp(\beta i^{\gamma})} \\
&\le \int_0^\infty \frac{\alpha}{\alpha + \lambda \exp(\beta x^{\gamma})}\, dx \\
&= \left[\int_0^{\left[\frac{1}{\beta}\ln\left(\frac{\alpha}{\lambda}\right)\right]^{\frac{1}{\gamma}}}+\int_{\left[\frac{1}{\beta}\ln\left(\frac{\alpha}{\lambda}\right)\right]^{\frac{1}{\gamma}}}^\infty\right] \frac{\alpha}{\alpha + \lambda \exp(\beta x^{\gamma})}\, dx\\
&\le \left[\frac{1}{\beta}\ln\left(\frac{\alpha}{\lambda}\right)\right]^{\frac{1}{\gamma}} + 
\frac{\alpha}{\lambda}\int_{\left[\frac{1}{\beta}\ln\left(\frac{\alpha}{\lambda}\right)\right]^{\frac{1}{\gamma}}}^\infty  \exp(-\beta x^{\gamma})\, dx\\
&\lesssim \left[\frac{1}{\beta}\ln\left(\frac{\alpha}{\lambda}\right)\right]^{\frac{1}{\gamma}} 
+ \frac{\alpha}{\lambda} \cdot \frac{\exp\left(-\beta \left[\frac{1}{\beta}\ln\left(\frac{\alpha}{\lambda}\right)\right]\right)}{\beta \gamma \left[\frac{1}{\beta}\ln\left(\frac{\alpha}{\lambda}\right)\right]^{\frac{\gamma-1}{\gamma}}} \\
&= \left[\frac{1}{\beta}\ln\left(\frac{\alpha}{\lambda}\right)\right]^{\frac{1}{\gamma}} 
+ \frac{1}{\beta \gamma} \cdot \left[\frac{1}{\beta}\ln\left(\frac{\alpha}{\lambda}\right)\right]^{\frac{1}{\gamma}-1} \\
&\lesssim \left(\log \frac{1}{\lambda}\right)^{\frac{1}{\gamma}}.
\end{aligned}
\end{equation*}
The last comparison holds for all sufficiently small \(\lambda>0\), which
gives constants \(\lambda_0\in(0,1)\) and
\(C_{\alpha,\beta,\gamma}>0\) as stated.
\end{proof}

\subsection{Proof of Lemma~\ref{lem:sigma-bound-exponential-decay}}

\begin{proof}
  \begin{equation*}
\begin{aligned}
\sigma^2(m, \oplq)
&= \sum_{i=m+1}^{\infty} \mu_i\\
&\le \int_m^{\infty} \alpha \cdot \exp(-\beta t^{\gamma})\, dt \\
&\leq C_{\alpha,\beta,\gamma}^2
m^{1-\gamma}\exp(-\beta m^\gamma).
\end{aligned}
\end{equation*}
Here Lemma~\ref{lem:tail-bound-exponential-decay-integral} gives the last
inequality for all sufficiently large \(m\); increasing
\(C_{\alpha,\beta,\gamma}\) covers the finitely many remaining positive
integers. Taking square roots proves the claim.
\end{proof}

\subsection{Proof of Lemma~\ref{lem:tail-sum-sigma-bound}}

\begin{proof}
  As in the proof of Lemma \ref{lem:sigma-l2-trace-operator-norm-bound},
  \begin{equation*}
\begin{aligned}
\opv &= \opv_1 + \opv_2, \\
\opv_1 &\preceq_{\splq} \tau_n \overline{\sigma}^2 \cdot \opt_1, \\
\opv_2
&\preceq_{\splq}
\tau_n \kappa^2
\max\{\kappa^{2s-1},1\}^{2}
\lambda^{2\min\{s,1/2\}}
\|g_{\ast}\|^2\, \opt_1, \\
\opv_1 + \opv_2
&\preceq_{\splq}
\tau_n\!\left(
  \overline{\sigma}^2
  + \kappa^2 \max\{\kappa^{2s-1},1\}^{2}
  \lambda^{2\min\{s,1/2\}}\|g_{\ast}\|^2
\right) \opt_1.
\end{aligned}
\end{equation*}
where
\begin{equation*}
\opt_1 = (\oplpw+\lambda \opid)^{-1} \oplpw (\oplpw+\lambda \opid)^{-1} \oplq.
\end{equation*}
$\forall\, f \in \splq$,
\begin{equation*}
\begin{aligned}
\langle f, \opt_1 f \rangle_{\splq}
&= \left\langle f,\,
(\oplpw+\lambda \opid)^{-1} \oplpw (\oplpw+\lambda \opid)^{-1} \oplq f
\right\rangle_{\splq} \\
&= \left\langle \oplq f,\,
(\oplpw+\lambda \opid)^{-1} \oplpw (\oplpw+\lambda \opid)^{-1} \oplq f
\right\rangle \\
&\le \left\langle \oplq f, \frac{1}{\lambda} \oplq f \right\rangle \\
&= \frac{1}{\lambda} \langle f, \oplq f \rangle_{\splq}.
\end{aligned}
\end{equation*}

Thus,
\begin{equation*}
\opt_1 \preceq_{\splq} \frac{1}{\lambda} \oplq.
\end{equation*}

By the min--max eigenvalue monotonicity for compact positive operators and the
identity
\begin{equation*}
\sigma^2_{\splq}(m,\oplq) = \sigma^2(m,\oplq),
\end{equation*}
we have
\begin{equation*}
\sigma^2_{\splq}(m,\opv)
\le
\frac{\tau_n}{\lambda}
\left(
  \overline{\sigma}^2
  + \kappa^2 \max\{\kappa^{2s-1},1\}^{2}
  \lambda^{2\min\{s,1/2\}}\|g_{\ast}\|^2
\right)
\sigma^2(m,\oplq).
\end{equation*}
\end{proof}

\subsection{Proof of Lemma~\ref{lem:sub-gaussianity-of-ui}}

\begin{proof}
\begin{equation*}
\begin{aligned}
\widetilde{U}_i
&= (\oplpw+\lambda \opid)^{-1}
\Bigl[
w(X_i)\,\eps_i\,k_{X_i}
+ w(X_i)\bigl(k_{X_i}\otimes k_{X_i}\bigr)\,(f_*-\fpw)
\Bigr] \\
&= w(X_i)\bigl(\eps_i+f_*(X_i)-\fpw(X_i)\bigr)\,
(\oplpw+\lambda \opid)^{-1} k_{X_i} .
\end{aligned}
\end{equation*}

Fix \(t\in\splq\). Since
\((\oplpw+\lambda\opid)^{-1}\oplq t\in\mathsf H\) is deterministic, the
weighted-feature part of Assumption~\ref{assump:sub-gaussian-feature} and
{\protect\cite[Proposition~2.5.2]{vershynin2018high}} imply
\begin{equation*}
\begin{aligned}
&\left\|
\left\langle w(X_i)k_{X_i},
(\oplpw+\lambda\opid)^{-1}\oplq t
\right\rangle
-
\mathbb E\!\left[
\left\langle w(X_i)k_{X_i},
(\oplpw+\lambda\opid)^{-1}\oplq t
\right\rangle
\right]
\right\|_{\psi_2}\\
&\qquad\leq C_{sg}b_0
\left(
\mathbb E\left[
\left\langle w(X_i)k_{X_i},
(\oplpw+\lambda\opid)^{-1}\oplq t
\right\rangle^2
\right]
\right)^{1/2}.
\end{aligned}
\end{equation*}
Here the condition applies with \(w=w_n=w_{\tau_n}\), since
\(\tau_n\geq\tau_0\). The same bound holds for the uncentered scalar
projection after changing the universal constant, because \(b_0\geq1\) and
the absolute mean is bounded by the square root of the second moment.

\begin{equation*}
\begin{aligned}
\bigl|\langle \widetilde{U}_i, t \rangle_{\splq}\bigr|
&= \bigl|\langle \widetilde{U}_i, \oplq t \rangle\bigr| \\
&\le (\overline{\sigma}+\kappa\|f_*\|)\,
\bigl|\langle w(X_i)\,(\oplpw+\lambda \opid)^{-1} k_{X_i},\, \oplq t \rangle\bigr| \\
&= (\overline{\sigma}+\kappa\|f_*\|)\,
\bigl|\langle w(X_i)k_{X_i},\, (\oplpw+\lambda \opid)^{-1} \oplq t \rangle\bigr| .
\end{aligned}
\end{equation*}

\begin{equation*}
\begin{aligned}
\bigl\|\langle \widetilde{U}_i, t \rangle_{\splq}\bigr\|_{\psi_2}
&\le (\overline{\sigma}+\kappa\|f_*\|)\,
\bigl\|\langle w(X_i)k_{X_i},\, (\oplpw+\lambda \opid)^{-1} \oplq t \rangle\bigr\|_{\psi_2} \\
&\le C_{sg} \cdot b_0\,(\overline{\sigma}+\kappa\|f_*\|)\,
\left(
\mathbb{E}\!\left[
\left\langle w(X_i)k_{X_i},\, (\oplpw+\lambda \opid)^{-1} \oplq t \right\rangle^{2}
\right]
\right)^{1/2} \\
&\le C_{sg} \cdot b_0\,\frac{\overline{\sigma}+\kappa\|f_*\|}{\underline{\sigma}}\,
\left\{
\mathbb{E}\!\left[\langle U_i, t \rangle_{\splq}^{2}\right]
\right\}^{1/2} .
\end{aligned}
\end{equation*}
where we used Lemma~\ref{lem:sigma-lower-operator-bound}.

Since \(U_i=\widetilde U_i-\mathbb E\widetilde U_i\), centering changes only
the universal constant. Hence
\begin{equation*}
\begin{aligned}
\bigl\|\langle U_i, t \rangle_{\splq}\bigr\|_{\psi_2}
&\le C_{sg}^2 \cdot b_0\,\frac{\overline{\sigma}+\kappa\|f_*\|}{\underline{\sigma}}\,
\left(
\mathbb{E}\!\left[\langle U_i, t \rangle_{\splq}^{2}\right]
\right)^{1/2} .
\end{aligned}
\end{equation*}

Again by {\protect\cite[Proposition~2.5.2]{vershynin2018high}}, \(U_i\) is $\splq$-\(C_{sg}^3  b_0 \dfrac{\overline{\sigma}+\kappa\|f_*\|}{\underline{\sigma}}\)-sub-Gaussian.
\end{proof}

\subsection{Proof of Lemma~\ref{lem:tq-self-adjoint-in-l2q}}

\begin{proof}
  By the spectral decomposition of \(\oplq\),
  \[
    \splq
    =
    \overline{\ran(\oplq)}
    \oplus_{\splq}^{\perp}
    \ker(\oplq),
    \qquad
    \overline{\ran(\oplq)}
    =
    \overline{\mathrm{span}\{\phi_j:\mu_j>0\}} .
  \]
  Moreover, when \(\mathsf{H}\) is viewed as a subspace of \(\splq\), Lemma~\ref{lem:eigendecomposition_LQ} gives
  \(\mathsf{H}\subset\overline{\ran(\oplq)}\). We first record the identity
  \begin{equation}\label{eq:lq-l2-h-duality}
    \langle u,h\rangle_{\splq}
    =
    \langle \oplq u,h\rangle_{\mathsf{H}},
    \qquad u\in\splq,\ h\in\mathsf{H}.
  \end{equation}
  Indeed, writing \(u=\sum_j u_j\phi_j\) in \(\splq\) and
  \(h=\sum_{\mu_j>0}h_j\phi_j\in\mathsf{H}\), the kernel component of \(u\) is
  orthogonal to \(h\), and
  \[
    \langle \oplq u,h\rangle_{\mathsf{H}}
    =
    \sum_{\mu_j>0}\frac{\mu_j u_jh_j}{\mu_j}
    =
    \sum_{\mu_j>0}u_jh_j
    =
    \langle u,h\rangle_{\splq}.
  \]
  Now take \(f,g\in\splq\) and decompose
  \(f=f_r+f_0\), \(g=g_r+g_0\), with
  \(f_r,g_r\in\overline{\ran(\oplq)}\) and
  \(f_0,g_0\in\ker(\oplq)\). Since \(\oplq f_0=\oplq g_0=0\) and
  \(\opt\oplq f_r,\opt\oplq g_r\in\mathsf{H}\subset\overline{\ran(\oplq)}\),
  the kernel components do not contribute to the \(\splq\) inner products. Hence,
  by \eqref{eq:lq-l2-h-duality} and the self-adjointness of \(\opt\) on
  \(\mathsf{H}\),
  \begin{align*}
    \langle f,\opt\oplq g\rangle_{\splq}
    &=
    \langle f_r,\opt\oplq g_r\rangle_{\splq} \\
    &=
    \langle \oplq f_r,\opt\oplq g_r\rangle_{\mathsf{H}} \\
    &=
    \langle \opt\oplq f_r,\oplq g_r\rangle_{\mathsf{H}} \\
    &=
    \langle \opt\oplq f_r,g_r\rangle_{\splq} \\
    &=
    \langle \opt\oplq f,g\rangle_{\splq}.
  \end{align*}
  Thus \(\opt\oplq\) is self-adjoint in \(\splq\).
\end{proof}

\subsection{Proof of Lemma~\ref{lem:bound-bahadur-2-2}}

\begin{proof}
As $\oplpw$ and $\oplpnw$ are self-adjoint, we have
\begin{equation*}
\begin{aligned}
  &\bigl\|
    (\oplpnw + \lambda \opid)^{-\frac{1}{2}}
    (\oplpw + \lambda \opid)^{\frac{1}{2}}
  \bigr\|^2\\
  =&
  \sup_{\|f\| = 1}
  \left\langle
    (\oplpnw + \lambda \opid)^{-\frac{1}{2}}
    (\oplpw + \lambda \opid)^{\frac{1}{2}} f,
    (\oplpnw + \lambda \opid)^{-\frac{1}{2}}
    (\oplpw + \lambda \opid)^{\frac{1}{2}} f
  \right\rangle \\
  =&
  \sup_{\|f\| = 1}
  \left\langle
    f,\,
    (\oplpw + \lambda \opid)^{\frac{1}{2}}
    (\oplpnw + \lambda \opid)^{-1}
    (\oplpw + \lambda \opid)^{\frac{1}{2}} f
  \right\rangle \\
  =&
  \bigl\|
    (\oplpw + \lambda \opid)^{\frac{1}{2}}
    (\oplpnw + \lambda \opid)^{-1}
    (\oplpw + \lambda \opid)^{\frac{1}{2}}
  \bigr\|.
\end{aligned}
\end{equation*}
Combining this with Lemma \ref{lem:bound-bahadur-2}, we have the desired result.
\end{proof}

\subsection[Derivation of the bootstrap matrix representation]
{Derivation of Equation~\eqref{eq:B-def}}

Starting from the first representation of \(B\) in
\eqref{eq:B-def}, we have
    \begin{equation*}
    \begin{aligned}
    B &= \frac{1}{n} \sum_{i=1}^n \sum_{j=1}^n h_{ij} \frac{\what V_i - \what V_j}{\sqrt{2}} \\
    &= \frac{1}{\sqrt{2}} \cdot \frac{1}{n} \left[ \sum_{i=1}^n \what V_i \sum_{j=1}^n h_{ij} - \sum_{j=1}^n \what V_j \sum_{i=1}^n h_{ij} \right] \\
    &= \frac{1}{\sqrt{2}} \cdot \frac{1}{n} \left[ \sum_{i=1}^n \what V_i \sum_{j=1}^n h_{ij} - \sum_{i=1}^n \what V_i \sum_{j=1}^n h_{ji} \right]\\
    &= \frac{1}{\sqrt{2}} \cdot \frac{1}{n} \sum_{i=1}^n \what V_i \sum_{j=1}^n (h_{ij} - h_{ji})\\
    &= \frac{1}{\sqrt{2}} \vvx^{\top}
    \diag(\what\beps) (\mah - \mah^{\top}) \vone .
    \end{aligned}
    \end{equation*}

where
\begin{equation*}
    \begin{aligned}
    \vvx^{\top} &= \left[ \frac{1}{n} w(X_1) (\oplpnw + \lambda \opid)^{-1} k_{X_1}, \dots, \frac{1}{n} w(X_n) (\oplpnw + \lambda \opid)^{-1} k_{X_n} \right] \\
&= \frac{1}{n} (\oplpnw + \lambda \opid)^{-1} \vkx^{\top} \maw.
    \end{aligned}
\end{equation*}
As
\begin{equation*}
    \begin{aligned}
    (\oplpnw + \lambda \opid) \vkx^{\top} &= [(\oplpnw + \lambda \opid) k_{X_1}, \dots, (\oplpnw + \lambda \opid) k_{X_n}]\\
    \oplpnw k_{X_j} &= \frac{1}{n} \sum_{i=1}^n w(X_i) k(X_i, X_j) k_{X_i} \\
    &= \frac{1}{n} (\vkx^{\top} \maw \mak)_j,
\end{aligned}
\end{equation*}
we have
\begin{equation*}
    \begin{aligned}
    (\oplpnw + \lambda \opid) \vkx^{\top} &= \vkx^{\top} \left( \frac{1}{n} \maw \mak + \lambda \maid \right) \\
    (\oplpnw + \lambda \opid)^{-1} \vkx^{\top} &= \vkx^{\top} \left( \frac{1}{n} \maw \mak + \lambda \maid \right)^{-1}.
    \end{aligned}
\end{equation*}
Then,
\begin{equation*}
    \begin{aligned}
    \vvx^{\top} = \vkx^{\top} (\maw \mak + n \lambda \maid)^{-1} \maw &= \vkx^{\top} (\mak + n \lambda \maw^{-1})^{-1}.
    \end{aligned}
\end{equation*}

\end{document}